\documentclass[11pt]{article}      
\pdfoutput=1
\usepackage[margin=1in,bottom=1in]{geometry}  

\usepackage[utf8]{inputenc}

\usepackage{libertine}
\usepackage{amsthm}
\usepackage{amsmath}

\usepackage{amssymb}
\usepackage{stmaryrd}
\usepackage{mathrsfs} 
\usepackage{mathtools}
\usepackage{thm-restate}
\usepackage[colorlinks = true,linkcolor = blue,urlcolor = blue, citecolor = blue]{hyperref}
\usepackage[capitalize, nameinlink]{cleveref}
\usepackage{url}
\usepackage{todonotes}
\usepackage{mathrsfs} 
\usepackage{soul}
\usepackage{relsize}
\usepackage{enumerate}
\usepackage[all,defaultlines=3]{nowidow}
\usepackage[mathlines]{lineno}
\usepackage{multicol}
\usepackage{enumitem}
\usepackage{xspace}
\usepackage{lineno}

\usepackage{textpos}
\setlist[itemize]{topsep=4pt,itemsep=3pt,parsep=0pt} 
\setlist[enumerate]{topsep=4pt,itemsep=3pt,parsep=0pt} 

\crefname{claim}{Claim}{Claims}
\crefname{figure}{Figure}{Figures}

\renewcommand{\subset}{\subseteq}
\newtheorem{theorem}{Theorem}[section]

\newtheorem{corollary}[theorem]{Corollary}

\newtheorem{observation}[theorem]{Observation}
\newtheorem{lemma}[theorem]{Lemma}

\newtheorem{claim}[theorem]{Claim}

\theoremstyle{definition}
\newtheorem{definition}[theorem]{Definition}
\theoremstyle{plain}

\theoremstyle{definition}

\newtheorem*{example*}{Example}

\numberwithin{equation}{section}

\newenvironment{claimproof}[1][Proof of the claim.]{%
  \begin{proof}[#1]%
}{%
  \end{proof}%
}

\newcommand{\cc}{\mathsf{cc}}
\newcommand{\bd}{\mathsf{bd}}
\newcommand{\chipfamily}{\mathcal{C}}
\newcommand{\ins}{\mathcal{I}}
\newcommand{\recurse}{\triangleleft}
\newcommand{\xprof}{\ensuremath{\mathcal{X}}}
\newcommand{\reachx}{\ensuremath{\mathcal{R}}}
\newcommand{\rpmap}{\phi}
\newcommand{\comps}{\mathcal{C}}
\newcommand{\mcomps}{\mathcal{C}}
\newcommand{\fcomps}{\mathcal{C}}
\newcommand{\rcomps}{\mathcal{C}}
\newcommand{\algcasecc}{Case~1\xspace}
\newcommand{\algcasenwl}{Case~2\xspace}
\newcommand{\algcasewlsmall}{Case~3\xspace}
\newcommand{\algcasebalsep}{Case~3.1\xspace}
\newcommand{\algcaserwl}{Case~3.2\xspace}
\newcommand{\algcasecarve}{Case~4\xspace}
\newcommand{\conna}{\mu}
\newcommand{\connb}{\bar{\mu}}
\newcommand{\numcarv}{\gamma}
\newcommand{\combfolio}{\xi}

\newcommand{\tw}{\mathsf{tw}}

\newcommand{\Pp}{\mathcal{P}}
\newcommand{\Ww}{\mathcal{W}}
\newcommand{\bag}{\mathsf{bag}}
\newcommand{\intr}{\mathsf{int}}

\newcommand{\fM}{\mathsf{folio}}
\newcommand{\cM}{\mathsf{clique}}

\newcommand{\CMSO}{\mathsf{CMSO}}
\newcommand{\MSO}{\mathsf{MSO}}
\newcommand{\tup}[1]{{\bar{#1}}}

\def\phi{\varphi}
\newcommand{\N}{\mathbb{N}}

\newcommand{\Z}{\mathbb{Z}}

\def\tp{\textnormal{tp}}

\newcommand{\Cc}{\mathscr{C}}
\newcommand{\Ee}{\mathcal{E}}

\newcommand{\Oof}{\mathcal{O}}

\newcommand{\wh}{\widehat}

\newcommand{\Ff}{\mathcal{F}}

\newcommand{\Aa}{\mathcal{A}}

\newcommand{\Tt}{\mathcal{T}}
\newcommand{\Ll}{\mathcal{L}}

\newcommand{\LL}{\mathcal L}

\def\N{\mathbb N} 
\def\epsilon{\varepsilon}

\newcommand{\Oh}{\Oof}

\renewcommand{\emptyset}{\varnothing}

\renewcommand{\leq}{\leqslant}
\renewcommand{\geq}{\geqslant}
\renewcommand{\le}{\leq}
\renewcommand{\ge}{\geq}

\renewcommand{\setminus}{-}

\newcommand{\roots}{\pi}
\newcommand{\model}{\eta}

\newcommand{\TermCarv}{{\sc{Terminal Carving}}\xspace}
\newcommand{\Folio}{{\sc{Folio}}\xspace}
\newcommand{\FolioApex}{{\sc{Folio or Apex-Grid Minor}}\xspace}
\newcommand{\FolioApexG}{{\sc{Folio or Apex Minor}}\xspace}
\newcommand{\FolioClique}{{\sc{Folio or Clique Minor}}\xspace}
\newcommand{\FolioCliqueEx}{{\sc{Folio or Clique Minor Existence}}\xspace}

\usepackage{tikz}

\usetikzlibrary{backgrounds}
\tikzset{node/.style={draw, circle, fill = black, minimum size = 3pt, inner sep=0pt, line width=1pt}}
\tikzset{nodesubwall/.style={draw, circle, fill = blue, minimum size = 3.1pt, inner sep=0pt}}
\tikzset{corner/.style={draw=magenta, fill = red!20!white, minimum size = 6pt,inner sep=0pt,line width=1pt}}
\tikzset{central/.style={draw=orange, fill = red!20!white, minimum size = 6pt,inner sep=0pt,line width=1pt}}
\tikzset{edge/.style={draw=white!60!black,line width=1.5pt}}
\tikzset{edgesubwall/.style={draw=blue!60!white,line width=2pt}}

\tikzset{subnode/.style={draw, circle, fill = yellow!50!red!50!white, minimum size = 2pt, inner sep=0pt}}
\tikzset{edgethin/.style={draw=white!60!black,line width=1.3pt}}

\begin{document}

\pagenumbering{Alph}

\newcommand{\funding}{T.K. was supported by the Research Council of Norway via the project BWCA (grant no. 314528). M.P. and G.S. were supported by the project BOBR that is funded from the European Research Council (ERC) under the European Union’s Horizon 2020 research and innovation programme with grant agreement No. 948057.}

\title{Minor Containment and Disjoint Paths in almost-linear time\thanks{\funding}}
\date{}
 \author{
   Tuukka Korhonen \\
   \small{University of Bergen} \\
   \small{tuukka.korhonen@uib.no}
   \and
   Michał Pilipczuk \\
   \small{University of Warsaw} \\
   \small{michal.pilipczuk@mimuw.edu.pl}
   \and
   Giannos Stamoulis \\
   \small{University of Warsaw} \\
   \small{giannos.stamoulis@mimuw.edu.pl}
 }
\maketitle

\begin{abstract}
We give an algorithm that, given graphs $G$ and $H$, tests whether $H$ is a minor of $G$ in time $\Oh_H(n^{1+o(1)})$; here, $n$ is the number of vertices of $G$ and the $\Oh_H(\cdot)$-notation hides factors that depend on $H$ and are computable.
By the Graph Minor Theorem, this implies the existence of an $n^{1+o(1)}$-time membership test for every minor-closed class of graphs.

More generally, we give an $\Oh_{H,|X|}(m^{1+o(1)})$-time algorithm for the rooted version of the problem, in which $G$ comes with a set of roots $X\subseteq V(G)$ and some of the branch sets of the sought minor model of $H$ are required to contain prescribed subsets of $X$; here, $m$ is the total number of vertices and edges of $G$.
This captures the {\sc{Disjoint Paths}} problem, for which we obtain an $\Oh_{k}(m^{1+o(1)})$-time algorithm, where $k$ is the number of terminal pairs. For all the mentioned problems, the fastest algorithms known before are due to Kawarabayashi, Kobayashi, and Reed [JCTB 2012], and have a time complexity that is quadratic in the number of vertices of $G$.

Our algorithm has two main ingredients: First, we show that by using the {\em{dynamic treewidth}} data structure of Korhonen, Majewski, Nadara, Pilipczuk, and Soko{\l}owski [FOCS 2023], the irrelevant vertex technique of Robertson and Seymour can be implemented in almost-linear time on apex-minor-free graphs.
Then, we apply the recent advances in almost-linear time flow/cut algorithms to give an almost-linear time implementation of the {\em{recursive understanding}} technique, which effectively reduces the problem to apex-minor-free graphs.
\end{abstract}

 \begin{textblock}{20}(-1.75, 3.6)
 \includegraphics[width=40px]{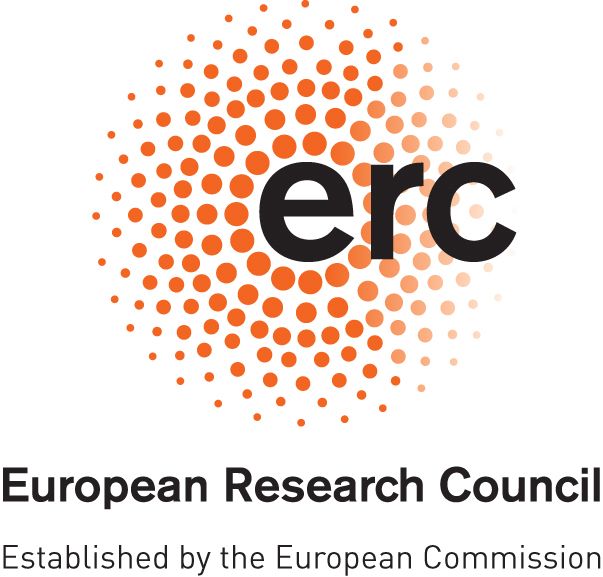}%
 \end{textblock}
 \begin{textblock}{20}(-1.75, 4.6)
 \includegraphics[width=40px]{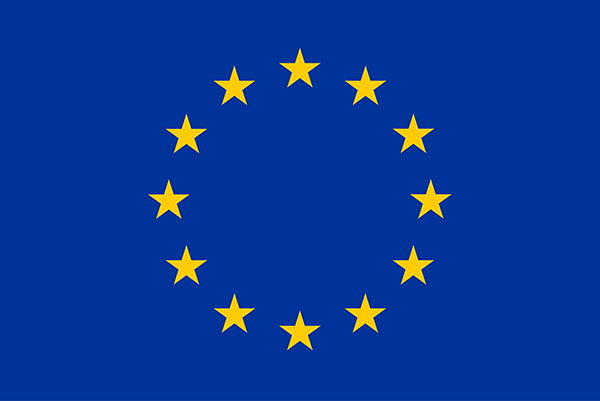}%
 \end{textblock}

\thispagestyle{empty}

\newpage

\setcounter{tocdepth}{2}

\tableofcontents
\thispagestyle{empty}

\newpage

\pagenumbering{arabic}

\clearpage
\setcounter{page}{1}

\section{Introduction}

A graph $H$ is a {\em{minor}} of a graph $G$ if $H$ can be obtained from $G$ by an iterative application of vertex deletions, edge deletions, and edge contractions. Equivalently, there is a mapping $\eta$, called the {\em{minor model}}, that maps every vertex $u$ of $H$ to the respective {\em{branch set}} $\eta(u)\subseteq V(G)$ so that (i) the branch sets $\eta(u)$ are pairwise disjoint, (ii) $G[\eta(u)]$ is connected for every vertex $u$ of $H$, and (iii) whenever $u$ and $v$ are adjacent in $H$, in $G$ there is an edge with one endpoint in $\eta(u)$ and the other in $\eta(v)$. The minor relation is transitive and reflexive, hence it constitutes a quasi-order on all finite graphs, called the {\em{minor order}}.

The minor order is perhaps the most fundamental notion of topological embedding for graphs considered in structural graph theory. Its study dates back to the theorem of Wagner, who characterized planar graphs as exactly those that do not contain $K_5$ and $K_{3,3}$ as a minor~\cite{Wagner37}; see also the closely related theorem of Kuratowski~\cite{Kuratowski30}. Another fundamental statement about the minor order is {\em{Wagner's Conjecture}} or the {\em{Graph Minor Theorem}}, which asserts that there is no infinite antichain in the minor order (or, equivalently in the obvious absence of infinite descending chains, that finite graphs are well-quasi-ordered by the minor relation). Initially posed in the 60s, Wagner's Conjecture was eventually confirmed by Robertson and Seymour in their monumental {\em{Graph Minors}} series~\cite{RobertsonS04}, which is probably the most profound achievement of modern structural graph~theory.


On the algorithmic front, the most important outcome of the Graph Minors series is a cubic-time fixed-parameter algorithm for the {\sc{Minor Containment}} problem. More precisely, in Graph Minors XIII~\cite{GM13}, Robertson and Seymour gave an algorithm with running time\footnote{The $\Oh_{\tup k}(\cdot)$ notation, where $\tup k$ is a vector of parameters, hides factors that depend on $\tup k$ in a computable manner.} $\Oh_H(n^3)$ that determines whether a given $n$-vertex graph $G$ contains a given graph $H$ as a minor. Since from the Graph Minor Theorem it follows that every minor-closed class of graphs $\Cc$ can be characterized by a finite number of forbidden minors, that is, $\Cc=\{G~|~G \textrm{ does not contain any member of }\Ff \textrm{ as a minor}\}$ for some finite set of graphs $\Ff$, this implies an $\Oh(n^3)$-time membership test for every minor-closed class of graphs $\Cc$: given $G$, to decide whether $G\in \Cc$ just test whether $G$ contains any graph from $\Ff$ as a minor, each in cubic time.


This provides a powerful technique for deriving (the existence of) polynomial-time algorithms for topological problems, as many concrete problems of interest can be directly cast as deciding membership in a certain minor-closed class of graphs. Classical examples include verifying embeddability of a graph into a fixed surface, deciding whether a graph admits a linkless embedding in $\mathbb{R}^3$, or computing the apex number of a graph (the minimum number of vertices to delete to obtain a planar graph). This technique has been particularly impactful for the early development of parameterized complexity, see e.g. the work of Fellows and Langston~\cite{FellowsL88,FellowsL87,FellowsL94}, and by now is a part of standard textbook material, see~\cite[Chapter~18]{DowneyF13},~\cite[Section 11.7]{FlumG06}, and~\cite[Section~6.3]{platypus}.

The algorithm of Robertson and Seymour in fact works for the more general problem of {\sc{Rooted Minor Containment}}. In this setting, $G$ is provided together with a specified set of roots $X\subseteq V(G)$, while $H$ is supplied with a function $\pi$ that prescribes, for each $u\in V(H)$, the set of roots $\pi(u)\subseteq X$ that should be contained in $\eta(u)$. Then, the running time for this problem is $\Oh_{H,|X|}(n^3)$. The canonical example of a problem that can be expressed as an instance of {\sc{Rooted Minor Containment}} is {\sc{Disjoint Paths}}: given a graph $G$ and $k$ pairs of vertices $(s_1,t_1),\ldots,(s_k,t_k)$, decide whether there exists vertex-disjoint paths $P_1,\ldots,P_k$ such that $P_i$ connects $s_i$ with $t_i$. Indeed, we take $X=\{s_1,t_1,\ldots,s_k,t_k\}$, $H$ to be an edgeless graph on $k$ vertices, and $\pi$ mapping vertices of $H$ to distinct terminal pairs.

Given the significance of {\sc{(Rooted) Minor Containment}} and {\sc{Disjoint Paths}}, both in structural graph theory and in parameterized complexity, it is natural to ask whether the cubic running time of the algorithm of Robertson and Seymour can be improved. And indeed, Kawarabayashi, Kobayashi, and Reed~\cite{KawarabayashiKR12} gave an $\Oh_k(n^2)$-time algorithm for {\sc{Disjoint Paths}} and argued that the same technique gives also an $\Oh_{H,|X|}(n^2)$-time algorithm for {\sc{Rooted Minor Containment}}. There is, however, a significant conceptual barrier for further improvements below quadratic time complexity. Namely, both the algorithms of~\cite{GM13} and of~\cite{KawarabayashiKR12} are based on the {\em{irrelevant vertex technique}}: iteratively find a vertex that can be removed without changing the answer, until the treewidth of the graph becomes bounded and the problem can be solved using dynamic programming on a tree decomposition.

The algorithm of Robertson and Seymour~\cite{GM13} locates every consecutive irrelevant vertex in quadratic fixed-parameter time, giving cubic time complexity in total. The main contribution of Kawarabayashi, Kobayashi, and Reed~\cite{KawarabayashiKR12} is improving the complexity of the procedure locating an irrelevant vertex to linear fixed-parameter time. To achieve further improvements, one needs to deviate from the general framework of locating up to $n$ irrelevant vertices, each using a separate procedure requiring at least linear time. So far, no such deviation working for general {\sc{Rooted Minor Containment}} has been proposed. However, for several concrete problems of interest it has been shown that irrelevant vertices can be located and removed in larger batches, leading to overall linear time complexity; examples are the aforementioned problems of computing the genus~\cite{Mohar99,KawarabayashiMR08}, the crossing number~\cite{KawarabayashiR07}, and the apex number of a graph~\cite{Kawarabayashi09,JansenLS14}. In each of these cases, the arguments were problem-specific and relied on the existence of suitable topological embeddings.

We remark that there has been also a substantial amount of work on {\sc{Disjoint Paths}} and {\sc{Minor Containment}} in planar graphs~\cite{ReedRSS91,Reed95,AdlerDFST12,AdlerKKLST17,LokshtanovMP0Z20,WlodarczykZ23}, aimed at reducing both the parametric factor and the dependence on the size of the host graph $G$ in the running time. In particular, Reed, Robertson, Schrijver, and Seymour gave a linear-time fixed-parameter algorithm for {\sc{Disjoint Paths}} in planar graphs~\cite{ReedRSS91,Reed95}, obtained again by removing irrelevant vertices in larger batches.

\paragraph*{Our contribution.} In this work we prove that the {\sc{Rooted Minor Containment}} problem can be solved in almost-linear fixed-parameter time.

\begin{restatable}{theorem}{mainthm}
\label{thm:main}
The {\sc{Rooted Minor Containment}} problem can be solved in time $\Oh_{H,|X|}(m^{1+o(1)})$, where $m$ is the total number of vertices and edges of $G$. In case of a positive answer, the algorithm also provides a witnessing minor model within the same running time.
\end{restatable}

A disclaimer is required about the computability claim in \cref{thm:main}. Our proof relies on several standard results from the theory of graph minors, including the Structure Theorem. While the original proof of the Structure Theorem, due to Robertson and Seymour~\cite{RobertsonS03a}, does not directly provide computable bounds on the obtained constants, Kawarabayashi, Reed, and Thomas~\cite{KawarabayashiTW20} recently gave a new, improved proof that yields explicit, computable bounds. The work~\cite{KawarabayashiTW20} has not yet been peer-reviewed and published, hence our computability claim is contingent on~\cite{KawarabayashiTW20} being correct. Without it, we can rely on standard statements from the Graph Minors series, and in the same way derive a non-uniform fixed-parameter algorithm, that is, a separate $m^{1+o(1)}$-time algorithm for every fixed $H$ and $|X|$.

As mentioned, by the Graph Minor Theorem, for every minor-closed class of graphs $\Cc$ there exists a finite set of graph $\Ff$ such that testing whether a given graph $G$ belongs to $\Cc$ reduces to verifying that $G$ does not contain any member of $\Ff$ as a minor. Since in graphs from proper minor-closed classes the number of edges is linear in the number of vertices~\cite{DBLP:journals/combinatorica/Kostochka84,thomason_1984}, from \cref{thm:main} we obtain the following.

\begin{theorem}\label{thm:minor-testing}
 For every minor-closed class of graph $\Cc$, there exists an $n^{1+o(1)}$-time algorithm that given an $n$-vertex graph $G$, decides whether $G\in \Cc$.
\end{theorem}

\cref{thm:minor-testing} provides a unified non-constructive technique for  obtaining almost-linear time algorithms for problems that can be reduced to deciding membership in a minor-closed class, including parameterized problems where this applies to every slice (instances with a fixed value of the parameter) of the problem. As we mentioned, linear-time algorithms for some of concrete problems of this type were known, including computing the genus or the apex number of a graph, but there are some concrete applications for which \cref{thm:minor-testing} seems to improve the state of the art. For instance, recognition of graphs admitting a linkless embedding in $\mathbb{R}^3$ was known to be solvable in time $\Oh(n^2)$~\cite{KawarabayashiKM12}, while \cref{thm:minor-testing} provides an $n^{1+o(1)}$-time algorithm for this problem.

Finally, we note that since {\sc{Disjoint Paths}} can be modelled as an instance of {\sc{Rooted Minor Containment}} for a specific pattern graph $H$, we have the following.

\begin{theorem}\label{thm:disjoint-paths}
 The {\sc{Disjoint Paths}} problem can be solved in time $\Oh_k(m^{1+o(1)})$, where $m$ is the total number of vertices and edges of the input graph.
\end{theorem}

\paragraph*{Our approach.} While a technical overview of the proof of \cref{thm:main} will follow in \cref{sec:overview}, we now briefly comment on how our approach fundamentally differs from the previous works.

As mentioned, so far most of linear-time implementations of irrelevant vertex arguments focused on identifying irrelevant vertices in larger batches, and removing many of them within a single linear-time iteration. This puts the bulk of the work on the graph-theoretic side, as one needs strong arguments to reason about the irrelevance of a significant portion of the graph. Given the immense complexity of irrelevant vertex arguments in general graphs, this makes this route towards subquadratic algorithms for {\sc{Rooted Minor Containment}} very difficult.

Our approach is to do the opposite: we still remove irrelevant vertices one by one, but we make use of the recent advances in data structures to detect every next irrelevant vertex in amortized time $n^{o(1)}$. As a result, we rely only on the most classic results from the graph minors theory --- the Flat Wall Theorem~\cite{GM13}, the Structure Theorem for minor-free graphs~\cite{RobertsonS03a}, and the Irrelevant Vertex Rule~\cite{GM13,DBLP:journals/jct/RobertsonS12} --- as black-boxes, while majority of the work lies on the algorithmic side.

The key ingredient that allows our improvement is the {\em{dynamic treewidth}} data structure due to Korhonen, Majewski, Nadara, Pilipczuk, and Soko\l{}owski~\cite{KorhonenMNPS23}. Namely, Korhonen et al. proved that given a dynamic graph $G$ of treewidth at most $k$, updated by edge insertions and deletions, one can maintain a tree decomposition of $G$ of width at most $6k+5$ in amortized update time $\Oh_k(n^{o(1)})$, together with the run of any reasonable dynamic programming algorithm on the maintained  decomposition. We use this data structure in a procedure that iteratively ``uncontracts'' the graph and detects and removes irrelevant vertices ``on the fly'', always keeping the treewidth bounded. This provides a conceptually transparent algorithm for {\sc{Rooted Minor Containment}} on apex-minor-free graphs, that is, graphs excluding a fixed apex graph (graph that can be made planar by removing one vertex) as a minor. With some technical work using the tools from the graph minor theory, particularly the Structure Theorem, this algorithm can be lifted to clique-minor-free graphs assuming that they are {\em{compact}}: the set of roots $X$ is well-linked, and one cannot separate a large portion of the graph from $X$ using a small separator. Next, we reduce the general clique-minor-free case to the compact case by providing an almost-linear time implementation of a recursive scheme in the spirit of {\em{recursive understanding}}~\cite{DBLP:journals/siamcomp/ChitnisCHPP16,DBLP:conf/stoc/GroheKMW11,DBLP:conf/focs/KawarabayashiT11}. To achieve this, we rely on multiple recent results on almost-linear time algorithms for flows and cuts~\cite{DBLP:conf/focs/Brand0PKLGSS23,DBLP:conf/focs/LiP20,DBLP:conf/focs/SaranurakY22}. Finally, the general case is reduced to the clique-minor-free case using a similar recursive scheme.

\paragraph*{Outlook.} We believe that the techniques introduced in this paper are of much wider applicability in the context of fixed-parameter algorithms and computational problems in the theory of graph minors. This in particular applies to
\begin{itemize}[nosep]
 \item the almost-linear time implementation of the irrelevant vertex rule on apex-minor-free graphs, using the dynamic treewidth data structure; and
 \item the almost-linear time implementation of recursive understanding.
\end{itemize}
In fact, a superset of the authors believes that they found a way to use this toolbox to give an almost-linear time algorithm for computing the tree decomposition provided by the Structure Theorem for minor-free graphs~\cite{RobertsonS03a}; a manuscript containing this result is under preparation. Another problem that could be approached using the introduced methodology is {\sc{Topological Minor Containment}}, for which a fixed-parameter algorithm was given by Grohe, Kawarabayashi, Marx, and Wollan~\cite{DBLP:conf/stoc/GroheKMW11}.

More concretely about this work, note that our main result, \cref{thm:main}, measures the running time in the total number of vertices and of edges. In the context of \cref{thm:minor-testing} one can assume that the number of edges is linear in the number of vertices, but we do not know of any technique that would allow making such an assumption in the context of solving the general {\sc{Rooted Minor Containment}} problem in almost-linear time. We remark that Kawarabayashi et al.~\cite{KawarabayashiKR12} showed how to reduce the number of edges to linear in the number of vertices  within their quadratic-time algorithm, and for this they used the sparsification method of Nagamochi and Ibaraki~\cite{NagamochiI92}. However, their approach does not seem to be directly applicable within the regime of almost-linear time complexity.


We remark that the running time of our algorithm would improve to $\Oh_{H,|X|}(m \log^{\Oh(1)} n)$ directly if the $n^{o(1)}$-factors in the running times of the subroutines from~\cite{KorhonenMNPS23,DBLP:conf/focs/Brand0PKLGSS23,DBLP:conf/focs/SaranurakY22} would be improved to $\log^{\Oh(1)} n$-factors.
However, it seems that obtaining a linear $\Oh_{H,|X|}(m)$ running time would require different techniques, as for example the dynamic treewidth data structure cannot be expected to be optimized to work in $\Oh_k(1)$ amortized update time (see~\cite{DBLP:journals/siamcomp/PatrascuD06}).

\section{Overview}\label{sec:overview}
In this section we present an informal overview of the proof of \Cref{thm:main}.
This overview is presented in three parts comprising \Cref{subsec:overview:apexminorfree,subsec:overview:cliqueminorfree,subsec:overview:general}.
We assume that the reader has basic familiarity with treewidth (see e.g. \cite[Chapter~7]{platypus}), but we do not assume much knowledge about graph minor theory.
Let us start with a short explanation of the general plan.

First, instead of focusing on {\sc{Rooted Minor Containment}}, we solve the more general {\sc{Folio}} problem, defined as follows. Given a graph $G$, a set of roots $X\subseteq V(G)$, and a parameter $\delta\in \N$, compute the {\em{$(X,\delta)$-folio}} of $G$: the set of all $X$-rooted minors of \emph{detail} at most $\delta$. Here, the detail of a rooted minor $(H,\pi)$ is the number of unrooted vertices, that is, vertices $u\in V(H)$ with $\pi(u)=\emptyset$. We give an algorithm for this problem with running time $\Oh_{|X|,\delta}(m^{1+o(1)})$. Moreover, with every rooted minor contained in the folio, the algorithm also computes a witnessing minor model. Clearly, this result implies \cref{thm:main}.

We build our algorithm by giving three successive algorithms for \Folio that work on increasingly more general graph classes.
First, we give an algorithm for \Folio on \emph{apex-minor-free} graphs, i.e., on graphs that exclude a bounded-size apex graph as a minor.
An \emph{apex graph} is a graph that can be made planar by deleting one vertex.
The main idea of this algorithm is to implement the irrelevant vertex technique of Robertson and Seymour~\cite{GM13} efficiently with the recent dynamic treewidth data structure of Korhonen, Majewski, Nadara, Pilipczuk, and Soko{\l}owski~\cite{KorhonenMNPS23}.

Second, we give an algorithm for \Folio on clique-minor-free graphs by reducing the problem from clique-minor-free graphs to apex-minor-free graphs.
This reduction works by implementing (a version of) the ``recursive understanding'' technique~\cite{DBLP:journals/siamcomp/ChitnisCHPP16,DBLP:conf/stoc/GroheKMW11,DBLP:conf/focs/KawarabayashiT11} in almost-linear time, by applying, among other techniques, several recent results on almost-linear time algorithms for flows and cuts~\cite{DBLP:conf/focs/Brand0PKLGSS23,DBLP:conf/focs/LiP20,DBLP:conf/focs/SaranurakY22}.
In particular, recursive understanding reduces the problem to graphs that are highly connected in a certain sense.
Graphs that are highly connected in this sense and exclude a clique as a minor are known to also exclude an apex graph as a minor after the deletion of bounded number of vertices~\cite{RobertsonS03a,LokshtanovPP022}.

Finally, we reduce the general case of \Folio to clique-minor-free graphs by again using (a simpler version of) recursive understanding, and exploiting a lemma of Robertson and Seymour~\cite{GM13} that guarantees that the $(X,\delta)$-folio of a graph $G$ is {\em{generic}}, i.e., contains all possible rooted graphs, if $G$ contains a clique minor into which $X$ is well-connected enough.

\subsection{Apex-minor-free graphs}
\label{subsec:overview:apexminorfree}
The first ingredient of our result is the desired algorithm for input graphs $G$ that exclude an apex graph as a minor.
In the \FolioApexG problem, we are given an instance $(G,X,\delta)$ of \Folio and an apex graph $R$, and the task is to either return a solution to $(G,X,\delta)$, or a minor model of $R$ in the graph $G-X$.
Our goal is to give an $\Oh_{|X|,\delta,R}(m^{1+o(1)})$ time algorithm for \FolioApexG.
The requirement to return the apex-minor in $G-X$ instead of $G$ means the algorithm in fact works on graphs that are apex-minor-free after the removal of $X$; this is the generality needed for the next step of clique-minor-free graphs.
In this overview, let us assume that $G-X$ indeed excludes $R$ as a minor.
The same ideas, after minor technical adjustments, work for the case when the model of $R$ needs to be returned.

\paragraph{The irrelevant vertex technique.}
Our algorithm for \FolioApexG is based on the \emph{irrelevant vertex technique}, that was introduced by Robertson and Seymour for their $\Oh_{|X|,\delta}(n^3)$ time algorithm for \Folio~\cite{GM13}.
Robertson and Seymour showed that if the treewidth of the graph $G$ is large compared to $|X|$ and $\delta$, then $G$ contains a vertex $v$ that is \emph{irrelevant} for the $(X,\delta)$-folio of $G$, meaning that the $(X,\delta)$-folios of $G$ and $G - v$ are the same.
In particular, they gave an algorithm that in $\Oh_{|X|,\delta}(n^2)$ time either finds such an irrelevant vertex $v$, or produces a tree decomposition of $G$ of width bounded by a function of $|X|$ and~$\delta$.
This yields their $\Oh_{|X|,\delta}(n^3)$ time algorithm simply by repeatedly removing irrelevant vertices until we obtain a tree decomposition of bounded width, and at that point solving the problem by dynamic programming on the tree decomposition.
Kawarabayashi, Kobayashi, and Reed gave their $\Oh_{|X|,\delta}(n^2)$ time algorithm by optimizing the algorithm for finding an irrelevant vertex or a tree decomposition to work in $\Oh_{|X|,\delta}(n)$ time~\cite{KawarabayashiKR12}.

In our algorithm for \FolioApexG, we also use the combinatorial argument for irrelevant vertices given by~\cite{GM13} (proof completed in \cite{DBLP:journals/jct/RobertsonS12}), but implement the repeated removal of irrelevant vertices more efficiently.
Let us then explain some details of the irrelevant vertex argument of~\cite{GM13} that are needed for explaining our algorithm.

\begin{figure}[!t]
\centering
\input{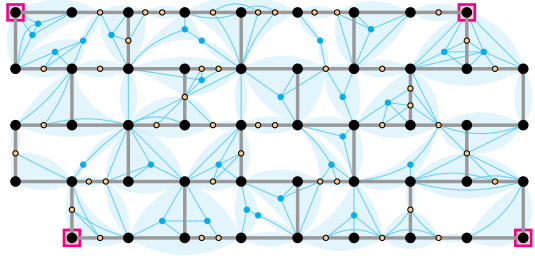}
\caption{An example of a wall $W$ of height four, consisting of the black and yellow vertices and the grey edges between them. The wall $W$ and the blue vertices and edges form a graph $G[C]$ such that $(W,C)$ is a flat wall in a graph $G$.}
\label{fig:flatwall}
\end{figure}
On graphs excluding a minor, the irrelevant vertex argument is based on flat walls.
A \emph{wall} of height $h$ is a subcubic graph similar to the $(h+1) \times (h+1)$-grid, see \Cref{fig:flatwall}.
The \emph{perimeter} of a wall $W$ is the cycle of $W$ bounding the outer face in the drawing shown in \Cref{fig:flatwall}, and is denoted by $\mathsf{Per}(W)$.
The four \emph{corners} of $W$ are the vertices highlighted with red in \Cref{fig:flatwall}.
By the Grid Minor Theorem of Robertson and Seymour~\cite{RobertsonS86}, a graph of large treewidth necessarily contains a wall of large height as a subgraph.
For the purpose of this overview, let us give a simple informal definition of a flat wall: a \emph{flat wall} of height $h$ in a graph $G$ is a pair $(W,C)$, where $W$ is a wall of height $h$ and $C$ is a connected set of vertices such~that
\begin{itemize}[nosep]
\item $W$ is a subgraph of $G[C]$;
\item a vertex in $C$ is adjacent to $V(G) - C$ only if it is in $\mathsf{Per}(W)$; and
\item $G[C]$ can be drawn in a certain planar-like way such that the corners of $W$ are on the outer face in this drawing (see \Cref{fig:flatwall}).
\end{itemize}

The Flat Wall Theorem of Robertson and Seymour~\cite{GM13} states that if a graph excludes a minor, then any wall that is large enough can be turned into a (smaller) flat wall after the removal of a small number of vertices from $G$.
In particular, there is a function $f(t, h)$ so that if $G$ has treewidth at least $f(t,h)$ and excludes $K_t$ as a minor, then $G$ contains a set $Z \subseteq V(G)$ of size $|Z| \le t^2$ and a flat wall $(W,C)$ of height $h$ in $G-Z$.
On apex-minor-free graphs, one can assume that the set $Z$ is empty, and therefore, in our setting we can assume that the set $Z$ in fact equals the set $X$ received as an input for the \Folio problem.

Robertson and Seymour~\cite{GM13} showed that the central vertices of a flat wall that satisfies certain extra conditions, called {\em{homogeneity}}, are irrelevant for the $(X,\delta)$-folio of $G$.
In this overview, let us not focus on these extra conditions, as any large enough flat wall $(W,C)$ can be turned into a flat wall $(W',C')$ that satisfies them and has $C' \subseteq C$, but let us instead assume that all flat walls satisfy these conditions.

To sum up, now our goal is to repeatedly find large flat walls and remove their central vertices from~$G$, until the treewidth of $G$ becomes small.
We also know that if $G$ has high treewidth, then $G - X$ is guaranteed to contain a large enough flat wall because $G-X$ is apex-minor-free.
Here, what is ``high treewidth'' and ``large enough flat wall'' depends only on $|X|$, $\delta$, and the excluded apex graph $R$.

\paragraph{Dynamic treewidth.}
To remove irrelevant vertices efficiently, we devise an algorithm using the recent \emph{dynamic treewidth} data structure by Korhonen, Majewski, Nadara, Pilipczuk, and Soko\l{}owski~\cite{KorhonenMNPS23}.
The dynamic treewidth data structure maintains a bounded-width tree decomposition of a  dynamic graph of bounded tree\-width under edge additions and deletions.
In particular, it is initialized with an $n$-vertex graph $G$ and an integer $k$.
It maintains a tree decomposition of $G$ of width at most $6k+5$, under edge additions and deletions, assuming that the treewidth of $G$ never exceeds $k$.
The amortized update time is $\Oh_k(n^{o(1)})$, and the initialization time is $\Oh_k(n^{1+o(1)})$.

Within the same running time, the data structure can also maintain any dynamic programming scheme on the tree decomposition, provided the dynamic programming state of a node can be computed from the states of its (at most $2$) children in $\Oh_k(1)$-time.
This is formalized in~\cite{KorhonenMNPS23} by stating that the data structure can maintain whether $G$ satisfies a fixed property expressible in the counting monadic second-order logic ($\CMSO_2$), analogously to Courcelle's theorem~\cite{courcelle1990monadic}.
In our algorithm, we also use $\CMSO_2$ as the formalism for describing dynamic programming on tree decompositions, and extend the interface provided by~\cite{KorhonenMNPS23} a little, observing that their data structure can be easily extended to also maintain labels on vertices, and to provide an operation that outputs a vertex $v \in V(G)$ that satisfies a $\CMSO_2$-formula $\varphi(v)$ (if such $v$ exists).
By using vertex-labels, we can assume that the data structure also supports deletions of vertices and insertions of isolated vertices in $\Oh_k(n^{o(1)})$ amortized time.

\paragraph{Fast removal of irrelevant vertices.}
With the preliminaries about irrelevant vertices and dynamic treewidth stated, we can now describe our algorithm.
The idea of the algorithm is to first contract the graph $G$ so that it contains only the vertices in $X$, and then start iteratively uncontracting it, removing irrelevant vertices as the treewidth grows.
Throughout the uncontracting process, the graph is stored in the dynamic treewidth data structure, which allows to quickly find irrelevant vertices.

More precisely, our algorithm works as follows.
Let us assume that $G-X$ is connected, as the $(X,\delta)$-folios of the graphs $G[C_i \cup X]$, where $C_i$ ranges over the connected components of $G-X$, can be easily combined to compute the $(X,\delta)$-folio of $G$.
Let $v_1,\ldots,v_{|X|}$ be an arbitrary ordering of $X$, and let $v_{|X|+1}, \ldots, v_n$ be an ordering of $V(G)-X$ so that any suffix $v_i,\ldots,v_n$ (for $i>|X|$) induces a connected subgraph of~$G$.
Note that such an ordering of $V(G)-X$ can be found, for example, as a post-order traversal of a spanning tree of $G-X$.

Then, for $i \in [|X|,n-1]$, let $G_i$ be the graph obtained from $G$ by contracting the vertices $\{v_{i+1},\ldots,v_n\}$ into a single ``mega-vertex'' $v^\star$.
In other words, $G_i$ has the vertex set $V(G_i) = \{v_1,\ldots,v_i\} \cup \{v^\star\}$, with the subgraph induced by $\{v_1,\ldots,v_i\}$ being equal to that of $G$, and with $v^\star$ being adjacent to $v_j$ with $j \le i$ if $v_j$ is adjacent to some vertex $v_d$ with $d > i$.
Note that because each suffix of the ordering is connected, each $G_i$ is a minor of $G$, and in particular, $G_i - X$ is a minor of $G-X$, and therefore apex-minor-free.

Our algorithm starts by initially setting $i = |X|$ and initializing the dynamic treewidth data structure to contain $G_i$.
Then, it iteratively transforms $G_i$ into $G_{i+1}$, by inserting the vertex $v_{i+1}$ and inserting and deleting edges as needed, and then increasing $i$ by one.

Suppose then that $i$ is the first iteration at which the treewidth of $G_i$ is at least some threshold, say $k \in \N$ depending on $|X|$, $\delta$, and $R$.
Now our goal is to decrease the treewidth of $G_i$ to below $k$ by removing irrelevant vertices.
The crucial observation is: if $(W,C)$ is a large flat wall in $G_i-X$, then by ``zooming'' into a subwall, we can find a (still large) flat wall $(W',C')$ in $G_i-X$ so that $C'$ does not contain the mega-vertex $v^\star$.
Now, because $C'$ does not contain $v^\star$, we have that $(W',C')$ is also a flat wall in $G-X$!
Therefore, we know that a central vertex $v$ of $W'$ is irrelevant for the $(X,\delta)$-folio of $G$, and thus can be safely removed from $G$.

We mark the irrelevant vertex $v$ as removed by adding it to a set $I \subseteq V(G)$ of irrelevant vertices that we start accumulating, and remove it from the graph held by the dynamic treewidth data structure.
In particular, the description of the algorithm above should be adjusted so that $G_i$ is replaced by $G_i - I$ and $G$ by $G-I$, where $I$ is the current set of irrelevant vertices accumulated.
The set $I$ always satisfies the invariant that the $(X,\delta)$-folio of $G-I$ is equal to the $(X,\delta)$-folio of $G$.
The algorithm keeps removing irrelevant vertices like this at the $i$th iteration until the treewidth of $G_i - I$ shrinks below $k$, and then continues the loop with $i$ increased by one.
In particular, always when $i$ increases, the treewidth of $G_i - I$ is guaranteed to be less than $k$, and therefore the treewidth of $G_{i+1} - I$ is guaranteed to be at most $k$.

At the end the treewidth of $G-I$ is at most $k$, and therefore we can solve \Folio on $G-I$ by dynamic programming in $\Oh_{|X|,\delta,k}(n)$ time, which by our invariant is equivalent to solving \Folio on $G$.

Now, there are some details to figure out.
First, to actually find the irrelevant vertex using the dynamic treewidth data structure, we need to write a $\CMSO_2$-formula $\varphi(v)$ that checks whether a vertex $v$ is a central vertex of a large flat wall $(W,C)$, where $C$ does not contain the mega-vertex $v^\star$.
Writing such a formula $\varphi(v)$ requires fixing several technical details, but ultimately it is not surprising that it can be done. We remark that the usage of $\CMSO_2$ in the context of computation problems in graph minor theory is not entirely new, it was done for instance by Grohe, Kawarabayashi, and Reed in~\cite{GroheKR13}.

Another detail to verify is that the iterative transformations of $G_i - I$ into $G_{i+1} - I$ cause at most $\Oh(m)$ edge insertions and deletions throughout the algorithm.
This is quite simple: An edge of the form $v_i v_j \in E(G)$ with $j > i$ is added when transforming $G_{j-1} - I$ into $G_j - I$ (if $i \notin I$), and is deleted only when $i$ or $j$ is added to $I$.
An edge of the form $v_i v^\star$ incident to the mega-vertex $v^\star$ gets always added when transforming $G_{i-1} - I$ into $G_i - I$ because the connectedness of every suffix implies that $v_i$ is adjacent to a vertex $v_j$ with $j>i$.
The edge $v_i v^\star$ gets deleted when transforming $G_{j-1} - I$ into $G_j - I$ if $v_j$ is the largest-index vertex to which $v_i$ is adjacent to.
To sum up, each possible edge gets added and deleted at most once, and it is easy to compute the iterations at which this happens.
The possible edges are those of $G$ and the edges $v_i v^\star$ for all $i$, so the total number of them is at most $|E(G)|+|V(G)|$.
This concludes the overview of our algorithm for \FolioApexG.

After this description, the reader may wonder why the same idea does not work for minor-free graphs.
The problem with minor-free graphs is that there we obtain a flat wall only after deleting a bounded-size set $Z \subseteq V(G)$ of vertices, which is not necessarily contained in the set $X$ like in our algorithm for apex-minor-free graphs.
This would be otherwise fine, but it is problematic if the set $Z$ obtained by applying the Flat Wall Theorem in the graph $G_i - I$ contains the mega-vertex $v^\star$.
In that case, we have no guarantee that the flat wall obtained in $G_i - I$ can be transformed into a flat wall in $G - I$.

\subsection{Clique-minor-free graphs}
\label{subsec:overview:cliqueminorfree}
The second step is to give the desired algorithm for input graphs $G$ that exclude a minor, which in the full generality is equivalent to excluding a clique minor.
In the \FolioClique problem, we are given an instance $(G,X,\delta)$ of \Folio and an integer $h$, and the task is to either return a solution to $(G,X,\delta)$, or a minor model of $K_h$ in $G$.
Our goal now is to give an $\Oh_{|X|,\delta,h}(m^{1+o(1)})$-time algorithm for \FolioClique.
By employing a trick using binary search, we can assume that in the latter case it suffices to only return the conclusion that $G$ contains $K_h$ as a minor, and thus in this overview let us simply just assume that $G$ does not contain $K_h$ as a minor.

\paragraph{Compact graphs.}
We use the following definition of ``highly-connected'' graphs.
Let $X \subseteq V(G)$ be a set of vertices (which indeed will be the same set $X$ for which we are solving $(X,\delta)$-folio), and $k,\alpha \in \N$ positive integers.
A set $C \subseteq V(G)$ is an \emph{$(X,k,\alpha)$-chip} if $C$ is connected and disjoint with $X$, $|N(C)| < k$, and $|C| \ge \alpha$.\footnote{$N(C)$ denotes the neighborhood of $C$.}
A graph $G$ is \emph{$(X,k,\alpha)$-compact} if it does not have any $(X,k,\alpha)$-chips.
The set $X \subseteq V(G)$ is \emph{well-linked} if there are no separations $(A,B)$ of $G$ so that $|A \cap X|,|B \cap X| > |A \cap B|$.\footnote{A separation of $G$ is a pair $(A,B)$ so that $A \cup B = V(G)$ and there are no edges between $A \setminus B$ and $B \setminus A$.}

Our algorithm for \FolioClique can be separated into two ingredients: First we employ a recursive process, that reduces the problem to graphs that are $(X,k,\alpha)$-compact, where $k$ is a large enough constant depending on $h$, $|X| \ge k$, $\alpha$ depends on $k$ and $\delta$, and $X$ is well-linked.
Second, we observe that such graphs are apex-minor-free after the removal of bounded number of vertices, and solve \Folio by using the algorithm for \FolioApexG described in \Cref{subsec:overview:apexminorfree}.

\paragraph{The algorithm for compact graphs.}
Let us start with the second ingredient, as it is simpler.
The algorithm in this case boils down to the observation that if $X$ is well-linked, and $G$ is $(X,k,\alpha)$-compact and excludes $K_h$ as a minor, for $k$ large enough depending on $h$, the Structure Theorem for graphs excluding a minor of Robertson and Seymour~\cite{RobertsonS03a} degenerates into a case where the tree decomposition provided by it has only one large bag, and other bags are adjacent to this central bag and have size bounded by $\alpha$. This observation was also used recently by Lokshtanov, Pilipczuk, Pilipczuk, and Saurabh~\cite{LokshtanovPP022} in the context of an isomorphism test for graphs excluding a minor.
Let us not delve into the statement of the Structure Theorem here, but what we derive from this observation is the following: if an $(X,k,\alpha)$-compact graph $G$ excludes $K_h$ as a minor, then there is a set $Z$ of bounded size so that $G-Z$ excludes an apex graph $R$ as a minor. Here, $|Z|$ and $R$ depend on only $|X|$, $k$, and $\alpha$.

Now if we would know the set $Z$, we could directly apply the algorithm for \FolioApexG with the instance $(G,X \cup Z, \delta)$ and the apex graph $R$, and guarantee to receive the solution for \Folio.
In reality we do not know $Z$, but take a different iterative approach, where in each application of the algorithm for \FolioApexG, we either receive the solution to \Folio, or find a vertex that must be in $Z$.

An \emph{apex-grid} of order $p$ is the apex graph formed by adding an universal vertex to the $p \times p$-grid.
Let $R'$ be an apex-grid of large order depending on $R$ and $h$.
We apply the algorithm for \FolioApexG with the instance $(G,X \cup Z', \delta)$ and the apex-grid $R'$, where $Z' \subseteq Z$ is the subset of $Z$ found so far.
If it returns a solution for \Folio we are done.
If it returns a minor model of $R'$ in $G-(X \cup Z')$, we can extract from it (in $\Oh_{R'}(m)$-time) either a minor model of $K_h$, or a minor model of a large grid with a vertex $a$ adjacent to all branch sets of the model (in $G-(X \cup Z')$).
In the former case we can terminate the algorithm immediately, and in the latter case we argue that vertex $a$ must belong to $Z$, so we add it to $Z'$ and repeat.
After $|Z| = \Oh_{|X|,k,\alpha}(1)$ iterations we are done.

\paragraph{Making $X$ well-linked.}
The more challenging ingredient is the process of reducing the problem to the case assumed in the above explanation.
This consists of making two assumptions hold: First, making $X$ a well-linked set of size at least $k$, and second, making $G$ into an $(X,k,\alpha)$-compact graph.
Let us start with the easier of them, which is making $X$ into a well-linked set of size at least $k$.

First, with the Ford-Fulkerson algorithm, we can in $\Oh_{|X|}(m)$ time test if $X$ is well-linked, and if not, find a separation $(A,B)$ of $G$ so that $|A \cap X|,|B \cap X| > |A \cap B|$.
When we find such a separation, we can recursively solve the instances $(G[A], (A \cap X) \cup (A \cap B), \delta)$ and $(G[B], (B \cap X) \cup (A \cap B), \delta)$, and combine the solutions.
Note that it is guaranteed that $|(A \cap X) \cup (A \cap B)|, |(B \cap X) \cup (A \cap B)| < |X|$.

The recursive procedure described above stops when $X$ becomes well-linked.
At that point, there are two cases: either $|X| < k$, or $|X| \ge k$.

The former case of $|X| < k$ is the simpler one of them. In that case, we wish to add vertices to $X$ to continue this recursion.
If we added arbitrary vertices, then the recursion could continue like this for depth $\Omega(n)$, and the algorithm could run in time $\Omega(n^2)$, essentially emulating the $\Oh_k(n^2)$-time treewidth approximation algorithm of Robertson and Seymour~\cite{GM13}.
Instead, we wish to emulate the $\Oh_k(n \log n)$-time treewidth approximation algorithm of Reed~\cite{DBLP:conf/stoc/Reed92}, which is essentially a modification of the Robertson-Seymour algorithm that achieves recursion-depth $\Oh_k(\log n)$.
In particular, by modifying the approach of Reed~\cite{DBLP:conf/stoc/Reed92}, we obtain an algorithm that either finds a balanced separator of $G$ of size at most $3k$, or a well-linked set of size $3k$, so that for any separator of $G$ of size at most $k$, the large side with respect to this well-linked set is also the large side with respect to the number of vertices.
Now if we find a small balanced separator, we recurse using it (causing the set $X$ of the recursive call to have size up to $4k-1$), and if we find a well-linked set of size $3k$ with such property, we add it to $X$ (also causing the set $X$ of the recursive call to have size up to $4k-1$).
This guarantees that the recursion-depth stays at most $\Oh_k(\log n)$.

\paragraph{Making $G$ compact.}
The case that we did not cover is what to do when $X$ is well-linked and has size $|X| \ge k$.
This is the main case, in which we now wish to make the graph $G$ into an $(X,k,\alpha)$-compact graph and apply the reduction to apex-minor-free graphs described before.

The basic approach for making $G$ into an $(X,k,\alpha)$-compact graph follows the standard technique of \emph{recursive understanding}: if there exists an $(X,k,\alpha)$-chip $C$, then compute the $(N(C), \delta)$-folio of $G[N[C]]$\footnote{$N[C]$ denotes $N(C) \cup C$.} by calling the algorithm recursively, and replace the graph $G[N[C]]$ by a graph $H_C$ having an equivalent $(N(C),\delta)$-folio and $H_C[N(C)] = G[N(C)]$, but with the size $\|H_C\|$ bounded by a function of $|N(C)|$ and $\delta$.
Note that $|N(C)| < k$, and at this point, let us fix $\alpha$ (based on $k$ and $\delta$) so that it is guaranteed that $\|H_C\| \le \alpha/2$, which implies that this replacement operation decreases the size of $G$ by at least $|C|/2$.
The computability of such bounded-size replacements from the $(N(C),\delta)$-folio of $G[N[C]]$ is a relatively standard fact, and analogous statements have been used before, for example in \cite{DBLP:conf/stoc/GroheKMW11}.

Because $C$ is disjoint from $X$, and the graph $H_C$ has the same $(N(C),\delta)$-folio as $G[N[C]]$, it follows that replacing $G[N[C]]$ by $H_C$ does not change the $(X,\delta)$-folio of $G$.
Note that this implies that it also does not change the fact that $X$ is well-linked, as Menger's theorem implies that this is certified by collections of disjoint paths between vertices in $X$.

By repeating this replacement-scheme until no more chips exist, the graph $G$ turns into an $(X,k,\alpha)$-compact graph, while preserving its $(X,\delta)$-folio.
This general strategy, known as recursive understanding, has been employed in several parameterized algorithms before~\cite{DBLP:journals/siamcomp/ChitnisCHPP16,DBLP:journals/siamcomp/CyganLPPS19,DBLP:conf/stoc/GroheKMW11,DBLP:conf/focs/KawarabayashiT11}.
However, all of its previous implementations had at least quadratic running times.
A major novel insight of our work is to show that it can be implemented with almost-linear running time.

\paragraph{Fast recursive understanding.}
We say that a vertex $v \in V(G)$ is \emph{$(X,k,\alpha)$-carvable} if $v$ is in some $(X,k,\alpha)$-chip $C$.
We first show that a family $\chipfamily$ of pairwise non-touching\footnote{Two chips $C_1$ and $C_2$ are {\em{touching}} if they intersect or have an edge between.} $(X,k,\alpha)$-chips that contains at least a fraction of $1/\Oh_{k,\alpha}(\log n)$ of all $(X,k,\alpha)$-carvable vertices can be computed in $\Oh_{k,\alpha}(m^{1+o(1)})$ time.
This is however not enough, because a priori we have almost no control on how replacing the chips could create new carvable vertices, meaning that implementing the replacement operation with such a family is not guaranteed to decrease the number of carvable vertices, implying that this process could go on for potentially $\Omega(n)$ iterations.

We fix this by showing that one can compute replacements $H_C$, of size bounded by a function of $k$ and $\delta$, that not only preserve that $(N(C),\delta)$-folio, but also are guaranteed to not create new $(X,k,\alpha)$-carvable vertices outside of the replaced subgraphs.
By using replacements with these stronger properties, we indeed guarantee that one iteration of the replacement-scheme decreases the number of $(X,k,\alpha)$-carvable vertices by a factor of $1/\Oh_{k,\alpha}(\log n)$, implying that $\Oh_{k,\alpha}(\log^2 n)$ iterations are sufficient to make $G$ into an $(X,k,\alpha)$-compact graph.

Let us say a few words about the techniques for achieving the claims of the previous two paragraphs.
For finding a large family of $(X,k,\alpha)$-chips, we use the approach of Chitnis, Cygan, Hajiaghayi, Pilipczuk, and Pilipczuk~\cite{DBLP:journals/siamcomp/ChitnisCHPP16} that applies color coding~\cite{DBLP:journals/jacm/AlonYZ95} to reduce the problem to a problem about finding \emph{terminal $k$-chips}: Let $G$ be a graph, $z \in V(G)$ a vertex, and $T \subseteq V(G) - \{z\}$ a set of terminals so that $T \cup \{z\}$ is an independent set.
A \emph{terminal $k$-chip} is a connected set $C \subseteq V(G)$, so that $C$ contains at least one terminal from $T$, does not contain $z$, $|N(C)| < k$, and $N(C)$ is disjoint with $T \cup \{z\}$.
The task is to find a collection $\chipfamily$ of non-touching terminal $k$-chips that contains all terminals $v \in T$ that are contained in some terminal $k$-chip.

The existence of such a collection $\chipfamily$ follows from the existence of Gomory-Hu trees for symmetric submodular functions~\cite{DBLP:journals/combinatorica/GoemansR95}.
Moreover, we could in fact quite directly apply the algorithm of Pettie, Saranurak, and Yin~\cite{DBLP:conf/stoc/PettieSY22} for Gomory-Hu trees for element connectivity to find this collection in $\Oh_k(m^{1+o(1)})$ time.
This algorithm is randomized, so to keep our algorithm deterministic, we instead devise our own deterministic $\Oh_k(m^{1+o(1)})$-time algorithm for finding such a collection $\chipfamily$.
For this, we apply the technique of isolating cuts of Li and Panigrahi~\cite{DBLP:conf/focs/LiP20}, together with the deterministic version of the almost-linear time algorithm for max-flow by van den Brand, Chen, Peng, Kyng, Liu, Gutenberg, Sachdeva, and Sidford~\cite{DBLP:conf/focs/Brand0PKLGSS23} (see~\cite{DBLP:conf/focs/ChenKLPGS22} for the original randomized version), and the mimicking network construction of Saranurak and Yingchareonthawornchai~\cite{DBLP:conf/focs/SaranurakY22}.

For the replacement graphs $H_C$ that do not create new carvable vertices, we show the existence of what we call \emph{intrusion-preserving} replacements.
They achieve the property that if $D'$ is a connected set in the graph $G'$ obtained by replacing $G[N[C]]$ by $H_C$, and $D'$ has $|N(D')| < k$, then a corresponding set $D$ also exists in $G$, having $|D| \ge |D'|$, $|N(D)| \le |N(D')|$, and so that $D \cap (V(G) - C) = D' \cap (V(G) - C)$.
In particular, any $(X,k,\alpha)$-chip in $G'$ can be translated into a no-smaller chip in $G$.
We prove the existence of intrusion-preserving replacements of bounded size by a quite intricate argument using preservers for folios as building blocks.
After their existence is proven, they are quite easy to compute: we only need to construct an algorithm that given $G[N[C]]$ and the $(N(C),\delta)$-folio of $G[N[C]]$, verifies in linear time if a given bounded-size $H_C$ indeed satisfies the desired properties.
For this we use the technique of important separators~\cite{DBLP:journals/tcs/Marx06}.

This finishes the description of our implementation of recursive understanding.
Let us mention a couple more details that need attention.
First, we need to make sure that the whole recursion tree  has depth $\Oh_k(\log n)$. Recall here that our recursion alternates between recursing along separators to make $X$ well-linked, and using recursive understanding when $X$ is a large well-linked set.
The bound on the recursion depth is guaranteed by the fact that when using recursive understanding, the size of $X$ in the recursive call is less than $k$, so before applying the recursive understanding step again, the ``Reed step'' must occur, making sure that the size of the graph significantly decreases on any path of the recursion tree containing multiple recursive understanding steps.

It is not enough that the recursion tree has depth $\Oh_k(\log n)$, but to not blow up the total size of all graphs in the tree, we need to guarantee that the total number of vertices in the ``child instances'' of an instance is at most a factor of $(1+1/\log n)$ larger.
This is trivially true in other steps of the overall recursive scheme, but in the recursive understanding step we need to use a trick that if a chip has size $|C| \le \Oh_k(\log n)$, then instead of solving it recursively, we just solve it with the Robertson-Seymour algorithm, running in time $\Oh_{k,\delta}(\log^3 n)$~\cite{GM13}.

\subsection{General case}
\label{subsec:overview:general} 
Finally, let us discuss how to reduce the general case of \Folio to \FolioClique.
For this, the main ingredient is a lemma of Robertson and Seymour~\cite{GM13}, that states that if a graph $G$ contains a minor model of a large enough clique, whose no branch set can be separated from $X$ by a separator of size $<|X|$, then the $(X,\delta)$-folio of $G$ is \emph{generic}, meaning that it contains all $X$-rooted graphs with detail $\delta$.

Our idea is that if $G$ is $(X,|X|,\alpha)$-compact, and contains a minor model of a large enough clique (depending on $\alpha$, $|X|$, and $\delta$), then by combining its branch sets to form a minor model of a smaller clique (but with larger branch sets), the preconditions of the lemma of Robertson and Seymour from~\cite{GM13} will be satisfied.
We give an algorithmic version of their lemma, which in this situation also recovers the models of the rooted minors in the $(X,\delta)$-folio, when given a model of the clique minor.
Therefore, on $(X,|X|,\alpha)$-compact graphs, \Folio directly reduces to \FolioClique, where the size of the sought clique minor depends on $\alpha$, $|X|$, and $\delta$.

Now the remaining goal is to turn the graph $G$ into an $(X,|X|,\alpha)$-compact graph, for some constant $\alpha$ depending on $|X|$ and $\delta$.
For this, we again employ a similar recursive understanding scheme with $(X,|X|,\alpha)$-chips and replacements as for minor-free graphs, but this time it is simpler because we do not require the properties that $X$ is well-linked or large.
In particular, recall that if $C$ is an $(X,|X|,\alpha)$-chip, then $|N(C)|<|X|$, implying that the depth of the recursion in this case is at most $|X|$.

\section{Preliminaries}

We denote the set of non-negative integers by $\N$. For integers $a,b$, we define $[a,b]\coloneqq \{a,a+1,\ldots,b\}$ when $a\leq b$ and $[a,b]\coloneqq \emptyset$ otherwise. Also, we use the shorthand $[a]\coloneqq [1,a]$. For a set $\Omega$, by $2^{\Omega}$ we denote the powerset of $\Omega$.

\paragraph*{Computation model.} We assume standard word RAM model with words of length $\Oh(\log n)$, where $n$ is the input length. As for representation of graphs, we assume that a vertex identifier fits within a single machine word and we use standard graph representation using adjacency lists. In particular, given a vertex $u$ of degree $d$, all edges incident to $u$ can be listed in time $\Oh(d)$.



\paragraph*{Graph terminology.}
All graphs in this paper are finite, simple (loopless and without parallel edges), and undirected, unless explicitly stated otherwise. By $V(G)$ and $E(G)$ we denote the vertex set and the edge set of a graph $G$.
By $\|G\|\coloneqq |V(G)|+|E(G)|$ we denote the total number of vertices and~edges.

By $N_G(u)$ we denote the set of neighbors of a vertex $u$ in a graph $G$.
For a vertex subset $X\subseteq V(G)$, its {\em{open neighborhood}} $N_G(X)$ comprises all vertices outside of $X$ that have a neighbor in $X$, and the {\em{closed neighborhood}} is defined as $N_G[X] = X \cup N_G(X)$. We may drop the subscript if it is clear from the context.

A set $A$ is {\em{connected}} in a graph $G$ if $G[A]$ is a connected graph. 
We define that a connected component of $G$ is an inclusion-wise maximal set $C \subseteq V(G)$ that is connected.
The set of connected components of $G$ is denoted by $\cc(G)$.
Two sets of vertices $A,B$ in a graph $G$ {\em{touch}} if $A\cap B\neq \emptyset$ or there exists a vertex in $A$ that has a neighbor belonging to $B$.

The clique on $h$ vertices is denoted by $K_h$. For $p\in \N$, the {\em{grid}} of order $p$ (or the $p\times p$ grid) is the graph on the vertex set $[p]\times [p]$ where vertices $(a,b)$ and $(a',b')$ are adjacent if $|a-a'|+|b-b'|=1$. We define the {\em{apex-grid}} of order $p$ to be the following graph: take a $p\times p$ grid and add one vertex adjacent to all the vertices of the grid. We will usually call this vertex the {\em{apex}} of an apex-grid. More generally, an {\em{apex graph}} is a graph that can be made planar by removing one vertex.

\paragraph{Treewidth.}
Let us recall the definition of treewidth.
A \emph{tree decomposition} of a graph $G$ is a pair $(T,\mathsf{bag})$ consisting of a tree $T$ and a function $\mathsf{bag}\colon V(T)\to 2^{V(G)}$ such that (i) for each edge $uv$ of $G$ there is a node $x\in V(T)$ with $\{u,v\}\subseteq \mathsf{bag}(x)$ and (ii) for each vertex $u$ of $G$, the set $\{x\in V(T)\colon u\in \mathsf{bag}(x)\}$ induces a non-empty connected subtree of $T$.
The \emph{width} of the decomposition $(T,\mathsf{bag})$ is $\max_{x\in V(T)}|\mathsf{bag}(x)|-1$. The \emph{treewidth} of $G$, denoted $\tw(G)$, is the minimum possible width of any tree decomposition of $G$.

{\em{Path decompositions}} and {\em{pathwidth}} are defined in the same way, except we require that tree $T$ underlying the decomposition is a path. In this context, we may simply assume that a path decomposition is a sequence of bags such that (i) every edge of $G$ has both endpoints in some bag, and (ii) for every vertex of~$G$, the bags containing this vertex form a non-empty interval in the sequence.

It is well-known that the treewidth of the $p\times p$ grid is equal to $p$.

\paragraph*{Separations and separators.}
A {\em{separation}} of a graph $G$ is a pair of vertex subsets $(A,B)$ such that $A\cup B=V(G)$ and there are no edges with one endpoint in $A\setminus B$ and the other in $B\setminus A$. The {\em{order}} of separation $(A,B)$ is the size of its {\em{separator}} $A\cap B$. With these definitions in place, we can recall the standard notion of a well-linked set.

\begin{definition}
 A set $X$ in a graph $G$ is {\em{well-linked}} if for every separation $(A,B)$ of $G$, we have
 $$|A\cap X|\leq |A\cap B|\qquad\textrm{or}\qquad |B\cap X|\leq |A\cap B|.$$
\end{definition}

We need the following well-known lemma about testing well-linkedness.

\begin{lemma}\label{lem:testing-well-linked}
 There is an algorithm that given a graph $G$ and a subset of vertices $X$, in time $2^{\Oh(|X|)}\cdot \|G\|$ decides whether $X$ is well-linked in $G$. In case of a negative outcome, the algorithm also provides a separation witnessing that $X$ is not well-linked.
\end{lemma}
\begin{proof}
 For every tri-partition of $X$ into $(A_X,S_X,B_X)$, verify whether in $G-S_X$ there is a vertex cut of size smaller than $\min(|A_X|,|B_X|)-|S_X|$ separating $A_X$ and $B_X$. This can be done by $\min(|A_X|,|B_X|)\leq |X|$ iterations of flow augmentation in the Ford-Fulkerson algorithm, which take time $\Oh(|X|\cdot \|G\|)$ per tri-partition, so $2^{\Oh(|X|)}\cdot \|G\|$ in total. If for any tri-partition such a cut is found, it naturally gives rise to a separation witnessing that $X$ is not well-linked, and otherwise $X$ is well-linked.
\end{proof}

We will also use the following notion of an $(S,T)$-separator.
Suppose $G$ is a graph and $S,T$ are two subsets of vertices of $G$.
An $(S,T)$-path is a path with one endpoint in $S$ and the other in $T$.
An \emph{$(S,T)$-separator} is a set $X$ that intersects every $(S,T)$-path in $G$.
Note that in this definition $S$ and $T$ are allowed to overlap, but $X \supseteq S \cap T$, and $X$ can be the empty set if there is no connected component of $G$ that intersects both $S$ and $T$.
A \emph{minimal $(S,T)$-separator} is an $(S,T)$-separator such that no subset of it is an $(S,T)$-separator.
By Menger's theorem, the size of a smallest $(S,T)$-separator is equal to the maximum number of vertex-disjoint $(S,T)$-paths.

\paragraph*{Minors.} We use the following definition of (rooted) minors.

\begin{definition}\label{def:rooted-minor}
 For a set $X$, an {\em{$X$-rooted graph}} is a graph $H$ together with a mapping $\roots\colon V(H)\to 2^X$ such that $\{\roots(u)\colon u\in V(H)\}$ are pairwise disjoint subsets of $X$.
 The \emph{detail} of $(H,\roots)$ is the number of vertices $v \in V(H)$ so that $\roots(v) = \emptyset$.
 For a graph $G$ and a vertex subset $X\subseteq V(G)$, we say that an $X$-rooted graph $(H,\roots)$ is an {\em{$X$-rooted minor}} of $G$ if there exists a mapping $\model\colon V(H) \to 2^{V(G)}$ such that
 \begin{itemize}[nosep]
 \item the sets $\{\model(u) \colon u \in V(H)\}$ are connected and pairwise disjoint;
  \item for each edge $uv\in E(H)$, there is an edge in $G$ with one endpoint in $\model(u)$ and the other in $\model(v)$;~and
  \item for each vertex $u\in V(H)$, $\roots(u)=\model(u)\cap X$.
 \end{itemize}
 A mapping $\model$ as above will be called a {\em{minor model}} of $(H,\roots)$ in $G$. The sets $\model(v)$ for $v \in V(H)$ are called the \emph{branch sets} of the model. By restricting the definition to unrooted graphs (i.e. $X=\emptyset$), we obtain the standard definitions of minors and minor models.
\end{definition}

It is well-known that the minor relation is transitive, and if $H$ is a minor of $G$, then $\tw(H)\leq \tw(G)$.

We will use the following two results on minors and grids. The first one is the classic Grid Minor Theorem of Robertson and Seymour~\cite{RobertsonS86}, see also later works for explicit and computable bounds~\cite{DiestelJGT99,RobertsonST94,LeafS15,ChekuriC16,ChuzhoyT19}.

\begin{theorem}[Grid Minor Theorem]\label{thm:grid-minor}
 There exists a computable function $g\colon \N\to \N$ such that for every $k\in \N$, every graph of treewidth at least $g(k)$ contains the $k\times k$ grid as a minor.
\end{theorem}


The second result is about finding grid minors within a grid.

\begin{lemma}[{\cite[Lemma 3.1]{DemaineFHT04}}, see also~{\cite[Lemma 30]{SauST23kapicesI}}]\label{prop:grid-in-grid}
 Let $m,k\in \N$ with $m\ge k^2 +2k$, and let $H$ be the $m\times m$ grid, say with vertex set $[m]\times [m]$. Further, let $X$ be a subset of vertices in the central $(m-2k)\times (m-2k)$ subgrid of $H$ (that is, $X\subseteq [k+1,m-k]\times [k+1,m-k]$), with $|X|\ge k^4$.
 Then in $H$ there is a minor model of the $k\times k$ grid in which every branch set intersects $X$.
\end{lemma}

\paragraph*{Folios.}
Next, we define the {\em{folio}} of a rooted graph, which is essentially the set of all rooted minors of bounded detail that appear in the graph.

\begin{definition}
Let $G$ be a graph and $X$ be a set of vertices of $G$.
For $\delta\in \N$, the {\em{$(X,\delta)$-folio}} of $G$ is the set of all $X$-rooted minors of $G$ with detail at most $\delta$.
We say that the $(X,\delta)$-folio of $G$ is {\em{generic}} if it contains all $X$-rooted graphs with detail at most $\delta$.
An \emph{$(X,\delta)$-model-folio} of $G$ is a pair that contains the $(X,\delta)$-folio of $G$, and a mapping that maps each $X$-rooted graph in the folio to a model of it in $G$.
\end{definition}

In the context of folios, we implicitly and without loss of generality assume that the set of vertices $V(H)$ of an $X$-rooted graph $(H,\roots)$ in the folio is the set $[|V(H)|]$, in order to make the $(X,\delta)$-folio of $G$ finite.
Note that indeed the number of $X$-rooted graphs in the $(X,\delta)$-folio of $G$ is bounded by a computable function of $|X|$ and $\delta$.

The definition of folio naturally gives rise to the following computational problem, which is the main point of focus in this work.

\begin{definition}
 In the \Folio problem we are given a parameter $\delta\in \N$, a graph $G$, and a vertex subset $X\subseteq V(G)$. The task is to compute the $(X,\delta)$-model-folio of $G$.
\end{definition}

An \emph{instance} of the \Folio problem is the triple $(G,X,\delta)$. With these definitions in place, we may state formally the main result of this work.

\begin{restatable}{theorem}{thmmainreal}
\label{thm:main-real}
 The \Folio problem can be solved in time $\Oh_{|X|,\delta}(\|G\|^{1+o(1)})$.
\end{restatable}

Note that \cref{thm:main-real} implies \cref{thm:main}, because given a graph $G$, a set of roots $X\subseteq V(G)$, and an $X$-rooted graph $(H,\pi)$, to decide whether $(H,\pi)$ is an $X$-rooted minor of $G$ (and compute a witnessing minor model) it suffices to compute an $(X,\delta)$-model-folio of $G$ for $\delta=|\{u\in V(H)\colon \pi(u)=\emptyset\}|$, verify whether $(H,\pi)$ is in this folio, and if so, return the witnessing model. Therefore, for the remainder of this paper we focus on proving \cref{thm:main-real}.




Finally, we will need the following standard lemma about compositionality of folios.
As its proof is long but straightforward, we present it in \Cref{sec:appendix_foliocompose}.
\begin{restatable}{lemma}{lemfoliooverseparator}
\label{lem:folio_over_separator}
Suppose $G$ is a graph, $X$ is a subset of vertices of $G$, and $(A,B)$ is a separation of $G$. Denote $X_A=(X \cap A) \cup (A \cap B)$ and $X_B=(X \cap B) \cup (A \cap B)$. Let $\Ff_A$ be the $(X_A,\delta)$-folio of $G[A]$, and $\Ff_B$ the $(X_B,\delta)$-folio of $G[B]$. Then we have the following:
\begin{itemize}
\item The $(X,\delta)$-folio of $G$ is uniquely determined by $\Ff_A$ and $\Ff_B$. In particular, there is a function $\combfolio$ so that $\combfolio(\{\Ff_A, \Ff_B\}, X, \delta)$ is the $(X,\delta)$-folio of $G$.
\item The function $\combfolio$ can be evaluated in time $\Oh_{|X|+|A \cap B|,\delta}(1)$.
\item Given an $(X_A,\delta)$-model-folio of $G[A]$ and an $(X_B,\delta)$-model-folio of $G[B]$, an $(X,\delta)$-model-folio of $G$ can be computed in time $\Oh_{|X|+|A \cap B|,\delta}(|V(G)|)$.
\end{itemize}
\end{restatable}

\paragraph{Embeddings.} In several parts of this work we will consider surface-embedded graphs.
By a {\em{surface}} we mean a compact two-dimensional manifold $\Sigma$ without boundary. It is well-known that every such surface is homeomorphic to a sphere with a finite number of handles and a finite number of cross-caps attached. Therefore, whenever in an algorithmic context we say that a surface is given, it is always specified by providing the respective number of handles and cross-caps. An {\em{embedding}} of a graph $G$ in $\Sigma$ (resp. in a subset $X$ of $\Sigma$) is a drawing $\gamma$ of vertices and edges of $G$ on $\Sigma$ (resp.~$X$) in a non-crossing manner: $\gamma$ maps vertices to pairwise distinct points on $\Sigma$ (resp. $X$) and edges to pairwise internally disjoint curves (homeomorphic images of the interval $[0,1]$) connecting images of the corresponding endpoints.

\section{Apex-minor-free graphs}\label{sec:apex-minor-free}
We start with providing an algorithm for the variant of the problem where we assume that the input graph excludes some apex-grid as a minor. Formally, we will work with the following problem.

\begin{definition}\label{def:apexFolio}
 In the \FolioApex problem we are given an instance $(G,X,\delta)$ of the \Folio problem and a parameter $p\in \N$. The task is to either solve the instance $(G,X,\delta)$, or to output a minor model of the apex-grid of order $p$ in $G\setminus X$.
\end{definition}

Our main result in this section is the following.

\begin{restatable}{theorem}{thmapexfree}\label{thm:apexFree}
    The \FolioApex problem can be solved in time $\Oh_{|X|,\delta,p}(\|G\|^{1+o(1)})$.
\end{restatable}

This section is organized as follows.
In~\Cref{subsec:prelim-logic}, we present the main algorithmic tools that are used in the proof of~\Cref{thm:apexFree}. These are~\Cref{prop:courcelle,prop:dynamic-tw}, on evaluating $\CMSO_2$-formulas on (static) graphs and on maintaining $\CMSO_2$-definable queries on dynamically changing graphs, both under the restriction of bounded treewidth.
Definitions that are needed for logic are also presented in this subsection.
The proof of~\Cref{thm:apexFree} follows a scheme for iteratively uncontracting the graph. This is described in~\Cref{subsec:iterative-decomp}, where we state~\Cref{lem:main-subroutine-apexfree}, which supports the main subroutine of the algorithm of~\Cref{thm:apexFree}.
In particular~\Cref{lem:main-subroutine-apexfree} provides a formula certifying that a vertex is irrelevant for the $(X,\delta)$-folio of a given graph $G$.
The combinatorial conditions that allow a vertex to be characterized irrelevant are given in~\Cref{subsec:existence-irre}. Then, in~\Cref{subsec:defirvlogic} we show that these conditions can be expressed in $\CMSO_2$, and we prove~\Cref{lem:main-subroutine-apexfree}.

\subsection{Preliminaries on logic}
\label{subsec:prelim-logic}

Let us recall the standard logic $\CMSO_2$ on graphs. In this logic, there are variables of four sorts:
\begin{itemize}[nosep]
    \item single vertex variables;
    \item single edge variables;
    \item vertex set variables; and
    \item edge set variables.
\end{itemize}
Variables of the last two sorts are called {\em{monadic}}. We follow the convention that single vertex/edge variables are denoted using lowercase letters, whereas uppercase letters are reserved for monadic variables. 

Atomic formulas of $\CMSO_2$ are of the following form:
\begin{itemize}[nosep]
    \item equality of two single vertex/edge variables: $x=y$;
    \item membership test between a single vertex/edge variable and a monadic variable: $x\in X$;
    \item test of the incidence relation between a single vertex variable and a single edge variable: $\mathsf{inc}(x,e)$;
    \item modular predicates of the form $|X|\equiv a \bmod m$, where $a,m$ are nonnegative integers with $m\neq 0$. 
\end{itemize}
Larger formulas can be formed out of these atomic formulas by means of standard Boolean connectives, negation, and quantification of all sorts of variables (both universal and existential). The semantics is defined naturally. The {\em{size}} of a formula $\varphi$ of $\CMSO_2$, denoted by $\|\varphi\|$, is the number of subformulas of $\varphi$ plus the sum of the moduli $m$ of all the modular predicates present in $\varphi$.
We also consider the logic $\mathsf{MSO}_2$, which is a restriction of $\CMSO_2$ that does not allow modular predicates.

As usual, {\em{sentences}} are formulas without free variables, and sentence $\varphi$ being satisfied in a graph $G$ is denoted by $G\models \varphi$. More generally, if a formula $\varphi$ has free variables $\tup X$, we often denote it by $\varphi(\tup X)$. Here, $\tup X$ is a {\em{tuple of variables}}, which is just a finite set of variables, of all possible sorts. If $G$ is a graph, by $G^{\tup X}$ we denote the set of all {\em{evaluations}} of variables of $\tup X$ in $G$; these are mappings that associate with each variable $X$ of $\tup X$ a feature $\tup U(X)$ of appropriate sort (vertex, edge, vertex subset, or edge subset) in~$G$.
Then for an evaluation $\tup U\in G^{\tup X}$ we can write $G\models \varphi(\tup U)$ and say that $\varphi(\tup U)$ is satisfied in $G$, with the obvious meaning.

As proved by Courcelle~\cite{courcelle1990monadic}, $\CMSO_2$ formulas can be efficiently evaluated on graphs of bounded treewidth (see also~\cite{ArnborgLS91easy}).

\begin{theorem}[\cite{courcelle1990monadic,ArnborgLS91easy}]\label{prop:courcelle}
    There is an algorithm that, given a $\CMSO_2$ formula $\varphi(\tup X)$ and a graph $G$, runs in time $\Oh_{\|\varphi\|,\mathsf{tw}(G)}(\|G\|)$, and either finds an evaluation $\tup U\in G^{\tup X}$ such that $G\models \varphi(\tup U)$, or correctly concludes that such an evaluation does not exist.
\end{theorem}

We remark that the original statement of~\cite{courcelle1990monadic,ArnborgLS91easy} assumed that $G$ is provided together with a tree decomposition of bounded width, but this can be computed in time $2^{\Oh(\mathsf{tw}(G)^3)}\cdot \|G\|$ using the algorithm of Bodlaender~\cite{Bodlaender96}.

\paragraph{Dynamic maintenance of queries in graphs of bounded treewidth.}

Let $\Sigma$ be a finite set of colors. A {\em{$\Sigma$-vertex-colored graph}} is a graph $G$ together with, for every color $U\in \Sigma$, a subset of vertices $U^G\subseteq V(G)$ that are deemed to be of color $U$ in $G$. Note that the colors are not necessarily disjoint: a vertex can be simultaneously of several different colors (belongs to several sets $U^G$) or of no color at all (belongs to none of the sets $U^G$). For readers familiar with the terminology of relational structures, colors are just unary predicates on the vertex set of the graph. We can naturally consider the $\CMSO_2$ logic on $\Sigma$-vertex-colored graphs by allowing atomic formulas of the form $U(x)$ for $U\in \Sigma$, where $x$ is a single vertex variable, with the semantics that $U(v)$ holds in $G$ if $v\in U^G$.

As mentioned, in order to detect irrelevant vertices fast, we use the dynamic treewidth data structure of Korhonen et al.~\cite{KorhonenMNPS23}.
This was originally proved for Boolean queries (i.e., when $\varphi$ is a sentence, without free variables, and the query only reports whether $G\models \varphi$) on graphs without vertex colors.
The form presented here is for queries with one free single vertex variable on colored graphs.
In~\Cref{app:dyn-tw}, we explain how to adapt the proof of Korhonen et al.~\cite{KorhonenMNPS23} to the setting presented here.

\begin{restatable}[\cite{KorhonenMNPS23}, adapted]{theorem}{dynamictw}\label{prop:dynamic-tw}
 Consider an integer $k\in \N$, a finite set of colors $\Sigma$, and a formula $\varphi(x)$ of  $\CMSO_2$ over $\Sigma$-vertex-colored graphs, where $x$ is a vertex variable. There is a data structure that maintains a dynamic $\Sigma$-vertex-colored graph $G$ under the following~updates:
 \begin{itemize}[nosep]
  \item given vertices $u,v$, add the edge $uv$;
  \item given vertices $u,v$, delete the edge $uv$;
  \item given a vertex $u$ and a color $U\in \Sigma$, flip the membership of $u$ to $U^G$;
 \end{itemize}
 and supports the following query:
 \begin{itemize}[nosep]
  \item provide any vertex $a$ such that $G\models \varphi(a)$, or correctly report that no such vertex exists.
 \end{itemize}
 The query may report \textsf{Treewidth too large} instead of providing the correct answer, which guarantees that the current treewidth of $G$ exceeds $k$.

The initialization on an edgeless and colorless $n$-vertex graph and for given $k,\varphi$ takes time $\Oh_{k,\varphi,\Sigma}(n)$.
 Then, each update takes amortized time $\Oh_{k,\varphi,\Sigma}(n^{o(1)})$, and each query takes worst-case time $\Oh_{k,\varphi,\Sigma}(1)$.
\end{restatable}

Obviously, \cref{prop:dynamic-tw} can be also used when $\varphi$ is a sentence, rather than a formula with one free variable, to maintain whether $\varphi$ is satisfied in the graph. Just add a dummy non-used free variable to $\varphi$ and test whether the data structure returns some vertex or not.

We will also use the following observation.

\begin{observation}\label{obs:mso-minor}
    Let $\ell,h\in \mathbb{N}$. Also,
    let $\tup x$ be a tuple of $\ell$ single vertex variables and $\tup Y$ be a tuple of $h$ vertex set variables.
    For every graph $H$ on $h$ vertices and every mapping $\pi: V(H)\to 2^{\tup x}$,
    there is an $\mathsf{MSO}_2$-formula $\varphi_{(H,\pi)}(\tup x,\tup Y)$ such that the following holds.
    Given a graph $G$, a tuple $\tup v\in G^{\tup x}$, and a tuple $\tup U\in G^{\tup Y}$, we have $G\models \varphi_{(H,\pi)}(\tup v,\tup U)$ if and only if
    \begin{itemize}[nosep]
        \item $(H,\pi_{\tup v})$ is an $X$-rooted minor of $G$; and
        \item $\tup U$ is a minor model of $(H,\pi_{\tup v})$ in $G$,
    \end{itemize}
    where $X=\{\tup v(x): x\in\tup x\}$ and $\pi_{\tup v}$ maps each $u\in V(H)$ to the subset $\{\tup v(x): x\in \pi(u)\}$ of $X$.
\end{observation}

The proof is obvious:
the formula of~\Cref{obs:mso-minor} checks whether $(H,\pi)$ is an $X$-rooted minor of $G$ by checking whether $\tup U$ satisfies the conditions in~\Cref{def:rooted-minor}, which can be easily expressed in $\mathsf{MSO}_2$.

We also use the following result on maintaining bounded treewidth on graphs updated by edge additions and deletions. Its proof is a direct consequence of~\Cref{prop:dynamic-tw},
applied for a formula $\varphi$ that expresses that the treewidth of the graph is at most $k$. The latter is certified by a sentence as the one in~\Cref{obs:mso-minor} encoding the \emph{obstructions} for graphs of treewidth at most $k$; see~\cite[Theorem 5.9]{Lagergren98}.

\begin{lemma}\label{lem:ds-tw}
    Consider an integer $k\in \N$. There is a data structure that for an $n$-vertex graph $G$ updated by edge additions and deletions,
    maintains whether the treewidth of $G$ is at most $k$.
   The initialization time is $\Oh_{k}(n)$, the amortized update time is $\Oh_{k}(n^{o(1)})$, and the query time is $\Oh_{k}(1)$.
\end{lemma}

\begin{proof}
We set up the data structure obtained from~\Cref{prop:dynamic-tw} for a sentence $\varphi_k$ of $\CMSO_2$ such that $G\models\varphi_k$ if and only if the treewidth of $G$ is at most $k$. In order to obtain such sentence for every fixed~$k$, we use the fact that graphs of bounded treewidth graphs can be characterized in terms of forbidden minors.
In particular, by~\cite[Theorem 5.9]{Lagergren98},
we infer the existence of a collection $\mathcal{F}_k$ of graphs, which is computable from $k$, such that $G$ has treewidth at most $k$ if and only if $G$ does not contain any graph in~$\mathcal{F}$ as a minor.
Let us stress that the size of $\mathcal{F}_k$ and the number of vertices in each graph in $\mathcal{F}_k$ are both bounded by some computable function of $k$. 
The claimed sentence $\varphi_k$ is obtained after considering for each graph $H\in\mathcal{F}_k$ the sentences obtained from~\Cref{obs:mso-minor} for such $H$ (expressing that $H$ is a minor of $G$) and then conjoining the negations of all such formulas.
\end{proof}    

\subsection{Iterative decompression}\label{subsec:iterative-decomp}

In this subsection we present~\Cref{lem:main-subroutine-apexfree} and how to use it in order to show~\Cref{thm:apexFree}.
The proof of~\Cref{thm:apexFree} is based on an iterative decompression scheme that allows us to use the data structure provided by~\Cref{prop:dynamic-tw} in order to detect \emph{irrelevant} vertices.
Given an instance $(G,X,\delta)$ of the \textsc{Folio} problem, we say that a vertex $v\in V(G)\setminus X$ is \emph{irrelevant} for $(G,X,\delta)$ if the $(X,\delta)$-folio of $G$ is the same as the $(X,\delta)$-folio of $G-v$.

\paragraph*{Contracting a connected set of edges.}
Let $G$ be a graph and let $F\subseteq E(G)$.
We use $V(F)$ to denote the set of all endpoints of edges in $F$.
We say that $F$ is \emph{connected} if the graph $(V(F),F)$ is connected.
We use $G/F$ to denote the graph obtained from $G$ after contracting all edges in $F$.
If $F$ is connected, then we use $v_F$ to denote the (unique) vertex of $G/F$ that is the result of the contraction of the edges in $F$.


As explained in the overview in~\Cref{subsec:overview:apexminorfree}, in the uncontracting process of our algorithm, our aim is to detect irrelevant vertices for the instance $(G,X,\delta)$.
The following result states that the existence of vertices that are irrelevant for $(G,X,\delta)$ can be defined by a logical formula which is evaluated on the contracted graph $G/F$, for some connected subset $F$ of edges that are not incident to $X$.
Moreover, this lemma specifies the conditions that allow such formula to be true, which in turn implies the existence of an irrelevant vertex for $(G,X,\delta)$.

\begin{restatable}{lemma}{subroutineapexfree}\label{lem:main-subroutine-apexfree}
    For every $\ell,\delta\in \N$ there is an $\mathsf{MSO}_2$-formula $\phi(\tup x,y,z)$, computable from $\ell$ and $\delta$, where $\tup x$ is a tuple of $\ell$ vertex variables and $y,z$ are vertex variables, such that the following holds.
    Let $(G,X,\delta)$ be an instance of the \textsc{Folio} problem, where $|X|=\ell$,
    and let $F$ be a connected subset of $E(G)$ with $V(F)\cap X=\emptyset$.
    Then, for every $\tup v\in G^{\tup x}$ such that $X=\{\tup v(x): x\in\bar{x}\}$ and every
    vertex $u\in V(G)\setminus (V(F)\cup X)$, if $G/F\models \varphi(\tup v,v_F,u)$, then $u$ is irrelevant for $(G,X,\delta)$.

    Moreover, for every $p\in\mathbb{N}$ there is a constant $k\in\mathbb{N}$, computable from $\ell, \delta$, and $p$, such that for every $\tup v\in G^{\tup x}$ such that $X=\{\tup v(x)\colon x\in\bar{x}\}$ if
    \begin{itemize}[nosep]
        \item $(G/F)\setminus X$ does not contain the apex-grid of order $p$ as a minor; and
        \item the treewidth of $G/F$ is larger than $k$;
    \end{itemize}
    then there is a vertex $u\in V(G)\setminus (V(F)\cup X)$ such that $G/F\models\varphi(\tup v,v_F,u)$.
\end{restatable}


The proof of~\Cref{lem:main-subroutine-apexfree} is presented in the rest of this section.
Before presenting the proof, 
we show how to use~\Cref{lem:main-subroutine-apexfree} in order to prove~\Cref{thm:apexFree}.
We restate~\Cref{thm:apexFree} here for convenience.

\thmapexfree*

\begin{proof}
    We start by assuming that $G\setminus X$ is connected, since otherwise we solve 
    \FolioApex in $G[C\cup X]$, for each connected component $C$ of $G\setminus X$ and combine the outcomes using~\Cref{lem:folio_over_separator}.
    
    Let $n\coloneqq |V(G)|$ and $\ell\coloneqq |X|$.
    The algorithm performs the following iterative decompression scheme.
    \begin{enumerate}
        \item Set $\tup v\coloneqq (v_1,\ldots,v_\ell)$ be an arbitrary ordering of $X$ and consider an ordering $v_{\ell+1}, \ldots, v_n$ of $V(G)-X$ so that any suffix $v_i,\ldots,v_n$ (for $i>\ell$) induces a connected subgraph of~$G$.
        Also, for every $i\in [\ell,n]$ compute $\tau(i)$ to be the highest index of a neighbor of $v_i$, and for each $j\in [\ell,n]$ compute $E_j$ as the set of edges $v_iv_j$ for $i<j$ such that $j=\tau(i)$.

        \item Set $\Sigma = \{U_1,\ldots,U_\ell,I,A,M\}$ to be a set of $\ell+3$ colors. 
        Let us mention here that in the colors $U_1,\ldots,U_\ell$ are used to encode $X$ (i.e., $U_i^G$ will be the singleton $\{v_i\}$, for $i\le \ell$), $A$ encodes the active vertices in each step of the iterative procedure, $I$ encodes the irrelevant vertices (i.e., vertices that are removed without affecting the $(X,\delta)$-folio of the instance)  and $M$ encodes the megavertex $v^\star$ that is used in the decompression scheme.  
        
        \item Set $k$ and $\varphi(\tup x,y,z)$ be the constant and the formula from~\Cref{lem:main-subroutine-apexfree}, computed for $\ell,\delta$, and $p$.
        Set also $\widehat{\varphi}(z)$ to be the formula obtained from $\varphi(\tup x,y,z)$ after existentially quantifying single vertex variables $x_1,\ldots,x_{\ell},y$, where $(x_1,\ldots,x_\ell)=\tup x$, and appending the formula $\varphi(\tup x,y,z)$ modified as follows:
        \begin{itemize}[nosep]
            \item ask that for each $i\in [\ell]$, $x_i\in U_i$, $y\in M$, and $z\in A\setminus (I\cup M \cup \bigcup_{i\in [l]} U_i)$; and
            \item restrict the scope of all quantifications in $\varphi(\tup x,y,z)$ to $A\setminus I$.
            More precisely, modify $\varphi(\tup x,y,z)$ as follows: for a quantified single vertex (resp. edge) variable, demand additionally that it belongs to (resp. its endpoints belong to) $A\setminus I$; for a quantified vertex (resp. edge) set variable, demand additionally that it is a subset of (resp. the endpoints of its edges belong to) $A\setminus I$.
        \end{itemize}

        \item 
        For $k,\Sigma$ and $\widehat{\varphi}(z)$ as above, let $D$ be the data structure of~\Cref{prop:dynamic-tw} that is initialized on the edgeless and colorless graph $G'$ with vertex set $V(G)\cup \{v^\star\}$, where $v^\star$ is a fresh vertex that does not belong to $V(G)$.

        \item For $k$ as above, also set $D_{\mathsf{tw}}$ be the data structure given by~\Cref{lem:ds-tw}, initialized also on $G'$ (we indifferently use $G'$ for both data structures $D$ and $D_{\mathsf{tw}}$ as the updates on the edge set of $G'$ will always be the same in both $D$ and $D_{\mathsf{tw}}$).
        
        \item In both data structures $D$ and $D_{\mathsf{tw}}$, do the following updates in $G'$: add in $G'$ all edges of $G$ with both endpoints in $X$, and for every vertex $v \in X$ that has a neighbor in $V(G) \setminus X$, add the edge $vv^\star$ to $G'$.
        Also, in the data structure $D$, add $v_i$ to $U_i^{G'}$ for each $i\leq \ell$ and set
        $I^{G'}\coloneqq \emptyset$, $A^{G'}\coloneqq X\cup \{v^\star\}$, and $M^{G'}\coloneqq \{v^\star\}$.
        \item Initialize $i\coloneqq \ell$.
        \item While $i\leq n$, ask the data structure $D_{\mathsf{tw}}$ whether the treewidth of $G'$ is at most $k$. Depending on the outcome of $D_{\mathsf{tw}}$, do the following.
        \begin{enumerate}
            \item[(a)] If  $D_{\mathsf{tw}}$ reports that the treewidth of $G'$ is at most $k$, then do the following updates:
            \begin{itemize}[nosep]
                \item In both $D$ and $D_{\mathsf{tw}}$, update the edge set of $G'$ as follows:
                \begin{enumerate}
                    \item[(i)] add the edge $v_{i+1}v^\star$ to $G;$
                    \item[(ii)] for every edge $e$ in $G$ incident to $v_{i+1}$ and a vertex $v_j$, where $j\le i$,
                    add $e$ to $G'$; and
                    \item[(iii)] for every edge in $v_jv_{i+1}\in E_{i+1}$ delete the edge $v_jv^\star$ from $G'$;
                \end{enumerate}
                \item in $D$, add $v_{i+1}$ to $A^{G'}$;
                \item set $i\coloneqq i+1$; and continue the loop.
            \end{itemize}
            \item[(b)] If $D_{\mathsf{tw}}$ reports that the treewidth of $G'$ exceeds $k$, then ask the data structure $D$ for a vertex $u$ such that $G'\models\widehat{\varphi}(u)$.
            Depending on the outcome of $D$, do the following.
            \begin{itemize}
                \item[-] If $D$ returns such a vertex $u$,
                add $u$ to $I^{G'}$. Also, in both data structures $D$ and $D_{\mathsf{tw}}$, delete from $G'$ all edges incident to $u$, and continue the loop.
                \item[-] If $D$ reports that no such vertex exists or \textsf{Treewidth too large}, then apply the algorithm of~\Cref{prop:courcelle} with input $G'\setminus (I^{G'}\cup X)$ and an $\mathsf{MSO}_2$-formula encoding a minor model of the apex-grid of order $p$ (see~\Cref{obs:mso-minor}). If the algorithm of~\Cref{prop:courcelle} outputs such a minor model, then extract from it a minor model of the apex-grid of order~$p$ in $G\setminus X$
                and terminate the algorithm by returning this minor model. Here, if there is a branch set in $G'$ that contains $v^\star$ then extend this branch set to $G$ by removing $v^\star$ and including all vertices in $\{v_{i+1},\ldots,v_n\}$; this intuitively corresponds to uncontracting the~megavertex.
            \end{itemize}
        \end{enumerate}
              
        \item Apply the algorithm of~\Cref{prop:courcelle} in $G\setminus I^G$ to compute its $(X,\delta)$-model-folio (see~\Cref{obs:mso-minor}) and terminate the algorithm by returning this outcome.
    \end{enumerate}

\paragraph*{Correctness of the algorithm.}
For the correctness of the above algorithm, we show that the following two invariants are satisfied after each iteration of the loop:
\begin{itemize}
    \item $G\setminus I^G$ has the same $(X,\delta)$-folio as $G$; and
    \item $G'$ has treewidth at most $k+1$.
\end{itemize} 

We first argue that the invariants are maintained if we deal with case (a) of Step 8.
In this case, observe that $I^{G'}$ does not change and therefore the first invariant is trivially maintained.
For the second invariant, we first observe that at the end of the iteration for case (a), the former $G'$ (before this iteration) can be obtained from the updated $G'$ (at the end of this iteration) after contracting the edge $v_{i+1}v^\star$. This implies that the treewidth of $G'$ (after the update) cannot increase by more than one. Therefore, assuming that 
$D_{\mathsf{tw}}$ reports that the treewidth of $G'$ (at the beginning of the considered iteration) is at most $k$, the second invariant is maintained.

Let us now focus on case (b) of Step 8 and prove that the two invariants from above are maintained.
In particular, we show that these invariants are preserved when $G'$ is updated because the data structure~$D$ returns a vertex $u$ such that $\widehat{\varphi}(z)$.

To show the first invariant for case (b) of Step 8, we employ~\Cref{lem:main-subroutine-apexfree}.
We next argue how the setting of the algorithm fits the  interface of~\Cref{lem:main-subroutine-apexfree}.
Recall that
for every $i>\ell$, the graph $G[\{v_{i},\ldots,v_{n}\}]$ is connected.
For every $i\le n$, let $F_i$ be the set of edges with both endpoints in the set $\{v_{i+1},\ldots,v_n\}$ and note that $F_i$ is a connected subset of $E(G)$.
Observe that, at the beginning of each iteration of the loop, $A^G\setminus \{v^\star\}=V(G) \setminus V(F_i)$ and
$G'[A^G]$ is isomorphic to $G/F_i$, via an isomorphism that maps $v^\star$ to $v_F$.
Also, in each iteration of the loop,
$V(F_i)$ is disjoint from both sets $I^G$ and $X$.
By~\Cref{lem:main-subroutine-apexfree}, if there is a vertex $u\in V(G\setminus I^G)$ that does not belong to $V(F_i)\cup X$ such that $(G\setminus  I^G)/F_i\models\varphi(\tup v,v_F,u)$, then $u$ is irrelevant for $(G\setminus I^G,X,\delta)$.
By definition of $\widehat{\varphi}(z)$, this is equivalent to the statement $G'\models \widehat{\varphi}(u)$.
Therefore, assuming that the $(X,\delta)$-folio of $G\setminus I^G$ is the same as the $(X,\delta)$-folio of $G$, we have that the same holds for the $(X,\delta)$-folio of $G\setminus(I^G\cup\{u\})$. This completes the proof of the first invariant for case~(b) of Step 8.
Note that the treewidth bound (second invariant) follows trivially from the fact that, in this case, edges of $G'$ can be only deleted; therefore treewidth cannot increase.

To conclude the proof of correctness of the algorithm, we stress that, in case (b) of Step 8, if the data structure $D$ does not return a vertex $u$ such that $G'\models \widehat{\varphi}(u)$, 
then~\Cref{lem:main-subroutine-apexfree} implies that $G'\setminus (I^G\cup X)$ contains 
the apex-grid of order $p$ as a minor. This is correctly detected by the algorithm of~\Cref{prop:courcelle} in this case.
Also, note that if the algorithm exits the loop
without terminating, then $A^{G'}=V(G)$ and therefore $G'[A^{G'}\setminus I^{G'}]\setminus \{v^\star\} = G\setminus I^{G'}$.
Since the first invariant above is maintained, the algorithm correctly returns the $(X,\delta)$-model-folio of $G$, by computing the $(X,\delta)$-model-folio of $G\setminus I^{G'}$.

\paragraph*{Running time.}
For the running time of the algorithm, we do the following observations. In Step 1, the claimed ordering of $V(G)-X$ can be found in linear time, for example, as a post-order traversal of a spanning tree of $G-X$, while the sets $E_j$, where $j\in [\ell,n]$ can be computed in total time $\Oh(\|G\|)$.
Also, Steps 2-3 can be implemented in time $\Oh_{\ell,\delta,p}(1)$.
The initialization of the data structures $D$ and $D_{\mathsf{tw}}$ in Steps 4-5 can be implemented in time $\Oh_{k,\varphi,\Sigma}(n)$; see~\Cref{prop:dynamic-tw,lem:ds-tw}.
Step 6 can be implemented in time  $\Oh_{\ell,\delta,p}(\|G\| + n^{o(1)})$, since the insertion of each edge in $G'$ takes time $\Oh_{\ell,\delta,p}(n^{o(1)})$ (see~\Cref{prop:dynamic-tw}) and $G'$ initially contains $\Oh(\ell^2)$ edges.
Step 9 also takes time $\Oh_{\ell,\delta,p}(\|G\|)$; see~\Cref{prop:courcelle}.

We conclude the analysis by focusing on the loop of Step 8.
By~\Cref{lem:ds-tw}, asking $D_{\mathsf{tw}}$ to report whether the treewidth of $G'$ is at most $k$ takes time $\Oh_k(1)$.
In order to measure the total number of edge updates in Step 8, we distinguish into two cases: whether an update concerns an edge incident to $v^\star$ or~not.

For each edge incident to $v^\star$,
we observe that it is added at most once: either in Step 6 (making $v^\star$ adjacent to a vertex from $X$) or
in item (i) of Step 8a. Such edges are not added only if the loop of Step 8 stops in case (b). 
Also, such edges $v_jv^\star$ are deleted exactly once, in item (iii) of Step 8a (if $i+1$ is the largest index in the neighborhood of $v_j$, or, alternatively, $v_jv_{i+1}\in E_{i+1}$). Therefore, the total number of updates concerning edges incident to $v^\star$ is at most $2n$.

Let us now focus on the edges that are added or deleted from $G'$ and are not incident to $v^\star$. First observe that such edges are also edges in $G$.
Observe that they each edge of $G$ is added either in Step 6 (if it has both endpoints in $X$) or in item (ii) of Step 8a (if it is incident to $v_{i+1}$).
Edges of $G$ are deleted from $G'$ only in case (b) of Step 8, i.e., if one of the endpoints of this edge is added to $I^{G'}$.
In this case, the indices of both endpoints of this edge are upper bounded by $i+1$ and since the loop continues with $i$ incremented by one, this edge is never added back again.
Therefore, we have a total linear (in $|E(G)|$) number of updates for edges not incident to $v^\star$.
This implies that the total number of edge updates in the data structures $D$ and $D_{\mathsf{tw}}$ is $\Oh(\|G\|)$.
Also, every such update takes amortized time $\Oh_{\ell,\delta,p}(n^{o(1)})$, resulting to total time $\Oh_{\ell,\delta,p}(\|G\|^{1+o(1)})$ in what concerns edge updates in the algorithm.

The remaining analysis for Step 8 is the following.
In case (a), the update of $A^{G'}$ takes amortized time $\Oh_{\ell,\delta,p}(n^{o(1)})$ and can happen at most $n$ times.
In case (b), querying the data structure $D$ takes 
worst-case time $\Oh_{\ell,\delta,p}(1)$ and the update
of $I^{G'}$ takes amortized time $\Oh_{\ell,\delta,p}(n^{o(1)})$. Observe that once a vertex $u$ in removed from $I^{G'}$, the specifications of the formula $\widehat{\varphi}(z)$ imply that such vertex was not in $I^{G'}\cup X\cup\{v^\star\}\cup \{v_{i+1},\ldots,v_n\}$.  
Also, in case (b) of Step 8, the application of~\Cref{prop:courcelle} is done once in total and takes time $\Oh_{\ell,\delta,p}(n)$.
This concludes the analysis of the algorithm.
\end{proof}

\subsection{Existence of irrelevant vertices}
\label{subsec:existence-irre}

In~\cite{GM13}, Robertson and Seymour provide conditions that allow a vertex to be declared irrelevant for the \textsc{Folio} problem. In this subsection, we present the definitions from~\cite{GM13} that are necessary to state their results.

\subsubsection{(Rural) divisions of societies}

The definitions presented here are taken from~\cite{GM13}.

\paragraph{Societies.}
A \emph{society} is a pair $(G,\Omega)$, where $G$ is a graph and $\Omega$ is an ordered subset of $V(G)$, i.e., a set equipped with an ordering\footnote{We stress here that in~\cite{GM13},
societies are defined for $\Omega$ being a \textsl{cyclic} ordering of a subset of $V(G)$. In this paper, we work with \textsl{orderings} and provide the arguments needed in order to fit the framework of~\cite{GM13}. In particular, we refer the reader to the discussion on the definition of attachment schemes later in this subsection.}
$(u_1,\ldots,u_{|\Omega|})$ on its elements.
A society $(G,\Omega)$ is \emph{rural} if $G$ can be drawn in a closed disk
$\Delta$ so that the elements of $\Omega$ are embedded along the boundary of $\Delta$, in an order coinciding with the ordering $(u_1,\ldots,u_{|\Omega|})$.

Let $(G,\Omega)$ be a society.
For a subgraph $H$ of $G$ we define $\partial H$ to be the set of all $v\in V(H)$ such that either $v\in \Omega$ or $v$ is incident with an edge in $E(G)\setminus E(H)$.

\begin{definition}\label{def:division}
Let $(G,\Omega)$ be a society. A \emph{division} of $(G,\Omega)$ is a set $\mathcal{F}$ of subgraphs of $G$, called \emph{flaps}, satisfying the following conditions.
\begin{enumerate}[label=(\arabic*), ref=(\arabic*)]
    \item $\bigcup_{F\in\mathcal{F}}F = G$;\label{div:cond1}
    \item for all distinct $F,F'\in \mathcal{F}$, $E(F)\cap E(F')=\emptyset$ and $V(F)\cap V(F')=\partial F \cap \partial F'$;\label{div:cond2}
    \item if $F,F'\in \mathcal{F}$ are distinct then $\partial F\neq \partial F'$;\label{div:cond3}
    \item for each $F\in\mathcal{F}$ and all $v,u\in \partial F$,
    there is an $(u,v)$-path in $F$ with no internal vertex in $\partial F$;\label{div:cond4}
    \item for each $F\in\mathcal{F}$, there is either a vertex or an edge of $F$ that belongs to no other $F'\in\mathcal{F}$; and\label{div:cond5}
    \item for each $F\in\mathcal{F}$, $|\partial F|\le 3$ and there are $|\partial F|$ pairwise vertex-disjoint paths in $G$ between $\partial F$ and~$\Omega$.\label{div:cond6}
\end{enumerate}
\end{definition}
Note that if $G$ does not contain isolated vertices, then we can assume that in a division $\mathcal{F}$ of $(G,\Omega)$, every flap $F\in\mathcal{F}$ has a distinct edge that belongs to no other $F'\in\mathcal{F}$ (see condition~\ref{div:cond5} of~\Cref{def:division}).

Let $(G,\Omega)$ be a society and let $\mathcal{F}$ be a division of $(G,\Omega)$.
We write $\partial \mathcal{F}$ to denote $\bigcup_{F \in \mathcal{F}} \partial F$.
Note that $\Omega\subseteq \partial \mathcal{F}$.
We use $G_{\mathcal{F}}$ to denote the bipartite graph with vertex set $\mathcal{F}\cup\partial \mathcal{F}$, in which $F\in\mathcal{F}$ is adjacent to $v\in\partial\mathcal{F}$ if $v\in\partial F$.
We say that $\mathcal{F}$ is \emph{rural} if $(G_{\mathcal{F}},\Omega)$ is rural.

\paragraph{Walls.}
Let $h$ be an even integer with $h\ge 2$.
An \emph{elementary wall of height $h$}
is the graph with the vertex set $\{0,\ldots,2h+1\}\times \{0,\ldots,h\}\setminus \{(0,0),(2h+1,h)\}$ and an edge between any vertices 
$(i,j)$ and $(i',j')$
if either
\begin{itemize}[nosep]
    \item $|i-i'|=1$ and $j=j'$; or
    \item $i=i'$, $|j-j'|=1$ and $i$ and $\max\{j,j'\}$ have the same parity; see~\Cref{fig:wall}.
\end{itemize}
The \emph{perimeter} of an elementary wall is the cycle consisting of all vertices $(i,j)$ with $i\in\{0,1,2h,2h+1\}$ or $j\in\{0,h\}$ (except $(0,0)$ and $(2h+1,h)$) and the edges between them.
Given an elementary wall $\overline{W}$, its \emph{corners} are the vertices $(1,0), (2h+1,0),(2h,h),(0,h)$.
Its \emph{central vertices} are the vertices $(h,h/2+1),(h+1,h/2+1)$. See~\Cref{fig:wall} for an illustration of an elementary wall, its corners, and its central vertices.
Given an odd $i\ge 1$, the \emph{$i$-th vertical path} of  $\overline{W}$ is the path of $\overline{W}$ with vertices $(i,j-1),(i,j),(i-1,j),(i-1,j+1)$ for every odd $j\ge 1$.
Given an $i\ge 2$, the \emph{$i$-th horizontal path} of $\overline{W}$ is the path of $\overline{W}$ with vertices $(0,i-1),\ldots,(2h+1,i-1)$.

Let $\overline{W}$ be an elementary wall of height $h$, for even $h\ge 2$.
Let $h'$ be even with $2\le h'\le h$ and let $i,j$ be even non-negative integers with $0\le i\le 2h-2h'$
and $0\le j\le h-h'$.
Also, let $\overline{W}'$ be the subgraph of $\overline{W}$ induced by the vertex set $\{(x,y): i\le x\le i+2h'+1, j\le y \le j+h'\}\setminus \{(i,j),(i+2h'+1,j+h')\}$.
Note that $\overline{W}'$ is an elementary wall of height $h'$ consisting of ``consecutive'' vertices of $\overline{W}$, with corners $(i+1,j),(i+2h'+1,j),(i,j+h'),(i+2h,j+h')$.
We call $\overline{W}'$ a \emph{subwall of $\overline{W}$ of height $h'$}. See~\Cref{fig:wall} for an example of a subwall.
Given an even $h'$ with $2\le h'\le h$,
we define the \emph{central $h'$-subwall} of $\overline{W}$ to be the subwall of $\overline{W}$ of height $h'$ that does not intersect the first and last $(h-h')/2$ horizontal and vertical paths of $\overline{W}$.

\begin{figure}[ht]
    \centering
    \begin{tikzpicture}[scale=0.6]

        \pgfmathtruncatemacro\h{6}
        \pgfmathtruncatemacro\hh{2*\h+1}

        \foreach \i in {0,...,\hh}{
            \foreach \j in {0,...,\h}{
            \pgfmathtruncatemacro\z{\i+\j}
            \pgfmathtruncatemacro\y{\i+\j}
            \pgfmathtruncatemacro\o{\hh+\h}
    
            \ifnum \z>0
                \ifnum \y<\o
                \node[node] () at (\i,\j) {};
            \fi
            \fi}}

        \begin{scope}[on background layer]
    
        \foreach \i/\j in {1/0,\hh/0,\hh-1/\h,0/\h}{
            \node[corner] () at (\i,\j){};}
    
        \node[central] at (\h,\h/2){};
        \node[central] at (\h+1,\h/2){};
        
        \pgfmathtruncatemacro\hhminus{\hh-1}
        \foreach \i in {0,...,\hhminus}{
            \foreach \j in {0,...,\h}{
            \pgfmathtruncatemacro\z{\i+\j}
            \pgfmathtruncatemacro\y{\i+\j}
            \pgfmathtruncatemacro\o{\h+\hhminus}
            \ifnum \z>0
                \ifnum \y<\o
                \draw[edge] (\i,\j) -- (\i+1,\j);
            \fi
            \fi}}
    
            \pgfmathtruncatemacro\hminus{\h-1}
           \foreach \i in {0,...,\hh}{
            \foreach \j in {0,...,\hminus}{
            \pgfmathtruncatemacro\z{mod(\i,2)}
            \pgfmathtruncatemacro\y{mod(\j,2)}
            \ifnum \z=1
                \ifnum \y =0
                    \draw[edge] (\i,\j) -- (\i,\j+1);
                \fi
            \else
                \ifnum \y=1
                    \draw[edge] (\i,\j) -- (\i,\j+1);
                \fi
            \fi
         }} 
         \end{scope}
    
         \begin{scope}
             \pgfmathtruncatemacro\h{4}
        \pgfmathtruncatemacro\hh{2*\h+1}
    
        \foreach \i in {0,...,\hh}{
            \foreach \j in {0,...,\h}{
            \pgfmathtruncatemacro\z{\i+\j}
            \pgfmathtruncatemacro\y{\i+\j}
            \pgfmathtruncatemacro\o{\hh+\h}
    
            \ifnum \z>0
                \ifnum \y<\o
                \node[nodesubwall] () at (\i,\j) {};
            \fi
            \fi}}
    
        \begin{scope}[on background layer]
        \pgfmathtruncatemacro\hhminus{\hh-1}
        \foreach \i in {0,...,\hhminus}{
            \foreach \j in {0,...,\h}{
            \pgfmathtruncatemacro\z{\i+\j}
            \pgfmathtruncatemacro\y{\i+\j}
            \pgfmathtruncatemacro\o{\h+\hhminus}
            \ifnum \z>0
                \ifnum \y<\o
                \draw[edgesubwall] (\i,\j) -- (\i+1,\j);
            \fi
            \fi}}
    
            \pgfmathtruncatemacro\hminus{\h-1}
           \foreach \i in {0,...,\hh}{
            \foreach \j in {0,...,\hminus}{
            \pgfmathtruncatemacro\z{mod(\i,2)}
            \pgfmathtruncatemacro\y{mod(\j,2)}
            \ifnum \z=1
                \ifnum \y =0
                    \draw[edgesubwall] (\i,\j) -- (\i,\j+1);
                \fi
            \else
                \ifnum \y=1
                    \draw[edgesubwall] (\i,\j) -- (\i,\j+1);
                \fi
            \fi
         }} 
         \end{scope}
         \end{scope}
    
         \end{tikzpicture}
         \caption{An elementary wall $\overline{W}$ of height six and a subwall of $\overline{W}$ of height four. The corners and the central vertices of $\overline{W}$ are depicted as squared vertices.}
         \label{fig:wall}
\end{figure}
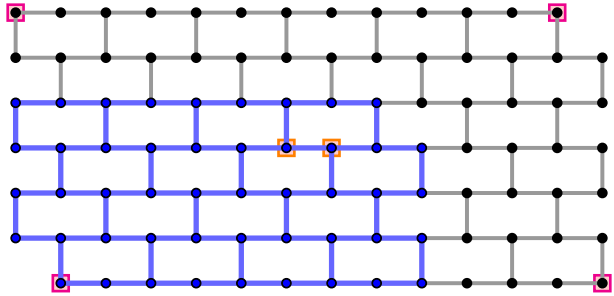

A \emph{wall of height $h$} (for even $h\geq 2$) is a subdivision $W$ of an elementary wall $\overline{W}$ of height $h$.
We always assume that $W$ is given with a choice of an (ordered) set of corners $\Omega_W$, 
that correspond to the corners of the elementary wall $\overline{W}$, ordered in the clockwise ordering $(0,h),(2h,h),(2h+1,0),(1,0)$. 
Given a wall $W$ that is a subdivision of an elementary wall $\overline{W}$,
its \emph{perimeter}, denoted by $\mathsf{Per}(W)$, is the cycle of $W$ whose non-subdivision vertices are the vertices of the perimeter of $\overline{W}$.
The \emph{central vertices} of $W$ are the central vertices of $\overline{W}$.
Given a wall $W$ that is a subdivision of an elementary wall $\overline{W}$,
a \emph{subwall $W'$ of $W$ of height $h'$}
is a subdivision of a subwall $\overline{W}'$ of $\overline{W}$ of height $h'$.
The notions of \emph{horizontal/vertical paths} and \emph{central subwalls} of a wall are defined analogously.
Given a graph $G$, a \emph{wall in $G$} is a subgraph of $G$ which is isomorphic to a wall.
We restate here the \emph{Grid Minor Theorem} (\Cref{thm:grid-minor}) stated for walls instead of grids.

\begin{theorem}\label{prop:grid-thm}
    There exists a computable function $g\colon \N\to \N$ such that for every even $h\ge 2$, every graph of treewidth at least $g(h)$ contains a wall of height $h$ as a subgraph.
\end{theorem}

\paragraph{Compass and society of a wall.}
Let $W$ be a wall in a graph $G$.
Let $Y$ be the union of $V(\mathsf{Per}(W))$ and the vertex set of the unique connected component of $G\setminus V(\mathsf{Per}(W))$ that contains a vertex of $W$; uniqueness follows from the fact that $W\setminus V(\mathsf{Per}(W))$ is connected and non-empty.
We define the \emph{compass} of $W$ in $G$, which we denote by $\mathsf{Compass}_G(W)$, to be the graph $G[Y]$.
It follows that $W$ is a subgraph of its compass and its compass is connected.
The \emph{society} of $W$ in $G$ is the society $(\mathsf{Compass}_G(W),\Omega_W)$, where $\Omega_W$ is the (ordered) set of corners of $W$.

\paragraph{Flat walls.}
Let $G$ be a graph and let $W$ be a wall in $G$.
We say that $W$ is \emph{flat} in $G$ if the society of $W$ in $G$ admits a rural division.
The existence of a flat wall in a minor-free graph of large treewidth is guaranteed by the \emph{Flat Wall Theorem}, which was proved in~\cite[(9.8)]{GM13}; see also~\cite{KawarabayashiTW18,GiannopoulouT13,Chuzhoy15}. We use here the version with multiple flat walls of Robertson and Seymour~\cite[(9.7)]{GM13}.
Computability of the constants in the results from~\cite{GM13}, including (9.7) and the minor-testing algorithm, is argued at the end of that paper. See also~\cite[Appendix A]{LokshtanovPPS22a} for a discussion of explicit bounds for~\cite[(3.5)]{RobertsonS88}, which is a major non-constructive argument appearing in the Graph Minors series, mentioned at the end of~\cite{GM13}.

\begin{theorem}[{\cite[(9.7)]{GM13}}]\label{prop:many-flat-walls}
    For all $t,q\ge 0$ and even $h\ge 2$, there is a constant $h'$, computable from $t,q$ and $h$, such that the following holds. Let $G$ be a graph and let $W$ be a wall of height $h'$ in $G$. Then either $G$ contains $K_t$ as a minor or
    there is a vertex subset $A\subseteq V(G)$ with $|A|<\binom{t}{2}$ and a collection $\mathcal{W}$ of $q$ subwalls of $W$ of height $h$,
    whose compasses are pairwise vertex-disjoint and disjoint from $A$, so that each subwall in $\mathcal{W}$ is flat in $G\setminus A$.
\end{theorem}

\subsubsection{Homogeneous subwalls and irrelevant vertices}
In this subsection, we provide additional definitions needed in order to extend rural divisions to include a set $X$ of fixed size, as well as a way to define an ordering on the boundary of each flap.
These additional technical tools are used in order to define homogeneous subwalls. These definitions are taken from Section~10 of~\cite{GM13} but the way that they are presented here differs slightly from the ones of~\cite{GM13}.
The reason behind this reformulation is to make a clear connection on how to define all such notions in logic, which is done in the next subsection.

\paragraph{Extension schemes.}
Let $(G,\Omega)$ be a society, let $X\subseteq V(G)\setminus \Omega$, and let $\mathcal{F}$ be a rural division of $(G\setminus X,\Omega)$.
Before defining the notion of an \textsl{attachment scheme}, we define \textsl{$X$-extension schemes} and \textsl{boundary ordering schemes}.

Let $F\in\mathcal{F}$.
An \emph{edge interface} of $F$,
denoted by $E_{F}$,
is a subset of edges of $G$ with one endpoint in $X$ and one endpoint in $\partial F$.
An \emph{$X$-extension scheme} of $\mathcal{F}$ is a choice of $E_F$ for every $F\in\mathcal{F}$,
such that $\{E_F\cup E(F):F\in\mathcal{F}\}$ is a partition\footnote{We note that in \cite{GM13}, the analogous definition is stated for partitions of the whole set $E(G)$. In order to fit the setting of~\cite{GM13}, one could extended the current definition to partitions of $E(G)$ as follows: assume the existence of a flap $F\in\mathcal{F}$ with $|\partial F|=1$ and $\partial F\subseteq \Omega$ and demand
that every edge with both endpoints in $X$ belongs to $E_F$; if such $F$ does not exist, we add it to $\mathcal{F}$ and observe that it preserves the property of being a (rural) division.}
of $E(G)\setminus \binom{X}{2}$.
Given an $X$-extension scheme $\xi$ of $\mathcal{F}$ and a flap $F\in\mathcal{F}$, we use $\xi(F)$ to denote the graph $(V(F)\cup X,E(F)\cup E_F)$.

\paragraph{Boundary ordering schemes.}
For each flap $F\in\mathcal{F}$ with  $|\partial F|= 3$ we need to define a particular type of orderings of $\partial F$, which is described in the next paragraph.

Let $F\in\mathcal{F}$ with $|\partial F|= 3$.
Also let $\mathcal{P}$ be a collection of three pairwise vertex-disjoint paths of $G\setminus X$ from $\partial F$ to $\Omega$ and observe that such collection $\mathcal{P}$ exists by condition~\ref{div:cond6} of~\Cref{def:division}.
Recall that $\Omega$ is an ordered set.
For an ordering $\tau$ of a subset of $\Omega$, we say that it is \emph{aligned} with $\Omega$ if $\tau$ coincides with the order its elements appear in the cyclic ordering induced\footnote{Note that every ordering $(a_1,\ldots,a_m)$ of a set induces a unique cyclic ordering of the same set, where $a_1$ comes after $a_m$.} by $\Omega$.
We use $T_\Omega(\mathcal{P})$ be the set of the endpoints of the paths in $\mathcal{P}$ in $\Omega$.
Given an ordering $(s_1,s_2,s_3)$ of $\partial F$, the \emph{ordering of $T_\Omega(\mathcal{P})$ given by $(s_1,s_2,s_3)$} is
the ordering $(t_1,t_2,t_3)$ of $T_\Omega(\mathcal{P})$ where for each $i\in\{1,2,3\}$, $s_i$ and $t_i$ are connected via a path in $\mathcal{P}$. 
We say that an ordering $\lambda$ of $\partial F$ is \emph{$\mathcal{P}$-compatible} with $\Omega$ if the ordering of $T_\Omega(\mathcal{P})$ given by $\lambda$ is aligned with~$\Omega$.
We say that an ordering $\lambda$ of $\partial F$ is \emph{compatible} with $\Omega$, if it is $\mathcal{P}$-compatible with $\Omega$ for every 
collection~$\mathcal{P}$ of three pairwise vertex-disjoint paths of $G\setminus X$ from $\partial F$ to $\Omega$.

We define a \emph{boundary ordering} of a flap $F\in\mathcal{F}$, denoted by $\lambda_F$, to be an ordering $(v_1,\ldots,v_{|\partial F|})$ of $\partial F$ which in the case $|\partial F|=3$ is also compatible with $\Omega$.
A \emph{boundary ordering scheme} of $\mathcal{F}$
is a choice of $\lambda_F$ for each $F\in\mathcal{F}$.
We stress here that for every $F\in\mathcal{F}$,
in order to fix a boundary ordering of $F$, it suffices to fix the first element $v_1$.
This is proved in the following lemma.

\begin{lemma}\label{lem:boundary-ordering}
    Let $(G,\Omega)$ be a society, let $X\subseteq V(G)\setminus \Omega$, and let $\mathcal{F}$ be a rural division of $(G\setminus X,\Omega)$.
    Also, let $F\in\mathcal{F}$ such that $|\partial F|=3$.
    For every $v\in \partial F$, there exists a unique ordering $(s_1,s_2,s_3)$ of $\partial F$ that is compatible with $\Omega$ such that $s_1=v$.
\end{lemma}

\begin{proof}
    Let $v\in \partial F$.
    The existence of an ordering $\lambda=(s_1,s_2,s_3)$ of $\partial F$ that is $\mathcal{P}$-compatible with $\Omega$ and $s_1=v$, for some collection $\mathcal{P}$ of three pairwise vertex-disjoint paths in $G\setminus X$ from $\partial F$ to $\Omega$, follows from condition~\ref{div:cond2} of~\Cref{def:division}.
    We next show that $\lambda$ is $\mathcal{P}$-compatible with $\Omega$ for every such collection $\mathcal{P}$.
    
    \begin{claim}\label{claim:compati}
        Let $\mathcal{P}$, $\mathcal{P}'$ be collections of paths in $G\setminus X$, each containing three pairwise vertex-disjoint paths from $\partial F$ to $\Omega$.
        If $\lambda$ is an ordering of $\partial F$ that is $\mathcal{P}$-compatible with $\Omega$, then it is also $\mathcal{P}'$-compatible with $\Omega$.
    \end{claim}
    \begin{claimproof}
    We set $\lambda=(s_1,s_2,s_3)$.
    Let also $(t_1,t_2,t_3)$ and $(t_1',t_2',t_3')$ be
    the orderings of $T_\Omega(\mathcal{P})$ and $T_\Omega(\mathcal{P}')$ given by $\lambda$, respectively. For every $i\in\{1,2,3\}$, let $P_i'$ be the path in $\mathcal{P}$ with endpoints $s_i,t_i'$.

    Since $\mathcal{F}$ is rural,
    the graph $G_{\mathcal{F}}$ (recall that this is the bipartite graph with vertex set $\mathcal{F}\cup\partial \mathcal{F}$, in which $F\in\mathcal{F}$ is adjacent to $v\in\partial\mathcal{F}$ if $v\in\partial F$) can be drawn in a closed disk $\Delta$ so that the elements of $\Omega$ are embedded along the boundary of $\Delta$ according to the cyclic ordering induced by $\Omega$.
    We fix such an embedding $\gamma$ of $G_{\mathcal{F}}$ in $\Delta$.
    Consider a circle $C$ in $\Delta$ (i.e., a subset of $\Delta$ that is homeomorphic to $\{(x,y)\in\mathbb{R}^2\mid x^2+y^2=1\}$)
    that intersects $\gamma$ only in $\partial F$.
    Assume without loss of generality that the cyclic ordering induced by $\Omega$ is counter-clockwise in the boundary of $\Delta$ in the embedding $\gamma$.
    Then, the fact that $\lambda$ is $\mathcal{P}$-compatible with $\Omega$ implies that $\lambda$ induces a counter-clockwise cyclic ordering of $\partial F$ on $C$.

    We show that $\lambda$ is $\mathcal{P}'$-compatible with $\Omega$ by arguing that
    $(t_1',t_2',t_3')$ induces a counter-clockwise cyclic ordering of $T_\Omega(\mathcal{P}')$ on the boundary of $\Delta$.
    In our argument, we use ``projections'' of the paths of $\mathcal{P}$ in~$G_{\mathcal{F}}$.
    More formally,
    for each $i\in\{1,2,3\}$,
    we use $\widehat{P}_i'$ to denote 
    the path of $G_{\mathcal{F}}$ obtained from $P_i'$ after performing the following replacement for each $F'\in\mathcal{F}$ that contains an edge of $P_i'$: replace the subpath $H$ of $P_i
    $ in $F'$ whose endpoints, say $x,y$, are in $\partial F'$ with the corresponding path in $G_{\mathcal{F}}$ with vertices $x,F',y$.
    Observe that the fact that for each flap $F'\in\mathcal{F}$, $|\partial F'|\leq 3$, implies that $\widehat{P}_1',\widehat{P}_2',\widehat{P}_3'$ are pairwise vertex-disjoint~in~$G_{\mathcal{F}}$.

    To conclude the proof of the claim, it suffices to prove that for every distinct $i,j,k\in\{1,2,3\}$ such that $s_i,s_j,s_k$ appear in this order when traversing $C$ counter-clockwise, it holds that $t_i',t_j',t_k'$ appear in this order when traversing the boundary of $\Delta$ counter-clockwise.
    To show this, let $I_{i,j}$ be the closed arc in the boundary of $\Delta$ with endpoints $t_i'$ and $t_j'$ and note that the union (of the images via $\gamma$) of the path with vertices $s_iFs_j$ and the paths $\widehat{P}_i',\widehat{P}_j'$ and the arc $I_{i,j}$ bounds a subset $\Delta'$ of $\Delta$ that is a closed disk.
    Since $\widehat{P}_1',\widehat{P}_2',\widehat{P}_3'$ are pairwise vertex-disjoint, $\Delta'$ does not contain (the image via $\gamma$ of) $t_k'$.
    Therefore $t_i',t_j',t_k'$ appear in this order when traversing the boundary of $\Delta$ counter-clockwise. This implies that  $(t_1',t_2',t_3')$ induces a counter-clockwise cyclic ordering of $T_\Omega(\mathcal{P}')$ on the boundary of $\Delta$, which in turn implies that $\lambda$ is $\mathcal{P}'$-compatible with $\Omega$.
    \end{claimproof}

    By~\Cref{claim:compati}, we derive that $\lambda$ is compatible with $\Omega$. To prove uniqueness, consider a distinct ordering $\lambda'=(s_1',s_2',s_3')$ that is compatible with $\Omega$ and $s_1'=v$.
    Also pick a collection of paths $\mathcal{P}$ as above.
    Since both $\lambda$ and $\lambda'$ are $\mathcal{P}$-compatible with $\Omega$, the ordering of $T_{\Omega}(\mathcal{P})$ given by $\lambda$ is the same as the one given by~$\lambda'$, implying that $s_i=s_i'$ for each $i\in\{1,2,3\}$.
\end{proof}

\paragraph{Attachment schemes.}
We now define \emph{attachment schemes} as triples consisting of an extension scheme, a boundary ordering scheme, and an ordering of $X$.

\begin{definition}
    Let $(G,\Omega)$ be a society, let $X\subseteq V(G)\setminus V(\Omega)$, and let $\mathcal{F}$ be a rural division of $(G\setminus X,\Omega)$.
An \emph{attachment scheme $\sigma$ of $(\mathcal{F},X)$} is a triple $(\xi,\Lambda,\rho_X)$, where 
$\xi$ is an $X$-extension scheme of $\mathcal{F}$,
$\Lambda$ is a boundary ordering scheme of $\mathcal{F}$,
and $\rho_X$ is an ordering of $X$.
\end{definition}

Let us comment on the following.
The notion of \emph{attachment schemes} roughly corresponds to the notion of \emph{visions} defined in~\cite{GM13}. However, in~\cite{GM13}, a vision also encodes an ``anchor'' vertex for each flap (which intuitively corresponds to a vertex of the wall that is the closest to the flap).
In our work, we do not fix such ``anchor vertices''. Instead, we existentially quantify one such vertex, using the notion of a flap being \emph{captured by a subwall} defined in the next paragraph. Note that this is also the situation in the statement in~\cite[(8.3)]{DBLP:journals/jct/RobertsonS12} which implies~\Cref{prop:existence-irrelevant} given below.

Given an attachment scheme $\sigma=(\xi,\Lambda,\rho_X)$,
for each $F\in\mathcal{F}$, we use $\omega(F)$ to denote the ordering of $\partial F\cup X$, where the vertices of $\partial F$ appear first, in the order given by $\lambda_F$ and are followed by the vertices of $X$, ordered by $\rho_X$.
We use $\mathbf{F}_\sigma$ to denote the graph $\xi(F)$ where the set $\partial F\cup X$ is ordered via $\omega(F)$.

Let $F,F'\in\mathcal{F}$. Also, let $\delta\ge 0$.
We say that  $\mathbf{F}_\sigma$ and $\mathbf{F}_\sigma'$ 
\emph{have the same $\delta$-folio} if
there is a bijection $\rho$ from the $(X\cup \partial F,\delta)$-folio of $\xi(F)$ to the $(X\cup \partial F',\delta)$-folio of $\xi(F')$ mapping every $(X\cup \partial F)$-rooted minor
of $\xi(F')$ of detail at most $\delta$ to an \textsl{isomorphic} $(X\cup \partial F')$-rooted minor of $\xi(F')$ of detail at most $\delta$ via an isomorphism that respects the orderings $\omega(F)$ and $\omega(F')$.

\paragraph{Homogeneity.}
Let $G$ be a graph and let $W$ be a wall in $G$.
Let $\mathcal{F}$ be a rural division of the society of $W$ in $G$. We say that a flap $F\in\mathcal{F}$ is \emph{$W$-internal} if $\partial F\cap V(\mathsf{Per}(W))=\emptyset$.
Let $W'$ be a subwall of $W$.
We say that a flap $F\in\mathcal{F}$ is \emph{captured} by $W'$ if there is a vertex $v\in V(W')$ and a path in $\mathsf{Compass}_G(W)$ from $\partial F$ to $v$ with no vertex in $V(W)$ except $v$.
Note that the fact that $\mathsf{Compass}_G(W)$ is connected
implies that for every flap $F\in\mathcal{F}$ there is a vertex $v\in V(W)$ and a path in $\mathsf{Compass}_G(W)$ from $\partial F$ to $v$ with no vertex in $V(W)$ except $v$.
The following definition of homogeneity corresponds to the one of~\cite{GM13}.

\begin{definition}\label{def:homogeneity}
Let $G$ be a graph, let $X\subseteq V(G)$, let $W$ be a wall in $G\setminus X$,
and let $\mathcal{F}$ be a rural division of the society of $W$ in $G\setminus X$.
Also, let $\sigma$ be an attachment scheme of $(\mathcal{F},X)$.
Given a $\delta\ge 0$ and an even $h\ge 2$,
we say that a subwall $W'$ of $W$ is \emph{$(\delta,h)$-homogeneous
with respect to $(\mathcal{F},\sigma)$} if the following holds:
for every $W$-internal flap $F\in \mathcal{F}$ that is captured by $W'$ and for every subwall $W''$ of $W'$ of height $h$,
there is a $W$-internal flap $F'\in\mathcal{F}$ that is captured by $W''$ such that $\mathbf{F}_\sigma$ and $\mathbf{F}_\sigma'$ have the same $\delta$-folio.
\end{definition}

The next result indicates that a central vertex of a sufficiently large homogeneous wall is \emph{irrelevant}. 


\begin{theorem}[{\cite[(10.2)]{GM13}; see also \cite[(8.3)]{DBLP:journals/jct/RobertsonS12}}]\label{prop:existence-irrelevant}
    There is a computable function $f:\mathbb{N}^3\to\mathbb{N}$, whose images are even and $f(x,y,z)\ge z$,
    such that for every $\ell,\delta\ge 0$ and even $h\ge 2$,
    the following holds.
    Let $G$ be a graph, let $X$ be a set of vertices of $G$ of size $\ell$, let $W$ be a wall in $G\setminus X$, let $\mathcal{F}$ be a rural division of the society of $W$ in $G\setminus X$, and let $\sigma$ be an attachment scheme of $(\mathcal{F},X)$.
    Also, let $W'$ be a subwall of $W$ that is $(\delta,h)$-homogeneous with respect to $(\mathcal{F},\sigma)$ and of height $f(\ell,\delta,h)$.
    If $v$ is a central vertex of $W'$, then
    the $(X,\delta)$-folio of $G\setminus v$ is the same as the $(X,\delta)$-folio of $G$.
\end{theorem}

In order to find a large homogeneous subwall, one has to start with a large enough wall.
In~\cite[(10.3)]{GM13}, Robertson and Seymour show how to obtain a sufficiently large homogeneous subwall $W'$ inside some large enough wall $W$, when the relation between the homogeneity of $W'$ and its height is given by a fixed function $f$ (which in fact is the function from~\Cref{prop:existence-irrelevant}).
By allowing $f$ to be an arbitrary function from $\mathbb{N}^3$ to $\mathbb{N}$ whose images are even and $f(x,y,z)\ge z$, it is easy to derive the following statement from the proof of~\cite[(10.3)]{GM13}.

\begin{theorem}\label{prop:existence-homogeneous}
    For every $\ell,\delta\ge 0$ and every computable function $f:\mathbb{N}^3\to\mathbb{N}$ whose images are even and $f(x,y,z)\ge z$, there are even $r,h\ge 2$, computable from $\ell,\delta$, and a given Turing machine computing $f$, such that the following holds.
    Let $G$ be a graph, let $X\subseteq V(G)$ of size $\ell$, let $W$ be a wall in $G\setminus X$, let $\mathcal{F}$ be a rural division of the society of $W$ in $G\setminus X$, and let $\sigma$ be an attachment scheme of $(\mathcal{F},X)$.
    If the height of $W$ is at least $r$, then there is a subwall $W'$ of $W$ that is $(\delta,h)$-homogeneous with respect to $(\mathcal{F},\sigma)$ and of height~$f(\ell,\delta,h)$.
\end{theorem}

After assembling all ingredients for expressing the existence of an irrelevant vertex in a graph, we next claim that all above conditions can be expressed in logic.

\begin{restatable}{lemma}{formulairrelevant}\label{lem:formula-irrelevant}
    For every $\ell,\delta\ge 0$ and every even $r,r',h\ge 2$,
    there is an $\mathsf{MSO}_2$-formula $\mathsf{irr}(\tup x,y,z)$,
    where $\tup x$ is a tuple of $\ell$ single vertex variables and $y$ and $z$ are single vertex variables,
    such that the following holds.~Given
    \begin{itemize}[nosep]
        \item a graph $G$ and a vertex set $X\subseteq V(G)$, where $X=\{\tup v(x):x\in\tup x\}$ for some $\tup v\in G^{\tup x}$; and
        \item vertices $u_0,u\in V(G)$;
    \end{itemize}
    then we have $G\models \mathsf{irr}(\tup v,u_0,u)$
    if and only if all of the following conditions hold.
    \begin{itemize}[nosep]
        \item There is a wall $W$ of height $r$ in $G\setminus X$ such that $u_0$ is disjoint from the compass of $W$ in $G\setminus X$.
        \item There is a rural division $\mathcal{F}$ of the society of $W$ in $G\setminus X$ and an attachment scheme $\sigma$ for $(\mathcal{F},X)$.
        \item There is a subwall $W'$ of $W$ of height $r'$ that is $(\delta,h)$-homogeneous with respect to $(\mathcal{F},\sigma)$.
        \item $u$
        is a central vertex of $W'$.
    \end{itemize}
\end{restatable}

Note that the proof of~\Cref{lem:formula-irrelevant} is not trivial, as some of the definitions provided above are not directly expressible by $\mathsf{MSO}_2$-formulas. For example, see~\Cref{def:division}, where divisions are defined as collections of arbitrarily many subgraphs.
In~\Cref{subsec:defirvlogic} we explain how to define 
all above conditions for situating an irrelevant vertex in $\mathsf{MSO}_2$ and we present the proof of~\Cref{lem:formula-irrelevant}.
Before this, we show how to prove~\Cref{lem:main-subroutine-apexfree} assuming~\Cref{lem:formula-irrelevant}.
We restate~\Cref{lem:main-subroutine-apexfree} for convenience.

\subroutineapexfree*

\begin{proof}
The formula of~\Cref{lem:main-subroutine-apexfree} is set to be the formula of~\Cref{lem:formula-irrelevant} by fixing $r$ and $h$ to be the constants obtained from~\Cref{prop:existence-homogeneous} when setting $f$ to be the function of~\Cref{prop:existence-irrelevant} and $r'\coloneqq f(\ell,\delta,h)$.
We use $G'$ to denote $G/F$.

Assume that $G'\models \mathsf{irr}(\tup v,v_F,u)$.
Let $W$ be a wall of height $r$ in $G'\setminus X$ such that $v_{F}$ does not belong to the vertex set of the compass of $W$ in $G'\setminus X$ and the society of $W$ in $G'\setminus X$ admits a rural division $\mathcal{F}$.
Note that in this case $\mathsf{Compass}_{G'\setminus X}(W) = \mathsf{Compass}_{G\setminus X}(W)$ and in particular
$\mathcal{F}$ is a rural division of the society of $W$ in $G\setminus X$.
Therefore, by~\Cref{prop:existence-homogeneous,prop:existence-irrelevant,lem:formula-irrelevant}, the $(X,\delta)$-folio of $G\setminus u$ is the same as the $(X,\delta)$-folio of $G$.

To show the ``moreover'' part of~\Cref{lem:main-subroutine-apexfree}, let $p\in\mathbb{N}$.
We set $$t\coloneqq p^2+1, \qquad d\coloneqq p^4, \qquad\textrm{and}\qquad q\coloneqq d\cdot \binom{t}{2}+1.$$ 
Let $h'$ be the constant provided by~\Cref{prop:many-flat-walls} for the parameters $q$, $t$, and $r$ and let $$\hat{h}\coloneqq h'+2p.$$
Also, let $w$ be the constant provided by~\Cref{prop:grid-thm} for the parameter $\hat{h}$ and let $$k\coloneqq w+|X|.$$
As in the statement of the lemma, we assume that $G'\setminus X$ does not contain the apex-grid of order $p$ as a minor and that the treewidth of $G'$ is larger than $k$.
First observe that the treewidth of $G'\setminus X$ is larger than $k-|X|=w$.
Then from~\Cref{prop:grid-thm} we infer that there is a wall $\widehat{W}$ of height $\hat{h}$ in $G'\setminus X$. Recall that $\hat{h}=h'+2p$.
Let $W$ be the central $h'$-subwall of $\widehat{W}$.

Since $G'\setminus X$ does not contain the apex-grid of order $p$ as a minor, it also does not contain the complete graph $K_{t}$ as a minor, for $t=p^2+1$.
Therefore,~\Cref{prop:many-flat-walls} gives a vertex subset $A\subseteq V(G'\setminus X)$ where $|A|<\binom{t}{2}$ and a collection $\mathcal{W}$ of $q$ subwalls of $W$ of height $r$, with the following additional properties:
\begin{itemize}[nosep]
    \item for every two distinct subwalls in $\mathcal{W}$, their compasses in $G'\setminus X$ are pairwise vertex-disjoint and disjoint from $A$; and
    \item for each subwall in $\mathcal{W}$, its society in $G'\setminus (X\cup A)$ admits a rural division.
\end{itemize}
We show that there are at least two distinct subwalls $W_1,W_2$ in $\mathcal{W}$ such that there is no vertex in $A$ that is adjacent to the compass of either of them in $G'\setminus X$.

Let $\mathcal{W}'$ be a collection of $d$ subwalls in $\mathcal{W}$. 
Assume, towards a contradiction, that there is a vertex $v\in A$
that is adjacent
to at least one vertex of the compass (in $G'\setminus X$) of each subwall in $\mathcal{W}'$.
Recall that~$W$ is the central $h'$-subwall of the wall $\widehat{W}$ and that $\widehat{W}$ has height $\hat{h}=h'+2p$.
We contract all edges of~$\widehat{W}$ that belong to the intersection of horizontal and vertical paths, as well as every edge of the compass of $\widehat{W}$ that is not in $E(\widehat{W})$.
Observe that this way we obtain the $\hat{h}\times \hat{h}$ grid as a minor of $\widehat{W}$ and moreover its central $h' \times h'$ subgrid (which intuitively corresponds to the contraction of $W$) contains a vertex set~$X$ of size $d=p^4$ such that every vertex in $X$ is adjacent to $v$.
By~\Cref{prop:grid-in-grid}, we infer that $G'\setminus X$ contains the apex-grid of order $p$ as a minor, a contradiction to our initial assumption.
Therefore, for every vertex $v$ of $A$, the number of subwalls in $\mathcal{W}$ whose compass (in $G'\setminus X$) contains a vertex adjacent to $v$ is less than $d$.
Since $|\mathcal{W}|=q$ and $|A|< \binom{t}{2}$, we conclude that there are at least two distinct subwalls in $\mathcal{W}$ such that no vertex in $A$ is adjacent to a vertex in their compasses in $G'\setminus X$.

Let $W_1$ and $W_2$ be such subwalls in $\mathcal{W}$. Observe that the society of each of these walls in $G'\setminus X$ admits a rural division. 
Since $\mathsf{Compass}_{G'\setminus X}(W_1)$ and $\mathsf{Compass}_{G'\setminus X}(W_2)$ have no vertices in common,
for one of them, say $W_1$, its compass in $G'\setminus X$ does not contain the vertex $v_{F}$. Using the fact that the height of~$W_1$ is $r$ and applying~\Cref{prop:existence-homogeneous,prop:existence-irrelevant}, we conclude that there is a vertex $u\in V(G'\setminus X)$ such that $G'\models\mathsf{irr}(\tup v,v_F,u)$.  
\end{proof}

\subsection{Defining irrelevant vertices in logic}
\label{subsec:defirvlogic}

This subsection is devoted to the proof of~\Cref{lem:formula-irrelevant}.
In particular, we show how to define rural divisions, attachment schemes, and homogeneity in $\mathsf{MSO}_2$.

\subsubsection{Defining rural divisions in logic}

In order to define the existence of a rural division of a society in $\mathsf{MSO}_2$, we consider a \textsl{canonical} way to obtain a division $\mathcal{F}$ of a society, for some fixed set $Z=\partial \mathcal{F}$. This construction was considered in the algorithm (8.6) in~\cite{GM13}.

\paragraph{Canonical splits.}
Let $G$ be a graph and let $Z\subseteq V(G)$ be a subset of vertices of $G$.
Let $\mathcal{H}$ be the collection of all subgraphs $H$ of $G$ such that one of the following holds:
\begin{itemize}[nosep]
    \item $H$ is the graph with a single vertex $x\in Z$ which is isolated in $G$ and no edges; or
    \item $H=(\{x,y\},\{xy\})$, for some $x,y\in Z$ such that $xy\in E(G)$; or
    \item $H$ consists of a connected component $C$ of $G\setminus Z$ together with the vertices from $N(C)$ and all edges of $G$ between $C$ and $N(C)$.
\end{itemize}
We stress that in some statements and proofs in the rest of this section, we assume that $G$ does not contain isolated vertices. This assumption implies that the subgraphs in the collection $\mathcal{H}$ given above are only of the latter two kinds.

We say that two graphs $H,H'\in\mathcal{H}$
are \emph{boundary equivalent}, which we denote by $H\sim_Z H'$, if $\partial H = \partial H'$.
Note that $\sim_Z$ is an equivalence relation.
For every equivalence class $\mathcal{K}$ of $\sim_Z$, we denote by $F_{\mathcal{K}}$ the subgraph of $G$ obtained by the union of all subgraphs $H\in\mathcal{K}$.

We also use $\mathcal{F}$ to denote the collection of all graphs $F_{\mathcal{K}}$, for all equivalence classes $\mathcal{K}$ of $\sim_Z$.
We call $\mathcal{F}$ the \emph{canonical $Z$-split} of $G$.

Observe that, given a graph $G$ and a vertex subset $Z\subseteq V(G)$, the canonical $Z$-split of $G$ satisfies all the conditions of~\Cref{def:division} except~\ref{div:cond6}. Therefore, we derive the following:
\begin{observation}\label{obs:can-div}
    Let $(G,\Omega)$ be a society and let $Z\subseteq V(G)$ be a subset of vertices of $G$  such that $\Omega\subseteq Z$. 
    If for every connected component $C$ of $G\setminus Z$, it holds that $|N(C)|\le 3$ and
    there are $|N(C)|$ pairwise vertex-disjoint paths in $G$ between $N(C)$ and $\Omega$, then the canonical $Z$-split of $G$ is a division of $(G,\Omega)$.
\end{observation}



The next lemma shows that from each (rural) division $\mathcal{F}$ of a society $(G,\Omega)$, if we consider the canonical $Z$-split of $G$ where $Z=\partial \mathcal{F}$, this canonical $Z$-split is also a (rural) division of $(G,\Omega)$.

\begin{lemma}\label{lem:division-to-canonical}
    Let $(G,\Omega)$ be a society and let 
    $\mathcal{F}$ be a division of $(G,\Omega)$.
    Then the canonical $\partial \mathcal{F}$-split $\mathcal{F}'$ of $G$ is also a division of $(G,\Omega)$.
    Moreover, if $\mathcal{F}$ is rural, then $\mathcal{F}'$ is also rural.
\end{lemma}

\begin{proof}
    Let $\mathcal{F}$ be a division of $(G,\Omega)$. We set $Z\coloneqq \partial \mathcal{F}$.
    Let $\mathcal{F}'$ be the canonical $Z$-split of $G$.
    Observe that, by conditions~\ref{div:cond1} and~\ref{div:cond2} of~\Cref{def:division}, for every connected component $C$ of $G\setminus Z$, there is an
    $F\in\mathcal{F}$ such that $C$ is a subgraph of $F$ and $N(C) \subseteq \partial F$.
    Therefore, condition~\ref{div:cond6} of~\Cref{def:division} implies that for every connected component $C$ of $G\setminus Z$, $|N(C)|\le 3$ and there are $|N(C)|$ pairwise vertex-disjoint paths in $G$ between $N(C)$ and $V(\Omega)$.
    Due to~\Cref{obs:can-div}, the latter condition
    implies that $\mathcal{F}'$ is a division of $(G,\Omega)$.

    To show that if $\mathcal{F}$ is rural, then $\mathcal{F}'$ is also rural, we fix an embedding $\gamma$ of $G_{\mathcal{F}}$ in a closed disk $\Delta$ certifying that $\mathcal{F}$ is rural. Our aim is to define an embedding $\gamma'$ of $G_{\mathcal{F}'}$ in $\Delta$ certifying that $\mathcal{F}'$ is rural.
    First, we set $\gamma'(v)\coloneqq \gamma(v)$, for every $v\in Z$.
    We distinguish cases for every (vertex of $G_{\mathcal{F}'}$) $F\in \mathcal{F}'$, depending on the degree of vertex $F$ in  $G_{\mathcal{F}'}$;
    recall that every vertex in $V(G_{\mathcal{F}'})\cap \mathcal{F}'$ has degree at most three because of condition~\ref{div:cond6} of~\Cref{def:division}.

    First observe that because of condition~\ref{div:cond3} of~\Cref{def:division}, if $C_1,C_2$ are connected components of $G\setminus Z$, where $N(C_1)=N(C_2)$ and $|N(C_1)|=3$, then there is a flap $F\in\mathcal{F}$ such that both $C_1$ and~$C_2$ are subgraphs of $F$.
    This implies that for every $F'\in\mathcal{F}'$ such that $|\partial F'|=3$, there is a flap $F\in\mathcal{F}$ such that $F'$ is a subgraph of $F$ and $\partial F = \partial F'$.
    In other words, in the graph $G_{\mathcal{F}'}$, for each vertex $F'$ in $V(G_{\mathcal{F}'})\cap \mathcal{F}'$ of degree three in $G_{\mathcal{F}'}$, there is a vertex $F$ in $V(G_{\mathcal{F}})$ with the same neighborhood as $F'$.
    We set $\gamma'(F')\coloneqq \gamma(F)$.
    Also, note that is no other $F''\in\mathcal{F}'$ that is a subgraph of $F$ and $\partial F = \partial F''$. This implies that there is no other $F''\in\mathcal{F}'$ so that $\gamma'(F'') = \gamma(F)$.

    Also, note that for each $F'\in\mathcal{F}'$ such that $|\partial F'|\le 2$,
     there is a collection $A_{F'}$ of flaps in $\mathcal{F}$ such that $F'$ is a subgraph of the union of the graphs in $A_{F'}$ and for every $F\in A_{F'}$, 
    $\partial F'\subseteq \partial F$.
    If there is some $F\in A_{F'}$ such that $\partial F=\partial F'$, then we set $\gamma'(F')\coloneqq \gamma(F)$ and observe that there is no other $F''\in\mathcal{F}'$ with $\partial F=\partial F''$.
    Otherwise, we pick a face of $G_{\mathcal{F}}$ that is incident to a vertex in $V(G_{\mathcal{F}}) \cap A_{F'}$ and we set $\gamma'(F')$ to be a point inside this face. It is now easy to see that $\gamma'$ can be extended to an embedding of $G_{\mathcal{F}'}$ into $\Delta$ certifying that $\mathcal{F}'$ is rural.
\end{proof}

By \Cref{lem:division-to-canonical}, we have the following.
\begin{corollary}\label{cor:existence-canonical-rural}
    For every society $(G,\Omega)$,
    there is a rural division of $(G,\Omega)$
    if and only if there is a vertex subset $Z\subseteq V(G)$ such that $\Omega \subseteq Z$ and the canonical $Z$-split of $G$ is a rural division of $(G,\Omega)$.
\end{corollary}

For a given society $(G,\Omega)$ and vertex subset $Z\subseteq V(G)$ with $\Omega\subseteq Z$,
\Cref{cor:existence-canonical-rural}
implies that the existence of a rural division $\mathcal{F}$ of $(G,\Omega)$ with $Z=\partial \mathcal{F}$ can be equivalently translated to checking whether the canonical $Z$-split of $G$ is a rural division of $(G,\Omega)$.
The latter is easier to handle in logic, since the flaps of the canonical $Z$-split can be bijectively mapped to a set of representative edges of $G$ (such representative edges exist because of condition~\ref{div:cond5} of~\Cref{def:division} and under the assumption that $G$ does not contain isolated vertices) and each flap $F$ of the canonical $Z$-split can be determined only by its boundary $\partial F$ and the corresponding representative edge.





\paragraph{Defining rural divisions.}
We now show that we can define the fact that a canonical $Z$-split is a rural division.
For this, we first note that one can define in $\mathsf{MSO}_2$ the fact that a vertex $v$ belongs to $\partial F_e$, for the flap $F_e$ of the canonical $Z$-split of $G$ that contains an edge $e\in E(G)$.
We state this in the following observation.
Given a graph $G$, a vertex subset $Z\subseteq V(G)$, and an edge $e=vu\in E(G)$,
we call the \emph{$Z$-boundary} of $e$ in $G$ to be the set $S$ defined as follows:
\begin{itemize}[nosep]
    \item if $\{v,u\}\subseteq Z$, then $S\coloneqq \{v,u\}$; and
    \item if $C$ is the connected component of $G\setminus Z$ that contains an endpoint of $e$, then $S\coloneqq N(C)$. 
\end{itemize}
Translating this definition to $\mathsf{MSO}_2$ gives immediately the following.

\begin{observation}\label{obs:def-sep}
    There is an $\mathsf{MSO}_2$-formula $\mathsf{boundary}(X,w,z)$, where $X$ is a vertex set variable, $w$ is a single edge variable and $z$ is a single vertex variable,
    such that the following holds.
    Given a graph $G$, a vertex subset $Z\subseteq V(G)$, an edge $e\in E(G)$ and a vertex $v\in V(G)$,
    $G\models \mathsf{boundary}(Z,e,v)$ if and only if $v$ belongs to the $Z$-boundary of $e$ in $G$.
\end{observation}

\paragraph{Certifying a rural division.}
Let $(G,\Omega)$ be a society, where $G$ does not contain isolated vertices.
Also, let $Z$ be a subset of $V(G)$ and $R$ be a subset of $E(G)$.
We say that $(Z,R)$ \emph{certifies} a rural division of $(G,\Omega)$
if $\Omega\subseteq Z$, the canonical $Z$-split $\mathcal{F}$ of $G$ is a rural division of $(G,\Omega)$,
and $R$ is a set consisting of exactly one edge for each $F\in\mathcal{F}$.

\begin{lemma}\label{lem:define-rural}
    For every $\ell\in\mathbb{N}$
    there is an $\mathsf{MSO}_2$-formula $\mathsf{rural}(X,Y,\tup x)$, where $X$ is a vertex set variable, $Y$ is an edge set variable, and $\tup x=(x_1,\ldots,x_\ell)$ is a tuple of $\ell$ vertex variables,
    such that for every society $(G,\Omega)$ where $|\Omega|=\ell$ and $G$ does not contain isolated vertices, every vertex subset $Z\subseteq V(G)$, and every edge subset $R\subseteq E(G)$,
    it holds that 
    $G\models \mathsf{rural}(Z,R,\tup u)$ if and only if $(Z,R)$ certifies a rural division of $(G,\Omega)$, where $\tup u$ is the enumeration of the elements in $\Omega$ in their order in $\Omega$.
\end{lemma}

\begin{proof}
First observe that $\Omega\subseteq Z$ can be directly checked by asking that $x_i \in X$ for every $i\in[\ell]$.
    
        

    
    By the definition of a canonical $Z$-split and~\Cref{obs:can-div},
    we derive that the canonical $Z$-split of $G$ is a division of $(G,\Omega)$ if and only if
    the following holds:
    \begin{quote}
        for every connected component $C$ of $G\setminus Z$, 
        $|N(C)|\le 3$, and there are $|N(C)|$ pairwise vertex-disjoint paths in $G$ between $N(C)$ and $\Omega$.
    \end{quote}
    It is easy to observe that the above property can be expressed in $\mathsf{MSO}_2$ via a formula with $X$ and $\tup x$ as free variables.
    To express that $R$ is a set consisting of exactly one edge for each flap in $\mathcal{F}$, we use~\Cref{obs:def-sep} to write an $\mathsf{MSO}_2$-formula asking that
    for every $e\in E(G)$ there is exactly one $e'\in R$ such that for each $v\in V(G)$, it holds that
    $G\models \mathsf{boundary}(Z,e,v)\iff \mathsf{boundary}(Z,e',v)$.

    To check whether $\mathcal{F}$ is rural,
    we have to verify that the society $(\widehat{G},\Omega)$, where $\widehat{G}\coloneqq (V,E)$ with $V=Z\cup R$ and $E=\{(v,e)\in Z\times R: G\models \mathsf{boundary}(Z,e,v)\}$, is rural.
    Our aim is to write a formula $\varphi(X,Y,\tup x)$ such that the following statements are equivalent.
    \begin{itemize}[nosep]
        \item[(a)] $G\models\varphi(Z,R,\tup u)$;
        \item[(b)] $\widehat{G}$ can be embedded in a closed disk $\Delta$ with the elements of $\Omega$ appearing in the boundary of $\Delta$ in their order in $\Omega$.
    \end{itemize}
    For this, we first transform $\widehat{G}$ to a new graph $\widehat{G}^+$ as follows: we add an extra vertex $v^+$ (i.e., $V(\widehat{G}^+)\coloneqq V(\widehat{G})\cup\{v^+\}$) and make it adjacent to all elements of $\Omega$, while adding an edge between two vertices in $\Omega$ if they appear consecutively in the ordering of $\Omega$ (i.e., $E(\widehat{G}^+)=E(\widehat{G})\cup E^+$, where $E^+$ is the set of all pairs $v^+u$, for $u\in \Omega$ and all $uu'$ for $u$ and $u'$ in $\Omega$ in consecutive order).
    Observe that (b) is equivalent to $\widehat{G}^+$ being planar. Also note that planarity can be expressed by an $\mathsf{MSO}_2$-sentence $\varphi_{\mathsf{planar}}$, by expressing that the graph does not contain $K_5$ or $K_{3,3}$ as a minor. In the next claim we show that planarity of $\widehat{G}^+$ can be expressed already in $\widehat{G}$, by ``emulating'' the features of $\widehat{G}^+$ not present in $\widehat{G}$. In fact, this can be argued very easily using more elaborate logic tools: observing that $\widehat{G}^+$ can be $\MSO_2$-interpreted in $\widehat{G}$ with $\Omega$ marked using a tuple of constants, and using the Backward Translation Theorem for $\MSO_2$ interpretations, see e.g.~\cite{CourcelleE12}. However, for the sake of concreteness we give a direct proof~here.


    \begin{claim}\label{cl:formula-removing-vertex}
        For every $\mathsf{MSO}_2$ sentence $\varphi$, there is an $\mathsf{MSO}_2$-formula $\widehat{\varphi}(\tup x)$ such that we have
        $$\widehat{G}\models \widehat{\varphi}(\tup u) \qquad\textrm{if and only if}\qquad
        \widehat{G}^+\models \varphi,$$
        where $\tup u\in \Omega^{\tup x}$ is the enumeration of $\Omega$ in order.
    \end{claim}

\begin{claimproof}
 Let $F\coloneqq \{v^+\}\cup E^+\cup \{\bot\}$ be the set of features present in $\widehat{G}^+$ but not in $\widehat{G}$, extended by a marker $\bot$.
 Suppose $\tup X$ is a tuple of variables (of all possible sorts). A {\em{potential extension}} of $\tup X$ is any evaluation $\tup R\in F^{\tup X}$ with the property that single vertex variables are evaluated to elements of $\{v^+,\bot\}$, single edge variables are evaluated to elements of $E^+\cup \{\bot\}$, monadic vertex variables are evaluated to subsets of $\{v^+\}$, and monadic edge variables are evaluated to subsets of $E^+$. For an evaluation $\tup U\in \widehat{G}^{\tup X}$ and a potential extension $\tup R$ of $\tup X$, we define the evaluation $\tup U[\tup R]\in (\widehat{G}^+)^{\tup X}$ as follows:
 \begin{itemize}[nosep]
  \item If $x\in \tup X$ is a single vertex/edge variable, then $\tup U[\tup R](x)=\tup R(x)$ provided $\tup R(x)\neq \bot$, and otherwise $\tup U[\tup R](x)=\tup U(x)$.
  \item If $X\in \tup X$ is a monadic variable, then $\tup U[\tup R](X)=\tup U(X)\cup \tup R(X)$.
 \end{itemize}
 Now, for every subformula $\psi(\tup X)$ of $\varphi$ and for every potential extension $\tup R$, we construct a formula $\psi_{\tup R}(\tup X,\tup x)$ with the following property: for every $U\in \widehat{G}^{\tup X}$, we have
 \[\widehat{G}\models \psi_{\tup R}(\tup U,\tup u)\qquad\textrm{if and only if}\qquad \widehat{G}^+\models \psi(\tup U[\tup R]).\]
 Note that for the whole sentence $\varphi$ there is only one empty possible extension $\emptyset$, because there are no free variables, and then we can set $\widehat{\varphi}(\tup x)\coloneqq \varphi_{\emptyset}(\tup x)$.

 The construction is inductive on the structure of $\varphi$. First, if $\psi(\tup X)$ is an atomic formula, then $\psi_{\tup R}(\tup X,\tup x)$ can be constructed by a direct case study. For instance, if $\psi(\tup X)=(y\in Y)$ for some variables $y,Y\in \tup X$, then we have $\psi_{\tup R}(\tup X,\tup x)=\psi(\tup X)$ if $\tup R(y)=\bot$, and otherwise $\psi_{\tup R}(\tup X,\tup x)$ is an always true or an always false formula depending on whether $\tup R(y)\in \tup R(Y)$. Similarly for other atomic formulas, where the construction for the incidence predicate requires the usage of variables $\tup x$ that are evaluated to the elements of $\Omega$ in order. Next, if $\psi(\tup X)=\psi'(\tup X)\vee \psi''(\tup X)$ then we can set $\psi_{\tup R}(\tup X,\tup x)=\psi'_{\tup R}(\tup X,\tup x)\vee \psi''_{\tup R}(\tup X,\tup x)$, and similarly if $\psi(\tup X)=\neg \psi'(\tup X)$, then we can set $\psi_{\tup R}(\tup X,\tup x)=\neg \psi'_{\tup R}(\tup X,\tup x)$. Finally, if $\psi(\tup X)=\exists_{Z}\ \psi'(\tup X,Z)$ (where $Z$ is a variable of any sort), then we write
 $$\psi_{\tup R}(\tup X,\tup x)=\bigvee_{\tup R'}\ \exists_Z\ \psi_{\tup R'}(\tup X,Z,\tup x),$$
 where the disjunction ranges over all potential extensions $\tup R'\in F^{\tup X,Z}$ that extend $\tup R$. As conjunction and universal quantification can be replaced with disjunction and existential quantification using de Morgan's law, this exhausts all the cases.
\end{claimproof}

Using~\Cref{cl:formula-removing-vertex}, we turn the sentence $\varphi_{\mathsf{planar}}$ to a formula $\widehat{\varphi}(\tup x)$ such that
$\widehat{G}\models \widehat{\varphi}(\tup u)$ if and only if $\widehat{G}^+$ is planar. The claimed formula $\varphi(X,Y,\tup x)$ that certifies that $(\widehat{G},\Omega)$ is rural can be obtained from $\widehat{\varphi}(\tup u)$ by restricting the quantification to $Z\cup R$ and interpreting the incidence relation using the formula $\mathsf{boundary}(\cdot,\cdot,\cdot)$, i.e., replacing each atomic formula $\mathsf{inc}(x,y)$ by $\mathsf{boundary}(X,x,y)$.
\end{proof}

\subsubsection{Defining attachment schemes in logic}
In this subsection, we provide additional definitions needed in order to define attachment schemes in~$\MSO_2$.


\paragraph{Samplings of flaps.}
Let $(G,\Omega)$ be a society, let $X\subseteq V(G)\setminus \Omega$, and let $\mathcal{F}$ be a rural division of $(G\setminus X,\Omega)$.
In order to obtain an $X$-extension scheme of $\mathcal{F}$, we define \emph{$k$-samplings} of $\mathcal{F}$.
Given $k\ge 0$, let $\tup L=(L_1,\ldots,L_k)$ be a tuple of $k$ subsets of $\partial \mathcal{F}$ and let 
$\tup S=(S_1,\ldots,S_k)$ be a tuple of $k$ collections of flaps in $\mathcal{F}$.
We say that $(\tup L,\tup S)$ is a \emph{$k$-sampling} of $\mathcal{F}$ if the following hold:
\begin{itemize}[nosep]
    \item $\tup L$ is a partition of $\partial \mathcal{F}$; and
    \item for every $i\in[k]$ and every $v\in L_i$, there is exactly one $F\in S_i$ such that $v\in\partial F$.
\end{itemize}
Observe that every $k$-sampling $(\tup L,\tup S)$ of $\mathcal{F}$ uniquely defines an $X$-extension scheme $\xi_{\tup L,\tup S}$ of $\mathcal{F}$ as follows.
Consider a vertex $v\in\partial \mathcal{F}$ and a flap 
$F\in\mathcal{F}$ such that $v\in \partial F$.
Let also $i\in[k]$ be such that $v\in L_i$ (such an $i$ exists since $\tup L$ is a partition of $\partial \mathcal{F}$).
If $F\in S_i$, we set $A_{v}^{F}$ to be the set of all edges of $G$ of the form $vu$, where $u\in X$.
If $F\notin S_i$, we set $A_{v}^{F}=\emptyset$.  
Note that the second condition of the definition of samplings implies that, given a sampling of $\mathcal{F}$, for every $v\in\partial \mathcal{F}$ there is at most one $F\in\mathcal{F}$ such that $A_{v}^{F}\neq \emptyset$. 
For every $F$, we set $E_F\coloneqq\bigcup_{v\in\partial F} A_{v}^{F}$ and observe that the collection $\xi_{\tup L,\tup S}\coloneqq\{E_F: F\in\mathcal{F}\}$ is an $X$-extension scheme of $\mathcal{F}$.

We say that an $X$-extension scheme $\xi$ of $\mathcal{F}$ is \emph{derived} from a sampling $(\tup L,\tup S)$ of $\mathcal{F}$, if $\xi = \xi_{\tup L,\tup S}$, where $\xi_{\tup L,\tup S}$ is defined from $(\tup L,\tup S)$ as above.

Samplings can be defined in logic, using a straightforward encoding of their definition and the formula from~\Cref{obs:def-sep}.

\begin{observation}\label{obs:formula-sampling}
    Given $k\ge 0$, there is an $\mathsf{MSO}_2$-formula $\mathsf{sampling}(X_0,Y_0,\tup X,\tup Y)$, where $X_0$ is a vertex set variable, $Y_0$ is an edge set variable, $\tup X=(X_1,\ldots,X_k)$ is a tuple of $k$ vertex set variables, and $\tup Y=(Y_1,\ldots,Y_k)$ is a tuple of $k$ edge set variables then the following holds. Given 
    \begin{itemize}[nosep]
        \item a graph $G$;
        \item a vertex subset $X\subseteq V(G)\setminus \Omega$ such that $G\setminus X$ does not contain isolated vertices;
        \item a vertex subset $Z\subseteq V(G\setminus X)$ and an edge subset
        $R\subseteq E(G\setminus X)$ such that
        $(Z,R)$ certifies a rural division $\mathcal{F}$ of $(G\setminus X,\Omega)$; and
        \item tuples $\tup L\in G^{\tup X}$ and $\tup E\in G^{\tup Y}$;
    \end{itemize}
    the following conditions are equivalent:
    \begin{itemize}[nosep]
        \item $G\models\mathsf{sampling}(Z,R,\tup L, \tup E)$, and 
        \item for every $i\in [k]$ we have $\tup E(Y_i)\subseteq R$, and  
        if
        $\tup S = (S(Y_1),\ldots,S(Y_k))$ where $S(Y_i)$ is the set of flaps of~$\mathcal{F}$ that contain edges of $\tup E(Y_i)$, then $(\tup L,\tup S)$ is a $k$-sampling of $\mathcal{F}$.
    \end{itemize}    
\end{observation}

We next show the existence of $k$-samplings of rural divisions, for some constant $k\in\mathbb{N}$.

\begin{lemma}\label{lem:existence-sampling}
    There is a constant $k\in\mathbb{N}$ such that if $(G,\Omega)$ is a society and $X\subseteq V(G)\setminus \Omega$, where $G\setminus X$ does not contain isolated vertices, then every rural division of $(G\setminus X,\Omega)$ admits a $k$-sampling.
\end{lemma}
The proof of \Cref{lem:existence-sampling} is derived directly from the following lemma on bipartite graphs of bounded \emph{star chromatic number}, combined with the fact that the bipartite graph $G_{\mathcal{F}}$ obtained by a given rural division $\mathcal{F}$ is planar and that planar graphs have constant star chromatic number (see e.g.~\cite{AlbertsonCKKR04}). 
The \emph{star chromatic number} of a graph $G$ is the minimum number of colors needed to vertex-color $G$ with a proper coloring
so that every vertex receives exactly one color and the subgraph induced by the union
of each pair of color classes is a star forest (a disjoint union of stars).
Such a coloring
is a \emph{star coloring} of G.

\begin{lemma}\label{lem:star-col}
    Let $G$ be a bipartite graph with parts $(A,B)$ with no isolated vertices in $A$ and such that the star chromatic number of $G$ is at most $k$, for some $k\in\mathbb{N}$.
    Then there is a partition $\mathcal{L}$ of $A$ with at most $k^2$ parts and a function $\rho$ mapping $\mathcal{L}$ to subsets of $B$
    such that for every $L\in\mathcal{L}$ and every $u\in L$, it holds that $|\rho(L)\cap N_G(u)|=1$.
\end{lemma}

\begin{proof}
    Let $\tau\colon V(G)\to [k]$ be a star coloring of $G$ with $k$ colors. Define a partition $\mathcal{L}$ of $A$ into parts $L_{i,j}$ for $(i,j)\in[k]\times [k]$ as follows: for every $u\in A$, pick any neighbor $v(u)\in B$ of $u$ (such a neighbor exists since $A$ has no isolated vertices), and assign $u$ to the part $L_{\tau(u),\tau(v(u))}$. Finally, for every $(i,j)\in [k]\times [k]$, define $\rho(L_{i,j})\coloneqq \{v(u)\colon u\in L_{i,j}\}$.

    To prove that the partition $\mathcal{L}=\{L_{i,j}\colon (i,j)\in [k]\times [k]\}$ has the required property, consider any $(i,j)\in [k]\times [k]$ and $u\in L_{i,j}$. From the definition of $L_{i,j}$ it follows that $\tau(u)=i$ and $\tau(v(u))=j$, hence $v(u)\in \rho(L_{i,j})\cap N_G(u)$, implying that $|\rho(L_{i,j})\cap N_G(u)|\geq 1$. On the other hand, since $G[\tau^{-1}(i)\cup \tau^{-1}(j)]$ is a star forest, we have that either $u$ is a leaf in this star forest, implying that $|\rho(L_{i,j})\cap N_G(u)|\leq |\tau^{-1}(j)\cap N_G(u)|\leq 1$, or $u$ is the center of a star in this star forest, implying that $v(u)$ is the only neighbor of $u$ that is contained in $\tau^{-1}(j)\cap v(\tau^{-1}(i))$, so again $|\rho(L_{i,j})\cap N_G(u)|\leq 1$. All in all, we have $|L_{i,j}\cap N_G(u)|=1$ for every $u\in L_{i,j}$.
\end{proof}




\paragraph{Ordering guides for flaps.}
Let $(G,\Omega)$ be a society, let $X\subseteq V(G)\setminus \Omega$, and let $\mathcal{F}$ be a rural division of $(G\setminus X,\Omega)$.
As discussed in the previous subsection, the 
choice of the first element in $\lambda_F$ uniquely determines the ordering of the rest vertices in $\partial F$; in the case where  $|\partial F|=3$, this follows from~\Cref{lem:boundary-ordering}. Therefore, in order to fix a boundary ordering scheme of $\mathcal{F}$, it suffices for each $F\in\mathcal{F}$ to fix the first element $v_1$ of the ordering of $\partial F$.

To define such a first element, we use the notion of \emph{ordering guides} for $\mathcal{F}$; these are labelings of the vertices in $\partial \mathcal{F}$ that allow to define a way to distinguish a vertex in each $\partial F$.
More formally, an \emph{ordering guide of size $k$} for $\mathcal{F}$ is a tuple $\tup U=(U_1,\ldots,U_k)$ of $k$ subsets of $\partial \mathcal{F}$ that form a partition of $\partial \mathcal{F}$ and for every flap $F\in\mathcal{F}$ and $i\in [k]$, we have $|U_i\cap \partial F|\leq 1$.
Every ordering guide $\tup U$ of $\mathcal{F}$ uniquely defines a boundary ordering scheme $\Lambda_{\tup U}$ of $\mathcal{F}$, defined as follows:
for every $F\in\mathcal{F}$, we pick $U_F$ to be the minimum indexed element of $\tup U$ that intersects $\partial F$ and we define the first element in the ordering $\lambda_F$ of $\partial F$ to be the unique element of $U_F\cap \partial F$.

We say that a boundary ordering scheme $\Lambda$ of $\mathcal{F}$ is \emph{derived} from $\tup U$ if $\Lambda=\Lambda_{\tup U}$, where $\Lambda_{\tup U}$ is defined from $\tup U$ as above.

Ordering guides can be defined in logic, using a straightforward encoding of their definition and the formula from~\Cref{obs:def-sep}. 

\begin{observation}\label{obs:formula-ordering}
    Given $k\ge 0$,
    there is an $\mathsf{MSO}_2$-formula
    $\mathsf{order\text{-}guide}(X_0,Y_0,\tup X)$,
    where $X_0$ is a vertex set variable, $Y_0$ is an edge set variable, and $\tup X$ is a tuple of $k$ vertex set variables such that the following holds.~Given
    \begin{itemize}[nosep]
        \item a graph $G$;
        \item a vertex subset $X\subseteq V(G)\setminus \Omega$ such that $G\setminus X$ does not contain isolated vertices;
        \item a vertex subset $Z\subseteq V(G\setminus X)$ and an edge subset
        $R\subseteq E(G\setminus X)$ such that
        $(Z,R)$ certifies a rural division $\mathcal{F}$ of $(G\setminus X,\Omega)$; and
        \item an $\tup A\in G^{\tup X}$;
    \end{itemize}
    the following conditions are equivalent.
    \begin{itemize}[nosep]
        \item $G\models\mathsf{order\text{-}guide}(Z,R,\tup A)$
        \item $\tup A$ is an ordering guide of size $k$ for $\mathcal{F}$.
    \end{itemize}     
\end{observation}

The existence of ordering guides of constant size for rural divisions is shown in the following lemma.

\begin{lemma}\label{lem:existence-ordering}
    Let $G$ be a graph, let $X\subseteq V(G)\setminus \Omega$ such that $G\setminus X$ does not contain isolated vertices, and let $\mathcal{F}$ be a rural division of $(G\setminus X,\Omega)$.
    Then $\mathcal{F}$ admits an ordering guide of size at most four.
\end{lemma}

The proof of~\Cref{lem:existence-ordering} is a direct consequence of the following lemma and the fact that the bipartite graph $G_{\mathcal{F}}$ obtained by a rural division $\mathcal{F}$ is planar and every vertex of $G_{\mathcal{F}}$ in $\mathcal{F}$ has degree at most three.

\begin{lemma}
    Let $G$ be a planar bipartite graph with parts $(A,B)$, where every vertex $v\in B$ has degree at most three and $A$ does not contain isolated vertices.
    Then there is a partition $\mathcal{U}$ of $A$ such that for every $v\in B$ and every $U\in\mathcal{U}$, $|N(v)\cap U|\le 1$ and moreover $|\mathcal{U}|\le 4$.
\end{lemma}

\begin{proof}
    We obtain a graph $G'$ from $G$ as follows:
    for every $u\in B$, we remove it from $G$ and make its neighborhood a complete graph (keep in mind that the degree of every vertex in $B$ is at most three).
    The obtained graph $G'$ has vertex set $A$ and is planar. Therefore it admits a proper coloring with four colors.
    The color classes of such coloring form the claimed partition $\mathcal{U}$.
\end{proof}

\paragraph{Tuples certifying attachment schemes.}
Let $(G,\Omega)$ be a society and let a vertex subset $X\subseteq V(G)\setminus \Omega$ be such that $G\setminus X$ does not contain isolated vertices.
Also let $Z\subseteq V(G\setminus X)$ and $R\subseteq E(G\setminus X)$ be such that $(Z,R)$ certifies a rural division $\mathcal{F}$ of $(G\setminus X,\Omega)$.
Let $\bar{L}$ and $\bar{A}$ be two tuples of vertex subsets of $Z$ of length $k$ and $q$ respectively and let $\bar{E}$ be a tuple of $k$ edge subsets of $R$.
Let also $\sigma$ be an attachment scheme of $(\mathcal{F},X)$.
We say that $(\bar{L},\bar{E},\bar{A})$ \emph{certifies} $\sigma$ if the following conditions are satisfied:
\begin{itemize}[nosep]
    \item $\bar{A}$ is an ordering guide for $\mathcal{F}$; 
    \item if $\tup S = \{S_E : E\in \tup E\}$, where $S_E$ is the set of flaps of $\mathcal{F}$ that contain edges of $E$, then $(\tup L,\tup S)$ is a $k$-sampling of $\mathcal{F}$; and 
    \item $\sigma=(\xi,\Lambda,\rho_X)$, where $\xi$ is the $X$-extension scheme of $\mathcal{F}$ derived from $(\tup L,\tup S)$, $\Lambda$ is the boundary ordering scheme of $\mathcal{F}$ derived from $\tup A$, and $\rho_X$ is an arbitrary ordering of $X$.
\end{itemize}

\subsubsection{Defining homogeneity in logic}

We next show how to define homogeneity in $\MSO_2$. First, we show that a wall of a fixed height can be encoded in $\MSO_2$ by picking its non-subdivision vertices as constants and the internally vertex-disjoint paths that connect them (corresponding to the edges of the elementary wall from which the wall is obtained).

\paragraph{Tuples certifying a wall.}
Let $G$ be a graph, let $\tup v$ be a tuple of vertices of $G$ and let $\tup P$ be a tuple of edge subsets of $E(G)$.
We say that $(\tup v,\tup P)$ \emph{certifies a wall of $G$ of height $h$} if
$\tup P$ consists of pairwise internally vertex-disjoint paths that contain elements of $\tup v$ only in their endpoints and 
there is a bijection $\rho$ from $\tup v$ to the vertex set of an elementary wall $\overline{W}$ of height $h$ and a bijection $\mu$ from $\tup P$ to the edge set of $\overline{W}$ such that if two vertices $x,y$ of $\overline{W}$ are adjacent, then the elements $\rho^{-1}(x),\rho^{-1}(y)$ of $\tup u$ are the two endpoints of the path $\mu^{-1}(xy)$ in $\tup P$.

It is easy to see that the fact that a tuple $(\tup v,\tup P)$ certifies a wall is definable in $\MSO_2$. Also, by simple comparison with the vertices in $\tup v$, one can also identify whether four given vertices are corners of a subwall of fixed height, while the vertices of the compass of a wall can be easily identified by writing the definition of the compass in $\MSO_2$.  Thus, we obtain the following.

\begin{observation}\label{obs:formulas-wall}
    Given an even $h\ge 2$, let $n_h$ and $m_h$ be the number of vertices and edges of an elementary wall of height $h$.
    Also, let $\tup x$ be a tuple of $n_h$ single vertex variables, let $\tup Y$ be a tuple of $m_h$ edge set variables, let $y_{\mathsf{sw}},y_{\mathsf{nw}},y_{\mathsf{ne}},y_{\mathsf{se}},x$ be five single vertex variables.
    Then the following holds.
    \begin{itemize}
        \item There is an $\mathsf{MSO}_2$-formula $\mathsf{wall}(\tup x,\tup Y)$ such that given a graph $G$, a $\tup v\in G^{\tup x}$, and a $\tup P\in G^{\tup Y}$,
        it holds that $G\models \mathsf{wall}(\tup v,\tup P)$ if and only if $(\tup v, \tup P)$ certifies a wall of height $h$ in $G$.


        \item There is an $\mathsf{MSO}_2$-formula $\mathsf{comp}(\tup x,\tup Y,x)$ such given a graph $G$,
        $\tup v\in G^{\tup x}$, $\tup P\in G^{\tup Y}$,
        where $(\tup v,\tup P)$ certify a wall $W$ of height $h$ in $G$, and a vertex $v\in V(G)$, it holds that $G\models \mathsf{comp}(\tup v,\tup P,v)$ if and only if $v$ is a vertex of the compass of $W$ in $G$.

        \item For every even $h'$ where $2\le h'\le h$,
        there is an $\mathsf{MSO}_2$-formula $\mathsf{subwall}_{h'}(\tup x,\tup Y,y_{\mathsf{sw}},y_{\mathsf{nw}},y_{\mathsf{ne}},y_{\mathsf{se}})$ such that given a graph $G$, $\tup v\in G^{\tup x}$, $\tup P\in G^{\tup Y}$,
        where $(\tup v,\tup P)$ certify a wall $W$ of height $h$ in $G$,
        and four vertices $v_{\mathsf{sw}},v_{\mathsf{nw}},v_{\mathsf{ne}},v_{\mathsf{se}}\in V(G)$,
        it holds that $G\models \mathsf{subwall}_{h'}(\tup v,\tup P,v_{\mathsf{sw}},v_{\mathsf{nw}},v_{\mathsf{ne}},v_{\mathsf{se}})$ if and only if $v_{\mathsf{sw}},v_{\mathsf{nw}},v_{\mathsf{ne}},v_{\mathsf{se}}$ are the corners of a subwall $W'$ of $W$ of height $h'$.

    \end{itemize} 
\end{observation}

In order to define the existence of an irrelevant vertex, we will define the existence of a large enough homogeneous subwall of a wall whose compass admits a rural division. To define homogeneity in logic, we show the following. 

\begin{lemma}\label{lem:formula-homogeneous}
    For every even $r\ge 2$, let $n_r$ and $m_r$ be the number of the vertices and edges of an elementary wall of height $r$, respectively.
    Given $\ell,k,q,\delta\ge 0$ and even $h,r,r'\ge 2$, there exists an $\mathsf{MSO}_2$-formula
    $$\mathsf{homogeneous}(\tup x,\tup y, \tup Q,V,\widehat{V},\tup S,\tup Y,
    \tup U,y_{\mathsf{sw}},y_{\mathsf{nw}},y_{\mathsf{ne}},y_{\mathsf{se}}),$$ where $\tup x$ and $\tup y$ are tuples of $\ell$ and $n_r$ single vertex variables respectively,
    $\tup Q$ is a tuple of $m_{r}$ edge set variables,
    $V$ is a vertex set variable, $\widehat{V}$ is an edge set variable, $\tup S,\tup Y, \tup U$ are tuples of $k$ vertex set variables, $q$ vertex set variables, and $k$ edge set variables, respectively, and $y_{\mathsf{sw}},y_{\mathsf{nw}},y_{\mathsf{ne}},y_{\mathsf{se}}$ are single vertex variables
    such that the following holds. Given 
    \begin{itemize}[nosep]
        \item a graph $G$;
        \item an evaluation $\tup u\in G^{\tup x}$ where $X\coloneqq \{\tup u(x): x\in\tup x\}$ has size $\ell$;
        \item a pair $(\tup v,\tup P)$, where $\tup v\in G^{\tup y}$ and $\tup P\in G^{\tup Q}$, certifying a wall $W$ in $G\setminus X$ of height $r$ and an ordering $\Omega_W$ of its corners;
        \item a pair $(Z,R)$, where $Z\in G^{V}$ and $R\in G^{\widehat{V}}$, certifying a rural division $\mathcal{F}$ of $(\mathsf{Compass}_{G\setminus X}(W),\Omega_W)$;
        \item tuples $\tup {L}=(L_1,\ldots,L_k)\in G^{\tup S},\tup A=(A_1,\ldots,A_q)\in G^{\tup Y}, \tup E=(E_1,\ldots,E_k)\in G^{\tup U}$, certifying an attachment scheme $\sigma$ for $(\mathcal{F},X)$;
        \item four vertices $v_{\mathsf{sw}},v_{\mathsf{nw}},v_{\mathsf{ne}},v_{\mathsf{se}}\in V(G)$, which are the corners of a subwall $W'$ of $W$ of height $r'$ in $G\setminus X$;
    \end{itemize}
    we have that the following conditions are equivalent:
    \begin{itemize}[nosep]
        \item $G\models\mathsf{homogeneous}(\tup u,\tup v,\tup P, Z,R,\tup L, \tup A, \tup E,v_{\mathsf{sw}},v_{\mathsf{nw}},v_{\mathsf{ne}},v_{\mathsf{se}})$.
        \item $W'$ is $(\ell,\delta,h)$-homogeneous with respect to $(\mathcal{F},\sigma)$.
    \end{itemize}
\end{lemma}

\begin{proof}
Following~\Cref{def:homogeneity},
the claimed formula should express that  
for every $W$-internal flap $F\in \mathcal{F}$ that is captured by $W'$ and for every subwall $W''$ of $W'$ of height $h$,
there is a $W$-internal flap $F'\in\mathcal{F}$ that is captured by $W''$ such that $\mathbf{F}_\sigma$ and $\mathbf{F}_\sigma'$ have the same $\delta$-folio.

Note that for each flap $F\in\mathcal{F}$, there is a unique edge $e\in R$ such that $e\in E(F)$. Given, $e\in R$, we denote by $F_e$ the flap that contains $e$.
In order to quantify existentially or universally in the flaps of $\mathcal{F}$, we quantify (existentially or universally) over elements of $R$.
In what follows, we show how to express the notions in the definition of homogeneity above (\Cref{def:homogeneity}) in logic.

\paragraph{The property of a flap being $W$-internal.}
Note that there is an $\mathsf{MSO}_2$-formula $\mathsf{internal}(X,\tup y,\widehat{x})$, where $X$ and $\tup y$ are as in the statement of the lemma and $\widehat{x}$ is a single edge variable, such that for every $e\in R$, $G\models \mathsf{internal}(Z,\tup v,e)$ if and only if the flap $F_e\in\mathcal{F}$ is $W$-internal.
In particular, this formula asks whether every vertex $v\in V(G)$ such that $G\models \mathsf{boundary}(Z,e,v)$ does not 
belong to the vertices on the perimeter of $W$.

\paragraph{Flaps captured by subwalls.}
We next define an $\mathsf{MSO}_2$-formula
$\mathsf{captured}(\tup y,\tup Q,V,y_{\mathsf{sw}},y_{\mathsf{nw}},y_{\mathsf{ne}},y_{\mathsf{se}},\widehat{x})$,
where $\tup y,\tup Q,V,y_{\mathsf{sw}},y_{\mathsf{nw}},y_{\mathsf{ne}},y_{\mathsf{se}}$ are as in the statement of the lemma and $\widehat{x}$ is a single edge variable
such that for every $e\in R$, $G\models \mathsf{captured}(\tup v,\tup P,Z,v_{\mathsf{sw}},v_{\mathsf{nw}},v_{\mathsf{ne}},v_{\mathsf{se}},e)$ if and only if $F_e$ is captured by $W'$. In particular, this formula asks whether there are vertices $v,u\in V(G)$ and an edge subset $E\subseteq E(G)$ such that $v\in V(W')$, $G\models \mathsf{boundary}(Z,e,u)$, and $E$ induces a path of $\mathsf{Compass}_G(W)$ from $u$ to $v$ with no vertex in $V(W)$ except $v$.

\paragraph{Defining flaps.}
We next define the $\mathsf{MSO}_2$-formula $\mathsf{flap}(V,\widehat{x},U,\widehat{U})$, where $V$ is as in the statement of the lemma, $\widehat{x}$ is a single edge variable, $U$ is a vertex set variable, and $\widehat{U}$ is an edge set variable such that for every edge $e\in R$, every vertex subset $V_{\mathsf{flap}}\subseteq V(G)$, and every edge subset $E_{\mathsf{flap}}\subseteq E(G)$,
$G\models\mathsf{flap}(Z,e,V_{\mathsf{flap}},E_{\mathsf{flap}})$ if and only if $F_e=(V_{\mathsf{flap}},E_{\mathsf{flap}})$.
In particular, this formula guesses at most three vertices that are in the boundary $\partial F_e$ of $F_e$ (using $\mathsf{boundary}(Z,e,v)$ from~\Cref{obs:def-sep}) and then checks whether $V_{\mathsf{flap}}\setminus \partial F_e$ is equal to the union of the vertex sets of all connected components $C$ of $G\setminus Z$ such that $N(C)=\partial F_e$. For $E_{\mathsf{flap}}$ we ask whether it is equal to the union of the edge sets of all connected components $C$ of $G\setminus Z$ such that $N(C)=\partial F_e$. If $|\partial F_e|=2$, we also ask that if there is an edge $\hat{e}$ in $G$ with both endpoints in $\partial F_e$, this edge $\hat{e}$ should also be included in $E_{\mathsf{flap}}$.

\paragraph{Folio equivalence.}
We conclude by defining the $\mathsf{MSO}_2$-formula 
$\mathsf{same\text{-}folio}(V,\tup S,\tup Y,\tup U,\widehat{x},\widehat{x}')$, where $V,\tup S,\tup Y,\tup U$ are as in the statement of the lemma and $\widehat{x},\widehat{x}'$ are single edge variables such that for every $e,e'\in R$,
$G\models\mathsf{same\text{-}folio}(Z,\bar{L},\bar{E},\bar{A},e,e')$ if and only if $\mathbf{F}_\sigma$ and $\mathbf{F}_\sigma'$ have the same $\delta$-folio, where $F=F_e$ and $F'=F_{e'}$.
In particular, this formula checks whether for every graph $H$ on at most $|X|+t+\delta$ vertices, where $t\coloneqq |\partial F_e|$, and every mapping $\pi: V(H)\to 2^{\tup x'}$,
where $\tup x'$ is a tuple of $|X|+t$ single vertex variables, the following holds:
\[\mbox{$\xi(F_e)$ satisfies $\varphi_{(H,\pi)}(\tup w)$ if and only if $\xi(F_{e'})$ satisfies $\varphi_{(H,\pi)}(\tup w')$,}\label{logic-hom}\tag{$\star$}
\]
where $\varphi_{(H,\pi)}(\tup z)$ is the formula of~\Cref{obs:mso-minor}, $\xi$ is the extension scheme given by the attachment scheme $\sigma$ and $\tup w$ and $\tup w'$ are the orderings $\omega(F_{e})$ and $\omega(F_{e}')$ of the vertices in $X\cup \partial F_e$ and $X\cup \partial F_{e'}$, respectively, given by $\sigma$.

Let us explain more formally how to express (\ref{logic-hom})
in $\mathsf{MSO}_2$.
We first show how to define the ordering $\omega(F_{e})$ (the argument for $\omega(F_{e}')$ is identical).
We fix the last $|X|$ elements of the ordering $\omega(F_{e})$ to be the vertices in $X$, ordered by their order of appearance in the ordered tuple $\tup u$.
To fix an ordering of the first $|\partial F_e|$ elements of $\omega(F_{e})$ we argue as follows. 
Since $(\tup L,\tup A,\tup E)$ certify an attachment scheme $\sigma$ for $(\mathcal{F},X)$, $\tup A$ is an ordering guide for $\mathcal{F}$.
Recall that in an ordering guide $(A_1,\ldots,A_q)$ for $\mathcal{F}$, the sets $A_1,\ldots,A_q$ form a partition of $\partial \mathcal{F}$ and for every flap $F\in\mathcal{F}$, $|A_i\cap \partial F|\le 1$, for every $i\in[k]$.
To define the ordering of the first $|\partial F_e|$ elements of $\omega(F_{e})$, we first guess the sets from $\tup A$ that have non-empty intersection with $\partial F_e$. Assume that $A_i$ is the minimum indexed element of $\tup A$ that intersects $\partial F$.
Then, we existentially quantify one single vertex variable $s_1$ and we ask this vertex to be contained in $A_i\cap \partial F$. This fixes the first element of the ordering $\omega(F_{e})$.
If $|\partial F_e|\le 2$, then this immediately gives the ordering of $\omega(F_{e})$.
In the case where $|\partial F_e|=3$,
we observe that due to~\Cref{lem:boundary-ordering},
the ordering of $\partial F_e$ starting with $s_1$ is given  by \textsl{any} collection of three pairwise vertex-disjoint paths in $\mathsf{Compass}_{G\setminus X}(W)$ between $\partial F$ and the corners $\Omega_W$ of $W$. We existentially quantify such a collection of paths $\mathcal{P}$ and we order $\partial F_e$ accordingly to the order of the corners of $W$ in $\Omega_W$.

We also show how to define an $\mathsf{MSO}_2$-formula $\psi(V,\widehat{x},U,\widehat{U},\tup S,\tup Y)$,
where $V,\widehat{x},U,\widehat{U}$ are as in the formula $\mathsf{flap}(V,\widehat{x},U,\widehat{U})$ above and $\tup S$ and $\tup Y$ are as in the statement of the lemma,
such that for every edge $e\in R$, every vertex subset $V_{\mathsf{flap}}\subseteq V(G)$, and every edge subset $E_{\mathsf{flap}}\subseteq E(G)$,
$G\models\psi(Z,e,V_{\mathsf{flap}},E_{\mathsf{flap}},\tup L,\tup E)$ if and only if $F_e=(V_{\mathsf{flap}},E_{\mathsf{flap}})$ and $\xi(F_e)\models \varphi_{(H,\pi)}(\tup w)$.
The claimed formula $\psi$ is the conjunction of the formula $\mathsf{flap}$ from above and the formula $\widehat{\varphi}_{(H,\pi)}(\cdot)$ obtained from $\varphi_{(H,\pi)}(\cdot)$ as follows:
for every $i\in[k]$,
we check whether $L_i$ intersects $\partial F_e$.
Then, for every such $i\in[k]$ such that $L_i\cap \partial F_e\neq \emptyset$,
we check whether $e\in E_i$.
Then we enhance $E_{\mathsf{flap}}$ by adding the edges
between $L_i\cap \partial F_e$ and $X$, for every $i\in[k]$ such that $e\in E_i$.
We use $E'$ to denote the obtained edge subset.
We modify $\varphi_{(H,\pi)}(\cdot)$ by asking that (the interpretations) of all quantified single edge (resp.edge set) variables belong to (resp. are subsets of) $E'$.

The formula $\mathsf{homogeneous}(\tup u,\tup v,\tup P, Z,R,\tup L, \tup A, \tup E,v_{\mathsf{sw}},v_{\mathsf{nw}},v_{\mathsf{ne}},v_{\mathsf{se}})$ can be derived easily by the above formulas, following~\Cref{def:homogeneity}.
\end{proof}

We conclude this section by proving~\Cref{lem:formula-irrelevant} which we restate here for convenience.

\formulairrelevant*

\begin{proof}
    The formula $\mathsf{irr}(\tup x,y,z)$ first asks for the existence of a wall $W$ of height $r$ in $G\setminus X$ whose compass in $G\setminus X$ is disjoint from $u_0$.
    This can be done using the formulas $\mathsf{wall}_{r}(\tup y,\tup Q)$ and $\mathsf{comp}(\tup y,\tup Q,y)$ from~\Cref{obs:formulas-wall}.
    The existence of a vertex subset $Z$ and an edge subset $R$ certifying a rural division of the society of $W$ in $G\setminus X$ can be asked using the formula from~\Cref{lem:define-rural}. Also, the existence of tuples $(\tup L,\tup A,\tup E)$ certifying an attachment scheme $\sigma$ for $(\mathcal{F},X)$ can be checked using the formulas of~\Cref{obs:formula-sampling} and~\Cref{obs:formula-ordering}. The existence of a $k$-sampling and an ordering guide of size at most four for $\mathcal{F}$ is guaranteed by~\Cref{lem:existence-sampling} and~\Cref{lem:existence-ordering}, respectively. 
    Finally the existence of a subwall $W'$ of $W$ of height $r'$ that is $(\delta,h)$-homogeneous with respect to $(\mathcal{F},\sigma)$ can be done using the formula $\mathsf{subwall}_{h'}(\tup y,\tup Q,y_{\mathsf{sw}},y_{\mathsf{nw}},y_{\mathsf{ne}},y_{\mathsf{se}})$ from~\Cref{obs:formulas-wall} and the formula from~\Cref{lem:formula-homogeneous}. Finally, having the vertices of $W'$ as parameters, it is easy to check whether $v$ is central in $W'$. 
\end{proof}

\section{Compact clique-minor-free graphs}\label{sec:compact}

In this and the next two sections we focus on a more general setting, where we would like to solve \Folio on graphs excluding a clique minor. Similarly to \cref{def:apexFolio}, we will focus on the following problem.

\begin{definition}\label{def:cliqueFolio}
 In the \FolioClique problem we are given an instance $(G,X,\delta)$ of the \Folio problem and a parameter $h\in \N$. The task is to either solve the instance $(G,X,\delta)$, or to output a minor model of the clique on $h$ vertices in $G$.
\end{definition}

We shall argue the following statement; its proof spans \cref{sec:compact,sec:carving,sec:full-clique-minor-free}.

\begin{theorem}\label{thm:cliqueFree}
 The \FolioClique problem can be solved in time $\Oh_{|X|,\delta,h}(\|G\|^{1+o(1)})$.
\end{theorem}

Before starting building tools towards the proof of \cref{thm:cliqueFree}, we observe that we can focus on a simpler variant where the clique minor model does not need to be output by the algorithm. Formally, we consider the \FolioCliqueEx problem that differs from \FolioClique (\cref{def:cliqueFolio}) in that in the second outcome, the algorithm may just correctly conclude that $G$ contains $K_h$ as a minor, without exposing a witnessing minor model. We prove that the \FolioClique problem can be reduced to \FolioCliqueEx, as described in the following statement.

\begin{restatable}{lemma}{binsearch}\label{lem:binsearch}
 Suppose the \FolioCliqueEx problem can be solved in time $f(|X|,\delta,h)\cdot \|G\|^{1+o(1)}$, for some function $f$ that is monotone with respect to each argument. Then the \FolioClique problem can be solved in time $f(|X|,\max(\delta,h),h+1)\cdot \|G\|^{1+o(1)}$.
\end{restatable}

The proof of \cref{lem:binsearch} is relegated to \cref{sec:finding} in order not to disturb the flow of the argumentation. The idea is to consider subgraphs $G_0,G_1,\ldots,G_m$ of $G$, where $G_0$ is edgeless, $G_m=G$, and each $G_i$ is obtained from $G_{i-1}$ by adding one edge. Using the assumed algorithm $\Aa$ for \FolioCliqueEx as a black-box, we find using binary search an index $a$ such that $\Aa$ applied to $G_a$ produces an $(X,\delta)$-model-folio, while $\Aa$ applied to $G_{a+1}$ reports that $G_{a+1}$ contains $K_h$ as a minor. If $G_a$ already contains $K_h$ as a minor, then a suitable minor model can be extracted from the (already computed) model-folio. Otherwise $G_a$ is $K_h$-minor-free, hence $G_{a+1}$ is $K_{h+1}$-minor-free. So applying $\Aa$ again to $G_{a+1}$, but with parameter $h+1$ instead of $h$, necessarily outputs an $(X,\delta)$-model-folio of $G_{a+1}$, from which a minor model of $K_h$ can be~extracted. Here we may assume that $\delta\geq h$ by replacing $\delta$ with $\max(\delta,h)$ at the very beginning of the argument.

Consequently, by \cref{lem:binsearch}, from now on we focus on designing an $\Oh_{|X|,\delta,h}(\|G\|^{1+o(1)})$-time algorithm for the \FolioCliqueEx problem.

In this section we focus on proving \cref{thm:cliqueFree} under the additional assumption that $X$ is a sizeable well-linked set and the graph $G$ is ``well-concentrated'' around $X$. This idea will be formalized in a moment in the definition of a {\em{compact}} graph. In \Cref{sec:carving,sec:full-clique-minor-free} we complete the proof of \Cref{thm:cliqueFree} by reducing the general case to the compact case using a recursive strategy that iteratively finds large subgraphs of $G$ that can be separated from $X$ by small separators, recursively understanding those subgraphs, and replacing them with smaller gadgets having the same folio.

\subsection{Chips and compactness}

The notion of compactness that we will use is the following.

\begin{definition}
 Let $\alpha,k\in \N$, $G$ be a graph, and $X$ be a subset of vertices of $G$.
 A vertex subset $C\subseteq V(G)$ will be called an {\em{$(X,k,\alpha)$-chip}} if it satisfies the following properties:
 \begin{itemize}[nosep]
  \item $C\cap X=\emptyset$;
  \item $G[C]$ is connected;
  \item $|C|\geq \alpha$; and
  \item $|N(C)|<k$.
 \end{itemize}
 We shall say that $G$ is {\em{$(X,k,\alpha)$-compact}} if there are no $(X,k,\alpha)$-chips in $G$.
\end{definition}

Thus, the main goal of this section is to prove the following statement, which explains that the \FolioCliqueEx problem can be solved efficiently in sufficiently compact graphs. Later, in \Cref{sec:carving,sec:full-clique-minor-free}, we will show how to reduce the general case to this compact case.

\begin{theorem}\label{thm:compact-solvable}
 For every $h\in \N$ there exists a constant $k\in \N$, computable from $h$, such that there is an algorithm for the \FolioCliqueEx problem that runs in time $\Oh_{|X|,\delta,h,\alpha}(\|G\|^{1+o(1)})$ under the guarantee that $X$ is well-linked, $|X|\geq k$, and the input graph $G$ is $(X,k,\alpha)$-compact. Here, we assume that $\alpha$ is also given to the algorithm on input.
\end{theorem}

The remainder of this section is devoted to the proof of \cref{thm:compact-solvable}.

\subsection{Top-apex-grids}
We need some combinatorial results on how minor models of apex graphs may behave in compact graphs. More precisely, we will work with a very special kind of models of apex graphs, where the branch set of the apex is a single vertex.

\begin{definition}
 A {\em{top-apex-grid}} of order $p$ in a graph $G$ consists of a minor model $\eta\colon [p]\times [p]\to 2^{V(G)}$ of a $p\times p$ grid in $G$ and a vertex $a$, called the {\em{apex}}, that lies outside of $\bigcup_{(i,j)\in [p]\times [p]} \eta(i,j)$ and has a neighbor in each of the sets $\{\eta(i,j)\colon (i,j)\in [p]\times [p]\}$.
\end{definition}

Clearly, if $G$ contains a top-apex-grid of order $p$, then $G$ also contains an apex-grid of order $p$ as a minor. The following statement gives a partial converse of this implication: the existence of a large apex-grid minor implies the existence of a large clique minor or of a large top-apex-grid.

\begin{lemma}\label{lem:getting-the-top}
 For all $h,p\in \N$ there is a constant $q\in \N$, computable from $h,p$, such that the following holds:
 Suppose $G$ contains the apex-grid of order $q$ as a minor. Then either $G$ contains $K_h$ as a minor, or $G$ contains a top-apex-grid of order $p$.

 Moreover, there is an algorithm that given $G$ and a minor model of an apex-grid of order $q$ in $G$, in time $\Oh_q(\|G\|)$ either finds a top-apex-grid of order $p$ in $G$, or correctly concludes that $G$ contains $K_h$ as a minor.
\end{lemma}
\begin{proof}
 We focus on proving the combinatorial claim. That the proof can be turned into a linear-time algorithm will be clear from the description; we comment on it at the end of the reasoning.

 We will need the following standard statement which intuitively says the following: if one adds to a grid a large matching consisting of planarity-breaking edges, then the obtained graph necessarily contains a large clique minor.

 \begin{claim}\label{cl:crossings-clique-minor}
  For every $h\in \N$ there exist constants $r,s\in \N$, computable from $h$, such that the following holds. Suppose $H$ is a grid (of some order) and $\widetilde{H}$ is a supergraph of $H$ obtained by adding to $H$ a matching $M$ consisting of $s$ edges with the following property: in $H$, the endpoints of the edges of $M$ are pairwise at distance at least $r$ from each other, and also at distance at least $r$ from the perimeter of $H$. Then $\widetilde{H}$ contains $K_h$ as a minor.
 \end{claim}
 \begin{claimproof}
  Follows directly from~\cite[(7.4)]{GM13} and also from \cite[Lemma~4.3]{KawarabayashiTW18}. The computability claim is not asserted in~\cite{GM13}, but explicit bounds with $r\in \Oh(h^2)$ and $s\in \Oh(h^6)$ are proved in~\cite{KawarabayashiTW18}.
 \end{claimproof}

 We proceed to the proof of the lemma. Let $r$ and $s$ be the constants provided by \cref{cl:crossings-clique-minor} for the parameter~$h$. We set
 $$d\coloneqq p^4\qquad\textrm{and}\qquad q\coloneqq 2p+rds.$$
 Let $R$ be the apex grid of order $q$. We assume that the vertex set of $R$ is $([q]\times [q])\cup \{a\}$, where $a$ is the apex. Further, let $\eta$ be a minor model of $R$ in $G$. Denote $J\coloneqq G[\eta(a)]$, and let $\widehat{J}$ be the supergraph of $J$ obtain by adding, for every $(i,j)\in [q]\times [q]$, an arbitrary edge with one endpoint in $\eta(a)$ and the other endpoint in~$\eta(i,j)$. Note that $\widehat{J}$ is connected and contains, for each $(i,j)\in [q]\times [q]$, exactly one vertex of~$\eta(i,j)$; call it $\sigma(i,j)$.

 Let
 $$Y\coloneqq \{(i,j)\in [p+1,q-p]\times [p+1,q-p]\ \mid\ i,j\textrm{ are divisible by }r\}\qquad \textrm{and}\qquad Z\coloneqq \sigma(Y).$$
 Note that $|Y|=\left(\frac{q-2p}{r}\right)^2=(ds)^2\geq ds$ and in the grid $R-a$, the vertices of $Y$ are pairwise at distance at least $r$ from each other and from the perimeter of $R-a$.
 Further, let $T$ be an inclusion-wise minimal connected subgraph of $\widehat{J}$ that contains all the vertices of $Z$. Then $T$ is a tree with the leaf set $Z$.

 Consider first the case when $T$ contains a vertex of degree at least $d$, say $b$. Then within $T$ we can find a set of paths $\Pp$ of size $d$ so that each path $P\in \Pp$ starts in $b$ and finishes at a different vertex of~$Z$, and the paths of $\Pp$ are pairwise vertex-disjoint, except for sharing the endpoint $b$. Let $X\subseteq Z$ be the set of endpoints within $Z$ of paths from $\Pp$. As $|X|=|\Pp|=d=p^4$ and $Y\subseteq [p+1,q-p]\times [p+1,q-p]$, from \cref{prop:grid-in-grid} we infer that in $G$ there exists a minor model $\eta'$ of the $p\times p$ grid that is vertex-disjoint with $\eta(a)$ and whose every branch set contains a vertex of $X$. By extending every branch set of $\eta'$ by the unique path of $\Pp$ that ends in this branch set, except for not including the endpoint $b$, we obtain a minor model of the $p\times p$ grid whose every branch set is adjacent to~$b$. This is a top-apex-grid of order $p$ in $G$.

 For the second case, when all vertices of $T$ have degrees smaller than $d$, we use the following claim.

 \begin{claim}\label{cl:lowdeg-matching}
  Let $S$ be a tree with $\ell\geq 2$ leaves and maximum degree at most $\Delta\geq 3$. Then in $S$ there exists a family $\Pp$ consisting of $\left\lceil \frac{\ell}{\Delta}\right\rceil$ vertex-disjoint paths such that for each $P\in \Pp$, the endpoints of $P$ are two different leaves of $S$.
 \end{claim}
 \begin{claimproof}
  We proceed by induction on $\ell$. The base case $\ell=2$ is trivial: just take any path connecting two different leaves of $S$.

  For the induction step, assume $\ell\geq 3$. First observe that we may assume that $S$ has no vertices of degree $2$, as we may iteratively contract such vertices onto any of their neighbors, and paths found after the contraction can be easily lifted to paths before the contraction. As $\ell\geq 3$, we can find a vertex $v$ that is not a leaf of $S$, but all except at most one neighbor of $v$ is a leaf. As the degree of $v$ is at least $3$, we may find two leaf neighbors $u_1,u_2$ of $v$, and construct the path $P\coloneqq u_1-v-u_2$. Now let $S'$ be the tree obtained from $S$ by removing $v$ and all leaf neighbors of $v$. As the degree of $v$ is at most $\Delta$, it follows that $S'$ has at least $\ell-\Delta$ leaves. Further, since $S$ has no vertices of degree $2$, all the leaves of $S'$ are also leaves of $S$. Hence, by applying the induction assumption to $S'$, we find a suitable family $\Pp'$ consisting of at least $\left\lceil \frac{\ell-\Delta}{\Delta}\right\rceil=\left\lceil \frac{\ell}{\Delta}\right\rceil-1$ paths in $S'$. Now $\Pp\coloneqq \Pp'\cup \{P\}$ is a family of vertex-disjoint paths in $S$ with all the desired properties.
 \end{claimproof}

 Recalling that $Y$ is the leaf set of $T$, from \cref{cl:lowdeg-matching} we infer that if all vertices of $T$ have degrees smaller than $d$, then in $T$ there exists a family of $\left\lceil \frac{|Y|}{d}\right\rceil=s$ vertex-disjoint paths, each connecting two different leaves of $T$. By contracting each of these paths to a single edge, and contracting each branch set of $\eta$ to a single vertex, we observe that $G$ contains as a minor a graph $\widetilde{H}$ obtained from the $q\times q$ grid by adding a matching of size $s$, where each edge of the matching connects two different vertices of $Y$ (or rather, their images under the contractions). Since in the grid, the endpoints of the matching are pairwise at distance at least $r$ from each other and from the perimeter of the grid, from \cref{cl:crossings-clique-minor} we conclude that $\widetilde{H}$ contains $K_h$ as a minor. Hence so does $G$.

 This proves the combinatorial claim. As for the algorithmic claim --- that there is an $\Oh_q(\|G\|)$-time algorithm that either produces a top-apex-grid of order $p$ or correctly reports that $G$ contains $K_h$ as a minor --- we construct the tree $T$ directly from the definition and investigate the degrees in $T$. If $T$ contains a vertex of degree at least $d$, then we may construct a top-apex-grid of order $p$ by computing $X$, finding the minor model $\eta'$ provided by \cref{prop:grid-in-grid} by brute-force in time $\Oh_q(1)$, and then following the remainder of the reasoning. Otherwise, if all vertices of $T$ have degrees smaller than $d$, then we have argued that it is safe to report that $G$ contains $K_h$ as a minor.
\end{proof}

\subsection{Few candidates for apices}

We now develop the key graph-theoretic lemma that underlies our approach: in compact graphs without large clique minors, there are only few candidates for apices of top-apex-grids.

\begin{lemma}\label{lem:few-candidates}
 For all $h,\alpha,\ell\in \N$ there exist  constants $k,p,c\in \N$, computable from $h,\alpha,\ell$ and such that $k$ depends only on $h$, such that the following holds. Suppose $\ell\geq k$.
 Let $G$ be a graph excluding $K_h$ as a minor, and let $X$ be a well-linked set in $G$ such that $|X|=\ell$ and
 $G$ is $(X,k,\alpha)$-compact. Call a vertex $u$ a {\em{candidate}} if $u$ is the apex of some top-apex-grid of order $p$ in~$G$. Then the number of distinct candidates is at most $c$.
\end{lemma}

Our proof of \cref{lem:few-candidates} heavily relies on the Structure Theorem of Robertson and Seymour for graphs excluding a fixed minor~\cite{RobertsonS03a}. More precisely, the idea, inspired by the work of Lokshtanov, Pilipczuk, Pilipczuk, and Saurabh~\cite{LokshtanovPP022}, is that in a compact graph that is $K_h$-minor-free, the Structure Theorem provides a very simple decomposition: the graph is nearly-embeddable in a fixed surface $\Sigma$ after the removal of a bounded number of apices, where nearly-embeddable means that the part not embedded in $\Sigma$ consists of a bounded number of {\em{vortices}} of bounded pathwidth and a (potentially unbounded) number of {\em{flaps}}, each with bounded maximum component size and attached to at most three vertices embedded in $\Sigma$. See also~\cite[Lemma 5.1, first bullet]{LokshtanovPPS22a} for a similar statement. Once this is understood, we prove that only the (boundedly many) apices of the near-embedding may serve the role of apices of large top-apex-grids. This proof relies on two technical lemmas observed by Cohen-Addad, Le, Pilipczuk, and Pilipczuk~\cite{Cohen-AddadLPP23}, which in turn exploit connections between excluding apex graphs as minors and having locally bounded treewidth, studied by Eppstein~\cite{Eppstein00}. Thus, our proof involves a significant amount of the Graph Minors machinery; we would be interested in seeing a more direct and lighter argument.

In our study of near-embeddings, we follow the terminology of Diestel, Kawarabayashi, M\"uller, and Wollan~\cite{DiestelKMW12}. We will work with near-embeddings defined as follows:

\begin{definition}
A {\em{near-embedding}} of a graph $G$ in a surface $\Sigma$ consists of a set of {\em{apices}} $Z\subseteq V(G)$, a subgraph $G_0$ of $G-Z$, an embedding $\gamma$ of $G_0$ in $\Sigma$, a family $\Ww$ of subgraphs  of $G-Z$ called {\em{vortices}}, and a family $\Ff$ of subgraphs  of $G-Z$ called {\em{flaps}}\footnote{In~\cite{DiestelKMW12}, vortices and flaps are called {\em{large vortices}} and {\em{small vortices}}, respectively. In~\cite{LokshtanovPP022,LokshtanovPPS22a}, flaps are called {\em{dongles}}.}. We require that the following properties are satisfied:
\begin{itemize}[nosep]
\item The edge sets of subgraphs of $\{G_0\}\cup \Ww\cup \Ff$ form a partition of the edge set of $G-Z$. In particular, these subgraphs are pairwise edge-disjoint.
\item For all distinct $H,H'\in \Ww\cup \Ff$, we have $V(H)\cap V(H')\subseteq V(G_0)$. For each $H\in \Ww\cup \Ff$, the set $\Omega(H)\coloneqq V(H)\cap V(G_0)$ will be called the {\em{society}} of $H$. In case $W\in \Ww$ is a vortex, $\Omega(W)$ is an {\em{ordered set}}, that is, a set equipped with an ordering $(u^W_1,u^W_2,\ldots,u^W_{n_W})$ on its elements, where we denote $n_W\coloneqq |\Omega(W)|$.\footnote{Compared to the terminology of \cref{sec:apex-minor-free}, here we prefer to use the term {\em{society}} only for the set $\Omega(H)$, and not for the whole subgraph $(H,\Omega(H))$. Also, note that in \cref{sec:apex-minor-free} we used notation $\partial F$ for the boundary of a flap, while here we use $\Omega(F)$ instead. This is in order to have a consistent notation for both vortices and flaps.}
\item For each $H\in \Ww\cup \Ff$ there is a closed disk $\Delta_H\subseteq \Sigma$ so that $\gamma$ embeds the society $\Omega(H)$ to points on the boundary of $\Delta_H$, and otherwise $\Delta_H$ does not intersect the image of $\gamma$. Moreover, the disks $\{\Delta_H\colon H\in \Ww\cup \Ff\}$ have pairwise disjoint interiors.
\item For each flap $F\in \Ff$, we have $|\Omega(F)|\leq 3$.
\item For each vortex $W\in \Ww$, the ordering $(u^W_1,u^W_2,\ldots,u^W_{n_W})$ coincides with the order in which the elements of $\Omega(W)$ are embedded along the boundary of $\Delta_W$. Moreover, there is a path decomposition $(\bag^W_1,\ldots,\bag^W_{p_W})$ of $W$ and an increasing function $j\colon [n_W]\to [p_W]$ such that $u^W_i\in \bag^W_{j(i)}$, for all $i\in [n_W]$.
\end{itemize}
When speaking about a near-embedding, we will assume that with every vortex $W$ we are given the society $\Omega(W)$ together with the ordering $(u^W_1,\ldots,u^W_{n_W})$, the {\em{decomposition}} $(\bag^W_1,\ldots,\bag^W_{p_W})$, and the increasing {\em{index function}} $j\colon [n_W]\to [p_W]$, as described above. The {\em{adhesion}} of a vortex $W$ is the maximum size of the intersection of two consecutive bags, $\max_{1\leq i<p_W} |\bag^W_i\cap \bag^W_{i+1}|$. The {\em{interior}} of a bag $\bag^W_i$ is the graph $\intr(\bag^W_i)\coloneqq G[\bag^W_i]-\left(\Omega(W)\cup \bag_{i-1}^W\cup \bag_{i+1}^W\right)$, where we write $\bag_0^W=\bag_{p_W+1}^W=\emptyset$ by convention. Similarly, the {\em{interior}} of a flap $F\in \Ff$ is the graph $\intr(F)\coloneqq F-\Omega(F)$.

The {\em{order}} of a near-embedding $(Z,G_0,\gamma,\Ww,\Ff)$ is the least integer $q$ such that $|Z|\leq q$, $|\Ww|\leq q$, and every vortex of $\Ww$ has adhesion at most $q$. The {\em{breadth}} of $(Z,G_0,\gamma,\Ww,\Ff)$ is the least integer $d$ such that for every flap $F\in \Ff$ and for every bag $\bag_i^W$ of a vortex $W\in \Ww$, every connected component of $\intr(F)$ and of $\intr(\bag_i^W)$ has as most $d$ vertices.
\end{definition}

Note that if $(Z,G_0,\gamma,\Ww,\Ff)$ is a near-embedding of a graph $G$ of order $q$ in a surface $\Sigma$, then we have the following:
\begin{itemize}[nosep]
 \item For every flap $F\in \Ff$, we have $N(V(\intr(F)))\subseteq Z\cup \Omega(F)$, and hence $|N(V(\intr(F)))|\leq q+3$.
 \item For every bag $\bag_i^W$ of a vortex $W\in \Ww$, we have $N(V(\intr(\bag_i^W)))\subseteq Z\cup (\bag_{i-1}^W\cap \bag_i^W)\cup (\bag_i^W\cap \bag_{i+1}^W)\cup \{u^W_{j^{-1}(i)}\}$ where the last term is omitted if $i$ is not in the image of $j$, and hence $|N(V(\intr(\bag_i^W)))|\leq 3q+1$.
\end{itemize}
We will use these observations implicitly in the sequel.

We remark that for a decomposition of a vortex, Diestel et al.~\cite{DiestelKMW12} simply assume that the length of the path decomposition is the same as the size of the society, and the $i$th vertex of the society belongs to the $i$th bag of the path decomposition. For us it will be convenient to allow the existence of intermediate bags that are not required to contain society vertices, as we will perform some minor surgery on vortices that may give rise to such bags.
Also, we use the term {\em{breadth}} to distinguish it from the {\em{depth}} of a near-embedding, which is defined somewhat differently in~\cite{DiestelKMW12} and also in other works.

Next, we need to recall the standard definition of a tangle.

\begin{definition}
 Let $G$ be a graph and $\theta\in \N$. A {\em{tangle}} of {\em{order}} $\theta$ in $G$ is a family $\Tt$ consisting of oriented separations in $G$, all of order smaller than $\theta$, satisfying the following properties:
 \begin{itemize}[nosep]
  \item For every separation $(A,B)$ of $G$ of order smaller than $\theta$, exactly one of the separations $(A,B)$ and $(B,A)$ belongs to $\Tt$.
  \item For all $(A_1,B_1),(A_2,B_2),(A_3,B_3)\in \Tt$, we have $A_1\cup A_2\cup A_3\subsetneq V(G)$.
 \end{itemize}
 If $(A,B)\in \Tt$, then we say that $A$ is the {\em{small side}} and $B$ is the {\em{large side}} of $(A,B)$ (with respect to $\Tt$).
\end{definition}

If $M$ is a subset of vertices of a graph $G$ and $\Tt$ is a tangle of $G$ of order $\theta>|M|$, then we can define $\Tt-M$ to be the family of all separations $(A,B)$ of $G-M$ for which $(A\cup M,B\cup M)\in \Tt$. It is well-known and straightforward to verify that then $\Tt-M$ is a tangle of order $\theta-|M|$ in $G-M$.

Next, large well-linked sets naturally give rise to tangles of larger order. Suppose $G$ is a graph, $\theta\in \N$, and $S$ is a well-linked set in $G$ with $|S|\geq 3\theta-2$. Define $\Tt_S$ to be the family of all separation $(A,B)$ of $G$ of order less than $\theta$ such that $|A\cap S|\leq |S|/2$. It is then easy to verify that $\Tt_S$ is a tangle of order $\theta$ in $G$; we call it the tangle {\em{induced}} by $S$.

We say that a near-embedding $(Z,G_0,\gamma,\Ww,\Ff)$ of a graph $G$ {\em{captures}} a tangle $\Tt$ if for every separation $(A,B)\in \Tt-Z$ of $G-Z$, there is no $H\in \Ww\cup \Ff$ such that $B\subseteq V(H)$. In other words, the large side of a separation cannot be entirely hidden in any vortex or flap.

With these definitions in place, we can state the Structure Theorem of Robertson and Seymour. As reported by Diestel et al.~\cite[Theorem~1]{DiestelKMW12}, the statement of~\cite[(3.1)]{RobertsonS03a} takes the following form in this terminology. (We drop here the assertion that $R$ cannot be drawn in $\Sigma$, which is present in~\cite{DiestelKMW12,RobertsonS03a} but irrelevant for us.)

\begin{theorem}[{\cite[(3.1)]{RobertsonS03a}}]\label{thm:RS-original}
 Let $R$ be a graph. Then there exist integers $\theta,q\in \N$ and a surface $\Sigma$ such that for every $R$-minor-free graph $G$ and tangle $\Tt$ of order at least $\theta$ in $G$, there exists a near-embedding of $G$ in $\Sigma$ of order at most $q$ that captures the tangle~$\Tt$.
\end{theorem}

A technical caveat with \cref{thm:RS-original} is that the original proof uses some non-constructive arguments of topological nature (see Graph Minors VII~\cite{RobertsonS88}), which makes it difficult to reason that $\theta,q,\Sigma$ are computable from $R$; and we need this assertion to claim computability of the bounds in \cref{lem:few-candidates}. Fortunately, recently Kawarabayashi, Thomas, and Wollan~\cite{KawarabayashiTW20} presented new, cleaned proofs of \cref{thm:RS-original} and of multiple related statements, which yield explicit and computable bounds. In particular, from their work it follows that in \cref{thm:RS-original}, $\theta,q,\Sigma$ can be assumed to be computable from $R$.

\begin{theorem}[follows from {\cite[Theorem~16.2]{KawarabayashiTW20}}]\label{thm:RS}
 Let $R$ be a graph. Then there exist integers $\theta,q\in \N$ and a surface $\Sigma$, all computable from $R$, such that for every $R$-minor-free graph $G$ and tangle $\Tt$ of order at least $\theta$ in $G$, there exists a near-embedding of $G$ in $\Sigma$ of order at most $q$ that captures the tangle~$\Tt$.
\end{theorem}

We remark that the work of Kawarabayashi et al.~\cite{KawarabayashiTW20} is so far available only as a preprint, and to the best of our knowledge it has not yet been peer-reviewed or published. Therefore, our computability claims are contingent on the correctness of the results reported in \cite{KawarabayashiTW20}. We remark that recently, some minor issues of technical nature in~\cite{KawarabayashiTW20} were reported by Arnon~\cite{Arnon23}, but these should only affect the numerical values contained in the obtained bounds, and not influence their computability.

We now prove that under the assumption that the graph is compact, the near-embedding provided by \cref{thm:RS} can be adjusted so that it also has bounded breadth.

\begin{lemma}\label{lem:shallow-embedding}
 For every $h\in \N$ there exist integers $k,q\in \N$ and a surface $\Sigma$, all computable from $h$, such that the following holds.
 Suppose $G$ is a $K_h$-minor-free graph that is $(X,k,\alpha)$-compact for some $\alpha\in \N$, where $X$ is a well-linked set in $G$ with $|X|\geq k$. Then $G$ admits a near-embedding in $\Sigma$ of order at most $q+|X|$ and breadth smaller than~$\alpha$.
\end{lemma}
\begin{proof}
 Let $\theta,q,\Sigma$ be the constants and the surface provided by \cref{thm:RS} for $R$, and let
 $$k'\coloneqq \theta+3q+3\qquad\textrm{and}\qquad k\coloneqq 3k'-2=3\theta+9q+7.$$

 Suppose now that $G$ is $R$-minor-free and $(X,k,\alpha)$-compact for a well-linked set $X$ of size at least~$k$; thus, $|X|\geq 3k'-2$. Let $\Tt\coloneqq \Tt_X$ be the tangle of order $k'$ induced by $X$. \cref{thm:RS} now gives us a near-embedding $\Ee=(Z,G_0,\gamma,\Ww,\Ff)$ of $G$ in $\Sigma$ of order at most $q$ that captures the tangle $\Tt$. Let $\Ee'=(Z',G'_0,\gamma',\Ww',\Ff')$ be the near-embedding obtained from $\Ee$ by setting $Z'\coloneqq X\cup Z$ and removing the vertices of $X$ from $G_0$ and all the flaps and all the vortices (and all the bags in the respective decompositions of vortices). Thus, $\Ee'$ is a near-embedding of order at most $q+|X|$. It remains to show that $\Ee'$ has breadth smaller than $\alpha$.

 First, consider any flap $F'\in \Ff'$ and any connected component $C\in \cc(\intr(F'))$. Note that there is a unique flap $F\in \Ff$ such that $F'=F-X$. We have $N(V(F)-\Omega(F))\subseteq Z\cup \Omega(F)$, hence
 $$N(C)\subseteq Z\cup \Omega(F)\cup (X\cap V(F)).$$
 Note that $(A,B)\coloneqq (V(F),V(G-Z)\setminus V(\intr(F)))$ is a separation of $G-Z$ of order at most $3$. As $(Z,G_0,\gamma,\Ww,\Ff)$ captures $\Tt$, we have $(A,B)\in \Tt-Z$, so $(A\cup Z,B\cup Z)\in \Tt$. In particular, $|(A\cup Z)\cap X|\leq |X|/2$. As $|(A\cup Z)\cap (B\cup Z)|\leq |Z|+3\leq q+3$ and $X$ is well-linked, we in fact~have
 $$|(A\cup Z)\cap X|\leq q+3.$$
 Therefore, we conclude that
 $$|N(C)|\leq |Z|+|\Omega(F)|+|X\cap V(F)|\leq q+3+|(A\cup Z)\cap X|\leq 2q+6<k.$$
 Moreover, as $C$ is disjoint with $Z'$ and $X\subseteq Z'$, we have $C\cap X=\emptyset$. It now follows that $C$ would be an $(X,k,\alpha)$-chip, unless $|C|<\alpha$. As $G$ is $(X,k,\alpha)$-compact, there are no such chips, so indeed $|C|<\alpha$.

 Second, consider any bag $\bag^{W'}_i$ of a vortex $W'\in \Ww'$ and any connected component $C\in \cc(\intr(\bag^{W'}_i))$. Again, $W'=W-X$ for some vortex $W\in \Ww$, and we have $\Omega(W')=\Omega(W)-X$ and $\bag^{W'}_j=\bag^W_{j}-X$ for $j\in \{i-1,i,i+1\}$. For convenience denote
 $$N\coloneqq \bag^{W}_i\cap (\Omega(W)\cup \bag^{W}_{i-1}\cup \bag^{W}_{i+1})\qquad\textrm{and}\qquad M\coloneqq \bag^{W}_i\setminus N=V(\intr(\bag^W_i)).$$ Thus $C\subseteq M$ and $N(C)\subseteq N\cup (X\cap M)$.

 Observe that since the index function associated with vortex $W$ is increasing, $\Omega(W)$ contains at most one vertex that belongs to $\bag^W_i$ but does not belong to $\bag^W_{i-1}\cup \bag^W_{i+1}$. Hence, as the adhesion of $W$ is at most $q$, we have
 $$|N|\leq 1+|\bag^W_{i-1}\cap \bag^W_i|+|\bag^W_i\cap \bag^W_{i+1}|\leq 2q+1.$$

 Next, consider the separation $(A,B)\coloneqq (\bag^W_i,V(G)\setminus M)$ of $G-Z$. We have $A\cap B=N$, hence the order of $(A,B)$ is at most $2q+1<k'-|Z|$. As $A\subseteq V(W)$ and $\Ee$ captures the tangle $\Tt$, we conclude that $(A,B)\in \Tt-Z$, implying that $(A\cup Z,B\cup Z)\in \Tt$. Consequently, we have $|(A\cup Z)\cap X|\leq |X|/2$, which together with well-linkedness of $X$ implies that
 $$|(A\cup Z)\cap X|\leq |(A\cup Z)\cap (B\cup Z)|\leq 2q+1+|Z|\leq 3q+1.$$
 We conclude that
 $$|N(C)|\leq |N|+|X\cap M|\leq (2q+1)+|(A\cup Z)\cap X|\leq 5q+2<k.$$
 Again, as $X\subseteq Z'$ and $C$ is disjoint with $X$, we infer that $C$ would be an $(X,k,\alpha)$-chip, unless $|C|<\alpha$. As there are no $(X,k,\alpha)$-chips in $G$, we conclude that indeed $|C|<\alpha$.
\end{proof}

Motivated by \cref{lem:shallow-embedding}, we study graphs that admit near-embeddings of bounded breadth. First, we show that in the absence of apices, they actually exclude some apex graph as a minor. Recall here that an {\em{apex graph}} is a graph that can be made planar by removing one vertex. The proof is an easy lift of a result of Cohen-Addad et al., who essentially proved the same statement under a stronger definition of the breadth; see~\cite[Lemmas B.1 and 6.2]{Cohen-AddadLPP23}. Precisely, in the setting considered in~\cite{Cohen-AddadLPP23}, every flap and every bag of a vortex has bounded size, while in our setting we only have an upper bound on the sizes of the connected components of the interiors of flaps and of bags of vortices.

\begin{lemma}\label{lem:shallow-apex-minor-free}
 For all $q,\alpha\in \N$ and surface $\Sigma$, there exists an integer $d$, computable from $q,\alpha,\Sigma$, such that the following holds. Suppose a graph $G$ admits a near-embedding in $\Sigma$ of order at most $q$, breadth less than~$\alpha$, and with an empty apex set. Then $G$ excludes the apex-grid of order $d$ as a minor.
\end{lemma}
\begin{proof}
 Let $\Ee=(\emptyset,G_0,\gamma,\Ww,\Ff)$ be the considered near-embedding in $\Sigma$. Let $G'$ be the graph obtained from $G$ by removing the interior of every flap $F\in \Ff$ and turning $\Omega(F)$ into a clique. Further, let $G'_0$ be obtained from $G_0$ by turning $\Omega(F)$ into a clique for every flap $F\in \Ff$, and let $\gamma'$ be the embedding of $G'_0$ in $\Sigma$ obtained by drawing all edges of $G'_0[\Omega(F)]$ inside the disk $\Delta_F$. (Note that this is always possible, as $|\Omega(F)|\leq 3$.) Thus, $\Ee'\coloneqq (\emptyset,G_0',\gamma',\Ww,\emptyset)$ is a near-embedding of $G'$ in $\Sigma$, with no apices and no flaps.

 We now make use of the following statement of Cohen-Addad et al.

 \begin{claim}[{\cite[Lemma B.1]{Cohen-AddadLPP23}}, with adjusted notation]\label{lem:B1}
  For every surface $\Sigma$ and integers $a,b\in \N$ there exists an apex graph $R$, computable from $a,b,\Sigma$, such that every graph nearly-embeddable on $\Sigma$ with at most $a$ vortices of width (i.e., maximum bag size) at most $b$, and no apices and no flaps, does not contain $R$ as a minor.
 \end{claim}

 We remark that the computability claim is not asserted explicitly in \cite{Cohen-AddadLPP23}, but the proof of \cite[Lemma~B.1]{Cohen-AddadLPP23} invokes the Grid Minor Theorem (\cref{thm:grid-minor}) and \cite[Lemma~22]{DiestelKMW12}, which in turn invokes \cite[Lemma~1]{DemaineH08}. All the bounds in those results are governed by computable functions, yielding also computability of the bounds provided by \cref{lem:B1}.


 It is straightforward to modify the path decomposition of every vortex $W\in \Ww$ to a decomposition where every bag has size at most $b\coloneqq \alpha+2q$: for every bag $\bag_i^W$, replace it with a sequence of bags, each containing the set $\bag_i^W-V(\intr(\bag_i^W))$ (which is contained in $\bag_i^W\cap (\Omega(W)\cup \bag_{i-1}^W\cup \bag_{i+1}^W)$, and thus has size at most $2q+1$) together with one connected component of $\intr(\bag_i^W)$ (of size smaller than~$\alpha$). As $|\Ww|\leq q\eqqcolon a$, from \cref{lem:B1} we conclude\footnote{There is a slight technical caveat that Cohen-Addad et al. assume, after Diestel et al.~\cite{DiestelKMW12}, that the length of the decomposition of a vortex equals the size of its society, and the $i$th bag contains the $i$th society vertex, while we relax this requirement by allowing intermediate bags that do not need to contain society vertices. This detail plays no role in their proof and can be easily worked around by adding a new artificial society vertex to every bag and embedding it on the boundary of the disk of the vortex.} that there exists an apex graph $R$, depending only on $q,\alpha,\Sigma$ in a computable way, such that $G'$ excludes $R$ as a minor.

 It is well-known that for every planar graph $P$ there exists a grid, of order computable from $P$, that contains $P$ as a minor. Applying this statement to $P$ equal to $R$ with its apex removed, we infer that there exists a constant $c$, computable from $R$, such that the apex-grid of order $c$ contains $R$ as a minor. Consequently, $G'$ excludes the apex-grid of order $c$ as a minor.

 Next, we need the following claim. Essentially the same statement was proven by Cohen-Addad et al. as~\cite[Lemma 6.2]{Cohen-AddadLPP23}, but using a more elaborate argument.

 \begin{claim}\label{lem:62+}
  For every integer $c\in \N$ there exists an integer $d\in \N$, computable from $c$, such that the following holds. Suppose a graph $G_0$ excludes the apex-grid of order $c$ as a minor. Suppose further that $G_1$ is a graph constructed from $G_0$ by adding, for every clique $A$ in $G_0$, an arbitrary number of vertices with neighborhood~$A$. Then $G_1$ excludes the apex-grid of order $d$ as a minor.
 \end{claim}
 \begin{claimproof}
  We set
  $$d\coloneqq g(c)+2,$$
  where $g$ is the function from the Grid Minor Theorem~(\cref{thm:grid-minor}). Observe that since $G_0$ excludes the apex-grid of order $c$ as a minor, every clique in $G_0$ has size at most $c^2$. Since we may assume that $g(c)\geq c$, we have $d\geq c+2$.

  Let $H$ be the apex-grid of order $d$ and let $a$ be the apex of $H$. Suppose, for contradiction, that $G_1$ contains a minor model $\eta$ of $H$. Recall that vertices of $G_1$ can be partitioned into original vertices of $G_0$ and new vertices, which we will call {\em{appendices}}; for every appendix $v\in V(G_1)\setminus V(G_0)$, the neighborhood of $v$ in $G_1$ is a clique in $G_0$. Call a vertex $u$ of $H$ {\em{lonely}} if $\eta(u)=\{v\}$ for some appendix $v$. Note that lonely vertices necessarily form an independent set in $H$, for their respective branch sets are pairwise non-adjacent in $G_1$. Note also that the apex $a$ is not lonely, because its degree in $H$ is $d^2$, while every appendix has degree at most $c^2<d^2$. Let $J$ be the graph obtained from $H$ by (i) turning the neighborhood of every lonely vertex into a clique, and then (ii) removing every lonely vertex.

  We first observe that $\eta'$ obtained from $\eta$ by restricting the domain to $V(J)$ and removing all appendices from all branch sets is a minor model of $J$ in $G_0$. This is because the neighborhoods of appendices are cliques, hence the branch sets of neighbors of a lonely vertex are pairwise adjacent in $G_0$, and the removal of an appendix from a branch set cannot break the connectivity of this branch set.

  We second observe that $\tw(J)\geq \tw(H)-1$. Indeed, given a tree decomposition $(T,\bag)$ of $J$, we can construct a tree decomposition $(T,\bag')$ of $H$ of width at most $1$ larger by locating, for every lonely vertex $v$, any node $x$ of $T$ whose bag contains all the neighbors of $v$ (which are a clique in $J$, so such a node exists), and adding $v$ to the bag of $x$.

  Since $H$ is the apex-grid of order $d$, we have $\tw(H)\geq d=g(c)+2$. Hence $\tw(J)\geq g(c)+1$, and so $\tw(J-a)\geq g(c)$. Consequently, by \cref{thm:grid-minor}, $J-a$ contains the $c\times c$ grid as a minor, implying that $J$ contains the apex-grid of order $c$ as a minor. But we argued that $J$ is a minor of $G_1$, and $G_1$ was supposed to exclude the apex-grid of order $c$ as a minor; a contradiction.
 \end{claimproof}

 It now remains to observe that we can obtain a supergraph of the graph $G$ by applying the operation described in \cref{lem:62+} $\alpha$ times to $G'$. It follows that there exists an integer $d$, depending only on $q,\alpha,\Sigma$ in a computable way, such that $G$ excludes the apex-grid of order $d$ as a minor.
\end{proof}


Next, we introduce apices back to the picture. Precisely, we prove that if a graph $G$ has a near-embedding $\Ee$ of bounded breadth, with apices possibly present, then every apex of a large enough top-apex-grid in $G$ must be in fact an apex of $\Ee$.

\begin{lemma}\label{lem:candidates-are-apices}
 For all $q,\alpha\in \N$ and surface $\Sigma$, there exists a constant $p\in \N$, computable from $q,\alpha,\Sigma$, such that the following holds. Suppose $G$ is a graph and $\Ee$ is a near-embedding of $G$ in $\Sigma$ of order at most $q$ and breadth less than $\alpha$. Suppose further that $G$ contains a top-apex-grid of order $p$, say with apex $a$. Then $a$ belongs to the apex set of $\Ee$.
\end{lemma}
\begin{proof}
 Let $\Ee=(Z,G_0,\gamma,\Ww,\Ff)$ be the considered near-embedding. Observe that $\Ee'\coloneqq (\emptyset,G_0,\gamma,\Ww,\Ff)$ is a near-embedding of $G-Z$ in $\Sigma$, of order at most $q$, breadth less than $\alpha$, and with an empty apex set. By \cref{lem:shallow-apex-minor-free}, there exists a constant $d$, depending only on $q,\alpha,\Sigma$, such that $G-Z$ excludes the apex-grid of order $d$ as a minor. We set
 $$p\coloneqq d+q.$$

 Assume now that there is a top-apex-grid of order $p$ in $G$, say with an apex $a$ and with branch sets $\{\eta(i,j)\colon i,j\in [p]\}$. For the sake of contradiction, suppose $a\notin Z$. Observe that at most $|Z|\leq q$ among the branch sets $\{\eta(i,j)\colon i,j\in [p]\}$ may intersect the set $Z$, so by removing those branch sets and merging some others we may find a grid minor model $\{\eta'(i,j)\colon i,j\in [d]\}$ such that every set $\eta'(i,j)$ is disjoint with $Z$ and adjacent to $a$. In other words, $a$ together with $\{\eta'(i,j)\colon i,j\in [d]\}$ form a top-apex-grid of order $d$ in $G-Z$. But this implies that $G-Z$ contains the apex-grid of order $d$ as a minor, a contradiction.
\end{proof}

Now \cref{lem:few-candidates} follows immediately by combining \cref{lem:shallow-embedding,lem:candidates-are-apices}.

\subsection{The algorithm}

With \cref{lem:few-candidates} established, we can prove \cref{thm:compact-solvable}.

\begin{proof}[Proof of \cref{thm:compact-solvable}]
 Let $k$ be the constant computable only from $h$ that is provided by \cref{lem:few-candidates}. Recall that we assume that we are given a graph $G$, a well-linked set $X$ of size at least $k$, and an integer $\alpha\in \N$ such that $G$ is guaranteed to be $(X,k,\alpha)$-compact. We let $p,c\in \N$ be the remaining constants provided by \cref{lem:few-candidates}, which are computable from $h$, $\alpha$, and $\ell\coloneqq |X|$. Further, let $q\in \N$ be the constant provided by \cref{lem:getting-the-top} for parameters $h$ and $p$.

 The algorithm works as follows.
 \begin{enumerate}
  \item Initialize $Z\coloneqq \emptyset$.
  \item While $|Z|\leq c$, apply the algorithm of \cref{thm:apexFree} with the graph $G$, the set $X'\coloneqq X\cup Z$, and the parameter~$q$. Depending on the outcome, do the following.
  \begin{itemize}
   \item If the algorithm of \cref{thm:apexFree} returns an $(X',\delta)$-model-folio of $G$, then extract from it an $(X,\delta)$-model-folio of $G$ and terminate the algorithm by outputting this model-folio.
   \item Otherwise, the algorithm of \cref{thm:apexFree} provides a minor model of an apex grid of order $q$ in $G-X'$. Apply the algorithm of \cref{lem:getting-the-top} to this minor model in the graph $G-X'$. If this application yields the conclusion that $G-X'$ contains $K_h$ as a minor, then terminate the algorithm by outputting the same conclusion about $G$. Otherwise, we obtain a top-apex-grid of order $p$ in $G-X'$, say with apex $z$. We add $z$ to $Z$ and continue the loop.
  \end{itemize}
  \item Once the loop finished with $|Z|>c$ without reaching any of the terminating outcomes, we infer by \cref{lem:few-candidates} that $G$ must contain $K_h$ as a minor, as there are more than $c$ candidates for an apex of a top-apex-grid of order $p$ in $G$. We return this conclusion.
 \end{enumerate}

 That the algorithm is correct follows immediately from \cref{lem:getting-the-top,lem:few-candidates}. As for the running time, observe that each iteration of the loop takes time $\Oh_{|X'|,\delta,q}(\|G\|^{1+o(1)})\leq \Oh_{|X|,\delta,h,\alpha}(\|G\|^{1+o(1)})$, and there are at most $c=\Oh_{|X|,h,\alpha}(1)$ iterations executed, because each iteration adds one new vertex to $Z$.
\end{proof}

\section{Carving chips}\label{sec:carving}

In \Cref{sec:full-clique-minor-free} we will prove \Cref{thm:cliqueFree} by reducing it to the setting of \Cref{thm:compact-solvable}, while in \Cref{sec:generalcase} we will prove \Cref{thm:main-real} by reducing it to the setting of \Cref{thm:cliqueFree}. Both these proofs will be based on a procedure that iteratively finds a large family of chips, solves the problem recursively on them, and then replaces the chips by small equivalent graphs.
In this section we give the two main ingredients for facilitating this procedure: an algorithm for finding a large family of chips in \Cref{subsec:computingchips} (\Cref{lem:computechipfamilyfinal}), and an algorithm for replacing chips by small equivalent graphs in \Cref{subsec:preserversandreplacements} (\Cref{lem:ultimatechipreplace}).

\subsection{Computing chips}
\label{subsec:computingchips}
Our first ingredient is an algorithm to compute a large family $\chipfamily$ of pairwise non-touching $(X,k,\alpha)$-chips.

Let $G$ be a graph, $X \subseteq V(G)$, and $k,\alpha \in \N$.
Let us say that a vertex $v \in V(G) - X$ is {\em{$(X,k,\alpha)$-carvable}} if there exists an $(X,k,\alpha)$-chip $C$ with $v \in C$.
The goal of this subsection is to prove the following~theorem.

\begin{restatable}{theorem}{computechipfamilyfinal}
\label{lem:computechipfamilyfinal}
There is an algorithm that, given an $n$-vertex graph $G$, a set $X \subseteq V(G)$, and integers $k,\alpha \in \N$, in time $\Oh_{k,\alpha}(\|G\|^{1+o(1)})$ computes a family $\chipfamily$ of pairwise non-touching $(X,k,\alpha)$-chips so that if $\numcarv$ is the number of $(X,k,\alpha)$-carvable vertices of $G$, then $|\bigcup \chipfamily| \ge \numcarv/\Oh_{k,\alpha}(\log n)$.
\end{restatable}

The proof of \Cref{lem:computechipfamilyfinal} splits into two parts.
First, we give an algorithm to a problem called \TermCarv, and then we use color coding to reduce the setting of \Cref{lem:computechipfamilyfinal} to that of \TermCarv.

\subsubsection{Terminal carving}
We consider the following problem called \TermCarv.
An instance of \TermCarv is a tuple $(G,z,T,k)$, consisting of a graph $G$, a root terminal $z$, a set of terminals $T \subseteq V(G)\setminus \{z\}$, and an integer $k$, where we require that $T^* = T \cup \{z\}$ is an independent set in $G$.
A set of vertices $C \subseteq V(G)$ is called a {\em{terminal $k$-chip}} if
\begin{itemize}[nosep]
\item $C$ is connected;
\item $z\notin C$;
\item $C\cap T\neq \emptyset$;
\item $N(C)\cap T^* = \emptyset$; and 
\item $|N(C)|<k$.
\end{itemize}
We call a terminal $t \in T$ \emph{$k$-carvable} if there exists a terminal $k$-chip that contains~$t$.

Our goal is to give the following algorithm for \TermCarv.

\begin{lemma}
\label{lem:carving-terminals}
There is an $\Oh_k(\|G\|^{1+o(1)})$-time algorithm that, given an instance $(G,z,T,k)$ of \TermCarv, outputs a family $\chipfamily$ of terminal $k$-chips such that
\begin{itemize}[nosep]
\item the chips of $\chipfamily$ are pairwise non-touching; and
\item every $k$-carvable terminal is contained in some $C \in \chipfamily$.
\end{itemize}
\end{lemma}

Note that in the statement of \Cref{lem:carving-terminals}, even the existence of a family $\chipfamily$ satisfying the asserted properties is a non-trivial. As we will see, it follows from submodularity of cuts.

We prove \Cref{lem:carving-terminals} by applying recent results on almost-linear-time graph algorithms, particularly from \cite{DBLP:conf/focs/Brand0PKLGSS23,DBLP:conf/focs/LiP20,DBLP:conf/focs/SaranurakY22}.
We note that a randomized version of the algorithm of \Cref{lem:carving-terminals} is implied quite directly by the randomized Gomory-Hu tree algorithm of Pettie, Saranurak, and Yin~\cite{DBLP:conf/stoc/PettieSY22}.

Let us introduce some definitions.
Let $G$ be a graph, $I \subseteq V(G)$ an independent set, and $A,B \subseteq I$ disjoint subsets of it.
Then, by $\connb_{G,I}(A,B)$ we denote the size of a smallest $(A,B)$-separator that is disjoint from $I$.
Note that $\connb_{G,I}(A,B) \le |V(G) - I|$ because $I$ is an independent set and $A$ and $B$ are disjoint.
Also, a terminal $t \in T$ is $k$-carvable if and only if $\connb_{G,T^*}(t,z) < k$ (recall that $T^* = T \cup \{z\}$).

Furthermore, for $A \subseteq I$ we denote $\connb_{G,I}(A) = \connb_{G,I}(A,I-A)$.
Note that
\[\connb_{G,I}(A,B) = \min_{A \subseteq S \subseteq I-B} \connb_{G,I}(S).\]
Let us then recall the standard fact that the function $\connb_{G,I}$ is submodular.

\begin{lemma}[Submodularity of $\connb_{G,I}$]
For all $A,B \subseteq I$,
\[\connb_{G,I}(A \cap B) + \connb_{G,I}(A \cup B) \le \connb_{G,I}(A) + \connb_{G,I}(B).\]
\end{lemma}
\begin{proof}
Let $S_A$ be an $(A,I-A)$-separator of size $|S_A| = \connb_{G,I}(A)$, disjoint with $I$, and $S_B$ a $(B,I-B)$-separator of size $|S_B| = \connb_{G,I}(B)$, disjoint with $I$.
Now $S_A \cup S_B$ is an $(A \cap B, I-(A \cap B))$-separator and an $(A \cup B, I-(A \cup B))$-separator.
If $\connb_{G,I}(A \cap B) = 0$ or $\connb_{G,I}(A \cup B) = 0$ we are done.

Otherwise, both $A \cap B$ and $I - (A \cup B)$ are non-empty, so let $R_s$ be the vertices reachable from $A \cap B$ in $G - (S_A \cup S_B)$, and $R_t$ the vertices reachable from $I - (A \cup B)$ in $G - (S_A \cup S_B)$.
Now, $N(R_s) \subseteq S_A \cup S_B$ is an $(A \cap B, I-(A \cap B))$-separator disjoint from $I$ and $N(R_t) \subseteq S_A \cup S_B$ is an $(A \cup B, I-(A \cup B))$-separator disjoint from $I$.

Suppose $v \in N(R_s) \cap N(R_t)$.
Then there exists a path from $A \cap B$ to $I-(A \cup B)$ that intersects $S_A \cup S_B$ only in $v$.
Therefore, in that case $v \in S_A \cap S_B$.
The facts that $N(R_s) \cup N(R_t) \subseteq S_A \cup S_B$ and $N(R_s) \cap N(R_t) \subseteq S_A \cap S_B$ imply that $|N(R_s)| + |N(R_t)| \le |S_A| + |S_B|$, finishing the proof.
\end{proof}

Let us then state a basic observation about $\connb$ following from submodularity.
This lemma follows from the existence of Gomory-Hu trees for all symmetric submodular cut functions \cite{DBLP:journals/combinatorica/GoemansR95}, but let us prove it~here.

\begin{lemma}
\label{lem:babygomoryhu}
Let $G$ be a graph, $T^* = T \cup \{z\}$ an independent set, and for each $t \in T$, let $T_t \subseteq T$ be a minimum-size set so that $t \in T_t$ and $\connb_{G,T^*}(T_t) = \connb_{G,T^*}(t,z)$.
Then, 
\begin{enumerate}[label=(\arabic*),ref=(\arabic*),nosep]
\item\label{lem:babygomoryhu:prop1} if $X \subseteq T$ with $t \in X$, $z \notin X$, and $\connb_{G,T^*}(X,z) = \connb_{G,T^*}(t,z)$, then $X \supseteq T_t$; and
\item\label{lem:babygomoryhu:prop2} for all $s,t \in T$, either $T_s \subseteq T_t$, or $T_t \subseteq T_s$, or $T_s \cap T_t = \emptyset$.
\end{enumerate}
\end{lemma}
\begin{proof}
To prove \Cref{lem:babygomoryhu:prop1}, suppose $X \subseteq T$ violates it, hence $X\cap T_t\subsetneq T_t$.
Then by submodularity,
\[\connb_{G,T^*}(X \cap T_t) \le \connb_{G,T^*}(X)+\connb_{G,T^*}(T_t)-\connb_{G,T^*}(T_t \cup X) \le \connb_{G,T^*}(T_t),\]
where $\connb_{G,T^*}(X) \le \connb_{G,T^*}(T_t \cup X)$ follows from $\connb_{G,T^*}(X) = \connb_{G,T^*}(t,z)$.
Therefore, $X \cap T_t$ contradicts the choice of $T_t$.

To prove \Cref{lem:babygomoryhu:prop2}, suppose first that $T_s$ and $T_t$ intersect, neither $T_s \subseteq T_t$ nor $T_t \subseteq T_s$, and $s \in T_s \cap T_t$.
Now by submodularity,
\[\connb_{G,T^*}(T_t \cap T_s) \le \connb_{G,T^*}(T_t) + \connb_{G,T^*}(T_s) - \connb_{G,T^*}(T_t \cup T_s) \le \connb_{G,T^*}(T_s) \]
where $\connb_{G,T^*}(T_t) \le \connb_{G,T^*}(T_t \cup T_s)$ again follows from $\connb_{G,T^*}(T_t) = \connb_{G,T^*}(t,z)$.
Therefore, $T_t \cap T_s$ contradicts the choice of $T_s$.

The remaining case is that $T_s$ and $T_t$ intersect, neither $T_s \subseteq T_t$ nor $T_t \subseteq T_s$, and $s \in T_s - T_t$ and $t \in T_t - T_s$.
For $X \subseteq T$, denote $\overline{X} = T - X$, and note that $\connb_{G,T^*}(X) = \connb_{G,T^*}(\overline{X})$.
Now by submodularity,
\[\connb_{G,T^*}(T_t \cap \overline{T_s}) \le \connb_{G,T^*}(T_t) + \connb_{G,T^*}(T_s) - \connb_{G,T^*}(T_t \cup \overline{T_s}) \le \connb_{G,T^*}(T_t),\]
where $\connb_{G,T^*}(T_s) \le \connb_{G,T^*}(T_t \cup \overline{T_s})$ follows from the facts that $\overline{T_t \cup \overline{T_s}} = T_s - T_t$ and $s \in T_s - T_t$.
Therefore, $T_t \cap \overline{T_s} = T_t - T_s$ contradicts the choice of $T_t$.
\end{proof}

In particular, \Cref{lem:babygomoryhu} implies that sets $T_t$ are unique.

Let $(G,z,T,k)$ be an instance of \TermCarv and $k \in \N$.
By using the notation of \Cref{lem:babygomoryhu}, let us introduce the definition of a \emph{representative set} of $(G,z,T,k)$.
We say that a set $R \subseteq T$ is a representative set of $(G,z,T,k)$ if
\begin{enumerate}[nosep]
\item \label{def:goodset:cond1} all terminals in $R$ are $k$-carvable, i.e., if $t \in R$ then $\connb_{G,T^*}(t,z) < k$; and
\item \label{def:goodset:cond2} for every $k$-carvable terminal $t \in T \setminus R$, there exists a terminal $s \in R$ such that $T_t \subseteq T_s$.
\end{enumerate}

In other words, a representative set is a subset of the $k$-carvable terminals, that keeps for every $k$-carvable terminal $t$ either $t$ itself, or a representative $s$ of $t$ with the property that the minimally-pushed minimum separator that separates $s$ from $z$ also separates $t$ from $z$.
It is useful to observe that the condition $T_t \subseteq T_s$ is equivalent to $t \in T_s$, and also to $\connb_{G,T^*}(s, \{t,z\}) > \connb_{G,T^*}(s,z)$.

We say that a representative set is \emph{minimal} if no subset of it is a representative set.
Note that it follows from \Cref{lem:babygomoryhu} that if $R$ is a minimal representative set, then the sets $T_t$ for $t \in R$ are disjoint, and in fact form a partition of all $k$-carvable terminals.

The algorithm of \Cref{lem:carving-terminals} works in two phases.
In the first phase, we compute a minimal representative set, and in the second phase, we solve the problem assuming a minimal representative set is~given.

For both of these phases, we use the vertex-version of the \emph{Isolating Cuts Lemma} of Li and Panigrahi~\cite{DBLP:conf/focs/LiP20}, combined with the deterministic almost-linear-time algorithm for max-flow of van den Brand, Chen, Kyng, Liu, Peng, Probst Gutenberg, Sachdeva, and Sidford~\cite{DBLP:conf/focs/Brand0PKLGSS23} (see the work of Chen, Kyng, Liu, Peng, Probst Gutenberg, and Sachdeva~\cite{DBLP:conf/focs/ChenKLPGS22} for the earlier randomized version).
A vertex-separator version of the Isolating Cuts Lemma was shown by Li, Nanongkai, Panigrahi, Saranurak, and Yingchareonthawornchai~\cite{DBLP:conf/stoc/LiNPSY21}, and its (very slight) generalization to the setting of undeletable terminals was given by Pettie, Saranurak, and Yin~\cite{DBLP:conf/stoc/PettieSY22}.

\begin{lemma}[Isolating Cuts, \cite{DBLP:conf/focs/LiP20,DBLP:conf/stoc/LiNPSY21,DBLP:conf/stoc/PettieSY22,DBLP:conf/focs/Brand0PKLGSS23}]
\label{lem:isolatingcuts}
There is an algorithm that, given a graph $G$, an independent set $T \subseteq V(G)$ of terminals, and a subset $T' \subseteq T$, in time $\|G\|^{1+o(1)}$ computes a family $\{C_t \colon t \in T'\}$ of connected pairwise non-touching sets, so that for each $t\in T'$,
\begin{enumerate}[nosep]
\item \label{lem:isolatingcuts:cond1} $C_t \cap T' = \{t\}$;
\item \label{lem:isolatingcuts:cond2} $N(C_t)$ is disjoint from $T$; and
\item \label{lem:isolatingcuts:cond3} $|N(C_t)| = \connb_{G,T}(t, T' \setminus \{t\})$.
\end{enumerate}
\end{lemma}

We then use \Cref{lem:isolatingcuts} to prove \Cref{lem:carving-terminals} under the additional assumption that a minimal representative set is given.

\begin{lemma}
\label{lem:carving-terminals-with-advice}
There is a $\|G\|^{1+o(1)}$-time algorithm that, given an instance $(G,z,T,k)$ of \TermCarv and a minimal representative set $R \subseteq T$, outputs a family $\chipfamily$ of terminal $k$-chips so that
\begin{itemize}[nosep]
\item the chips of $\chipfamily$ are pairwise non-touching; and
\item every $k$-carvable terminal is contained in $\bigcup \chipfamily$.
\end{itemize}
\end{lemma}
\begin{proof}
We apply \Cref{lem:isolatingcuts} with $G$, $T^*$, and $T'=R \cup \{z\} \subseteq T^*$.
Let $\chipfamily' = \{C_t \colon t \in R \cup \{z\}\}$ be the family returned by it.
We return $\chipfamily = \chipfamily' \setminus \{C_z\}$.
It remains to prove that each $C \in \chipfamily$ is a terminal $k$-chip and every $k$-carvable terminal is contained in $\bigcup \chipfamily$.

Let us define $T_t$ for $t \in T$ as in \Cref{lem:babygomoryhu}.
By the discussion above, because $R$ is a minimal representative set, the sets $T_t$ for $t \in R$ are disjoint.
It follows that $\connb_{G,T^*}(t, (R \cup \{z\}) \setminus \{t\}) = \connb_{G,T^*}(t,z)$ for all $t \in T'$.
Therefore, $|N(C_t)| = \connb_{G,T^*}(t,z) < k$, implying that $C_t$ is a terminal $k$-chip.
This also implies that $T_t \subseteq C_t$, implying that $\bigcup \chipfamily$ contains all $k$-carvable terminals.
\end{proof}

Then the remaining step is to compute a minimal representative set.
We design a divide-and-conquer algorithm for this.
The main ingredient is a theorem of Saranurak and Yingchareonthawornchai~\cite[Theorem 1.2]{DBLP:conf/focs/SaranurakY22} that we use to decrease the size of $G$ to be proportional to the number of terminals $|T|$.

Let us introduce one more connectivity definition before stating their theorem.
Let $G$ be a graph and $A,B \subseteq V(G)$.
By $\conna_G(A,B)$ we denote the size of a smallest $(A,B)$-separator.
Note that $\conna_G(A,B) \le \min(|A|,|B|)$, because we are allowed to delete vertices from $A$ and $B$.

Now their theorem is as follows.

\begin{theorem}[{\cite[Theorem 1.2]{DBLP:conf/focs/SaranurakY22}}]
\label{the:sparsifierorig}
There is an algorithm that, given a graph $G$, a set $W \subseteq V(G)$, and an integer $k \in \N$, computes in time $\Oh_k(\|G\|^{1+o(1)})$ a graph $H$, such that
\begin{itemize}[nosep]
\item $W \subseteq V(H)$;
\item for any $A,B \subseteq W$, it holds that $\min(\conna_G(A,B), k) = \min(\conna_H(A,B),k)$; and
\item $\|H\| \le \Oh_k(|W|)$.
\end{itemize}
\end{theorem}

We immediately observe that \Cref{the:sparsifierorig} translates from the setting of $\conna_G$ to the setting of $\connb_{G,T}$.

\begin{lemma}
\label{lem:terminalsparsifier}
There is an algorithm that, given a graph $G$, an independent set $T \subseteq V(G)$, a subset $T' \subseteq T$, and an integer $k \in \N$, computes in time $\Oh_k(\|G\|^{1+o(1)})$ a graph $H$, such that 
\begin{itemize}[nosep]
\item $T' \subseteq V(H)$ is an independent set in $H$;
\item for any disjoint $A,B \subseteq T'$, it holds that $\min(\connb_{G,T}(A,B), k) = \min(\connb_{H,T'}(A,B),k)$; and 
\item $\|H\| \le \Oh_k(|T'|)$.
\end{itemize}
\end{lemma}
\begin{proof}
We construct from $G$ a graph $G'$ by replacing each $v \in T$ with a clique $W_v$ of size $k$  whose all vertices are adjacent to all vertices in $N_G(v)$.
Note that $\|G'\| \le k^{\Oh(1)} \|G\|$ and for all disjoint sets $A,B \subseteq T$ it holds that 
\begin{equation}
\label{lem:terminalsparsifier:eq1}
\min(\connb_{G,T}(A,B),k) = \min\left(\conna_{G'}\left(\bigcup_{v \in A} W_v, \bigcup_{v \in B} W_v\right),k\right).
\end{equation}
We apply \Cref{the:sparsifierorig} with the graph $G'$, the set $W_{T'} = \bigcup_{v \in T'} W_v$, and the parameter $k$, and obtain a graph $H'$ with $\|H'\| \le \Oh_k(|T'|)$, so that if $A,B \subseteq T'$, then \eqref{lem:terminalsparsifier:eq1} holds also with $G'$ replaced by $H'$.
Then we construct the graph $H$ by adding each $v \in T'$ to $H'$ as a vertex adjacent to the set $W_v$.
Clearly, $H$ satisfies the required properties.
\end{proof}

We then use \Cref{lem:terminalsparsifier} to give an algorithm for computing a minimal representative set.

\begin{lemma}
\label{lem:computereps}
There is an algorithm that, given an instance $(G,z,T,k)$ of \TermCarv and an integer~$k$, in time $\Oh_k(\|G\|^{1+o(1)})$ computes a minimal representative set of $(G,z,T,k)$.
\end{lemma}
\begin{proof}
We apply a divide-and-conquer scheme on the terminals, using \Cref{lem:terminalsparsifier} to keep the size of the graph proportional to the number of terminals.
The base case is $|T|=1$, in which case we can use the Ford-Fulkerson algorithm to test if $\connb_{G,T^*}(T,z) < k$ in time $\Oh_k(\|G\|)$, and return $T$ if yes and $\emptyset$ if no.

Suppose then $|T| \ge 2$.
Let $T_1 \cup T_2 = T$ be a balanced partition of $T$, i.e., with $||T_1|-|T_2|| \le 1$, and let $T_1^* = T_1 \cup \{z\}$ and $T_2^* = T_2 \cup \{z\}$.
We apply \Cref{lem:terminalsparsifier} for each $i \in \{1,2\}$ with $G$, $T^*$, $T^*_i$, and $k$ to compute in time $\Oh_k(\|G\|^{1+o(1)})$ a graph $H_i$ so that $\|H_i\| \le \Oh_k(|T_i^*|)$, $T^*_i \subseteq V(H_i)$ is an independent set in $H_i$, and for any disjoint $A,B \subseteq T^*_i$ it holds that $\min(\connb_{G,T^*}(A,B), k) = \min(\connb_{H_i,T_i^*}(A,B),k)$.

Then we use the algorithm recursively to compute minimal representative sets of $(H_i,z,T_i,k)$ for both $i \in \{1,2\}$.
Let $R_1$ and $R_2$ be the returned sets.

By the properties of $H_i$, all vertices in $R_i$ are $k$-carvable in $G$.
Furthermore, for every $k$-carvable terminal $t \in T_i - R_i$, there exists $s \in R_i$ so that $\connb_{G,T^*}(s, \{t,z\}) > \connb_{G,T^*}(s,z)$.
These facts imply that $R_1 \cup R_2$ is a representative set of $(G,z,T,k)$.
It remains to turn $R_1 \cup R_2$ into a minimal representative~set.

Because $R_i$ is a minimal representative set of $(H_i,z,T_i,k)$, for every terminal $t \in R_i$ we have that $\connb_{G,T^*}(t, (R_i \cup \{z\}) - \{t\}) = \connb_{G,T^*}(t,z)$.
Therefore, the sets $T_t$ (defined in the context of $(G,z,T,k)$) for $t \in R_i$ are pairwise disjoint.
(But $T_{t_1}$ and $T_{t_2}$ may intersect when $t_1 \in R_1$ and $t_2 \in R_2$.)

Let us compute the sets $T_t$ explicitly for all $t \in R_1 \cup R_2$.
We invoke the following procedure first with $i = 1$ to compute the sets $T_t$ for $t \in R_1$, and then with $i = 2$ to compute these sets for $t \in R_2$.

We apply the Isolating Cuts Lemma (\Cref{lem:isolatingcuts}) with the graph $G$ and the sets $T^*$ and $R_i \cup \{z\}$.
Let $\{C_t \colon t \in R_i \cup \{z\}\}$ be the returned family.
By the fact that $\connb_{G,T^*}(t, (R_i \cup \{z\}) - \{t\}) = \connb_{G,T^*}(t,z)$, we have that $|N(C_t)| = \connb_{G,T^*}(t,z) < k$ for all $t \in R_i$.
This implies also that $T_t \subseteq C_t \cap T$ for all $t \in R_i$.

Then, for each $t \in R_i$, we use the Ford-Fulkerson algorithm in the graph $G[N[C_t]]$ to compute the unique connected set $C'_t$ with $\{t\} \subseteq C'_t \subseteq C_t$ so that $|N(C'_t)| \le |N(C_t)|$, and for all other connected sets $C''_t$ with $\{t\} \subseteq C''_t \subseteq C_t$ and $|N(C''_t)| \le |N(C_t)|$ it holds that $C'_t \subseteq C''_t$.
Recall that $|N(C_t)| < k$, so the Ford-Fulkerson algorithm runs in $\Oh_k(\|G[N[C_t]]\|)$ time in this case.
Now, $T_t = C'_t \cap T$.
Doing this for all $t \in R_i$ takes in total $\Oh_k(\|G\|)$ time.

We have now explicitly computed for all $t \in R_1 \cup R_2$ the sets $T_t$ in time $\Oh_k(\|G\|)$.
We then construct a set $R$ by processing the vertices $t \in R_1 \cup R_2$ in a non-increasing order of $|T_t|$, and for every $t$ such that no $s$ with $t \in T_s$ was processed before it, add $t$ to $R$.
Clearly, $R$ is a minimal representative set, so we return $R$.
Note that this can be implemented in time $\Oh(\sum_{t \in R_1 \cup R_2} |T_t|) = \Oh(|T|)$.

The running time of a recursive call is $\Oh_k(\|G\|^{1+o(1)})$, plus the running times of its child calls.
As $\|G\| = \Oh_k(|T|)$ on all calls except the first, and the sum of the sizes of $T$ across all recursive calls is $\Oh(|T| \log |T|)$, we obtain the total running time of $\Oh_k(\|G\|^{1+o(1)})$.
\end{proof}

The combination of \Cref{lem:carving-terminals-with-advice,lem:computereps} implies \Cref{lem:carving-terminals}.

\subsubsection{Color coding the chips}
We now use the classic color coding technique~\cite{DBLP:journals/jacm/AlonYZ95}, in the flavor resembling the technique of randomized contractions of Chitnis, Cygan, Hajiaghayi, Pilipczuk, Pilipczuk~\cite{DBLP:journals/siamcomp/ChitnisCHPP16}, to lift the result of \Cref{lem:carving-terminals} to prove \Cref{lem:computechipfamilyfinal}.

We use the following lemma from~\cite{DBLP:journals/siamcomp/ChitnisCHPP16}, which is a standard application of splitters~\cite{DBLP:conf/focs/NaorSS95} to derandomize color coding.

\begin{lemma}[{\cite[Lemma~1]{DBLP:journals/siamcomp/ChitnisCHPP16}}]
\label{lem:colorcode}
There is an algorithm that, given a set $U$ of size $n$, and integers $a,b \in \N$, in time $2^{\Oh(b \log (a+b))} n \log n$ constructs a family $\Ff$ of $2^{\Oh(b \log (a+b))} \log n$ subsets of $U$ so that for any sets $A,B \subseteq U$ with $A \cap B = \emptyset$, $|A| \le a$, and $|B| \le b$, there exists a set $S \in \Ff$ with $A \subseteq S$ and $B \cap S  = \emptyset$.
\end{lemma}

Recall that when $G$ is a graph, $X \subseteq V(G)$, and $k,\alpha \in \N$ are integers, an $(X,k,\alpha)$-chip is a connected set $C \subseteq V(G) - X$ so that $|C| \ge \alpha$ and $|N(C)| < k$.
A vertex $v \in V(G) - X$ is $(X,k,\alpha)$-carvable if there exists an $(X,k,\alpha)$-chip $C$ with $v \in C$.
We then use the combination of \Cref{lem:carving-terminals,lem:colorcode} to prove \Cref{lem:computechipfamilyfinal}, which we restate here.

\computechipfamilyfinal*
\begin{proof}
Let $D \subseteq V(G)$ denote the set of all $(X,k,\alpha)$-carvable vertices of $G$.
Then, for each $v \in D$ let us fix $C_v$ to be an arbitrarily selected $(X,k,\alpha)$-chip with $v \in C_v$.
Then we associate with each $v \in D$ an (again arbitrarily selected) connected set $A_v$ with $v \in A_v \subseteq C_v$ and $|A_v| = \alpha$.
We also let $B_v = N(C_v)$, which has size $|B_v| < k$.
We apply the algorithm of \Cref{lem:colorcode} with $U = V(G) \setminus X$, $a = \alpha$, and $b = k-1$, to obtain a family $\Ff$ of subsets of $U$.

Now, for a set $S \in \Ff$, we say that a vertex $v \in D$ is \emph{$S$-lucky} if $A_v \subseteq S$ and $B_v \cap S = \emptyset$.
By the properties of the family $\Ff$, for every $v \in D$ there exists at least one set $S \in \Ff$ so that $v$ is $S$-lucky.
Thus, there must exist $S \in \Ff$ so that at least $|D|/(2^{\Oh(k \log (\alpha+k))} \cdot \log n)$ vertices are $S$-lucky.

\begin{claim}
\label{lem:computechipfamilyfinal:claim}
Given $S$, we can find in time $\Oh_k(\|G\|^{1+o(1)})$ a family $\chipfamily_S$ of pairwise non-touching $(X,k,\alpha)$-chips so that $\bigcup \chipfamily_S$ contains all $S$-lucky vertices.
\end{claim}
\begin{claimproof}
We construct a graph $G'$ from $G$ by contracting each connected component of $G[S]$ of size at least $\alpha$ into a single vertex, and then adding a fresh vertex $z$ with $N_{G'}(z) = X$ (recall that $S$ and $X$ are disjoint).
Let $T \subseteq V(G')$ be the vertices corresponding to the contracted components, and observe that $T \cup \{z\}$ is an independent set in $G'$.
We apply \Cref{lem:carving-terminals} with the instance $(G',z,T,k)$ of \TermCarv to obtain a family $\chipfamily'_S$ of terminal $k$-chips that are pairwise non-touching and with every $k$-carvable terminal contained in some $C' \in \chipfamily'_S$.
By uncontracting the contractions, we map each $C' \in \chipfamily'_S$ into a set $C \subseteq V(G)$.

The facts that the contracted components of $G[S]$ have size at least $\alpha$ and $C'$ contains at least one vertex $t \in T$ imply $|C| \ge \alpha$.
The fact that $N_{G'}[C']$ is disjoint with $z$ implies that $C$ is disjoint with $X$.
The fact that $N(C')$ is disjoint with $T$ implies that $|N_G(C)| = |N_{G'}(C')| < k$.
Therefore, $C$ is an $(X,k,\alpha)$-chip of $G$.
Let $\chipfamily_S$ be the resulting family of $(X,k,\alpha)$-chips.
Because the sets in $\chipfamily'_S$ are pairwise non-touching, the sets in $\chipfamily_S$ are also pairwise non-touching.

We then show that $\bigcup \chipfamily_S$ contains all $S$-lucky vertices.
Let $v$ be an $S$-lucky vertex, and let $t \in T$ be the vertex of $G'$ that corresponds to the connected component of $G[S]$ that contains $A_v$.
Now $B_v$ separates $A_v$ from $X$ in $G$, and because it is disjoint with $S$, it also separates $t$ from $z$ in $G'$.
With $|B_v| < k$, this implies that $t$ is $k$-carvable, implying that $t \in C'$ for some $C' \in \chipfamily'_S$, and therefore $v \in \bigcup \chipfamily_S$.
\end{claimproof}

We run the algorithm of \Cref{lem:computechipfamilyfinal:claim} for all $S \in \Ff$, and return the family $\chipfamily_S$ that maximizes $|\bigcup \chipfamily_S|$.
This family has total size at least $|\bigcup \chipfamily_S| \ge |D|/(2^{\Oh(k \log (\alpha+k))} \cdot \log n) \ge \numcarv/\Oh_{k,\alpha}(\log n)$.
The total running time of the algorithm is $\Oh_{k,\alpha}(\|G\|^{1+o(1)})$.
\end{proof}

\subsection{Preservers and replacements}
\label{subsec:preserversandreplacements}
Once we have found a large family $\chipfamily$ of $(X,k,\alpha)$-chips, and computed $(N(C),\delta)$-folios recursively in the graphs $G[N[C]]$ for all chips $C \in \chipfamily$, we would like replace each $G[C]$ with a graph that is small but whose behavior is equivalent to the behavior of $G[C]$.
If we are concerned with just maintaining the $(X,\delta)$-folio of $G$ under these replacements, then this is a relatively standard task, see e.g.~\cite{DBLP:journals/jacm/BodlaenderFLPST16,DBLP:conf/stoc/GroheKMW11}.
However, to obtain the almost-linear running time of our algorithm, we also have to be careful to not cause any ``new'' chips to appear after doing this replacement.

Before going into the details, let us state in advance the ultimate result of this subsection, which encompasses everything we use from this subsection in the subsequent sections.
Stating it requires one more definition.
Let $G$ and $H$ be graphs and $X \subseteq V(G) \cap V(H)$.
An \emph{$(X,\delta)$-model-mapper} from $H$ to $G$ is a data structure that, given an $(X,\delta)$-model-folio of $H$, in time $\Oh_{|X|,\delta}(\|G\|+\|H\|)$ returns an $(X,\delta)$-model-folio~of~$G$.

\begin{restatable}{theorem}{lemultimatechipreplace}
\label{lem:ultimatechipreplace}
There is a computable function $f$, and an algorithm that takes as inputs
\begin{itemize}[nosep]
\item a graph $G$;
\item integers $k,\delta \in \N$; 
\item a set $X \subseteq V(G)$ with $|X| \ge k$;
\item a family $\chipfamily$ of pairwise non-touching $(X,k,f(k,\delta))$-chips of $G$; and 
\item for each $C \in \chipfamily$ an $(N(C), \delta)$-model-folio of $G[N[C]]$,
\end{itemize}
runs in time $\Oh_{|X|,\delta}(\|G\|)$, and outputs a graph $H$ so that
\begin{itemize}[nosep]
\item $X \subseteq V(H)$ and the $(X,\delta)$-folios of $H$ and $G$ are equal;
\item $|V(H)| \le |V(G)|$ and $\|H\| \le \|G\|$;
\item if the number of $(X,k,f(k,\delta))$-carvable vertices of $G$ is $\numcarv$, then the number of $(X,k,f(k,\delta))$-carvable vertices of $H$ is at most $\numcarv - |\bigcup \chipfamily| + f(k,\delta) \cdot |\chipfamily|/2 \le \numcarv - |\bigcup \chipfamily|/2$; and
\item the size of a largest $(X,k,f(k,\delta))$-chip of $H$ is at most the size of a largest $(X,k,f(k,\delta))$-chip of $G$.
\end{itemize}
Furthermore, the algorithm also outputs an $(X,\delta)$-model-mapper from $H$ to $G$.
\end{restatable}

This subsection is dedicated to the proof of \Cref{lem:ultimatechipreplace}.

\subsubsection{Folio-preservers}
Let us first consider the task of replacing a chip with a graph that is equivalent from the standpoint of preserving the folio.
To formally define this concept, we introduce the definition of a \emph{folio-preserver}.

A graph $H$ is an \emph{$(X,\delta)$-folio-preserver} of a graph $G$ if $X \subseteq V(G) \cap V(H)$, $G[X] = H[X]$, and the $(X,\delta)$-folios of $G$ and $H$ are equal.
Note that being a $(X,\delta)$-folio-preserver is an equivalence relation among graphs.
Note also that if $H$ is an $(X,\delta)$-folio-preserver of $G$, then $H$ is also an $(X',\delta')$-folio-preserver of $G$ for all $X' \subseteq X$ and $\delta' \le \delta$.

We bound the size of an $(X,\delta)$-folio-preserver as a function of $|X|$ and $\delta$.
This follows standard arguments in the area that appear, for example, in~\cite[Lemma~2.2]{DBLP:conf/stoc/GroheKMW11} and~\cite[Lemma~8.1]{DBLP:conf/stoc/FominLP0Z20}.
We present a proof because our definitions are slightly different.

\begin{lemma}
\label{lem:preserverbound}
There is a computable function $f$, so that if $G$ is a graph, $X \subseteq V(G)$, and $\delta \in \N$, then there exists an $(X,\delta)$-folio-preserver $H$ of $G$ with $|V(H)| \le f(|X|,\delta)$.
\end{lemma}
\begin{proof}
Let $H$ be an $(X,\delta)$-folio-preserver of $G$ with the smallest number of vertices.
Note that $H$ exists because $G$ is an $(X,\delta)$-folio-preserver of $G$.
By (6.1) of \cite{GM13} (irrelevant vertex when containing a clique-minor), there is a constant $h$ computable from $|X|$ and $\delta$ so that $H$ does not contain $K_h$ as a minor.
Then, by the combination of (9.9) and (10.3) of \cite{GM13} (irrelevant vertex when containing a flat wall), there is a constant $w$ computable from $h$, $|X|$, and $\delta$, so that the treewidth of $H$ is at most $w$.

A binary tree is a rooted tree where every node has at most two children.
It is well-known that every graph admits an optimum-width tree decomposition $(T,\bag)$ where $T$ is a rooted binary tree (see e.g. nice tree decompositions in~\cite[Chapter~7]{platypus}).
Let $(T,\bag)$ be a tree decomposition of $H$ so that (1) $T$ is a binary tree, (2) $X \subseteq \bag(x)$ for all $x \in V(T)$, (3) the width of $(T,\bag)$ is at most $w+|X|$, and (4) subject to (1), (2), and (3), $(T,\bag)$ minimizes $|V(T)|$.

For $x \in V(T)$, let $T_x$ denote the subtree of $T$ consisting of all the nodes that are descendants of $x$ (including $x$ itself), and let $H_x$ be the subgraph of $H$ induced by vertices in the bags of $T_x$, i.e., $H_x = H[\bigcup_{y \in T_x} \bag(y)]$.
Note that if $x \in V(T)$ is an ancestor of $y \in V(T)$, then $V(H_y) \subseteq V(H_x)$.
Furthermore, if also $x \neq y$, then it must be that $V(H_y) \subsetneq V(H_x)$, as otherwise we could decrease $|V(T)|$ by replacing $T_x$ by $T_y$.

\begin{claim}
\label{lem:preserverbound:claim1}
Let $S = \bag(x)$ for some node $x$ of $T$.
There is no graph $H'_x$ with fewer vertices than $H_x$ so that $S \subseteq V(H'_x)$, $H_x[S] = H'_x[S]$, and the $(S,\delta)$-folios of $H_x$ and $H'_x$ are equal.
\end{claim}
\begin{claimproof}
The pair $(V(H_x), (V(H) - V(H_x)) \cup S)$ is a separation of $H$.
Because $X \subseteq S$, by \Cref{lem:folio_over_separator} the $(X,\delta)$-folio of $H$ depends only on the $(S,\delta)$-folio of $H_x$ and the $(S,\delta)$-folio of the graph $H[(V(H) - V(H_x)) \cup S]$.

Towards a contradiction, suppose such $H'_x$ exists.
Now, let $H'$ be the graph obtained from $H$ by replacing $H_x$ by $H'_x$.
Because $H'_x[S] = H_x[S]$, we have that $H'[(V(H) - V(H_x)) \cup S] = H[(V(H) - V(H_x)) \cup S]$, meaning that their $(S,\delta)$-folios are equal.
As also the $(S,\delta)$-folios of $H_x$ and $H'_x$ are equal, the $(X,\delta)$-folio of $H'$ is equal to the $(X,\delta)$-folio of $H$.
Moreover, as $S \supseteq X$, it holds that $H'[X] = H[X]$, implying that $H'$ is an $(X,\delta)$-folio-preserver of $H$ and therefore also of $G$.
However, $H'$ has fewer vertices than $H$, contradicting the original choice of $H$.
\end{claimproof}

For a node $x \in V(T)$, let $\phi_x \colon \bag(x) \to \{u_1, \ldots, u_{|\bag(x)|}\}$ be an arbitrary mapping from $\bag(x)$ to a set $\{u_1,\ldots,u_{|\bag(x)|}\}$.
Now, for a node $x \in V(T)$, let $\Ff_x$ be the $(\phi_x(\bag(x)), \delta)$-folio of the graph $\phi_x(H_x)$, where the mapping $\phi_x(H_x)$ is defined by letting $\phi_x$ rename vertices in $\bag(x)$.
We define that the \emph{signature} of a node $x \in V(T)$ is the pair $(\phi_x(H[\bag(x)]), \Ff_x)$.
Note that the number of different signatures of nodes is bounded by a constant computable from $w+|X|$ and $\delta$.

\begin{claim}
If $x \in V(T)$ is an ancestor of $y \in V(T)$ and $x \neq y$, then their signatures are not equal.
\end{claim}
\begin{claimproof}
Towards contradiction, suppose $x$ is an ancestor of $y$ with $x \neq y$ and the signatures of $x$ and $y$ are equal.
Denote $S = \bag(x)$.
It follows that $H_y$ is isomorphic to a graph $H'_y$ so that $S \subseteq V(H'_y)$, $H'_y[S] = H_x[S]$, and the $(S,\delta)$-folios of $H'_y$ and $H_x$ are equal.
However, as $|V(H'_y)| = |V(H_y)| < |V(H_x)|$, this contradicts \Cref{lem:preserverbound:claim1}.
\end{claimproof}

It follows that the length of a longest root-leaf path in $T$ is bounded by the number of different signatures.
As $T$ is binary, this implies that the number of nodes of $T$ is bounded by a constant computable from $w+|X|$ and $\delta$.
The number of vertices of $H$ is at most $w+|X|$ times the number of nodes of $T$, implying that $|V(H)|$ is bounded by a constant computable from $w+|X|$ and $\delta$, i.e., by a constant computable from $|X|$ and $\delta$.
\end{proof}




For the proofs later in this section, we also need the following lemma about $(X,\delta)$-folio-preservers.

\begin{lemma}
\label{lem:unpumpcomponent}
There is a computable function $f$, so that if $G$ is a graph, $X \subseteq V(G)$, $\delta \in \N$, and $\fcomps \subseteq \cc(G - X)$ is a subset of the connected components of $G-X$ with $|\fcomps| \ge f(|X|,\delta)$, then there exists $C \in \fcomps$ so that $G - C$ is an $(X,\delta)$-folio-preserver of $G$.
\end{lemma}
\begin{proof}
Let $g(|X|,\delta)$ be the number of $X$-rooted graphs with detail $\delta$, and let $f(|X|,\delta) = g(|X|,\delta)+1$.
Assume $|\fcomps| \ge f(|X|,\delta)$.
Let us order the components as $\fcomps = \{C_1,\ldots,C_\ell\}$, and for $i \in [0,\ell]$, let $G_i = G - \left(\bigcup_{j={i+1}}^\ell C_j\right)$.
In particular, $G_\ell = G$, and for $i < \ell$, $G_i = G_{i+1} - C_{i+1}$.

The $(X,\delta)$-folio of $G_{i+1}$ is a superset of the $(X,\delta)$-folio of $G_i$ because any minor model in $G_i$ is also a minor model in $G_{i+1}$.
Now, because $\ell > g(|X|,\delta)$, there exists $i<\ell$ so that the $(X,\delta)$-folios of $G_i$ and $G_{i+1}$ are equal.
In other words, the $(X,\delta)$-folios of $G_{i+1}$ and $G_{i+1}-C_{i+1}$ are equal.
By \Cref{lem:folio_over_separator}, the $(X,\delta)$-folio of $G_j$ depends only on the $(X,\delta)$-folio of $G_{j-1}$ and the $(X,\delta)$-folio of $G[X \cup C_j]$.
It follows that for all $j>i$, the $(X,\delta)$-folios of $G_j$ and $G_j - C_{i+1}$ are equal.
In particular, the $(X,\delta)$-folios of $G$ and $G - C_{i+1}$ are equal.

Clearly, $G[X] = (G-C_{i+1})[X]$, so $G-C_{i+1}$ is an $(X,\delta)$-folio-preserver of $G$.
\end{proof}

\subsubsection{Replacements}
We then define \emph{replacements}.
Let $G$ be a graph and $Y \subseteq V(G)$.
An \emph{$Y$-replacement-map} of $G$ is a pair of functions $(\rpmap,\bd)$, where 
\begin{itemize}[nosep]
\item $\rpmap \colon \cc(G - Y) \to (\text{graphs})$ maps each connected component $C \in \cc(G - Y)$ to a graph $\rpmap(C)$; and
\item $\bd \colon \cc(G - Y) \to 2^{V(\rpmap(C))}$ maps $C \in \cc(G - Y)$ to a set of vertices $\bd(C) \subseteq V(\rpmap(C))$,
\end{itemize}
so that for every $C \in \cc(G-Y)$,
\begin{itemize}[nosep]
\item $\bd(C) = V(\rpmap(C)) \cap V(G)$;
\item $N_G(C) \subseteq \bd(C) \subseteq Y$;
\item $\rpmap(C)[\bd(C)] = G[\bd(C)]$; and
\item the sets $V(\rpmap(C)) - \bd(C)$ for distinct $C \in \cc(G - Y)$ are pairwise disjoint.
\end{itemize}
If $(\rpmap,\bd)$ is a $Y$-replacement-map of $G$, then the \emph{$(Y,\rpmap,\bd)$-replacement} of $G$ is the graph \[G[Y] \cup \bigcup_{C \in \cc(G - Y)} \rpmap(C).\]
Note that the $(Y,\rpmap,\bd)$-replacement of $G$ is uniquely determined by $Y$, $\rpmap$, $\bd$, and~$G$.
Note also that if $H$ is a $(Y,\rpmap,\bd)$-replacement of $G$, then $G[Y] = H[Y]$.

The function $\bd$ in this definition is in principle redundant, as it can be determined from $\rpmap$ and $G$ as $\bd(C) = V(\rpmap(C)) \cap V(G)$.
However, often the intersection of the vertex sets of two different graphs is not explicitly defined because it does not matter, but here it matters so we use $\bd$ to fix the intersection of $V(\rpmap(C))$ and $V(G)$ explicitly.

Observe that if $H$ is the $(Y,\rpmap,\bd)$-replacement of $G$ with $\cc(G - Y) = \{C_1,\ldots,C_\ell\}$ then $H$ can be constructed from $G$ through a sequence of graphs $H_0 = G, H_1, \ldots, H_{\ell-1}, H_\ell = H$, so that $H_i$ is the $(Y_i, \rpmap_i, \bd_i)$-replacement of $H_{i-1}$ with $\cc(H_{i-1} - Y_i) = \{C_i\}$, $\rpmap_i(C_i) = \rpmap(C_i)$, and $\bd_i(C_i) = \bd(C_i)$.
This is useful in proving that some properties are maintained in replacements, as it suffices to focus on the case when $|\cc(G-Y)| = 1$.

We then prove that if each $\rpmap(C)$ is a $(\bd(C),\delta)$-folio-preserver of $G[\bd(C) \cup C]$, then the $(Y,\rpmap,\bd)$-replacement of $G$ is a $(Y,\delta)$-folio-preserver of $G$.
We also construct a corresponding $(X,\delta)$-model-mapper for $X \subseteq Y$ when given the $(\bd(C),\delta)$-model-folios.

\begin{lemma}
\label{lem:foliopreservreplace}
Let $G$ be a graph, $Y \subseteq V(G)$, and $(\rpmap,\bd)$ a $Y$-replacement-map of $G$.
If $H$ is the $(Y,\rpmap,\bd)$-replacement of $G$, and every $\rpmap(C)$ for $C \in \cc(G - Y)$ is a $(\bd(C),\delta)$-folio-preserver of $G[C \cup \bd(C)]$, then $H$ is a $(Y,\delta)$-folio-preserver of $G$.

Moreover, there is an algorithm that, given $G$, $(Y,\rpmap,\bd)$, an integer $k$ so that $|\bd(C)| < k$ for all $C$, a set $X \subseteq Y$ with $|X| \ge k$, and $(\bd(C),\delta)$-model-folios of the graphs $G[C \cup \bd(C)]$ for all $C$, computes the graph $H$ and an $(X,\delta)$-model-mapper from $H$ to $G$, in time $\Oh_{|X|,\delta}(\|G\|+\|H\|)$
\end{lemma}
\begin{proof}
Let us first prove that $H$ is a $(Y,\delta)$-folio-preserver of $G$.
By the arguments in the discussion above, we can without loss of generality assume that $|\cc(G-Y)| = \{C\}$.
We have that $(C \cup \bd(C), Y)$ is a separation of $G$, and $(V(\phi(C)), Y)$ is a separation of $H$.
By \Cref{lem:folio_over_separator}, the $(Y,\delta)$-folio of $G$ depends only on the $(\bd(C), \delta)$-folio of $G[C \cup \bd(C)]$ and the $(Y,\delta)$-folio of $G[Y]$.
Similarly, the $(Y,\delta)$-folio of $H$ depends only on the $(\bd(C), \delta)$-folio of $\phi(C)$ and the $(Y,\delta)$-folio of $H[Y]$.
We have that $H[Y] = G[Y]$, implying that their $(Y,\delta)$-folios are equal.
We also have that the $(\bd(C),\delta)$-folio of $G[C \cup \bd(C)]$ is equal to the $(\bd(C),\delta)$-folio of $\rpmap(C)$.
Therefore, by \Cref{lem:folio_over_separator}, the $(Y,\delta)$-folios of $G$ and $H$ are equal.
Because also $G[Y] = H[Y]$, $H$ is a $(Y,\delta)$-folio-preserver of $G$.

Clearly, given $(Y,\rpmap,\bd)$, and $G$, the graph $H$ can be constructed in $\Oh(\|G\| + k^2 \|H\|)$ time.

The model mapper data structure stores $G$, $H$, $(Y,\rpmap,\bd)$, and the $(\bd(C),\delta)$-model-folios of the graphs $G[C \cup \bd(C)]$ for all $C$.
It remains to show that given a $Y$-rooted graph $(J,\roots)$ and a minor model $\model_H$ of $(J,\roots)$ in $H$, we can compute a minor model $\model_G$ of $(J,\roots)$ in $G$ in time $\Oh_{k,\delta}(\|G\|+\|H\|)$.
Note that as $X \subseteq Y$ and $|X| \ge k$, this implies that an $(X,\delta)$-model-folio of $H$ can be translated into an $(X,\delta)$-model-folio of $G$ in time $\Oh_{|X|,\delta}(\|G\|+\|H\|)$.

First assume $\cc(G-Y) = \{C\}$.
For each $v \in V(J)$ let $\mcomps_v = \cc(\phi(C)[\model_H(v) \cap V(\phi(C))])$.
Note that if $\model_H(v)$ intersects $Y$, then because $H[\model_H(v)]$ is connected, all sets $D \in \mcomps_v$ intersect $\bd(C)$.
Otherwise, $\mcomps_v = \{\model_H(v)\}$.
Let then $\mcomps = \bigcup_{v \in V(J)} \mcomps_v$.
Note that $\mcomps$ is a family of disjoint connected vertex sets in $\phi(C)$, and all but at most $\delta$ sets in $\mcomps$ intersect $\bd(C)$.
Now, the rooted graph $(J_C,\roots_C)$ with $V(J_C) = \mcomps$, $E(J_C)$ having an edge between $D_1,D_2 \in \mcomps$ whenever there is an edge between $D_1$ and $D_2$ in $\phi(C)$, and $\roots_C(D) = D \cap \bd(C)$, is a $\bd(C)$-rooted minor of $\phi(C)$ with detail at most $\delta$.
Therefore, $(J_C,\roots_C)$ is in the $(\bd(C),\delta)$-folio of $\phi(C)$, and therefore also in the $(\bd(C),\delta)$-folio of $G[C \cup \bd(C)]$.

We first compute the $\bd(C)$-rooted graph $(J_C,\roots_C)$ explicitly from $\model_H$, and then find from the $(\bd(C),\delta)$-model-folio of $G[C \cup \bd(C)]$ a model $\model_C$ of $(J_C,\roots_C)$.
Now, we set for each $v \in V(J)$ that 
\[\model_G(v) = (\model_H(v) \setminus C) \cup \bigcup_{D \in \mcomps_v} \model_C(D).\]
Clearly, $\model_G$ is a minor model of $(J,\roots)$ in $G$.

When $\cc(G-Y) = \{C\}$ this process can be implemented in $\Oh_{k,\delta}(\|G\| + \|H\|)$ time in a straightforward way.
When $|\cc(G-Y)| > 1$, we repeat this for all components, and note that when considering component $C$, we edit the model only on vertices in $\phi(C)$, and how we edit the model depends only on its intersection with $\phi(C)$, so with appropriate data structures, we can achieve the editing for a each $C$ in time $\Oh_{k,\delta}(\|\phi(C)\| + \|G[\bd(C) \cup C]\|)$, which over all $C \in \cc(G-Y)$ sums up to $\Oh_{k,\delta}(\|G\| + \|H\|)$.
\end{proof}

\subsubsection{Intrusion-preservers}
For an efficient implementation of the process of iteratively replacing chips, we need to control how replacements can create new carvable vertices.
For this we will define \emph{intrusions} and \emph{intrusion-preservers}.

To define them, we first need to define \emph{profiles}.
If $G$ is a graph and $X \subseteq V(G)$, then the \emph{$X$-profile} of $G$ is the graph $\xprof$ with the vertex set $V(\xprof) = X$ and having an edge between $u,v \in X$ whenever $u$ and $v$ are in the same connected component of $G$.

We then define \emph{intrusions}.
Let again $G$ be a graph and $X \subseteq V(G)$.
Let also $p \in \Z$ be an integer, $X_I,X_S \subseteq X$ be disjoint subsets of $X$, and $\xprof$ be a graph on the vertex set $X_I$.
We say that a set $C \subseteq V(G)$ is an \emph{$(X,p,X_I,X_S,\xprof)$-intrusion} to $G$ if
\begin{itemize}[nosep]
\item $C \cap X = X_I$;
\item every connected component of $G[C]$ intersects $X_I$;
\item the $X_I$-profile of $G[C]$ is a supergraph of $\xprof$; and
\item $|N(C) \setminus X_S| \le p$.
\end{itemize}

Note that when $p<0$, no intrusions exist.
We say that a graph $H$ is an \emph{$(X,p)$-intrusion-preserver} of a graph $G$ if 
\begin{itemize}[nosep]
\item $X \subseteq V(G) \cap V(H)$; 
\item $G[X] = H[X]$; and 
\item for every $p' \le p$, disjoint sets $X_I,X_S \subseteq X$, and a graph $\xprof$ on the vertex set $X_I$, if there exists an $(X,p',X_I,X_S,\xprof)$-intrusion $D$ to $H$, then there exists an $(X,p',X_I,X_S,\xprof)$-intrusion $C$ to $G$ with $|C| \ge |D|$.
\end{itemize}
Note that if $H$ is an $(X,p)$-intrusion-preserver of $G$, then $H$ is also an $(X',p')$-intrusion-preserver of $G$ for all $X' \subseteq X$, and $p' \le p$.

Let us then show that if we use intrusion-preservers for a replacement, then the replacement does not create new carvable vertices outside of the newly added graphs.
We will talk about chips without the $\alpha$ parameter, so let us call an $(X,k,0)$-chip an $(X,k)$-chip.

\begin{lemma}
\label{lem:replacebyintpresmaintain}
Let $G$ be a graph, $X \subseteq Y \subseteq V(G)$, and $(\rpmap,\bd)$ a $Y$-replacement-map of $G$.
Let also $k \in \N$, and let $H$ be the $(Y,\rpmap,\bd)$-replacement of $G$ so that every $\rpmap(C)$ for $C \in \cc(G - Y)$ is an $(\bd(C),k-1)$-intrusion-preserver of $G[C \cup \bd(C)]$.
Then, for any $(X,k)$-chip $D$ of $H$ that intersects $Y$, there exists an $(X,k)$-chip $D'$ of $G$ so that $D' \cap Y = D \cap Y$ and $|D'| \ge |D|$.
\end{lemma}
\begin{proof}
Assume without loss of generality that $\cc(G - Y) = \{C\}$, and denote $H_C = \rpmap(C)$ and $B = \bd(C)$.
Let $D$ be an $(X,k)$-chip of $H$ so that $D$ intersects $Y$.
We will construct an $(X,k)$-chip $D'$ of $G$ so that $D' \cap Y = D \cap Y$ and $|D'| \ge |D|$.

First, if $D$ is disjoint with $B$, then because $D$ intersects $Y$, we have $N_H[D] \subseteq Y$, implying that we can take $D' = D$.
Then assume that $D \cap B$ is non-empty.

Let $D_C = D \cap V(H_C)$.
Let also $N_Y = N_H(D) \cap Y$ and $N_C = N_{H_C}(D_C) \setminus N_Y$, and note that $N_H(D)$ is the disjoint union of $N_Y$ and $N_C$.
Let $\xprof$ be the $(D \cap B)$-profile of $H[D_C]$.
Because $D$ intersects $B$ and is connected, we have that $D_C$ is a $(B, |N_C|, D \cap B, N_Y \cap B, \xprof)$-intrusion to $H_C$.
Because $H_C$ is an $(B,k-1)$-intrusion-preserver of $G[C \cup B]$ and $|N_C| \le k-1$, there is a $(B, |N_C|, D \cap B, N_Y \cap B, \xprof)$-intrusion $D'_C$ to $G[C \cup B]$ with $|D'_C| \ge |D_C|$.
We claim that $D' = (D \setminus D_C) \cup D'_C$ is an $(X,k)$-chip in $G$ with $|D'| \ge |D|$.
Note that $D' \cap Y = D \cap Y$.

First, because $D'_C \cap B = D_C \cap B$ and $|D'_C| \ge |D_C|$, we have that $|D'| \ge |D|$.
Then, $D' \cap X = \emptyset$ is implied by the fact that $D' \cap Y = D \cap Y$.

\begin{claim}
$|N_G(D')| \le |N_H(D)| < k$.
\end{claim}
\begin{claimproof}
First, note that $N_G(D') = (N_G(D'_C) \cap (C \cup B)) \cup (N_G(D \cap Y) \cap Y)$.
It holds that $N_G(D \cap Y) \cap Y \subseteq N_Y$. 
By the definition of intrusion, $|N_G(D'_C) \cap (C \cup B) - N_Y| \le |N_C|$.
These imply that $|N_G(D')| \le |N_C| + |N_Y| \le |N_H(D)|$.
\end{claimproof}

\begin{claim}
$G[D']$ is connected.
\end{claim}
\begin{claimproof}
Let $\xprof_Y$ be the $(D \cap B)$-profile of $H[D \cap Y]$.
Because $H[D]$ is connected, the union of $\xprof$ and $\xprof_Y$ is connected.
Let $\xprof'$ be the $(D \cap B)$-profile of $G[D'_C]$.
Because $\xprof'$ is a supergraph of $\xprof$ and $H[D \cap Y] = G[D' \cap Y]$, the union of the $(D \cap B)$-profiles of $G[D'_C]$ and $G[D' \cap Y]$ is connected, which together with the facts that every connected component of $G[D'_C]$ intersects $D \cap B$ and every connected component of $G[D' \cap Y]$ intersects $D \cap B$ implies that $G[D']$ is connected. 
\end{claimproof}

These two claims finish the proof.
\end{proof}

We define that an {\em{$(X,\delta,p)$-preserver}} of a graph $G$ is a graph $H$ that is both an $(X,\delta)$-folio-preserver and an $(X,p)$-intrusion-preserver.
Note that an $(X,\delta)$-folio-preserver is an $(X,\delta,p)$-preserver for all $p<0$ because then no intrusions exist, but not necessarily for any $p \ge 0$.

Our goal is to give a linear-time algorithm that, given a graph $G$ and the $(X,\delta)$-folio of $G$, constructs an $(X,\delta,p)$-preserver of $G$ whose size is bounded by a function of $|X|$, $\delta$, and $p$.
We give this algorithm by first showing that an $(X,\delta,p)$-preserver of bounded size exists, and then giving a linear-time algorithm that verifies whether a given $H$ is an $(X,p)$-intrusion-preserver of $G$.
We start with the verification algorithm.

\paragraph{Important separators.}
For the verification algorithm, we use a standard tool called \emph{important separators}, introduced by Marx~\cite{DBLP:journals/tcs/Marx06}.
Let $G$ be a graph and $A,S \subseteq V(G)$.
We denote by $\reachx_G(A,S)$ the set of vertices that are reachable from $A \setminus S$ in $G \setminus S$, i.e., $\reachx_G(A,S)$ is the union of those connected components of $G \setminus S$ that contain vertices from $A$.

\begin{definition}[Important separator]
Let $A,B \subseteq V(G)$.
A minimal $(A,B)$-separator $S$ is an \emph{important $(A,B)$-separator} if there exists no $(A,B)$-separator $S'$ such that $|S'| \le |S|$ and $\reachx_G(A,S) \subsetneq \reachx_G(A,S')$.
\end{definition}

The following lemma follows directly from the definition of important separators.
\begin{lemma}
\label{lem:impsepdominate}
For any $(A,B)$-separator $S$, there exists an important $(A,B)$-separator $S'$ so that $|S'| \le |S|$ and $\reachx_G(A,S) \subseteq \reachx_G(A,S')$.
\end{lemma}

The following upper bound and enumeration algorithm for important separators was given implicitly in~\cite{DBLP:journals/algorithmica/ChenLL09} and explicitly in~\cite{DBLP:journals/siamcomp/MarxR14}. See also~\cite[Section 8.2]{platypus} for a textbook exposition for edge~cuts.

\begin{lemma}[\cite{DBLP:journals/algorithmica/ChenLL09,DBLP:journals/siamcomp/MarxR14}]
\label{lem:impsepenum}
There are at most $4^k$ important $(A,B)$-separators of size at most $k$.
Furthermore, there is an algorithm that, given $G$, $A$, $B$, and $k \in \N$, in time $2^{\Oh(k)} \cdot \|G\|$ returns all important $(A,B)$-separators of size at most $k$.
\end{lemma}

Now we are ready to give the verification algorithm by using important separators.

\begin{lemma}
\label{lem:presverialgo}
There is an algorithm that, given graphs $G$ and $H$, a set $X \subseteq V(G) \cap V(H)$, integers $\delta$ and $p$, and the $(X,\delta)$-folio of $G$, in time $\Oh_{\|H\|,\delta,p}(\|G\|)$ decides whether $H$ is an $(X,\delta,p)$-preserver of $G$.
\end{lemma}
\begin{proof}
We first compute the $(X,\delta)$-folio of $H$ by brute-force in time $\Oh_{\|H\|}(1)$, and return ``no'' if it is not equal to the given $(X,\delta)$-folio of $G$.
We also have to check that $G[X] = H[X]$, which can be done in time $\Oh(|X|^2 \cdot (\|G\|+\|H\|))$ by checking edges one by one.
Then it remains to check that $H$ is an $(X,p)$-intrusion-preserver of $G$.

To do this, it is sufficient to give an algorithm that, given an integer $p' \le p$, disjoint sets $X_I,X_S \subseteq X$, and a graph $\xprof$ on the vertex set $X_I$, computes the maximum size of an $(X,p',X_I,X_S,\xprof)$-intrusion $C$ to $G$, in time $\Oh_{|X|,p'}(\|G\|)$.
Note that the maximum size of an $(X,p',X_I,X_S,\xprof)$-intrusion to $G$ is equal to the maximum size of an $(X-X_S,p',X_I,\emptyset,\xprof)$-intrusion to $G \setminus X_S$, so from now assume that $X_S = \emptyset$.

Let $X_R = X \setminus X_I$.
We use the algorithm of \Cref{lem:impsepenum} to enumerate all important $(X_I, X_R)$-separators of size at most $p'$ in the graph $G$ in time $2^{\Oh(p')} \|G\|$, of which there are at most $4^{p'}$.
For such important separator $S$, we can compute the set $C_S = \reachx_G(X_I, S)$ in time $\Oh(\|G\|)$, and use the definition of an intrusion to check if $C_S$ is an $(X, p', X_I, \emptyset, \xprof)$-intrusion to $G$ in time $|X|^{\Oh(1)} \|G\|$.

We claim that the maximum size of an $(X, p', X_I, \emptyset, \xprof)$-intrusion to $G$ is the maximum size of such a set $C_S$ that is an $(X, p', X_I, \emptyset, \xprof)$-intrusion to $G$.
It remains to prove this claim.

Let $C$ be a maximum size $(X,p',X_I,\emptyset,\xprof)$-intrusion to $G$.
Now, $N(C)$ is an $(X_I,X_R)$-separator of size $|N(C)| \le p'$.
Moreover, $\reachx_{G}(X_I,N(C)) = C$, because $C \cap X = X_I$ and every connected component of $G[C]$ intersects $X_I$.
Therefore, we enumerated an important $(X_I,X_R)$-separator $S$ with $C_S \supseteq C$ and $|N(C_S)| \le |N(C)| \le p'$.
We claim that $C_S$ is an $(X,p',X_I,\emptyset,\xprof)$-intrusion to $G$.
Because $C \subseteq C_S$, $C_S = \reachx_G(X_I, S)$, and $S$ is an $(X_I,X_R)$-separator, it holds that $C_S \cap X = X_I$ and every connected component of $G[C_S]$ intersects $X_I$.
The fact that $C \subseteq C_S$ implies that the $X_I$-profile of $G[C_S]$ is a supergraph of the $X_I$-profile of $G[C]$, and therefore a supergraph of $\xprof$.
\end{proof}

Then the remaining goal is to prove that a bounded-size $(X,\delta,p)$-preserver exists.

\paragraph{Push-preservers.}
An important ingredient to this will be \emph{push-preservers}, that we define now.
Let $S,B \subseteq V(G)$ and $p \ge |S|$.
We say that $S$ is \emph{$p$-pushed towards} $B$ if for every $(S, B)$-separator $S'$ it holds that either $S' \supseteq S$ or $|S'| > p$.
Note that if there exists an $(A,B)$-separator of size at most $p$, then there exists a minimal $(A,B)$-separator $S$ that has size at most $p$ and is $p$-pushed towards $B$.

Let $D \subseteq V(G)$ be disjoint from $S \cup B$ and $S$ be $p$-pushed towards $B$ in the graph $G \setminus D$.
We say that a graph $H$ is an \emph{$(S,B,D,p)$-push-preserver} of $G$ if 
\begin{itemize}[nosep]
\item $S \cup B \cup D \subseteq V(H)$; and 
\item $S$ is $p$-pushed towards $B$ in $H \setminus D$.
\end{itemize}

We show that being a $(Y,0)$-folio-preserver for a suitable set $Y$ implies being a push-preserver.

\begin{lemma}
\label{lem:pushpreserver}
Let $G$ be a graph, $p \in \Z$ an integer, $S,B \subseteq V(G)$ sets with $|S| \le p$, $D \subseteq V(G) \setminus (S \cup B)$, and suppose $S$ is $p$-pushed towards $B$ in $G-D$.
Then there exists a set $Y \subseteq V(G)$ of size $|Y| \le \Oh(p^2) + |B| + |D|$ so that if $H$ is a $(Y,0)$-folio-preserver of $G$, then $H$ is also an $(S,B,D,p)$-push-preserver of $G$.
\end{lemma}
\begin{proof}
By Menger's theorem, $S$ is $p$-pushed towards $B$ in $G \setminus D$ if and only if $G$ contains for every $v \in S \setminus B$ a collection $\mathcal{P}_v$ of $p+1$ vertex-disjoint $(S \cup N(v), B)$-paths that are disjoint with $D$.
Let $N'_v \subseteq N(v)$ be the starting vertices of the paths in $\mathcal{P}_v$ that do not start from $S$.
We construct $Y$ by letting it contain $S \cup B \cup D \cup \bigcup_{v \in S \setminus B} N'_v$.
Note that $|Y| \le p (p+2) + |B| + |D|$.

Let $H$ be a $(Y,0)$-folio-preserver of $G$.
It remains to show that $H$ is an $(S,B,D,p)$-push-preserver.
For each $v \in S \setminus B$, $G$ contains the $Y$-rooted minor $(H,\roots)$ with $p+1$ vertices and no edges, each $u \in V(H)$ having $\roots(u) \cap D = \emptyset$, $\roots(u) \cap B \neq \emptyset$, and either $\roots(u) \cap S \neq \emptyset$ or $\roots(u) \cap N'_v \neq \emptyset$.
As $H$ is a $(Y,0)$-folio-preserver of $G$, $H$ also contains this $Y$-rooted minor.
As $H[Y] = G[Y]$, $v$ is adjacent to all vertices in $N'_v$ also in $H$, so the existence of such $Y$-rooted minor in $H$ witnesses the existence of a collection of $p+1$ vertex-disjoint $(S \cup N_H(v), B)$-paths in $H \setminus D$.
\end{proof}

Another ingredient of our construction of a bounded-size $(X,\delta,p)$-preserver is a proof that the number of connected components of $H \setminus X$ can be reduced to be bounded as a function of $|X|$, $\delta$, and $p$.
This uses the similar result about $(X,\delta)$-folio-preservers we proved as \Cref{lem:unpumpcomponent}.

\begin{lemma}
\label{lem:reducecomps}
Let $G$ be a graph, $X \subseteq V(G)$, and $\delta \in \N$, $p \in \Z$ integers.
There exists an induced subgraph $H$ of $G$ so that $H$ is an $(X,\delta,p)$-preserver of $G$ and the number of connected components of $H \setminus X$ is at most $f(|X|,\delta,p)$, for a computable function $f$.
\end{lemma}
\begin{proof}
Let $H$ be an induced subgraph of $G$ so that $H$ is an $(X,\delta,p)$-preserver of $G$ but no induced subgraph of $H$ is an $(X,\delta,p)$-preserver of $G$.
Let $g$ be the function from \Cref{lem:unpumpcomponent} and $h = \max(p+1, g(|X|,\delta))$.
Note that $h$ is computable.
We claim that $|\cc(H - X)| \le 2^{|X|} \cdot h$.

Suppose otherwise.
Then there exists $S \subseteq X$ so that there are at least $h+1$ components $C \in \cc(H \setminus X)$ with $N(C) = S$.
By \Cref{lem:unpumpcomponent}, there exists such $C$ so that $H-C$ is an $(X,\delta)$-folio-preserver of $H$.
To obtain a contradiction, it remains to prove that $H - C$ is an $(X,p)$-intrusion-preserver of $H$.

Let $D$ be an $(X,p',X_I,X_S,\xprof)$-intrusion to $H \setminus C$, for $p' \le p$.
We construct a set $D'$ as follows.
If $D$ intersects $S$, then we set $D' = D \cup C$.
Otherwise, we set $D' = D$.
We have that $|D'| \ge |D|$.
Clearly, $D' \cap X = D \cap X$, every connected component of $H[D']$ intersects $X_I$, and the $X_I$-profile of $H[D']$ is a supergraph of $\xprof$.
It remains to prove that $|N_H(D') - X_S| \le p'$.

If $D$ is disjoint with $S$, this holds because $D' = D$ does not touch $C$.
Suppose $D$ intersects $S$.
It suffices to prove that in this case, $S \subseteq N_{H \setminus C}[D]$.
There are at least $h \ge p+1$ components $C' \in \cc(H - (X \cup C))$ with $N(C') = S$.
Therefore, for at least one such component $C'$ it must hold that $C' \subseteq D$, implying that $S \subseteq N_{H \setminus C}[D]$.
\end{proof}

We are now ready to show the main lemma bounding the size of an $(X,\delta,p)$-preserver.

\begin{lemma}
\label{lem:presbound}
Let $G$ be a graph, $X \subseteq V(G)$, and $\delta \ge 0 $, $p \in \Z$ integers.
There exists an $(X,\delta,p)$-preserver $H$ of $G$ so that $|V(H)| \le f(|X|,\delta,p)$, for a computable function $f$.
\end{lemma}
\begin{proof}
We prove this by induction on $p$.
First, if $p < 0$, then every $(X,\delta)$-folio-preserver is a $(X,\delta,p)$-preserver because there are no intrusions, so a computable bound follows from \Cref{lem:preserverbound}.
Suppose then that $p \ge 0$.

We iterate over all disjoint subsets $X_I,X_S \subseteq X$.
Let $X_R = X \setminus (X_I \cup X_S)$.
If there exists an $(X_I, X_R)$-separator of size at most $p$ in $G \setminus X_S$, then let $S$ be a minimal $(X_I,X_R)$-separator in $G - X_S$ that is $p$-pushed towards $X_R$.
Let $Y_S$ be a set given by \Cref{lem:pushpreserver} so that every $(Y_S,0)$-folio-preserver of $G$ is an $(S,X_R,X_S,p)$-push-preserver of $G$.
Let $Y$ be the union of $X$ and all those sets $Y_S$.

For each connected component $C \in \cc(G - Y)$ let $H_C$ be a graph that is a $(Y,\delta,p-1)$-preserver of $G[Y \cup C]$ with $|V(H)| \le f(|Y|, \delta, p-1)$, which exists by induction.
Assume also that $|V(H_C)| \le |Y \cup C|$, as otherwise we can simply let $H_C = G[Y \cup C]$.
Let $\rpmap$ map each $C \in \cc(G - Y)$ to $\rpmap(C) = H_C$ and $\bd$ map each $C \in \cc(G - Y)$ to $\bd(C) = Y$, and let $H$ be the $(Y,\rpmap,\bd)$-replacement of $G$.
By \Cref{lem:foliopreservreplace}, $H$ is a $(Y, \delta)$-folio-preserver of $G$.
Note that each connected component of $H \setminus Y$ has size at most $f(|Y|,\delta,p-1)$.

We will then show that $H$ is an $(X,\delta,p)$-preserver of $G$.
After that, we will complete the proof by applying \Cref{lem:reducecomps} to $H$.

Suppose $H$ is not an $(X,p)$-intrusion-preserver of $G$, and let $D$ be an $(X,p',X_I,X_S,\xprof)$-intrusion to~$H$, so that either no $(X,p',X_I,X_S,\xprof)$-intrusion to $G$ exists, or every $(X,p',X_I,X_S,\xprof)$-intrusion to $G$ is smaller than $D$.

\begin{claim}
\label{lem:presbound:claimcomp}
There exists a component $C_H \in \cc(H \setminus Y)$ so that $|N_H(D) \cap C_H| = p$ and $N_H(D) - X_S \subseteq C_H$.
\end{claim}
\begin{claimproof}
Note that $|N_H(D) \cap C_H| = p$ implies $N_H(D) - X_S \subseteq C_H$ because $|N_H(D) - X_S| \le p$ and $X_S \subseteq Y$.

Suppose that for every component $C_H \in \cc(H \setminus Y)$ it holds that $|N_H(D) \cap C_H| < p$.
We construct an $(X,p',X_I,X_S,\xprof)$-intrusion $D'$ to $G$ with $|D'| \ge |D|$, contradicting the choice of $D$, as follows.

Let $N_Y = N_H(D) \cap Y$.
For each component $C_H \in \cc(H \setminus Y)$, corresponding to a component $C \in \cc(G \setminus Y)$ with $\rpmap(C) = H_C = H[Y \cup C_H]$, let $N_C = N_{H_C}(D \cap V(H_C)) \setminus N_Y$.
Note that such sets $N_C$ for different components $C_H$ are disjoint, and in fact they form a partition of $N_H(D) \setminus N_Y$.
Note also that our assumption implies $|N_C| < p$.
Let $\xprof_C$ be the $(D \cap Y)$-profile of $H[D \cap V(H_C)]$.

Now $D \cap V(H_C)$ is a $(Y, |N_C|, D \cap Y, N_Y, \xprof_C)$-intrusion to $H_C$.
Because $|N_C| < p$, and $H_C$ is a $(Y,p-1)$-intrusion-preserver of $G[Y \cup C]$, there exists a $(Y, |N_C|, D \cap Y, N_Y, \xprof_C)$-intrusion $D'_C$ to $G[Y \cup C]$ with $|D'_C| \ge |D \cap V(H_C)|$.
Let $D'$ be the union of $D'_C$ over all $C$ and note that $|D'| \ge |D|$.
We claim that $D'$ is an $(X,p',X_I,X_S,\xprof)$-intrusion to $G$.

We have that $D' \cap Y = D \cap Y$, implying that $D' \cap X = D \cap X = X_I$.
For each component $C$, the $(D' \cap Y)$-profile of $G[D' \cap (C \cup Y)]$ is a supergraph of $\xprof_C$.
This implies that the $(D' \cap Y)$-profile of $G[D']$ is a supergraph of the $(D \cap Y)$-profile of $H[D]$.
As $X_I \subseteq D' \cap Y$, this immediately implies that the $X_I$-profile of $G[D']$ is a supergraph of $\xprof$.
Furthermore, as every connected component of $H[D]$ intersects $X_I$, this implies that every connected component of $G[D']$ that intersects $Y$ also intersects $X_I$.
Now, for a component $C$, every connected component of $G[D'_C]$ intersects $Y$, so this implies that indeed every connected component of $G[D']$ intersects $X_I$.

For a component $C$, we have that $|N_{G[C \cup Y]}(D'_C) \setminus N_Y| \le |N_C|$.
This implies that $|N_G(D') \setminus N_Y| \le |N_H(D) \setminus N_Y|$.
Because $|N_H(D) \setminus N_Y| + |N_Y \setminus X_S| \le p'$, this implies that $|N_G(D') \setminus X_S| \le p'$.
\end{claimproof}

In the remaining case, $p' = p$ and there exists a component $C_H \in \cc(H \setminus Y)$ so that $N_H(D) \cap C_H = N_H(D) \setminus X_S$.
Let $X_R = X \setminus (X_I \cup X_S)$.

\begin{claim}
There exists an $(X_I,X_R)$-separator of size at most $p$ in $G \setminus X_S$.
\end{claim}
\begin{claimproof}
If $G \setminus X_S$ contains a collection of $p+1$ vertex-disjoint $(X_I,X_R)$-paths, then also $H \setminus X_S$ does because $H$ is a $(Y,0)$-folio-preserver of $G$.
However, $N_H(D)$ is an $(X_I,X_R)$-separator in $H \setminus X_S$.
\end{claimproof}

Let $S$ be the minimal $(X_I,X_R)$-separator of size at most $p$ in $G \setminus X_S$ that we fixed earlier, with $Y_S \subseteq Y$ so that every $(Y_S,0)$-folio-preserver of $G$ is an $(S,X_R,X_S,p)$-push-preserver of $G$.
In particular, $H$ is an $(S,X_R,X_S,p)$-push-preserver of $G$.

The fact that $S$ is a minimal $(X_I,X_R)$-separator in $G \setminus X_S$ is witnessed in $G$ by (1) the non-existence of an $(X_I,X_R)$-path that is disjoint from $S \cup X_S$, and (2) for all $v \in S$, the existence of an $(X_I,X_R)$-path that intersects $S \cup X_S$ only in $v$.
Therefore, because $X \cup S \subseteq Y$ and $H$ is a $(Y,0)$-folio-preserver of $G$, $S$ is also a minimal $(X_I,X_R)$-separator in $H \setminus X_S$.

Let $A_H$ be the vertices reachable from $X_I \setminus S$ in $H \setminus (S \cup X_S)$ and $A_G$ the vertices reachable from $X_I \setminus S$ in $G \setminus (S \cup X_S)$.
Because $X \cup S \subseteq Y$ and $H$ is a $(Y,0)$-folio-preserver of $G$, $A_G \cap Y = A_H \cap Y$.
Note that $S$ may intersect $X_I$, meaning that $A_H$ and $A_G$ may not contain all of $X_I$ and may even be empty.

\begin{claim}
\label{lem:presbound:claimdints}
$D$ intersects with $S$.
\end{claim}
\begin{claimproof}
Suppose $D$ is disjoint with $S$.
Because $D$ is also disjoint with $X_S$ and all components of $H[D]$ contain a vertex from $X_I$, this implies that $D \subseteq A_H$, implying also that $S$ is disjoint with $X_I$.

We claim that then $A_G$ is an $(X,p,X_I,X_S,\xprof)$-intrusion to $G$ with $|A_G| \ge |A_H| \ge |D|$, which would contradict the choice of $D$.

First we show that $|A_G| \ge |A_H|$.
Let $C'$ be a connected component of $H \setminus Y$ and $C$ the corresponding component of $G \setminus Y$, and recall that $|C'| \le |C|$.
Because $S \cup X_S \subseteq Y$, $C'$ is either entirely in $A_H$ or disjoint from $A_H$, depending on whether $N_H(C') \subseteq Y$ intersects $A_H$.
Similarly, $C$ is either entirely in $A_G$ or disjoint from $A_G$, depending on whether $N_G(C) = N_H(C')$ intersects $A_G$.
Since $A_G \cap Y = A_H \cap Y$, we have that $C \subseteq A_G$ if and only if $C' \subseteq A_H$, implying with $A_G \cap Y = A_H \cap Y$ that $|A_G| \ge |A_H|$.

Then we show that $A_G$ is an $(X,p,X_I,X_S,\xprof)$-intrusion to $G$.
We have that $X_I \subseteq A_G$ because $X_I \subseteq A_H$.
Also, $A_G$ is disjoint from $X_S$ by definition, and disjoint from $X_R$ because $S$ is an $(X_I,X_R)$-separator in $G \setminus X_S$.
By the definition of $A_G$, every connected component of $G[A_G]$ intersects $X_I$.
We have that $N_G(A_G) \setminus X_S \subseteq S$, implying $|N_G(A_G) \setminus X_S| \le p$.
It remains to argue that the $X_I$-profile of $H[D]$ is a subgraph of the $X_I$-profile of $G[A_G]$.

Because $N_G(A_G) \subseteq Y$, the $(Y \cap A_G, 0)$-folio of $G[A_G]$ is uniquely determined by the $(Y,0)$-folio of $G$ and the set $Y \cap A_G$.
Similarly, because $N_H(A_H) \subseteq Y$, the $(Y \cap A_H,0)$-folio of $H[A_H]$ is uniquely determined by the $(Y,0)$-folio of $H$ and the set $Y \cap A_H$.
Therefore, because $H$ is a $(Y,0)$-folio-preserver of $G$ and $Y \cap A_G = Y \cap A_H$, the $(Y \cap A_G,0)$-folios of $G[A_G]$ and $H[A_H]$ are the same, implying that the $(X_I,0)$-folios of $G[A_G]$ and $H[A_H]$ are the same.
Now, because $D \subseteq A_H$, this implies that the $X_I$-profile of $H[D]$ is a subgraph of the $X_I$-profile of $G[A_G]$.
\end{claimproof}

\begin{claim}
\label{lem:presbound:claimdnotsubd}
$A_H \not\subseteq D$
\end{claim}
\begin{claimproof}
Suppose $A_H \subset D$, and recall that by \Cref{lem:presbound:claimdints}, it holds that $D \cap S \neq \emptyset$.
Because $S$ is a minimal $(X_I,X_R)$-separator in $H \setminus X_S$ and $A_H$ is the vertices reachable from $X_I \setminus S$ in $H \setminus (S \cup X_S)$, it holds that $S \subseteq N_H[D]$.
Now $N_H(D) \setminus X_S$ is an $(S,X_R)$-separator in $H \setminus X_S$, having $N_H(D) \not\supseteq S$ and $|N_H(D) \setminus X_S| \le p$, meaning that $S$ is not $p$-pushed towards $X_R$ in $H \setminus X_S$.
This contradicts the fact that $H$ is an $(X_I,X_R,X_S,p)$-push-preserver of $G$.
\end{claimproof}

We then show that \Cref{lem:presbound:claimcomp,lem:presbound:claimdints,lem:presbound:claimdnotsubd} covered all of the cases, i.e., $D$ cannot satisfy them all, leading to a contradiction.

Because every connected component of $H[A_H]$ intersects $X_I$, and $X_S$ is disjoint with $A_H$, it follows from \Cref{lem:presbound:claimdnotsubd} that $N_H(D) \setminus X_S$ must intersect $A_H$.
Because $S$ is a minimal $(A_H,X_R)$-separator in $H \setminus X_S$, and $D$ intersects with $S$ (\Cref{lem:presbound:claimdints}) but is disjoint with $X_R \cup X_S$, it must be that $N_H(D) \setminus X_S$ intersects $V(H) \setminus A_H$.
In particular, $N_H(D) \setminus X_S$ intersects both $A_H$ and $V(H) \setminus A_H$.

\Cref{lem:presbound:claimcomp} implies that there is a single connected component $C_H$ of $H-Y$ so that $N_H(D) \setminus X_S \subseteq C_H$.
Therefore, $C_H$ intersects both $A_H$ and $V(H) \setminus A_H$.
However, because $N_H(A_H) \subseteq Y$, every connected component of $H \setminus Y$ is either a subset of $A_H$ or disjoint with $A_H$, so we obtain a contradiction.

This finishes the proof that $H$ is an $(X,\delta,p)$-preserver of $G$.

We then apply \Cref{lem:reducecomps} with $H$, $Y$, $\delta$, and $p$ to obtain an induced subgraph $H'$ of $H$ so that $H'$ is $(Y,\delta,p)$-preserver of $H$ and thus a $(X,\delta,p)$-preserver of $G$, and the number of connected components of $H' - Y$ is at most $g(|Y|,\delta,p)$ for a computable function $g$.
Because $H'$ is an induced subgraph of $H$, each connected component of $H' - Y$ has at most $f(|Y|,\delta,p-1)$ vertices, and therefore $|V(H')| \le |Y| + g(|Y|,\delta,p) \cdot f(|Y|,\delta,p-1)$.
Because $|Y| \le p^{\Oh(1)} \cdot 2^{\Oh(|X|)}$, we obtain by induction that the desired computable function $f$ bounding $|V(H')|$ as a function of $|X|$, $\delta$, and $p$ exists.
\end{proof}

Then we put together \Cref{lem:presverialgo,lem:presbound} to obtain a linear-time algorithm for constructing an $(X,\delta,p)$-preserver.

\begin{lemma}
\label{lem:presmainresult}
There is an algorithm that, given a graph $G$, a set $X \subseteq V(G)$, integers $\delta$ and $p$, and the $(X,\delta)$-folio of $G$, in time $\Oh_{|X|,\delta,p}(\|G\|)$ returns an $(X,\delta,p)$-preserver $H$ of $G$ with $\|H\| \le f(|X|,\delta,p)$, for a computable function $f$.
\end{lemma}
\begin{proof}
We iterate over all graphs $H$ with $X \subseteq V(H)$ and $|V(H)|$ bounded by the bound of \Cref{lem:presbound}, and for each check if $H$ is an $(X,\delta,p)$-preserver of $G$ with the algorithm of \Cref{lem:presverialgo}.
\end{proof}

We then combine \Cref{lem:presmainresult} with the algorithm of \Cref{lem:foliopreservreplace} to finish the proof of \Cref{lem:ultimatechipreplace}, which we restate here.

\lemultimatechipreplace*
\begin{proof}
Let $g$ be the function from \Cref{lem:presmainresult}.
We set $f(k,\delta) = 2 \cdot g(k-1,\delta,k-1)$, denote $\alpha = f(k,\delta)$, and therefore assume that the chips in $\chipfamily$ are $(X,k,\alpha)$-chips.

We use \Cref{lem:presmainresult} to compute an $(N(C),\delta, k-1)$-preserver $H_C$ of $G[N[C]]$ for each $C \in \chipfamily$.
Because $|N(C)| < k$, the sum of the sizes of the graphs $G[N[C]]$ is $k^{\Oh(1)} \cdot \|G\|$ and they can also be computed in that time, so therefore computing these preservers $H_C$ takes in total $\Oh_{k,\delta}(\|G\|)$ time.
We have that $\|H_C\| \le \alpha/2$ for all $C$.

Let $Y = V(G) \setminus \bigcup \chipfamily$ and note that $\chipfamily = \cc(G - Y)$ and $X \subseteq Y$.
We construct the $Y$-replacement-map $(\rpmap,\bd)$ with $\rpmap(C) = H_C$ and $\bd(C) = N_G(C)$, and apply the algorithm of \Cref{lem:foliopreservreplace} with $G$, $(Y,\rpmap,\bd)$, $X$, and the $(N(C), \delta)$-model-folios of each $G[N[C]]$, to compute the $(Y,\delta)$-replacement $H$ of $G$.
By \Cref{lem:foliopreservreplace}, $H$ is a $(X,\delta)$-folio-preserver of $G$.
The algorithm of \Cref{lem:foliopreservreplace} runs in time $\Oh_{|X|,\delta}(\|G\| + \|H\|)$, and returns also an $(X,\delta)$-model-mapper $\Delta$ from $H$ to $G$.
We return $H$ and $\Delta$.
The running time bound of the algorithm is clear, it remains to assert the rest of the properties of $H$.

That $\|H_C\| \le \alpha/2$ and $|C| \ge \alpha$ for each $C\in \chipfamily$ implies that $|V(H)| \le |V(G)|$ and $\|H\| \le \|G\|$.

By \Cref{lem:replacebyintpresmaintain}, for any $(X,k)$-chip $D$ of $H$ that intersects $Y$, there exists an $(X,k)$-chip $D'$ of $G$ so that $D' \cap Y = D \cap Y$ and $|D'| \ge |D|$.
Let $F_G \subseteq V(G)$ be the set of $(X,k,\alpha)$-carvable vertices in $G$ and $F_H \subseteq V(H)$ the set of $(X,k,\alpha)$-carvable vertices in $H$.
It follows that $F_H \cap Y \subseteq F_G \cap Y$.
Now, $F_G \setminus Y = \bigcup \chipfamily$, and from $|V(H_C)| \le \alpha/2$ we get that $|V(H) \setminus Y| \le \alpha |\chipfamily|/2$.
This implies $|F_H| \le |F_G| - |\bigcup \chipfamily| + \alpha |\chipfamily|/2 \le |F_G| - |\bigcup \chipfamily|/2$.

This assertion of \Cref{lem:replacebyintpresmaintain} also implies that a largest $(X,k)$-chip of $H$ that intersects $Y$ has size at most the size of a largest $(X,k)$-chip of $G$ that intersects $Y$.
Because all $(X,k,\alpha)$-chips of $H$ intersect $Y$, this implies that the size of a largest $(X,k,\alpha)$-chip of $H$ is at most the size of a largest $(X,k,\alpha)$-chip of~$G$.
\end{proof}

\section{Clique-minor-free graphs}\label{sec:full-clique-minor-free}
The goal of this section is to give a proof of \cref{thm:cliqueFree}.
Recall that by \Cref{lem:binsearch}, for this it suffices to give an $\Oh_{|X|,\delta,h}(\|G\|^{1+o(1)})$-time algorithm for the \FolioCliqueEx problem, which asks to either return an $(X,\delta)$-model-folio of $G$, or conclude that $G$ contains $K_h$ as a minor.

In \Cref{sec:compact}, we gave this algorithm under the assumptions that $X$ is well-linked, large enough, and $G$ is $(X,k,\alpha)$-compact for large enough $k$.
In this section we devise a recursive scheme to reduce the general case to this case.
Two ingredients of this recursive scheme are the algorithms for carving chips and replacing them from \Cref{sec:carving}, but we need one more ingredient to guarantee that the depth of the recursion tree is bounded by $\Oh_k(\log n)$, which we give now in \Cref{subsec:reedslemma}.

\subsection{Reed's Lemma}
\label{subsec:reedslemma}

Reed~\cite{DBLP:conf/stoc/Reed92} gave an algorithm that, assuming an $n$-vertex graph $G$ contains a separator $S$ with $|S| \le k$ so that every component of $G - S$ contains at most $\frac{2}{3} n$ vertices, finds in time $\Oh_k(\|G\|)$ a separator $S'$ with $|S'| \le k$ so that every component of $G - S'$ contains at most $\frac{3}{4} n$ vertices.
Next we give a version of Reed's algorithm, with worse guarantees on the separator, that in the negative case when no separator exists, outputs a well-linked set of $G$ that roughly captures the cardinality tangle of $G$ (the orientation of all the order-at-most-$k$ separations that always points to the side with larger cardinality).
We name our lemma the ``Reed's Lemma'', even though to the best of our knowledge the extension to output such well-linked set is novel to this work.

\begin{lemma}[Reed's Lemma]
\label{lem:reed}
 There is an algorithm that, given an $n$-vertex graph $G$ and an integer $k \ge 1$, in time $2^{\Oh(k)}\cdot \|G\|$ returns either
 \begin{enumerate}[nosep,label=(\arabic*)]
  \item\label{o:separation} a separation $(A,B)$ of $G$ of order at most $3k$ such that $|A\setminus B|\leq (1-\frac{1}{100k^2})n$ and $|B\setminus A|\leq (1-\frac{1}{100k^2})n$; or
  \item\label{o:welllinked} a well-linked set $X$ of size $3k$ such that for every separation $(A,B)$ of order at most $k$, if $|A\cap X|>k$ then $|A\setminus B|\geq \frac{n}{100k^2}$.
 \end{enumerate}
\end{lemma}
\begin{proof}
 For convenience, we let
 $$\beta\coloneqq \frac{n}{100k^2}.$$
 We may assume that $\beta>1$, or equivalently $n>100k^2$, because otherwise returning a separation $(A,B)=(\{u\},V(G))$ for any vertex $u$ is a valid outcome~\ref{o:separation}.

 We will define a procedure that in consecutive iterations constructs sets $S_0\subseteq S_1\subseteq S_2\subseteq \ldots$, starting from $S_0\coloneqq \emptyset$ and adding one vertex at a time; that is, $|S_i\setminus S_{i-1}|=1$ for all $i\geq 1$. The procedure always terminates before reaching iteration $3k+1$, so it will be always the case that $i\leq 3k$. We also let $G_i\coloneqq G-S_i$.

 The procedure works as follows. Suppose we have already defined $S_{i-1}$ for some~$i\in \{1,\ldots,3k\}$.
 The $i$th iteration proceeds in three steps.

 \paragraph*{Step 1:}
  We investigate the connected components of the graph $G_{i-1}$. If each of those connected components has at most $\frac{n}{2}$ vertices, then using a standard argument (see e.g.~\cite[Lemma 7.20]{platypus}) we find a partition of the components of $G_{i-1}$ into sets $\rcomps_A$ and $\rcomps_B$ so that $|\bigcup_{C \in \rcomps_A} C|\leq \frac{2}{3}n$ and $|\bigcup_{C \in \rcomps_B} C|\leq \frac{2}{3}n$.~Then
  \[(A,B)\coloneqq \left(S\cup \bigcup_{C\in \rcomps_A} C,S\cup \bigcup_{C\in \rcomps_B} C\right)\]
  is a separation of order $i\leq 3k$ that satisfies the requirements of outcome \ref{o:separation} and can be output by the algorithm.

  Otherwise, let $H = G_{i+1}[C]$, where $C$ is the unique connected component of $G_{i-1}$ with $|C| > n/2$.

  \newcommand{\chld}{\mathsf{chld}}

 \paragraph*{Step 2:}
  Construct any spanning tree $T$ of $H$ and root it at any vertex $r$. For a vertex $u$ of $T$, let $T_u$ be the subtree of $T$ induced by $u$ and all its descendants, $\chld(u)$ the set of the children of $u$ in $T$, and $n_u\coloneqq |V(T_u)|$. Note that all the numbers $n_u$ can be computed in time $\Oh(n)$ by bottom-up dynamic programming on $T$.

  Call a vertex $u$ of $T$ {\em{rich}} if the following condition holds:
  \begin{quote}
  Vertex $u$ has more than $3k$ children, and the sum of all the numbers $n_v$ for $v\in \chld(u)$, except for the $3k$ highest ones, is at least $\beta-1$.
  \end{quote}
  We note the following.
  \begin{claim}\label{cl:locateRich}
   One can in $\Oh(kn)$ time find all rich vertices in $T$.
  \end{claim}
  \begin{claimproof}
  For a vertex $u$, whether $u$ is rich can be decided in time $\Oh(k)$ by looking at $n_u$ and the $3k$ highest numbers among the multiset $\{n_v\colon v\in\chld(u)\}$. These can be computed in time $\Oh(k\cdot |\chld(u)|)$ per vertex~$u$, which sums up to $\Oh(kn)$ time in total.
  \end{claimproof}

  We apply the algorithm of \cref{cl:locateRich} to determine whether $T$ has any rich vertex.
  If this is the case, we select any rich vertex $u$, define
  $$S_i\coloneqq S_{i-1}\cup \{u\},$$
  and either proceed to iteration $i+1$ provided $i<3k$, or end the iteration with $S_{3k}$ defined provided $i=3k$.

  For the sake of future reasoning, we note the following.

  \begin{claim}\label{cl:rich-drags}
   For every separation $(A,B)$ of $G$ of order at most $3k$, if $u\in A\setminus B$, then $|A\setminus B|\geq \beta$.
  \end{claim}
  \begin{claimproof}
   Let $\Tt\coloneqq \{T_v\colon v\in \chld(u)\textrm{ and }T_v\textrm{ is disjoint with }A\cap B\}$. As $u\in A\setminus B$, we have $\{u\}\cup \bigcup_{L\in \Tt} V(L)\subseteq A\setminus B$.
   As $|A\cap B|\leq 3k$, $\Tt$ contains all but at most $3k$ subtrees of $T$ rooted at the children of $u$, which in particular, by the richness of $u$, contain at least $\beta-1$ vertices in total. Hence,
   $$|A\setminus B|\geq \left|\{u\}\cup \bigcup_{L\in \Tt} V(L)\right|\geq 1+(\beta-1)=\beta.\qedhere$$
  \end{claimproof}

 \paragraph*{Step 3:} Otherwise, $T$ contains no rich vertices. We break $T$ into small subtrees as follows.

 \begin{claim}
  Assuming $T$ has no rich vertex, one can in $\Oh(n)$ time compute a family $\Ll$ of vertex-disjoint subtrees of $T$ so that
  $$|\Ll|\geq 3k\qquad\textrm{and}\qquad |V(L)|\geq \beta\quad \textrm{for each }L\in \Ll.$$
 \end{claim}
 \begin{claimproof}
  We perform the following procedure that in consecutive iterations constructs subtrees $L^1,L^2,\ldots$ and $T=T^0,T^1,T^2,\ldots$ of $T$. In iteration $i$, we do the following:
  \begin{itemize}[nosep]
   \item Find the lowest vertex $u$ of $T^{i-1}$ such that the subtree of $T^{i-1}$ rooted at $u$, call it $T^{i-1}_u$, has at least $\beta$ vertices.
   \item Define $L^i\coloneqq T^{i-1}_u$, $T^i\coloneqq T^{i-1}-V(L^i)$, and proceed to the next iteration.
  \end{itemize}
  The procedure terminates once the currently handled tree $T^{i-1}$ already has less than $\beta$ vertices. We define $\Ll$ to be the set of all the constructed subtrees $L^i$; it is clear that they are vertex-disjoint and that each has at least $\beta$ vertices. Also, it is straightforward to implement the procedure so that it runs in time $\Oh(n)$, by a bottom-up scan on $T$.

  We are left with proving that $|\Ll|\geq 3k$. For this, we observe that from the construction and the assumption that $T$ has no rich vertex, it follows that for each $L\in \Ll$, we have $|V(L)|\leq \beta+3k\beta=(3k+1)\beta$. Indeed, assuming $L=L^i=T^{i-1}_u$ for some vertex $u$ of $T^{i-1}$, every subtree of $T^{i-1}$ rooted at a child of $u$ has less than $\beta$ vertices, and the total number of vertices contained in all except $3k$ largest such subtrees is upper bounded by $\beta-1$.
  Now, as $|V(T)|=|V(H)|>\frac{n}{2}$ and the iteration stops once there are fewer than $\beta$ vertices left, we conclude that $$|\Ll|\geq \frac{\frac{n}{2}-\beta}{(3k+1)\beta}\geq 3k.\qedhere$$
 \end{claimproof}

 Let $\Ll'$ be any subset of $\Ll$ of size exactly $3k$ and let $X$ be a set constructed by selecting one arbitrary vertex from each subtree $L\in \Ll'$. We use the algorithm of \cref{lem:testing-well-linked} to test whether $X$ is well-linked in~$G$, which takes time $2^{\Oh(k)}\cdot \|G\|$.
 Depending on the outcome of this test, we consider two options.

 \subparagraph*{Case 1: $X$ is not well-linked.} In this case, the algorithm of \cref{lem:testing-well-linked} outputs a separation $(A,B)$ such that $|A\cap X|>|A\cap B|$ and $|B\cap X|>|A\cap B|$. In particular, the order of $(A,B)$ is $|A\cap B|<|X|=3k$. Let $\Ll'_A$ consist of all the trees of $\Ll'$ that contain a vertex of $A\cap X$. Then $|\Ll'_A|=|A\cap X|>|A\cap B|$, hence there is at least one tree $L\in \LL'_A$ that is disjoint with $A\cap B$. Since $L$ contains also a vertex of $A$, it follows that $V(L)\subseteq A\setminus B$. Recalling that $|V(L)|\geq \beta$, we conclude that $|A - B| \ge \beta$, implying 
 \[|B-A| \le n-\beta \le \left(1-\frac{1}{100k^2}\right)n.\]
 Symmetrically, we also conclude that $|A-B| \le \left(1-\frac{1}{100k^2}\right)n$.
 So the separation $(A,B)$ satisfies all the requirements of outcome \ref{o:separation} and can be output by the algorithm.

 \subparagraph*{Case 2: $X$ is well-linked.} We claim that in this case, $X$ satisfies the requirements of outcome \ref{o:welllinked} and thus can be output by the algorithm. For this, it remains to prove that for every separation $(A,B)$ of order at most $k$, if $|A\cap X|>k$ then $|A|\geq \frac{n}{100k^2}$. As in the first case, let $\Ll'_A$ consist of all the trees of $\Ll'$ that intersect $|A\cap X|$. As $|\Ll'_A|=|A\cap X|>k\geq |A\cap B|$, there is a tree $L\in \Ll'_A$ that is not intersected by $A\cap B$. Then $L$ intersecting $A$ implies that $V(L)\subseteq A\setminus B$, hence $|A\setminus B|\geq |V(L)|\geq \beta=\frac{n}{100k^2}$.

 \paragraph*{After the iteration.} We are left with explaining what happens when the iterative procedure presented above managed to execute $3k$ iterations and constructed the set $S_{3k}$. Denote $X\coloneqq S_{3k}$.

 Again, we apply the algorithm of \cref{lem:testing-well-linked} to test whether $X$ is well-linked in $G$ in time $2^{\Oh(k)}\cdot \|G\|$. And again, we have two possible outcomes.

 \subparagraph*{Case 1: $X$ is not well-linked.} Again, the algorithm of \cref{lem:testing-well-linked} provides a separation $(A,B)$ such that $|A\cap X|>|A\cap B|$ and $|B\cap X|>|A\cap B|$, which in particular has order $|A\cap B|<|X|=3k$.

 Since $|A\cap X|>|A\cap B|$, there exists a vertex $u\in A\cap X$ that is also in $A\setminus B$. As $u\in X=S_{3k}$, we may apply \cref{cl:rich-drags} to $u$ and conclude that $|A\setminus B|\geq \beta$. This implies that $|B\setminus A|\leq |B|\leq n-\beta=(1-\frac{1}{100k^2})n$, and symmetrically also $|A\setminus B|\leq (1-\frac{1}{100k^2})n$. So separation $(A,B)$ satisfies all the requirements of outcome \ref{o:separation} and can be output by the algorithm.

 \subparagraph*{Case 2: $X$ is well-linked.} We again claim that $X$ satisfies the requirements of outcome \ref{o:welllinked}. For this, consider any separation $(A,B)$ of order at most $k$ such that $|A\cap X|>k$. Then there exists a vertex $u\in (A\setminus B)\cap X$, to which we can apply \cref{cl:rich-drags} and conclude that $|A\setminus B|\geq \beta=\frac{n}{100k^2}$, as required.
\end{proof}

\subsection{The algorithm}
\label{sec:omnomnomnomnom}
Finally, with all the tools prepared, we can prove \cref{thm:cliqueFree}.
Recall that we are given an instance $(G,X,\delta)$ of \Folio and an integer $h$, and the goal is to either compute an $(X,\delta)$-model-folio $G$, or return that $G$ contains $K_h$ as a minor.
The desired running time is $\Oh_{|X|,\delta,h}(\|G\|^{1+o(1)})$.

\subsubsection{Preparations}
Our algorithm is recursive, so let us now fix some parameters that will stay the same throughout the recursive procedure.
We set $\delta \coloneqq \max(\delta, h)$, to guarantee that preserving the $(X,\delta)$-folio preserves the existence of a $K_h$ minor model.
Let $x$ be the size of $X$ in the initial input, and let $k'$ be the constant from \Cref{thm:compact-solvable} (the algorithm for compact graphs), computable from $h$.
We let $k \coloneqq \max(x+1, k')$.
Note that for any~$\alpha$, being $(X,k,\alpha)$-compact implies being $(X,k',\alpha)$-compact.
Let us then fix $\alpha \coloneqq f(k, \delta)$, where $f$ is the function from \Cref{lem:ultimatechipreplace}, asserting that the algorithm of the lemma takes a family of $(X,k,f(k,\delta))$-chips as an input.
These parameters $h$, $\delta$, $k$, and $\alpha$ will stay the same throughout all recursive calls of the algorithm.
Note that they are bounded by a computable function of the parameters $|X|$, $\delta$, and $h$ of the initial input.

Now that $\delta$ is globally fixed, we can define that an instance is a pair $\ins = (G,X)$.
We denote $G(\ins) = G$ and $X(\ins) = X$.
Let us throughout use $n$ as a shorthand to $|V(G)|$.

Our algorithm uses quite complicated recursion with many cases.
Multiple of these cases uses the following recursive construction.
Let $(G,X)$ be an instance and $(A,B)$ a separation.
We denote by $(G,X) \recurse (A,B)$ the instance $(G[A], (X \cap A) \cup (A \cap B))$.
Note that in this notation, the order in which $A$ and $B$ are written matters.
Let $\ins_A = (G,X) \recurse (A,B)$ and $\ins_B = (G,X) \recurse (B,A)$.
Recall that by \Cref{lem:folio_over_separator}, given an $(X(\ins_A),\delta)$-model-folio of $G(\ins_A)$ and an $(X(\ins_B),\delta)$-model-folio of $G(\ins_B)$, we can compute an $(X,\delta)$-model-folio of $(G,X)$ in time $\Oh_{|X|+|A \cap B|,\delta}(\|G\|)$.

Then we define two invariants of the instance.
The \emph{size} of an instance $(G,X)$ is $|(G,X)| = n-|X|$.
Note that if $(A,B)$ is a separation, then $|(G,X) \recurse (A,B)| + |(G,X) \recurse (B,A)| \le |(G,X)|$.
Then, we define that the \emph{balance} of an instance $(G,X)$ is the size of a largest $(X,k,\alpha)$-chip of $G$, or $0$ if $G$ has no $(X,k,\alpha)$-chips.
Note that the balance of $(G,X)$ is at most $|(G,X)|$, and that if $|X| < k$, then the balance is the either the size of a largest connected component of $G - X$, or $0$ if all of them have size smaller than~$\alpha$.
Let us show that recursion by separations does not increase balance.

\begin{lemma}
\label{lem:maintbalance}
If $(A,B)$ is a separation of $G$, then the balance of $(G,X) \recurse (A,B)$ is at most the balance of $(G,X)$.
\end{lemma}
\begin{proof}
Let $C$ be an $((X \cap A) \cup (A \cap B), k, \alpha)$-chip of $G[A]$.
Now, $C \subseteq A \setminus (A \cap B)$, meaning that $N_{G[A]}[C] = N_G[C]$.
This implies that $C$ is also an $(X,k,\alpha)$-chip of $G$.
\end{proof}

We use the following result of~\cite{DBLP:journals/combinatorica/Kostochka84,thomason_1984} to assume that $\|G\| = \Oh_{|X|,h}(|(G,X)|)$.
\begin{lemma}[\cite{DBLP:journals/combinatorica/Kostochka84,thomason_1984}]
\label{lem:minorfreeedgebound}
There is a constant $c$ so that if $|E(G)| \ge c h \sqrt{\log h} \cdot |V(G)|$, then $G$ contains $K_h$ as a minor.
\end{lemma}
To avoid invoking \Cref{lem:minorfreeedgebound} explicitly many times, let us assume that whenever we construct a new instance and recurse into it, we first check if \Cref{lem:minorfreeedgebound} implies that it contains $K_h$-minor, and if yes, instead of making a recursive call, we immediately return that the graph contains $K_h$-minor.
We therefore always assume that $\|G\| = \Oh_h(n) = \Oh_{|X|,h}(|(G,X)|)$.

As another similar assumption, we assume that whenever $n \le \max(200k^3, \alpha+k)$, we do not use recursion, but instead solve the instance $(G,X)$ by brute-force in time $\Oh_{k,\alpha}(1)$.

\subsubsection{Overview}
We then briefly overview our recursive algorithm.
The algorithm takes as input an instance $(G,X)$, and returns the $(X,\delta)$-model-folio of $G$, or that $G$ contains a $K_h$-minor.
It has four main cases, one of them containing two subcases.
The summary of the cases is as follows.

\begin{enumerate}
\item[\textbf{\algcasecc:}] $G - X$ is not connected.
\begin{itemize}[nosep]
\item[Action:] Recurse into the instances $\ins_C = (G[X \cup C], X)$ for all $C \in \cc(G - X)$.
\item[Progress:] Each $G(\ins_C) - X(\ins_C)$ is connected.
\end{itemize}
\item[\textbf{\algcasenwl:}] $X$ is not well-linked.
\begin{itemize}[nosep]
\item[Action:] Let $(A,B)$ be a separation witnessing that $X$ is not well-linked. Recurse into the instances $\ins_A = (G,X) \recurse (A,B)$ and $\ins_B = (G,X) \recurse (B,A)$.
\item[Progress:] $|X(\ins_A)|,|X(\ins_B)| < |X|$
\end{itemize}
\item[\textbf{\algcasewlsmall:}] $X$ is well-linked but $|X| < k$.
\begin{itemize}[nosep]
\item[Note:] Because $G-X$ is connected, $|X|<k$, and $n \ge \alpha+k$, the balance of $(G,X)$ is $|(G,X)|>n-k$.
\item[Action:] Apply Reed's Lemma (\Cref{lem:reed}) with the parameter $k$. Based on whether it returns a separation or a well-linked set, go to \algcasebalsep or \algcaserwl.
\item[\textbf{\algcasebalsep:}] Separation $(A,B)$.
\begin{itemize}[nosep]
\item[Action:] Recurse into $\ins_A = (G,X) \recurse (A,B)$ and $\ins_B = (G,X) \recurse (B,A)$.
\item[Progress:] The balance of $\ins_A$ and $\ins_B$ is at most $n-\Omega(n/k^2)$.
\end{itemize}
\item[\textbf{\algcaserwl:}] Well-linked set $W$.
\begin{itemize}[nosep]
\item[Action:] Recurse into $(G,X \cup W)$.
\item[Progress:] The balance of $(G,X \cup W)$ is at most $n-\Omega(n/k^2)$.
\end{itemize}
\end{itemize}
\item[\textbf{\algcasecarve:}] $X$ is well-linked and $|X| \ge k$.
\begin{itemize}[nosep]
\item[Action:] Roughly speaking, implement the following process by using the tools from \Cref{sec:carving}: while there exists an $(X,k,\alpha)$-chip $C$, solve the instance $\ins_C = (G[N[C]], N(C))$ recursively and replace $G[N[C]]$ with a smaller graph. After no $(X,k,\alpha)$-chips exist, solve the instance by using \Cref{thm:compact-solvable} (the algorithm for compact graphs).
\item[Progress:] Every $\ins_C$ has $|X(\ins_C)| < k$. Moreover, we guarantee that the sum of the sizes of the instances $\ins_C$ is at most $|(G,X)| \cdot (1+1/\log n)$.
\end{itemize}
\end{enumerate}

Let us then briefly describe the main argument for bounding the running time.
The main idea is to show that the recursion-depth is at most $\Oh_k(\log n)$, and that the sum of the sizes of the instances we recurse into from an instance $(G,X)$ is at most $|(G,X)| \cdot (1+1/\log n)$.
This implies that the total size of all instances across the whole recursion tree can be bounded by $\Oh_k(n \log n)$.

The argument for bounding the recursion-depth by $\Oh_k(\log n)$ is that the balance never grows when going down the recursion tree, but \algcasewlsmall decreases the balance by a large factor.
Because $|X|$ always stays bounded by $4k$, decreases in \algcasenwl and \algcasecarve, and does not increase in \algcasecc, \algcasewlsmall must happen often enough.

\subsubsection{Details}
We then describe and analyze our recursive algorithm in detail.
We first consider individual recursive calls, describing and analyzing the algorithm on them.
Then we finish by bounding the overall running time across all recursive calls.

Each recursive call starts by testing if $G-X$ is connected.
\paragraph{\algcasecc: $G-X$ is not connected.}
If $G-X$ is not connected, we recursively solve the instance $\ins_C = (G[C \cup X], X)$ for each connected component $C \in \cc(G-X)$.
Note that $\ins_C = (G,X) \recurse (C \cup X, V(G) - C)$, so by \Cref{lem:maintbalance} the balance of $\ins_C$ is at most the balance of $(G,X)$.
Also, clearly the sum of the sizes of these instances $\ins_C$ is exactly the size of $(G,X)$.
Note that the graphs $G[C \cup X]$ for all $C$ can be constructed in total $\Oh_{|X|}(\|G\|)$ time.

If some recursive call returns that $G[C \cup X]$ contains a $K_h$-minor, then we can immediately return that $G$ contains $K_h$-minor.
Otherwise, all of them return an $(X,\delta)$-model-folio of $G[C \cup X]$.
To combine the model-folios into an $(X,\delta)$-model-folio of $G$, we run the following divide-and-conquer algorithm.
The algorithm takes as input a subset $\comps \subseteq \cc(G-X)$, and returns an $(X,\delta)$-model-folio of $G[X \cup \bigcup \comps]$.
When $|\comps| = 1$, we already have the answer and can return it directly.
When $|\comps| \ge 2$, we partition it in a balanced way into subsets $\comps = \comps_1 \cup \comps_2$, call the algorithm recursively with $\comps_1$ and $\comps_2$, and then use \Cref{lem:folio_over_separator} with the graph $G[X \cup \bigcup \comps]$ and the separation $(X \cup \bigcup \comps_1, X \cup \bigcup \comps_2)$ to compute the $(X,\delta)$-model-folio of $G[X \cup \bigcup \comps]$.
Clearly, this algorithm runs in $\Oh_{|X|,\delta}(\|G\| \log n)$ time.

The total running time of \algcasecc is $\Oh_{|X|,\delta}(\|G\| \log n)$, plus the running times of recursive calls.

\medskip

Then, if $G-X$ is connected, we use the algorithm of \Cref{lem:testing-well-linked} to test if $X$ if well-linked in time $2^{\Oh(|X|)} \cdot \|G\|$.
If $X$ is not well-linked, we go to \algcasenwl, and if $X$ is well-linked, we go to either \algcasewlsmall or \algcasecarve, based on whether $|X| < k$ or $|X|\geq k$.

\paragraph{\algcasenwl: $X$ is not well-linked.}
Let $(A,B)$ be a separation of $G$ such that $|A \cap X| > |A \cap B|$ and $|B \cap X| > |A \cap B|$, provided by the algorithm of \Cref{lem:testing-well-linked}.
We call the algorithm recursively on the instances $\ins_A = (G,X) \recurse (A,B)$ and $\ins_B = (G,X) \recurse (B,A)$, and then use \Cref{lem:folio_over_separator} to combine the results of the recursive calls.
Note that if either recursive call returns that the graph contains $K_h$ minor, then also $G$ contains $K_h$ minor because $G(\ins_A)$ and $G(\ins_B)$ are subgraphs of $G$.
Note also that
\[|X(\ins_A)|=|(A-B)\cap X|+|A\cap B|<|(A-B)\cap X|+|B\cap X|=|X|,\]
and similarly $|X(\ins_B)| < |X|$; that is, the size of $X$ shrinks in the recursive calls.
By \Cref{lem:maintbalance}, the balances of both $\ins_A$ and $\ins_B$ are at most the balance of $(G,X)$.
Note that $|\ins_A|+|\ins_B| \le |(G,X)|$.

The total running time of \algcasenwl is $\Oh_{|X|,\delta}(\|G\|)$, plus the running times of the two recursive calls.

\paragraph{\algcasewlsmall: $X$ is well-linked but $|X| < k$.}
Note that because $G-X$ is connected, $|X|<k$, and $n \ge \alpha+k$, the balance of $(G,X)$ is at this point $|(G,X)| > n-k$.
We apply Reed's Lemma (\Cref{lem:reed}) with $G$ and the parameter~$k$, taking $2^{\Oh(k)} \cdot \|G\|$ time.
There are two different subcases, \algcasebalsep and \algcaserwl, depending on its output.

\paragraph{\algcasebalsep: Separation.}
Suppose first that Reed's Lemma returns a separation $(A,B)$ of order at most $3k$ so that $|A - B|, |B - A| \le n-n/(100k^2)$.
In that case, we proceed similarly as in \algcasenwl: we call the algorithm recursively on the instances $\ins_A = (G,X) \recurse (A,B)$ and $\ins_B = (G,X) \recurse (B,A)$, and use \Cref{lem:folio_over_separator} to combine the results.
Similarly to \algcasenwl, we have that $|\ins_A| + |\ins_B| \le |(G,X)|$.
However, this time we have the additional property that the balance of both $\ins_A$ and $\ins_B$ is at most $n-n/(100k^2)$, simply because the sizes of both $\ins_A$ and $\ins_A$ are at most $n-n/(100k^2)$.
This time, $|X(\ins_A)|$ and $|X(\ins_A)|$ can be more than $|X|$, but because $|A \cap B| \le 3k$ and $|X| < k$, they are at most $4k-1$.

The running time of \algcasebalsep is the running time of the two recursive calls, plus $\Oh_{k,\delta}(\|G\|)$ (note that $|A \cap B|+|X| = \Oh(k)$).

\paragraph{\algcaserwl: Well-linked set.}
Suppose then that Reed's Lemma returns a well-linked set $W$ of size $3k$ so that for every separation $(A,B)$ of order at most $k$, if $|A \cap W| > k$, then $|A - B| \ge n/(100k^2)$.
We call the algorithm recursively with the instance $(G, X \cup W)$.
Because $|X| < k$ and $|W| \le 3k$, we have that $|X \cup W| < 4k$.
Note that an $(X \cup W,\delta)$-model-folio of $G$ can be turned into an $(X, \delta)$-model-folio of $G$ in time $\Oh_{k,\delta}(\|G\|)$ in a straightforward way.
Clearly, the size of $(G,X \cup W)$ is at most the size of $(G,X)$.
The main property we gain in this case is the following.

\begin{lemma}
The balance of $(G, X \cup W)$ is at most $n-n/(100k^2)$.
\end{lemma}
\begin{proof}
Let $C$ be an $(X \cup W, k, \alpha)$-chip of $G$.
We have that $(V(G) - C, N[C])$ is a separation of $G$ of order $|N(C)|<k$.
Because $C$ is disjoint with $W$, it holds that $|(V(G) - C) \cap W| = 3k > k$, and therefore by the properties of $W$, $|V(G) - N[C]| \ge n/(100k^2)$, implying $|C| \le n - n/(100k^2)$.
\end{proof}

The running time of \algcaserwl is the running time of the recursive call, plus $\Oh_{k,\delta}(\|G\|)$.

\paragraph{\algcasecarve: $X$ is well-linked and $|X| \ge k$.}
Let initially $G_1 \coloneqq G$ and $i \coloneqq 1$.
We repeat the following process, maintaining the invariants that the $(X,\delta)$-folio of $G_i$ is equal to the $(X,\delta)$-folio of $G$, $|V(G_i)| \le n$, $\|G_i\| \le \|G\|$, and the balance of $(G_i,X)$ is at most the balance of $(G,X)$.

We apply the algorithm of \Cref{lem:computechipfamilyfinal} with $G_i$, $X$, $k$, and $\alpha$.
In time $\Oh_{k,\alpha}(\|G_i\|^{1+o(1)})$ it returns a family $\chipfamily_i$ of pairwise non-touching $(X,k,\alpha)$-chips so that if $\numcarv_i$ is the number of $(X,k,\alpha)$-carvable vertices of $G_i$, then $|\bigcup \chipfamily_i| \ge \numcarv_i/\Oh_{k,\alpha}(\log |V(G_i)|)$.
If $\chipfamily_i = \emptyset$, then the process stops.

Otherwise, for each $C \in \chipfamily_i$, we compute the $(N(C), \delta)$-model-folio of $G_i[N[C]]$.
For this, there are two cases depending on the size of $C$.
We say that a chip $C \in \chipfamily_i$ is \emph{small} if $|C| \le 2 \alpha \log n$ and \emph{large} otherwise.
If $C$ is small, then we compute the $(N(C), \delta)$-model-folio of $G_i[N[C]]$ by calling the Robertson-Seymour algorithm~\cite{GM13}, which runs in $\Oh_{k,\delta}(|C|^3) \le \Oh_{k,\delta,\alpha}(\log^{3} n)$ time.
If $C$ is large, we call our algorithm recursively with $\ins_C = (G_i[N[C]], N(C))$.
If it returns that $G_i[N[C]]$ contains a $K_h$-minor, then because $\delta \ge h$ and the $(X,\delta)$-folio of $G_i$ is equal to the $(X,\delta)$-folio of $G$, $G$ also contains a $K_h$-minor, and we can return immediately.
Note that because $\ins_C = (G_i,X) \recurse (N[C], V(G_i)-C)$, the balance of $\ins_C$ is at most the balance of $(G_i,X)$, which is at most the balance of $(G,X)$.
The running time of this is bounded by the running times of the recursive calls on the large instances, plus $\Oh_{k,\delta,\alpha}(\|G_i\| \log^3 n)$.

We then apply the algorithm of \Cref{lem:ultimatechipreplace} with $G_i$, $k$, $\delta$, $X$, $\chipfamily_i$, and the $(N(C),\delta)$-model-folios of $G_i[N[C]]$ for each $C \in \chipfamily_i$.
It runs in time $\Oh_{|X|,\delta}(\|G_i\|)$ and outputs a graph $G_{i+1}$ so that $X \subseteq V(G_{i+1})$ and the $(X,\delta)$-folios of $G_i$ and $G_{i+1}$ are equal.
Moreover, $G_{i+1}$ has the properties that
\begin{itemize}[nosep]
\item $|V(G_{i+1})| \le |V(G_i)| \le n$ and $\|G_{i+1}\| \le \|G_i\| \le \|G\|$; 
\item the number of $(X,k,\alpha)$-carvable vertices of $G_{i+1}$ is at most $\numcarv_{i+1} \le \numcarv_i - |\bigcup \chipfamily_i| + \alpha \cdot |\chipfamily_i|/2 \le \numcarv_i - |\bigcup \chipfamily_i|/2$; and
\item the size of a largest $(X,k,\alpha)$-chip of $G_{i+1}$ is at most the size of a largest $(X,k,\alpha)$-chip of $G_i$.
\end{itemize}
The last property implies that the balance of $G_{i+1}$ is at most the balance of $G_i$, which is at most the balance of $G$.
It also outputs an $(X,\delta)$-model-mapper $\Delta_{i+1}$ from $G_{i+1}$ to $G_i$.
We then repeat the process with $i$ increased by one.

Let $\ell$ be the index at which the process stops, i.e., with $\chipfamily_\ell = \emptyset$, implying $\numcarv_\ell = 0$.
Now $G_\ell$ is $(X,k,\alpha)$-compact.
Moreover, $X$ is well-linked in $G_\ell$, because it is well-linked in $G$ and the $(X,\delta)$-folios of $G$ and $G_\ell$ are equal.
Now we can call the algorithm of \Cref{thm:compact-solvable} to return either the $(X,\delta)$-model-folio of $G_\ell$ or that $G_\ell$ contains a $K_h$-minor.
Again, because the $(X,\delta)$-folios of $G$ and $G_\ell$ are equal and $\delta \ge h$, if $G_{\ell}$ contains $K_h$-minor we can return that $G$ contains $K_h$-minor.
The algorithm runs in time $\Oh_{|X|,\delta,h,\alpha}(\|G_\ell\|^{1+o(1)})$.
Finally, we use the $(X,\delta)$-model-mappers $\Delta_\ell,\ldots,\Delta_2$ in sequence to turn the $(X,\delta)$-model-folio of $G_\ell$ to an $(X,\delta)$-model-folio of $G$ in time $\Oh_{|X|,\delta}(\ell \cdot \|G\|)$, and then return it.

To analyze the running time, let us first bound the number of rounds $\ell$ in this process.
The facts that $\numcarv_{i+1} \le \numcarv_i - |\bigcup \chipfamily_i|/2$ and $|\bigcup \chipfamily_i| \ge \numcarv_i/\Oh_{k,\alpha}(\log |V(G_i)|)$ imply that $\numcarv_{i+1} \le \numcarv_i - \numcarv_i/\Oh_{k,\alpha}(\log |V(G_i)|)$.
Because $\numcarv_1 \le n$, $\numcarv_{\ell-1} \ge 1$, and $|V(G_i)| \le n$, it follows that $\ell \le \Oh_{k,\alpha}(\log^2 n)$.
Therefore, as $\|G_i\| \le \|G\|$, the total running time of \algcasecarve is bounded by the running times of the recursive calls, plus 
\begin{gather*}
\Oh_{k,\alpha}(\log^2 n) \cdot (\Oh_{k,\alpha}(\|G\|^{1+o(1)}) + \Oh_{k,\delta,\alpha}(\|G\| \cdot \log^3 n) + \Oh_{|X|,\delta}(\|G\|)) + \Oh_{|X|,\delta,h,\alpha}(\|G\|^{1+o(1)})\\
\le \Oh_{|X|,\delta,h,\alpha}(\|G\|^{1+o(1)}).
\end{gather*}

To ultimately bound the total running time of the recursive algorithm, we need to bound the sum of the sizes of the instances to which we recurse into.
This is easily bounded by $2 n$, but it turns out that the number $2^{\Oh_{k}(\log n)}$ is large, and this is why we need to replace the $2$ in the base by $1+1/\log n$.
This is the reason why we treat large chips and small chips differently.

Let $\chipfamily'_i = \{C \in \chipfamily_i \colon |C| > 2 \alpha \log n\}$ be the family of large chips at the $i$th iteration.
The sum of the sizes of the instances we recurse into is $\sum_{i=1}^{\ell-1} |\bigcup \chipfamily'_i|$.
Because $|\bigcup \chipfamily'_i| \ge |\chipfamily'_i| \cdot 2 \alpha \log n$, we have that 
\begin{align*}
\numcarv_{i+1} &\le \numcarv_i - \left|\bigcup \chipfamily'_i\right| \cdot (1-1/(2\log n))\\
\Rightarrow \quad \left|\bigcup \chipfamily'_i\right| &\le (\numcarv_i - \numcarv_{i+1}) \cdot \frac{1}{1-1/(2\log n)} \le (\numcarv_i - \numcarv_{i+1}) \cdot (1+1/\log n)
\end{align*}

As $\sum_{i=1}^{\ell-1} \numcarv_i - \numcarv_{i+1} \le n - |X|$, it follows that $\sum_{i=1}^{\ell-1} |\bigcup \chipfamily'_i| \le (n-|X|) \cdot (1+1/\log n)$.
In particular, the sum of the sizes of the instances we recurse into is at most $|(G,X)| \cdot (1+1/\log n)$.

\paragraph{Overall running time analysis.}
We have established that the running time of each individual recursive call, discounting the running times of its child calls, is $\Oh_{|X|,k,\delta,h,\alpha}(\|G\|^{1+o(1)}) = \Oh_{|X|,k,\delta}(|(G,X)|^{1+o(1)})$.

Let us then analyze the running time of the whole algorithm.
First, recall that in the initial input it holds that $|X| < k$, and observe that only the recursion of \algcasewlsmall can increase $|X|$, but it increases it to at most $4k-1$.
Therefore, $|X| < 4k$ in all recursive calls.
Now, if $N$ is the sum of the sizes of all instances throughout all recursive calls, then the total running time is $\Oh_{k,\delta}(N^{1+o(1)})$.
To finish the proof of \Cref{thm:cliqueFree}, we show that $N = \Oh_k(n \log n)$, where $n$ is the number of vertices of the initial input graph $G$.

The first step to show this is to bound the depth of the recursion tree.

\begin{lemma}
The depth of the recursion tree is at most $\Oh(k^3 \log n)$.
\end{lemma}
\begin{proof}
Let $t_1,\ldots,t_\ell$ be a root-leaf path in the recursion tree, with $t_1$ being the root and $t_\ell$ a leaf.
Let $\ins_i$ be the instance corresponding to the node $t_i$.
By our assumption that instances with at most $200k^3$ vertices are solved by brute-force, we have that $|V(G(\ins_i))| > 200k^3$ for all $i < \ell$.

If $t_i$ corresponds to \algcasenwl or \algcasecarve, then $|X(\ins_{i+1})| < |X(\ins_i)|$.
If $t_i$ corresponds to \algcasecc, then $|X(\ins_{i+1})| \le |X(\ins_i)|$ and $t_{i+1}$ does not correspond to \algcasecc because $G(\ins_{i+1}) - X(\ins_{i+1})$ is connected.
Therefore, because the maximum value of $|X(\ins_i)|$ is $4k-1$, any subpath of $8k$ consecutive nodes must contain a node corresponding to \algcasewlsmall.

Let $t_i$ be a node corresponding to \algcasewlsmall with $i<\ell$.
The balance of $\ins_i$ is $|\ins_i|$, and the balance of $\ins_{i+1}$ is at most $|V(G(\ins_i))| - |V(G(\ins_i))|/(100k^2)$, which by $|X(\ins_i)| < k$ and $|V(G(\ins_i))| > 200k^3$ is at most $|\ins_i|-|\ins_i|/(200k^2)$.
Because the balance of $\ins_{i+1}$ is never larger than the balance of $\ins_i$, and the balance of the initial input instance is at most $n$, this implies that the path contains at most $\Oh(k^3 \log n)$ nodes.
\end{proof}

To bound the sum of the sizes of all instances across the recursion, we recall that in \algcasecc, \algcasenwl, and \algcasewlsmall, the sum of the sizes of the instances to which we recurse into from the instance $(G,X)$ is at most $|(G,X)|$, and in \algcasecarve it is at most $|(G,X)| \cdot (1+1/\log n)$.
Therefore, as the depth of the recursion tree is $\Oh(k^3 \log n)$, the sum of the sizes of all instances is at most
\begin{align*}
N &\le n \cdot (1+1/\log n)^{\Oh(k^3 \log n)} \cdot \Oh(k^3 \log n)\\
&\le n \cdot 2^{\Oh(k^3)} \cdot \Oh(k^3 \log n)\\
&\le \Oh_k(n \log n).
\end{align*}

This finishes the proof of \Cref{thm:cliqueFree}.

\section{General case}
\label{sec:generalcase}
Finally, we are ready to approach the proof of our main result, \Cref{thm:main-real}.
For this, we first show that under a certain assumption about the compactness of $G$, the problem reduces to the \FolioCliqueEx problem solvable by the algorithm of \Cref{thm:cliqueFree}.
Then we give a recursive scheme similar to the one in \Cref{sec:omnomnomnomnom}, but much simpler, to reduce the general case to this compact case.



\subsection{Generic folios in compact graphs}
Recall that the $(X,\delta)$-folio of a graph $G$ is called generic if it contains all $X$-rooted graphs with detail at most~$\delta$.
Our reduction from general graphs to minor-free graphs is based on the fact that if a graph is compact and contains a large enough clique as a minor, then its $(X,\delta)$-folio is generic.
A lemma that implies this was proven by Robertson and Seymour in Graph Minors~XIII~\cite[(5.3)]{GM13}.
However, in our algorithm we need not only the knowledge that the $(X,\delta)$-folio of $G$ is generic, but also to algorithmically obtain an $(X,\delta)$-model-folio of $G$ in that case.
A direct adaptation of the proof of Robertson and Seymour into an algorithm would result in a running time of $\Oh_{|X|,\delta}(\|G\|^2)$, so next we give an alternative version of the proof of theirs, giving an $\Oh_{|X|,\delta}(\|G\|)$ time algorithm.
Our proof still mostly follows their ideas.

\begin{lemma}
\label{lem:RSgenericAlgo}
Let $\delta \in \N$, $G$ be a graph, $X \subseteq V(G)$, and $\model$ a minor model of $K_h$ in $G$, where $h \ge 3 |X| + \delta$.
Suppose further that for every $v \in V(K_h)$, there is no separation $(A,B)$ of order $<|X|$ such that $X \subseteq A$ and $\model(v) \cap A = \emptyset$.
Then,
\begin{itemize}[nosep]
\item the $(X,\delta)$-folio of $G$ is generic; and
\item given $\model$, an $(X,\delta)$-model-folio of $G$ can be computed in time $\Oh_{|X|,\delta}(\|G\|)$.
\end{itemize}
\end{lemma}
\begin{proof}
Denote $k = |X|$.
By ignoring the sets $\model(v)$ that intersect  $X$, assume that each $\model(v)$ is disjoint with $X$ and $h = 2k + \delta$.
Let $(H,\roots)$ be the $X$-rooted graph with $H$ isomorphic to $K_{k+\delta}$, $X \subseteq V(H)$, and $\roots(x) = \{x\}$ for all $x \in X$.
It suffices to show that we can find a model of $(H,\roots)$ in $G$, as this implies that the $(X,\delta)$-folio of $G$ is generic, and can be turned into an $(X,\delta)$-model-folio of $G$ in time $\Oh_{|X|,\delta}(\|G\|)$.

Let $G'$ be a graph with $X \subseteq V(G')$, and $\model'$ a minor model of $K_h$ in $G'$.
We say that $v \in V(K_h)$ is \emph{inseparable} in $(G',\model')$ if there is no separation $(A,B)$ of $G'$ of order $<k$ so that $X \subseteq A$ and $\model'(v) \cap A = \emptyset$.
Note that every $v \in V(K_h)$ is inseparable in $(G,\model)$.

Suppose that there exists $W \subseteq V(K_h)$, with $|W|=k+\delta$, so that after obtaining $(G',\model')$ from $(G,\model)$ by contracting each set $\model(v)$ with $v \in W$ into a single vertex $v$, every $v \in W$ is inseparable in $(G',\model')$.
Then, because $|W| \ge k$, there is no $(X,W)$-separator of size $<k$ in $G'$, implying that there is a set of $k$ vertex-disjoint $(X,W)$-paths in $G'$.
These paths and $W$ can be easily turned into a model of $(H,\roots)$ in $G'$, which in turn can be translated into a model of $(H,\roots)$ in $G$.

We call such set $W \subseteq V(K_h)$ a \emph{witnessing clique}.
Given a witnessing clique $W$, the method of the previous paragraph for finding a model of $(H,\roots)$ in $G$ can be implemented in time $\Oh_k(\|G\|)$ with the Ford-Fulkerson algorithm.
We will show that a witnessing clique exists.
After this existence proof, our algorithm works simply by trying all subsets $W$ of $V(K_h)$ of size $k+\delta$.

Let $G'$ be a graph with $X \subseteq V(G')$ and $\model'$ a model of $K_h$ in $G'$.
We say that the pair $(G',\model')$ is \emph{good} if there exists 
\begin{itemize}[nosep]
\item a separation $(A,B)$ of $G'$ of order $\le k$ with $X \subseteq A$; and
\item at least one vertex $v \in V(K_h)$ with $\model'(v) \cap A = \emptyset$,
\end{itemize}
so that every $v \in V(K_h)$ with $\model'(v) \cap A = \emptyset$ is inseparable in $(G',\model')$.
Furthermore, we say that such a separation $(A,B)$ \emph{witnesses} that $(G',\model')$ is good.
Note that the order of $(A,B)$ must be exactly $k$.

Observe that $(G,\model)$ is good, witnessed by the separation $(X,V(G))$.
Suppose $(G_1,\model_1)$ is obtained from $(G,\model)$ by contracting edges inside the branch sets of $\model$, and assume that $(G_1,\model_1)$ is good, but any further contraction inside a branch set would make it non-good.

Let $(A_1,B_1)$ be a separation that witnesses that $(G_1,\model_1)$ is good, and subject to that, maximizes $|A_1|$.
Note that because $\model_1$ is a model of $K_h$ and there exists $v \in V(K_h)$ with $\model_1(v) \cap A_1 = \emptyset$, every branch set of $\model_1$ must intersect $B_1$.
Therefore, there are at least $k+\delta$ branch sets of $\model_1$ that do not intersect $A_1$.

There are two cases.
First, if all branch sets of $\model_1$ that do not intersect $A_1$ have size $1$, then the vertices of $K_h$ corresponding to these branch sets form a witnessing clique, and we are done.
Second, if $\model_1$ has a branch set $\model_1(w)$ of size $|\model_1(w)| \ge 2$ that does not intersect $A_1$, then we show that contracting an edge in it produces a good pair, and therefore contradicts the choice of $(G_1,\model_1)$.

We start with the following observation.

\begin{claim}
\label{lem:RSgenericAlgo:claim1}
Let $v \in V(K_h)$ with $\model_1(v) \cap A_1 = \emptyset$.
There is no separation $(A',B')$ of $G_1$ so that $A_1 \subsetneq A'$, the order of $(A',B')$ is $\le k$, and $\model_1(v) \cap A' = \emptyset$.
\end{claim}
\begin{claimproof}
Such $(A',B')$ would witness that $(G_1,\model_1)$ is good, and thus by $A_1 \subsetneq A'$ it would contradict the choice of $(A_1,B_1)$.
\end{claimproof}

Let $w \in V(K_h)$ so that $\model_1(w)$ is disjoint with $A_1$ and $|\model_1(w)| \ge 2$.
Let $xy \in E(G_1[\model_1(w)])$ be an edge in the branch set, and $(G_2,\model_2)$ the pair resulting from contracting $xy$ into a vertex $z$.
We claim that $(G_2,\model_2)$ is good and this is witnessed by the separation $(A_2,B_2) = (A_1,(B_1 - \{x,y\}) \cup \{z\})$ of $G_2$.
The order of $(A_2,B_2)$ is $k$, and $\model_2(w)$ is disjoint with $A_2$.
It remains to show that every $v \in V(K_h)$ with $\model_2(v) \cap A_2 = \emptyset$ is inseparable in $(G_2,\model_2)$.

\begin{claim}
\label{lem:RSgenericAlgo:claim2}
Let $v \in V(K_h)$ with $\model_2(v) \cap A_2 = \emptyset$.
There is no separation $(A',B')$ of $G_2$ so that $A_2 \subsetneq A'$, the order of $(A',B')$ is $< k$, and $\model_2(v) \cap A' = \emptyset$.
\end{claim}
\begin{claimproof}
By uncontracting the contracted edge, such separation $(A',B')$ of $G_2$ could be mapped into a separation of $G_1$ that would contradict \Cref{lem:RSgenericAlgo:claim1}.
\end{claimproof}

By Menger's theorem, there exists a collection of $k$ vertex-disjoint $(X,A_1 \cap B_1)$-paths in $G_1$.
Because these paths are contained in $G_1[A_1]$, they exist also in $G_2$ and are contained in $G_2[A_2]$.
By \Cref{lem:RSgenericAlgo:claim2} and Menger's theorem, for every $v \in V(K_h)$ with $\model_2(v) \cap A_2 = \emptyset$, there exists a collection of $k$ vertex-disjoint $(A_2 \cap B_2, N_{G_2}(\model_2(v)))$-paths in $G_2$.
Note that $A_2 \cap B_2 = A_1 \cap B_1$ and these paths are contained in $G_2[B_2]$.
By combining these two collections of paths, we conclude that for every such $v$, there exists a collection of $k$ vertex-disjoint $(X,N_{G_2}(\model_2(v)))$-paths in $G_2$.
This shows that $v$ is inseparable in $(G_2,\model_2)$.
\end{proof}

Now the following lemma about generic folios in compact graphs is a simple consequence of \Cref{lem:RSgenericAlgo}.

\begin{lemma}\label{lem:compact-generic}
Let $\alpha,\delta\in \N$ and $G$ be a graph that is $(X,k,\alpha)$-compact, for some $X \subseteq V(G)$ with $|X| = k$.
If $G$ contains $K_h$ as a minor, where $h=\alpha(3k+\delta)$, then the $(X,\delta)$-folio of $G$ is generic.
Moreover, when given a minor model of $K_h$ in $G$, we can in time $\Oh_{k,\delta}(\|G\|)$ compute an $(X,\delta)$-model-folio of $G$.
\end{lemma}
\begin{proof}
By grouping the branch sets of the minor model of $K_h$ into $t = 3k+\delta$ equisized groups, we infer that $G$ contains a minor model $\model$ of $K_t$ such that $|\model(v)| \geq \alpha$ for each $v \in V(K_t)$.
Since $t = 3k + \delta$, it suffices to verify that the premise of \Cref{lem:RSgenericAlgo} is satisfied.
In particular, that there is no separation $(A,B)$ of order smaller than $|X|=k$ with $X \subseteq A$ and $\model(v) \cap A = \emptyset$, for some $v \in V(K_t)$.
But if such a separation existed, then the connected component $C$ of $G[B \setminus A]$ that contains $\model(v)$ would satisfy $|C| \ge |\model(v)| \ge \alpha$ and $|N(C)| \le |A \cap B| < k$, and thus be an $(X,k,\alpha)$-chip in $G$, contradicting that $G$ is $(X,k,\alpha)$-compact.
\end{proof}

\subsection{The algorithm}
Having \cref{lem:compact-generic} in hand, we can give the algorithm promised in \cref{thm:main-real}.
It will be based on a recursive scheme that is similar to the one presented in \cref{sec:omnomnomnomnom}, but considerably simpler.

\thmmainreal*
\begin{proof}
Our algorithm solves \Folio recursively, similarly to the algorithm of \Cref{sec:omnomnomnomnom}, but this time, there is only one case in the recursion, corresponding to \algcasecarve of \Cref{sec:omnomnomnomnom}.
We first describe and analyze the algorithm on a single recursive call, and then present the overall analysis.

Let $(G,X,\delta)$ be the given instance of \Folio.
Let $k = |X|$, and let $f$ be the function from \Cref{lem:ultimatechipreplace}, asserting that the algorithm of the lemma takes a family of $(X,k,f(k,\delta))$-chips as an input, and let us fix $\alpha = f(k,\delta)$.
We proceed similarly as in \algcasecarve of the algorithm of \Cref{sec:omnomnomnomnom}.

Let initially $G_1 = G$ and $i = 1$.
We repeat the following process, maintaining the invariants that the $(X,\delta)$-folio of $G_i$ is equal to that of $G$, $|V(G_i)| \le |V(G)|$, and $\|G_i\| \le \|G\|$.

We apply the algorithm of \Cref{lem:computechipfamilyfinal} with $G_i$, $X$, $k$, and $\alpha$.
In time $\Oh_{k,\alpha}(\|G_i\|^{1+o(1)})$ it returns a family $\chipfamily_i$ of pairwise non-touching $(X,k,\alpha)$-chips so that if $\numcarv_i$ is the number of $(X,k,\alpha)$-carvable vertices of $G_i$, then $|\bigcup \chipfamily_i| \ge \numcarv_i/\Oh_{k,\alpha}(\log |V(G_i)|)$.
If $\chipfamily_i = \emptyset$, then the process stops.

Otherwise, for each $C \in \chipfamily_i$, we use recursion to compute the $(N(C), \delta)$-model-folio of $G_i[N[C]]$.
Note that here, $|N(C)| < k$, so in the recursion the size of $X$ strictly decreases.
Then we apply the algorithm of \Cref{lem:ultimatechipreplace} with $G_i$, $k$, $\delta$, $X$, $\chipfamily_i$, and the $(N(C),\delta)$-model-folios of $G_i[N[C]]$ for each $C \in \chipfamily_i$.
It runs in time $\Oh_{k,\delta}(\|G_i\|)$ and outputs a graph $G_{i+1}$ so that $X \subseteq V(G_{i+1})$ and the $(X,\delta)$-folios of $G_i$ and $G_{i+1}$ are equal.
Moreover, $G_{i+1}$ has the properties that
\begin{itemize}[nosep]
\item $|V(G_{i+1})| \le |V(G_i)|$ and $\|G_{i+1}\| \le \|G_i\|$; and
\item the number of $(X,k,\alpha)$-carvable vertices of $G_{i+1}$ is at most $\numcarv_{i+1} \le \numcarv_i - |\bigcup \chipfamily_i|/2$.
\end{itemize}
It also outputs an $(X,\delta)$-model-mapper $\Delta_{i+1}$ from $G_{i+1}$ to $G_i$.
We then repeat the process with $i$ increased by one.

Let $\ell$ be the index at which the process stops, i.e., with $\chipfamily_\ell = \emptyset$.
The graph $G_\ell$ is $(X,k,\alpha)$-compact.
We call the algorithm of \Cref{thm:cliqueFree} with $(G_{\ell},X,\delta)$ and the integer $h = \alpha (3k+\delta)$.
It runs in time $\Oh_{k,\delta,h}(\|G_\ell\|^{1+o(1)}) = \Oh_{k,\delta}(\|G\|^{1+o(1)})$, and returns either an $(X,\delta)$-model-folio of $G_{\ell}$, or a model of a $K_h$-minor in $G_{\ell}$.
In the latter case, we apply the algorithm of \Cref{lem:compact-generic} to compute an $(X,\delta)$-model-folio of $G_{\ell}$.
In both cases, we obtained an $(X,\delta)$-model-folio of $G_{\ell}$ in time $\Oh_{k,\delta}(\|G\|^{1+o(1)})$.
We then use the $(X,\delta)$-model-mappers $\Delta_\ell,\ldots,\Delta_2$ in sequence to turn the $(X,\delta)$-model-folio of $G_\ell$ to an $(X,\delta)$-model-folio of $G$ in time $\Oh_{k,\delta}(\ell \cdot \|G\|)$, and then return it.

By the same argument as in \Cref{sec:omnomnomnomnom}, the number of iterations $\ell$ is $\ell = \Oh_{k,\alpha}(\log^2 |V(G)|) = \Oh_{k,\delta}(\log^2 |V(G)|)$.
Therefore, the total running time of one call in the recursion is $\Oh_{k,\delta}(\|G\|^{1+o(1)})$, plus the running times of the child calls.

To bound the overall running time, let us bound the sum of the sizes $\|G_i[N[C]]\|$ of the graphs we recurse into.
Because the chips in $\chipfamily_i$ are non-touching, we have that $\sum_{C \in \chipfamily_i} \|G_i[N[C]]\| \le k^2 \cdot \|G_i\| \le \Oh_k(\|G\|)$.
Therefore, across all $i \in [\ell]$, this sums up to at most $\Oh_{k,\delta}(\|G\| \log^2 |V(G)|)$.

Recall that we recurse into instances $(G_i[N[C]], N(C), \delta)$, where $C$ is a $(X,|X|,\alpha)$-chip, meaning that $|N(C)| < |X|$.
Therefore, the overall recursion-depth is at most $k = |X|$.
This implies that the sum of the sizes of all graphs throughout the recursive algorithm is
\begin{equation}
\label{finalproof:eq1}
(f(k,\delta) \cdot \log^2 |V(G)|)^{k+1} \cdot \|G\|,
\end{equation}
for a computable function $f$.
Let us assume $f(k,\delta) \ge 2^{2^{k+1}}$.
When $k+1 \le \log \log |V(G)|$, \eqref{finalproof:eq1} is bounded by
\[f(k,\delta)^{k+1} \cdot \log^{2 \log \log |V(G)|} |V(G)| \cdot \|G\| \le \Oh_{k,\delta}(\|G\|^{1+o(1)}),\]
and when $k+1 > \log \log |V(G)|$, \eqref{finalproof:eq1} is bounded by $f(k,\delta)^{2k+4}$,
implying that the overall running time is $\Oh_{k,\delta}(\|G\|^{1+o(1)})$.

We note that by using a similar technique as in \Cref{sec:omnomnomnomnom}, the bound of \eqref{finalproof:eq1} could be improved to $\Oh_{k}(\|G\|)$, but this is not necessary in this proof.
\end{proof}

\bibliographystyle{alpha}
\bibliography{references}

\appendix

\section{Finding clique minor models}\label{sec:finding}

In this section we prove \cref{lem:binsearch}, restated below for convenience.

\binsearch*
\begin{proof}
 We may assume that $h\geq 2$, otherwise both problems are trivial. Also, let $\delta'\coloneqq \max(\delta,h)$.

 Let $m\coloneqq |E(G)|$ and $e_1,e_2,\ldots,e_m$ be an arbitrary ordering of $E(G)$. For $i\in \{0,1,\ldots,m\}$, let
 $$G_i\coloneqq (V(G),\{e_1,\ldots,e_i\});$$
 that is, $G_i$ is obtained from $G$ by restricting the edge set to the first $i$ edges in the ordering.

 Let $\Aa$ be the assumed algorithm for \FolioCliqueEx, and let $\phi\colon \{0,1,\ldots,m\}\to \{\fM,\cM\}$ be the Boolean function defined as follows: if $\Aa$ applied to $G_i$, $X$, and parameters $\delta',h$ returns the $(X,\delta')$-folio of $G_i$, then we set $\phi(i)=\fM$, and otherwise, if $\Aa$ reports that $G_i$ contains $K_h$ as a minor, then we set $\phi(i)=\cM$. Note that $\phi(0)=\fM$, for $K_h$ is not a minor of an edgeless graph, and we can compute the value of $\phi(i)$ for a given index $i$ in time $f(|X|,\delta',h)\cdot \|G\|^{1+o(1)}$, by just running~$\Aa$.

 We first compute the value of $\phi(m)$. If $\phi(m)=\fM$, then in fact we have just computed the $(X,\delta')$-folio of $G_m=G$, which we can report as the outcome of the algorithm (after trimming it to the $(X,\delta)$-folio). Hence, from now on we assume that $\phi(m)=\cM$.

 The idea now is to apply binary search to find an index $a\in \{0,1,\ldots,m-1\}$ such that
 $$\phi(a)=\fM\qquad\textrm{and}\qquad \phi(a+1)=\cM.$$
 We do this as follows:
 \begin{itemize}[nosep]
  \item Initialize $a\coloneqq 0$ and $b\coloneqq m$. We will maintain the invariant that $\phi(a)=\fM$ and $\phi(b)=\cM$.
  \item As long as $b>a+1$, let $r=\lceil\frac{a+b}{2}\rceil$ and compute the value $\phi(r)$. If $\phi(r)=\fM$, then set $a\coloneqq r$, and otherwise set $b\coloneqq r$.
  \item Once $b=a+1$, $a$ satisfies the sought property due to the invariant.
 \end{itemize}
 Note that in every iteration, the difference $b-a$ gets halved rounded up, so within $\lceil \log m\rceil\leq \lceil \log \|G\|\rceil$ iterations we find an appropriate index $a$. Every iteration requires one call to $\Aa$ to compute $\phi(r)$, hence the total running time is $f(|X|,\delta',h)\cdot \|G\|^{1+o(1)}\cdot \lceil \log \|G\|\rceil\leq f(|X|,\delta',h)\cdot \|G\|^{1+o(1)}$. Note that the search does not require the function $\phi$ to be monotone, the values $\fM$ and $\cM$ may be interleaved and the method still works.

 Once $a$ is computed, we investigate the $(X,\delta')$-folio of $G_a$, computed by applying $\Aa$ to $G_a$, $X$, and parameters $\delta',h$. Since $\delta'\geq h$, from this folio we can deduce whether $G_a$ contains $K_h$ as a minor, and if this is the case, also obtain a witnessing minor model, which is also a minor model of $K_h$ in $G$ that can be provided on output. So assume now that $G_a$ does not contain $K_h$ as a minor. Since $G_{a+1}$ is obtained from $G_a$ by adding one edge, it follows that $G_{a+1}$ does not contain $K_{h+1}$ as a minor. Hence, if we apply $\Aa$ to $G_{a+1}$, $X$, and parameters $\delta',h+1$ --- which takes time $f(|X|,\delta',h+1)\cdot \|G\|^{1+o(1)}$ --- the output of $\Aa$ is necessarily the $(X,\delta')$-folio of $G_{a+1}$. But $\phi(a+1)=\cM$, so $G_{a+1}$ does contain $K_h$ as a minor. As $\delta'\geq h$, from the obtained $(X,\delta')$-folio of $G_{a+1}$ we can extract a minor model of $K_h$ in $G_{a+1}$, which is also a minor model of $K_h$ in $G$.
\end{proof}

\section{Dynamic maintenance of queries in graph of bounded treewidth}\label{app:dyn-tw}


We show~\Cref{prop:dynamic-tw}, which we restate here for convenience.

\dynamictw*

\begin{proof}
 As mentioned, the theorem has already been proved by Korhonen et al.~\cite{KorhonenMNPS23} for Boolean queries (i.e., when $\varphi$ is a sentence, without free variables, and the query only reports whether $G\models \varphi$) on graphs without vertex colors. The formulation above is more general in that it allows vertex colors and queries for providing an example witness. The proof is a simple lift of the arguments provided in~\cite{KorhonenMNPS23}\footnote{Throughout this appendix, all references to~\cite{KorhonenMNPS23} concern the full version of this article, available on arXiv (version v1).}, using the toolbox already prepared there.
 In particular, we will use the framework of {\em{boundaried graphs}} and of {\em{tree decomposition automata}}, described in~\cite[Appendix A.1]{KorhonenMNPS23}. We assume the reader's familiarity with these notions, and with the concepts of {\em{annotated tree decompositions}} and of {\em{prefix-rebuilding data structures}}, described in~\cite[Section~3]{KorhonenMNPS23}.

 We first describe how to deal with the colors. We note that it is actually straightforward to lift the whole framework presented in \cite{KorhonenMNPS23} to the setting of vertex-colored graphs supporting recoloring updates. Namely, executing a recoloring update does not require rebuilding the tree decomposition, so it can be performed by recomputing the run of the automaton along the path from (the forget bag of) the recolored vertex to the root, which always has length $\Oh_k(n^{o(1)})$. However, as we would like to use the statements from \cite{KorhonenMNPS23} in a black-box manner, it is perhaps more convenient to apply a simple reduction that encodes vertex colors through small gadgets, as follows.

 Fix an arbitrary bijection $\alpha\colon 2^\Sigma\to \{2,3,\ldots,2^{|\Sigma|}+1\}$. The idea is that for every vertex $u$, if the set of colors to which $u$ belongs is $\Sigma(u)\subseteq \Sigma$, then in the reduction we will encode this by adding $\alpha(\Sigma(u))$ pendants (degree-one vertices) adjacent to $u$. More precisely, if $G$ is the given $\Sigma$-vertex-colored graph, then we consider its (uncolored) supergraph $\wh{G}$ obtained by adding, for every vertex $u$, a set $P_u$ consisting of $2^{|\Sigma|}+1$ vertices, and making exactly $\alpha(\Sigma(u))$ of them adjacent to $u$; the remaining members of $P_u$ stay in $\wh{G}$ as isolated vertices. Then, the vertices of the original graph $G$ can be recognized as those of degree at least $2$, edges of $G$ can be recognized as those whose both endpoints are vertices of $G$, and whether a vertex $u$ of $G$ is of color $U\in \Sigma$ can be recognized by verifying that the number of neighbors of degree $1$ of $u$ belongs to $\alpha^{-1}(\{\Gamma\subseteq \Sigma \mid U\in \Gamma\})$; this is a disjunction over $2^{|\Sigma|-1}$ possible numbers. We can now syntactically modify the given formula $\varphi(x)$ by requiring that $x$ and all the variables quantified throughout the formula belong to (resp. are subsets of) $V(G)$ (resp. $E(G)$), and replacing all atomic formulas verifying colors with the checks described above. This way, we obtain a formula $\wh{\varphi}(x)$ such that for all $a\in V(\wh{G})$, we have
 $$\wh{G}\models \wh{\varphi}(a)\qquad\textrm{if and only if}\qquad a\in V(G) \textrm{ and } G\models\varphi(a).$$
 Thus, it suffices to set up the data structure for the modified query $\wh{\varphi}(x)$ and the graph $\wh{G}$, which is uncolored. Observe that $G$ has treewidth at most $k$ if and only if $\wh{G}$ has treewidth at most $k$ (assuming $k\geq 1$, which is without loss of generality), and every update modifying the color of a vertex of $G$ can be modelled by at most $2^{|\Sigma|}$ edge insertion/deletion updates in $\wh{G}$.

 Hence, from now on we may focus on the case when $G$ is a graph without any vertex colors, and similarly $\varphi(x)$ is a $\CMSO_2$ formula over uncolored graphs. For this case, we prove the following statement, which is a lift of \cite[Lemma~3.8]{KorhonenMNPS23} to our level of generality.

 \begin{claim}\label{cl:prefix-rebuild}
  Fix $\ell\in \N$ and a $\CMSO_2$ formula $\varphi(x)$, where $x$ is a single vertex variable. Then there exists an $\ell$-prefix-rebuilding data structure with overhead $\Oh_{\ell,\varphi}(1)$ that additionally implements the following operation:
  \begin{itemize}
   \item $\mathsf{query}()$: returns any vertex $a$ such that $G\models \varphi(a)$, or correctly reports that no such vertex exists. Runs in worst-case time $\Oh_{\ell,\varphi}(1)$.
  \end{itemize}
 \end{claim}

 Once \cref{cl:prefix-rebuild} is proved, one can combine it with \cite[Lemma~B.3]{KorhonenMNPS23} exactly as in \cite[Proof of Theorem~1.1]{KorhonenMNPS23}, and thus establish \cref{prop:dynamic-tw}. Therefore, we are left with proving \cref{cl:prefix-rebuild}.

 For this, we need the following automaton construction, which is a lift of \cite[Lemma A.2]{KorhonenMNPS23} to our level of generality.

 \newcommand{\edges}{\mathsf{edges}}

 \begin{claim}\label{cl:automaton}
 For every integer $\ell$ and $\CMSO_2$ formula $\varphi(x)$, where $x$ is a single vertex variable, there exists a tree decomposition automaton $\Aa_{\ell,\varphi}$ of width $\ell$ and a mapping $\pi$ from the states of $\Aa_{\ell,\varphi}$ to vertices or a marker $\bot$ with the following property. Suppose $G$ is a graph and  $(T,\bag,\edges)$ is an annotated binary tree decomposition of $G$ of width at most $\ell$. Suppose further that the run of $\Aa_{\ell,\varphi}$ on $(T,\bag,\edges)$ labels the root with a state $q$. Then $\pi(q)$ is a vertex $a\in V(G)$ such that $G\models \varphi(a)$, and in case no such vertex exists, $\pi(q)=\bot$. The evaluation time of $\Aa_{\ell,\varphi}$ and the time to evaluate $\pi$ on a given state is bounded by $f(\ell,\varphi)$, for some computable function $f$.
 \end{claim}

 \cite[Lemma A.2]{KorhonenMNPS23} is a statement that shows that the run of any tree decomposition automaton can be efficiently maintained under prefix-rebuilding updates. Thus, similarly as \cite[Lemma~3.8]{KorhonenMNPS23} follows immediately by combining \cite[Lemma A.2]{KorhonenMNPS23} with \cite[Lemma A.6]{KorhonenMNPS23}, our \cref{cl:prefix-rebuild} follows immediately by combining \cref{cl:automaton} with \cite[Lemma A.2]{KorhonenMNPS23}. Hence, we are left with proving \cref{cl:automaton}.

 The construction of \cref{cl:automaton} is a standard application of the framework of {\em{types}}, which underlies essentially all the proofs of Courcelle's Theorem and related results on $\CMSO_2$ in graphs of bounded treewidth. This construction has been presented in the proof of \cite[Lemma A.6]{KorhonenMNPS23} for Boolean queries, that is, $\varphi$ being a sentence rather than a formula with one free variable. Here we lift this reasoning to our setting, which poses no real difficulties besides that of notational nature. In particular, the following explanation will feature some phrasing from \cite{KorhonenMNPS23} taken {\bf{verbatim}}.

 \newcommand{\Frm}{\mathsf{Formulas}}
 \newcommand{\Tp}{\mathsf{Types}}
 \newcommand{\bnd}{\partial}
 \newcommand{\adh}{\mathsf{adh}}

 We first recall the standard terminology of types. Let the {\em{rank}} of a formula $\psi$ (possibly with free variables) be the least integer that is an upper bound on the nesting of the quantifiers in $\psi$ and all the moduli appearing in modular predicates present in $\psi$. It is well known that for every $p\in \N$ and a finite tuple of variables $\tup x$, there exists a finite set
 $\Frm^p_{\tup x}$, consisting of formulas of quantifier rank at most $p$ with free variables $\tup x$, such that every formula $\psi(\tup x)$ of rank at most $p$ is equivalent to some member $\psi'(\tup x)$ of $\Frm^p_{\tup x}$, in the sense that $\psi(\tup x)$ and $\psi'(\tup x)$ select exactly the same tuples of vertices in every graph. Moreover, $\Frm^p_{\tup x}$ can be computed given $p$ and $\tup x$ (in particular, $|\Frm^p_{\tup x}|\leq \Oh_{p,|\tup x|}(1)$), and given $\psi(\tup x)$, an equivalent $\psi'(\tup x)\in \Frm^p_{\tup x}$ can be computed. Then, for a graph $G$ and a tuple of vertices $\tup a\in V(G)^{\tup x}$, we define the {\em{rank-$p$ type}} of $(G,\tup a)$ as
 $$\tp^p(G,\tup a)\coloneqq \{\psi(\tup x)\in \Frm^p_{\tup x}\mid G\models \psi(\tup a)\}.$$
 In other words, the rank-$p$ type of $G$ and $\tup a$ comprises of all normalized formulas of rank at most $p$ that are satisfied by $\tup a$ in $G$. We also let
 $$\Tp^p_{\tup x}\coloneqq 2^{\Frm^p_{\tup x}}$$
 be the set of all possible rank-$p$ types of $\tup x$-tuples (some of them might  not be satisfiable).

 As in \cite{KorhonenMNPS23}, we assume that all the vertices of all the considered graphs belong to some totally ordered (say, by an ordering $\leq$) countable set $\Omega$ whose elements can be handled in constant time in the RAM model; the reader may assume $\Omega=\N$ with the standard order. Then given a boundaried graph $G$ with $|\bnd G|=k$ we define the {\em{state}} of $G$ as the mapping $\xi_G\colon \Tp^p_{x\tup y} \to \Omega\cup \{\bot\}$, where $x\tup y=(x,y_1,\ldots,y_k)$ is a tuple of $k+1$ variables, defined as follows. Let $\tup b\in (\bnd G)^{\tup y}$ be the enumeration of $\bnd G$ in the order $\leq$: $\tup b(y_1)$ is the first element of $\bnd G$ in $\leq$, $\tup b(y_2)$ the second, and so on. Then for a type $q\in \Tp^p_{x\tup y}$,
 \begin{itemize}[nosep]
  \item $\xi_G(q)$ is the $\leq$-smallest vertex $a\in V(G)$ such that $(G,a\tup b)$ has type $q$; or
  \item $\xi_G(q)=\bot$ if there is no vertex $a$ as above.
 \end{itemize}
 Now, we define the tree decomposition automaton $\Aa=\Aa_{\ell,\varphi}$ of width $\ell$ so that for any graph $G$ and its given tree decomposition $(T,\bag,\edges)$ of width at most $\ell$, the run $\rho_\Aa$ of $\Aa$ on $(T,\bag,\edges)$ is defined as follows: for every node $x$ of $T$,
 $$\rho_\Aa(x)=\xi(G_x).$$
 Indeed, to argue that there exists such an automaton, it suffices to show that for a node $x$ with children $y$ and $z$, $\xi(G_x)$ uniquely depends on $\xi(G_y)$, $\xi(G_z)$, and the situation in the bag of $x$, consisting of $\bag(x)$, $\adh(x)$, $\adh(y)$, $\adh(z)$, and $\edges(x)$; and similarly for nodes with one child. Just as in \cite{KorhonenMNPS23}, this follows from a standard compositionality argument involving Ehrenfeucht-Fra\"isse games. Moreover, this argument provides an algorithm to compute $\xi(G_x)$ from the provided information, hence $\Aa_{\ell,\varphi}$ has evaluation time bounded by a computable function of $\ell$ and $\varphi$.

 Finally, the mapping $\pi$ is defined as follows: for a state $\xi\colon \Tp^p_{x\tup y} \to \Omega\cup \{\bot\}$, $\pi(\xi)$ is the $\leq$-smallest vertex $a$ of $G$ such that $\xi(q)=a$ for some $q\ni \phi(x)$, and $\pi(\xi)=\bot$ if no such $a$ exists.
\end{proof}

\section{Compositionality of folios}
\label{sec:appendix_foliocompose}
We present the proof of \Cref{lem:folio_over_separator}, which we restate here.

\lemfoliooverseparator*
\begin{proof}
Let $(H_A,\roots_A)$ be an $X_A$-rooted minor of $G[A]$ of detail at most $\delta$ and $(H_B,\roots_B)$ an $X_B$-rooted minor of $G[B]$ of detail at most $\delta$.
We define their \emph{$X$-combination} $\combfolio((H_A,\roots_A), (H_B,\roots_B),X)$ as follows.
Let $H_{AB}$ be the graph with the vertex set $V(H_{AB}) = V(H_A) \cup V(H_B)$, and with $u \in V(H_A)$ and $v \in V(H_B)$ connected by an edge if $\roots_A(u)$ and $\roots_B(v)$ intersect.
Now $\combfolio((H_A,\roots_A), (H_B,\roots_B), X)$ is the $X$-rooted graph $(H,\roots)$ with the vertex set $V(H) = \cc(H_{AB})$, with $C_1 \in V(H)$ adjacent to $C_2 \in V(H)$ if there is an edge in $H_A$ or $H_B$ with one endpoint in $C_1$ and another in $C_2$, and with \[\roots(C) = \left(\bigcup_{v \in C \cap V(H_A)} \roots_A(v) \cap X\right) \cup \left(\bigcup_{v \in C \cap V(H_B)} \roots_B(v) \cap X\right)\]
for each $C \in V(H)$.
We observe that $(H,\roots)$ is an $X$-rooted minor of $G$, and therefore is in the $(X,\delta)$-folio of $G$ if its detail is at most $\delta$.
We define the function $\combfolio(\{\Ff_A, \Ff_B\}, X, \delta)$ of the statement as the collection of all $X$-rooted graphs with detail at most $\delta$ obtained as $\combfolio((H_A,\roots_A), (H_B,\roots_B),X)$ for $(H_A,\roots_A) \in \Ff_A$ and $(H_B,\roots_B) \in \Ff_B$.
Clearly, $\combfolio$ can be evaluated in time $\Oh_{|X|+|A \cap B|,\delta}(1)$.

Given a model $\model_A$ of $(H_A,\roots_A)$ in $G[A]$ and a model $\model_B$ of $(H_B,\roots_B)$ in $G[B]$, we can obtain a model $\model$ of $(H,\roots) = \combfolio((H_A,\roots_A), (H_B,\roots_B),X)$ in $G$ in time $\Oh_{|X|+|A \cap B|,\delta}(|V(G)|)$ by letting
\[\model(C) = \left(\bigcup_{v \in C \cap V(H_A)} \model_A(v)\right) \cup \left(\bigcup_{v \in C \cap V(H_B)} \model_B(v)\right)\]
for each $C \in V(H)$.
Therefore, the models of all $X$-rooted minors in $\combfolio(\{\Ff_A, \Ff_B\}, X, \delta)$ can be computed from the $(X_A,\delta)$-model-folio of $G[A]$ and the $(X_B,\delta)$-model-folio of $G[B]$ in time $\Oh_{|X|+|A \cap B|,\delta}(|V(G)|)$.

It remains to show that $\combfolio(\{\Ff_A, \Ff_B\}, X, \delta)$ indeed contains all $X$-rooted minors of $G$ with detail at most $\delta$.

Let $(H,\roots)$ be an $X$-rooted minor of $G$ with detail at most $\delta$ and $\model$ a minor model of it in $G$.
For each $v \in V(H)$, let $\mcomps^A_v = \cc(G[\model(v) \cap A])$ and $\mcomps^B_v = \cc(G[\model(v) \cap B])$.
Then let $\mcomps^A = \bigcup_{v \in V(H)} \mcomps^A_v$ and $\mcomps^B = \bigcup_{v \in V(H)} \mcomps^B_v$.
Note that $\mcomps^A$ and $\mcomps^B$ are collections of disjoint connected subsets of $A$ and $B$, respectively.

Now, let $H_A$ be the graph with the vertex set $V(H_A) = \mcomps^A$, and having an edge between $C_1, C_2 \in \mcomps^A$ if they are adjacent in $G[A]$ and there exists $uv \in E(H)$ so that $C_1 \subseteq \model(u)$ and $C_2 \subseteq \model(v)$.
Also, let $\roots_A \colon V(H_A) \to 2^{X_A}$ be the function defined as $\roots_A(C) = X_A \cap C$.
Now, $(H_A,\roots_A)$ is an $X_A$-rooted minor of $G[A]$.
The detail of $(H_A,\roots_A)$ is at most $\delta$, because if $C \in \mcomps^A$ is disjoint from $X_A$, then it must be that $C = \model(v)$ for some $v \in V(H)$ with $\roots(v) = \emptyset$.

The $X_B$-rooted graph $(H_B,\roots_B)$ is defined analogously, and similar observations apply to it.
We have that $(H_A,\roots_A)$ is in the $(X_A,\delta)$-folio of $G[A]$ and $(H_B,\roots_B)$ is in the $(X_B,\delta)$-folio of $G[B]$, and observe that $\combfolio((H_A,\roots_A), (H_B,\roots_B), X) = (H,\roots)$.
\end{proof}

\end{document}